\documentclass[a4paper,11pt,onecolumn]{quantumarticle}
\pdfoutput=1

\usepackage{graphicx}
\usepackage{dcolumn}
\usepackage{bm}

\usepackage{appendix}
\usepackage{tikzit}
\usepackage{tikz, circuitikz}
\tikzstyle{env}=[copoint,regular polygon rotate=0,minimum width=0.2cm, fill=black]

\tikzstyle{probs}=[shape=semicircle,fill=white,draw=black,shape border rotate=180,minimum width=1.2cm]

\tikzstyle{nudge}=[yshift=0.6mm]

\tikzstyle{every picture}=[baseline=-0.25em,scale=0.5]
\tikzstyle{dotpic}=[] 
\tikzstyle{diredges}=[every to/.style={diredge}]
\tikzstyle{math matrix}=[matrix of math nodes,left delimiter=(,right delimiter=),inner sep=2pt,column sep=1em,row sep=0.5em,nodes={inner sep=0pt},text height=1.5ex, text depth=0.25ex]

\tikzstyle{inline text}=[text height=1.5ex, text depth=0.25ex,yshift=0.5mm]
\tikzstyle{label}=[font=\footnotesize,text height=1.5ex, text depth=0.25ex]
\tikzstyle{left label}=[label,anchor=east,xshift=2mm]
\tikzstyle{right label}=[label,anchor=west,xshift=-2mm]

\tikzstyle{braceedge}=[decorate,decoration={brace,amplitude=2mm,raise=-1mm}]
\tikzstyle{small braceedge}=[decorate,decoration={brace,amplitude=1mm,raise=-1mm}]

\tikzstyle{doubled}=[line width=1.6pt] 
\tikzstyle{boldedge}=[doubled,shorten <=-0.17mm,shorten >=-0.17mm]
\tikzstyle{boldedgegray}=[doubled,gray,shorten <=-0.17mm,shorten >=-0.17mm]
\tikzstyle{singleedgegray}=[gray]

\tikzstyle{semidoubled}=[line width=1.4pt] 
\tikzstyle{semiboldedgegray}=[semidoubled,gray,shorten <=-0.17mm,shorten >=-0.17mm]

\tikzstyle{boxedge}=[semiboldedgegray]

\tikzstyle{boldedgedashed}=[very thick,dashed,shorten <=-0.17mm,shorten >=-0.17mm]
\tikzstyle{vboldedgedashed}=[doubled,dashed,shorten <=-0.17mm,shorten >=-0.17mm]
\tikzstyle{left hook arrow}=[left hook-latex]
\tikzstyle{right hook arrow}=[right hook-latex]
\tikzstyle{sembracket}=[line width=0.5pt,shorten <=-0.07mm,shorten >=-0.07mm]

\tikzstyle{causal edge}=[->,thick,gray]
\tikzstyle{causal nondir}=[thick,gray]
\tikzstyle{timeline}=[thick,gray, dashed]

\tikzstyle{cedge}=[<->,thick,gray!70!white]

\tikzstyle{empty diagram}=[draw=gray!40!white,dashed,shape=rectangle,minimum width=1cm,minimum height=1cm]
\tikzstyle{empty diagram small}=[draw=gray!50!white,dashed,shape=rectangle,minimum width=0.6cm,minimum height=0.5cm]

\tikzstyle{dot}=[inner sep=0mm,minimum width=2mm,minimum height=2mm,draw,shape=circle]  
\tikzstyle{Wsquare}=[white dot, shape=regular polygon, rounded corners=0.8 mm, minimum size=3.3 mm, regular polygon sides=3, outer sep=-0.2mm]
\tikzstyle{Wsquareadj}=[white dot, shape=regular polygon, rounded corners=0.8 mm, minimum size=3.3 mm, regular polygon sides=3, outer sep=-0.2mm, regular polygon rotate=180]
\tikzstyle{ddot}=[inner sep=0mm, doubled, minimum width=2.5mm,minimum height=2.5mm,draw,shape=circle]

\tikzstyle{black dot}=[dot,fill=black]
\tikzstyle{white dot}=[dot,fill=white,,text depth=-0.2mm]
\tikzstyle{white Wsquare}=[Wsquare,fill=white,,text depth=-0.2mm]
\tikzstyle{white Wsquareadj}=[Wsquareadj,fill=white,,text depth=-0.2mm]
\tikzstyle{green dot}=[white dot] 
\tikzstyle{gray dot}=[dot,fill=gray!40!white,,text depth=-0.2mm]
\tikzstyle{red dot}=[gray dot] 

\tikzstyle{black ddot}=[ddot,fill=black]
\tikzstyle{white ddot}=[ddot,fill=white]
\tikzstyle{gray ddot}=[ddot,fill=gray!40!white]

\tikzstyle{gray edge}=[gray!60!white]

\tikzstyle{small dot}=[inner sep=0.5mm,minimum width=0pt,minimum height=0pt,draw,shape=circle]

\tikzstyle{small black dot}=[small dot,fill=black]
\tikzstyle{small white dot}=[small dot,fill=white]
\tikzstyle{small gray dot}=[small dot,fill=gray!40!white]

\tikzstyle{very small dot}=[inner sep=0.3mm,minimum width=0pt,minimum height=0pt,draw,shape=circle]

\tikzstyle{very small black dot}=[very small dot,fill=black]
\tikzstyle{very small white dot}=[small dot,fill=white]
\tikzstyle{very small gray dot}=[small dot,fill=gray!40!white]

\tikzstyle{causal dot}=[inner sep=0.4mm,minimum width=0pt,minimum height=0pt,draw=white,shape=circle,fill=gray!40!white]

\tikzstyle{phase dimensions}=[minimum size=5mm,font=\footnotesize,rectangle,rounded corners=2.5mm,inner sep=0.2mm,outer sep=-2mm]
\tikzstyle{dphase dimensions}=[minimum size=5mm,font=\footnotesize,rectangle,rounded corners=2.5mm,inner sep=0.2mm,outer sep=-2mm]

\tikzstyle{white phase dot}=[dot,fill=white,phase dimensions]
\tikzstyle{white phase ddot}=[ddot,fill=white,dphase dimensions]

\tikzstyle{white rect ddot}=[draw=black,fill=white,doubled,minimum size=5mm,font=\footnotesize,rectangle,rounded corners=2.5mm,inner sep=0.2mm]
\tikzstyle{gray rect ddot}=[draw=black,fill=gray!40!white,doubled,minimum size=6mm,font=\footnotesize,rectangle,rounded corners=3mm]

\tikzstyle{gray phase dot}=[dot,fill=gray!40!white,phase dimensions]
\tikzstyle{gray phase ddot}=[ddot,fill=gray!40!white,dphase dimensions]
\tikzstyle{grey phase dot}=[gray phase dot]
\tikzstyle{grey phase ddot}=[gray phase ddot]

\tikzstyle{small phase dimensions}=[minimum size=4mm,font=\tiny,rectangle,rounded corners=2mm,inner sep=0.2mm,outer sep=-2mm]
\tikzstyle{small dphase dimensions}=[minimum size=4mm,font=\tiny,rectangle,rounded corners=2mm,inner sep=0.2mm,outer sep=-2mm]

\tikzstyle{small gray phase dot}=[dot,fill=gray!40!white,small phase dimensions]
\tikzstyle{small gray phase ddot}=[ddot,fill=gray!40!white,small dphase dimensions]

\tikzstyle{small map}=[draw,shape=rectangle,minimum height=4mm,minimum width=4mm,fill=white]

\tikzstyle{cnot}=[fill=white,shape=circle,inner sep=-1.4pt]

\tikzstyle{asym hadamard}=[fill=white,draw,shape=NEbox,inner sep=0.6mm,font=\footnotesize,minimum height=4mm]
\tikzstyle{asym hadamard conj}=[fill=white,draw,shape=NWbox,inner sep=0.6mm,font=\footnotesize,minimum height=4mm]
\tikzstyle{asym hadamard dag}=[fill=white,draw,shape=SEbox,inner sep=0.6mm,font=\footnotesize,minimum height=4mm]

\tikzstyle{hadamard}=[fill=white,draw,inner sep=0.6mm,font=\footnotesize,minimum height=4mm,minimum width=4mm]
\tikzstyle{small hadamard}=[fill=white,draw,inner sep=0.6mm,minimum height=1.5mm,minimum width=1.5mm]
\tikzstyle{small hadamard rotate}=[small hadamard,rotate=45]
\tikzstyle{dhadamard}=[hadamard,doubled]
\tikzstyle{small dhadamard}=[small hadamard,doubled]
\tikzstyle{small dhadamard rotate}=[small hadamard rotate,doubled]
\tikzstyle{antipode}=[white dot,inner sep=0.3mm,font=\footnotesize]

\tikzstyle{scalar}=[diamond,draw,inner sep=0.5pt,font=\small]
\tikzstyle{dscalar}=[diamond,doubled, draw,inner sep=0.5pt,font=\small]

\tikzstyle{small box}=[rectangle,inline text,fill=white,draw,minimum height=5mm,yshift=-0.5mm,minimum width=5mm,font=\small]
\tikzstyle{small gray box}=[small box,fill=gray!30]
\tikzstyle{medium box}=[rectangle,inline text,fill=white,draw,minimum height=5mm,yshift=-0.5mm,minimum width=8mm,font=\small]
\tikzstyle{square box}=[small box] 
\tikzstyle{medium gray box}=[small box,fill=gray!30]
\tikzstyle{semilarge box}=[rectangle,inline text,fill=white,draw,minimum height=5mm,yshift=-0.5mm,minimum width=12.5mm,font=\small]
\tikzstyle{large box}=[rectangle,inline text,fill=white,draw,minimum height=5mm,yshift=-0.5mm,minimum width=15mm,font=\small]
\tikzstyle{large gray box}=[small box,fill=gray!30]

\tikzstyle{Bayes box}=[rectangle,fill=black,draw, minimum height=3mm, minimum width=3mm]

\tikzstyle{gray square point}=[small box,fill=gray!50]

\tikzstyle{dphase box white}=[dhadamard]
\tikzstyle{dphase box gray}=[dhadamard,fill=gray!50!white]
\tikzstyle{phase box white}=[hadamard]
\tikzstyle{phase box gray}=[hadamard,fill=gray!50!white]

\tikzstyle{point}=[regular polygon,regular polygon sides=3,draw,scale=0.75,inner sep=-0.5pt,minimum width=9mm,fill=white,regular polygon rotate=180]
\tikzstyle{copoint}=[regular polygon,regular polygon sides=3,draw,scale=0.75,inner sep=-0.5pt,minimum width=9mm,fill=white]
\tikzstyle{dpoint}=[point,doubled]
\tikzstyle{dcopoint}=[copoint,doubled]

\tikzstyle{wide copoint}=[fill=white,draw,shape=isosceles triangle,shape border rotate=90,isosceles triangle stretches=true,inner sep=0pt,minimum width=1.5cm,minimum height=6.12mm]
\tikzstyle{wide point}=[fill=white,draw,shape=isosceles triangle,shape border rotate=-90,isosceles triangle stretches=true,inner sep=0pt,minimum width=1.5cm,minimum height=6.12mm,yshift=-0.0mm]
\tikzstyle{wide point plus}=[fill=white,draw,shape=isosceles triangle,shape border rotate=-90,isosceles triangle stretches=true,inner sep=0pt,minimum width=1.74cm,minimum height=7mm,yshift=-0.0mm]

\tikzstyle{wide dpoint}=[fill=white,doubled,draw,shape=isosceles triangle,shape border rotate=-90,isosceles triangle stretches=true,inner sep=0pt,minimum width=1.5cm,minimum height=6.12mm,yshift=-0.0mm]

\tikzstyle{tinypoint}=[regular polygon,regular polygon sides=3,draw,scale=0.55,inner sep=-0.15pt,minimum width=6mm,fill=white,regular polygon rotate=180] 

\tikzstyle{white point}=[point]
\tikzstyle{white dpoint}=[dpoint]
\tikzstyle{green point}=[white point] 
\tikzstyle{white copoint}=[copoint]
\tikzstyle{gray point}=[point,fill=gray!40!white]
\tikzstyle{gray dpoint}=[gray point,doubled]
\tikzstyle{red point}=[gray point] 
\tikzstyle{gray copoint}=[copoint,fill=gray!40!white]
\tikzstyle{gray dcopoint}=[gray copoint,doubled]

\tikzstyle{white point guide}=[regular polygon,regular polygon sides=3,font=\scriptsize,draw,scale=0.65,inner sep=-0.5pt,minimum width=9mm,fill=white,regular polygon rotate=180]

\tikzstyle{black point}=[point,fill=black,font=\color{white}]
\tikzstyle{black copoint}=[copoint,fill=black,font=\color{white}]

\tikzstyle{tiny gray point}=[tinypoint,fill=gray!40!white]

\tikzstyle{diredge}=[->]
\tikzstyle{ddiredge}=[<->]
\tikzstyle{rdiredge}=[<-]
\tikzstyle{thickdiredge}=[->, very thick]
\tikzstyle{pointer edge}=[->,very thick,gray]
\tikzstyle{pointer edge part}=[very thick,gray]
\tikzstyle{dashed edge}=[dashed]
\tikzstyle{thick dashed edge}=[very thick,dashed]
\tikzstyle{thick gray dashed edge}=[thick dashed edge,gray!40]
\tikzstyle{thick map edge}=[very thick,|->]

\makeatletter
\newcommand{\boxshape}[3]{%
\pgfdeclareshape{#1}{
\inheritsavedanchors[from=rectangle] 
\inheritanchorborder[from=rectangle]
\inheritanchor[from=rectangle]{center}
\inheritanchor[from=rectangle]{north}
\inheritanchor[from=rectangle]{south}
\inheritanchor[from=rectangle]{west}
\inheritanchor[from=rectangle]{east}
\backgroundpath{
\southwest \pgf@xa=\pgf@x \pgf@ya=\pgf@y
\northeast \pgf@xb=\pgf@x \pgf@yb=\pgf@y

\@tempdima=#2
\@tempdimb=#3

\pgfpathmoveto{\pgfpoint{\pgf@xa - 5pt + \@tempdima}{\pgf@ya}}
\pgfpathlineto{\pgfpoint{\pgf@xa - 5pt - \@tempdima}{\pgf@yb}}
\pgfpathlineto{\pgfpoint{\pgf@xb + 5pt + \@tempdimb}{\pgf@yb}}
\pgfpathlineto{\pgfpoint{\pgf@xb + 5pt - \@tempdimb}{\pgf@ya}}
\pgfpathlineto{\pgfpoint{\pgf@xa - 5pt + \@tempdima}{\pgf@ya}}
\pgfpathclose
}
}}

\boxshape{NEbox}{0pt}{0pt}
\boxshape{SEbox}{0pt}{-5pt}
\boxshape{NWbox}{5pt}{0pt}
\boxshape{SWbox}{-5pt}{0pt}
\boxshape{EBox}{-3pt}{3pt}
\boxshape{WBox}{3pt}{-3pt}
\makeatother

\tikzstyle{cloud}=[shape=cloud,draw,minimum width=1.5cm,minimum height=1.5cm]

\tikzstyle{map}=[draw,shape=NEbox,inner sep=2pt,minimum height=6mm,fill=white]
\tikzstyle{dashedmap}=[draw,dashed,gray,shape=NEbox,inner sep=2pt,minimum height=6mm,fill=white]
\tikzstyle{medium dashedmap}=[draw,dashed,gray,shape=NEbox,inner sep=2pt,minimum height=6mm,fill=white,minimum width=7mm]
\tikzstyle{semilarge dashedmap}=[draw,dashed,gray,shape=NEbox,inner sep=2pt,minimum height=6mm,fill=white,minimum width=9.5mm]
\tikzstyle{large dashedmap}=[draw,dashed,gray,shape=NEbox,inner sep=2pt,minimum height=6mm,fill=white,minimum width=12mm]
\tikzstyle{very large dashedmap}=[draw,dashed,gray,shape=NEbox,inner sep=2pt,minimum height=6mm,fill=white,minimum width=17mm]

\tikzstyle{dashed map}=[fill=white, draw=gray, shape=rectangle, style=map, dashed]

\tikzstyle{mapdag}=[draw,shape=SEbox,inner sep=2pt,minimum height=6mm,fill=white]
\tikzstyle{mapadj}=[draw,shape=SEbox,inner sep=2pt,minimum height=6mm,fill=white]
\tikzstyle{maptrans}=[draw,shape=SWbox,inner sep=2pt,minimum height=6mm,fill=white]
\tikzstyle{mapconj}=[draw,shape=NWbox,inner sep=2pt,minimum height=6mm,fill=white]

\tikzstyle{medium map}=[draw,shape=NEbox,inner sep=2pt,minimum height=6mm,fill=white,minimum width=7mm]
\tikzstyle{medium map dag}=[draw,shape=SEbox,inner sep=2pt,minimum height=6mm,fill=white,minimum width=7mm]
\tikzstyle{medium map adj}=[draw,shape=SEbox,inner sep=2pt,minimum height=6mm,fill=white,minimum width=7mm]
\tikzstyle{medium map trans}=[draw,shape=SWbox,inner sep=2pt,minimum height=6mm,fill=white,minimum width=7mm]
\tikzstyle{medium map conj}=[draw,shape=NWbox,inner sep=2pt,minimum height=6mm,fill=white,minimum width=7mm]
\tikzstyle{semilarge map}=[draw,shape=NEbox,inner sep=2pt,minimum height=6mm,fill=white,minimum width=9.5mm]
\tikzstyle{semilarge map trans}=[draw,shape=SWbox,inner sep=2pt,minimum height=6mm,fill=white,minimum width=9.5mm]
\tikzstyle{semilarge map adj}=[draw,shape=SEbox,inner sep=2pt,minimum height=6mm,fill=white,minimum width=9.5mm]
\tikzstyle{semilarge map dag}=[draw,shape=SEbox,inner sep=2pt,minimum height=6mm,fill=white,minimum width=9.5mm]
\tikzstyle{semilarge map conj}=[draw,shape=NWbox,inner sep=2pt,minimum height=6mm,fill=white,minimum width=9.5mm]
\tikzstyle{large map}=[draw,shape=NEbox,inner sep=2pt,minimum height=6mm,fill=white,minimum width=12mm]
\tikzstyle{large map conj}=[draw,shape=NWbox,inner sep=2pt,minimum height=6mm,fill=white,minimum width=12mm]
\tikzstyle{very large map}=[draw,shape=NEbox,inner sep=2pt,minimum height=6mm,fill=white,minimum width=17mm]
\tikzstyle{very very large map}=[draw,shape=NEbox,inner sep=2pt,minimum height=6mm,fill=white,minimum width=50mm]
\tikzstyle{large map dag}=[draw,shape=SEbox,inner sep=2pt,minimum height=6mm,fill=white,minimum width=12mm]

\tikzstyle{medium dmap}=[draw,doubled,shape=NEbox,inner sep=2pt,minimum height=6mm,fill=white,minimum width=7mm]
\tikzstyle{medium dmap dag}=[draw,doubled,shape=SEbox,inner sep=2pt,minimum height=6mm,fill=white,minimum width=7mm]
\tikzstyle{medium dmap adj}=[draw,doubled,shape=SEbox,inner sep=2pt,minimum height=6mm,fill=white,minimum width=7mm]
\tikzstyle{medium dmap trans}=[draw,doubled,shape=SWbox,inner sep=2pt,minimum height=6mm,fill=white,minimum width=7mm]
\tikzstyle{medium dmap conj}=[draw,doubled,shape=NWbox,inner sep=2pt,minimum height=6mm,fill=white,minimum width=7mm]
\tikzstyle{semilarge dmap}=[draw,doubled,shape=NEbox,inner sep=2pt,minimum height=6mm,fill=white,minimum width=9.5mm]
\tikzstyle{semilarge dmap trans}=[draw,doubled,shape=SWbox,inner sep=2pt,minimum height=6mm,fill=white,minimum width=9.5mm]
\tikzstyle{semilarge dmap adj}=[draw,doubled,shape=SEbox,inner sep=2pt,minimum height=6mm,fill=white,minimum width=9.5mm]
\tikzstyle{semilarge dmap dag}=[draw,doubled,shape=SEbox,inner sep=2pt,minimum height=6mm,fill=white,minimum width=9.5mm]
\tikzstyle{semilarge dmap conj}=[draw,doubled,shape=NWbox,inner sep=2pt,minimum height=6mm,fill=white,minimum width=9.5mm]
\tikzstyle{large dmap}=[draw,doubled,shape=NEbox,inner sep=2pt,minimum height=6mm,fill=white,minimum width=12mm]
\tikzstyle{large dmap conj}=[draw,doubled,shape=NWbox,inner sep=2pt,minimum height=6mm,fill=white,minimum width=12mm]
\tikzstyle{large dmap trans}=[draw,doubled,shape=SWbox,inner sep=2pt,minimum height=6mm,fill=white,minimum width=12mm]
\tikzstyle{large dmap adj}=[draw,doubled,shape=SEbox,inner sep=2pt,minimum height=6mm,fill=white,minimum width=12mm]
\tikzstyle{large dmap dag}=[draw,doubled,shape=SEbox,inner sep=2pt,minimum height=6mm,fill=white,minimum width=12mm]
\tikzstyle{very large dmap}=[draw,doubled,shape=NEbox,inner sep=2pt,minimum height=6mm,fill=white,minimum width=19.5mm]

\tikzstyle{muxbox}=[draw,shape=rectangle,minimum height=3mm,minimum width=3mm,fill=white]
\tikzstyle{dmuxbox}=[muxbox,doubled]

\tikzstyle{box}=[draw,shape=rectangle,inner sep=2pt,minimum height=6mm,minimum width=6mm,fill=white]
\tikzstyle{dbox}=[draw,doubled,shape=rectangle,inner sep=2pt,minimum height=6mm,minimum width=6mm,fill=white]
\tikzstyle{dmap}=[draw,doubled,shape=NEbox,inner sep=2pt,minimum height=6mm,fill=white]
\tikzstyle{dmapdag}=[draw,doubled,shape=SEbox,inner sep=2pt,minimum height=6mm,fill=white]
\tikzstyle{dmapadj}=[draw,doubled,shape=SEbox,inner sep=2pt,minimum height=6mm,fill=white]
\tikzstyle{dmaptrans}=[draw,doubled,shape=SWbox,inner sep=2pt,minimum height=6mm,fill=white]
\tikzstyle{dmapconj}=[draw,doubled,shape=NWbox,inner sep=2pt,minimum height=6mm,fill=white]

\tikzstyle{ddmap}=[draw,doubled,dashed,shape=NEbox,inner sep=2pt,minimum height=6mm,fill=white]
\tikzstyle{ddmapdag}=[draw,doubled,dashed,shape=SEbox,inner sep=2pt,minimum height=6mm,fill=white]
\tikzstyle{ddmapadj}=[draw,doubled,dashed,shape=SEbox,inner sep=2pt,minimum height=6mm,fill=white]
\tikzstyle{ddmaptrans}=[draw,doubled,dashed,shape=SWbox,inner sep=2pt,minimum height=6mm,fill=white]
\tikzstyle{ddmapconj}=[draw,doubled,dashed,shape=NWbox,inner sep=2pt,minimum height=6mm,fill=white]

\boxshape{sNEbox}{0pt}{3pt}
\boxshape{sSEbox}{0pt}{-3pt}
\boxshape{sNWbox}{3pt}{0pt}
\boxshape{sSWbox}{-3pt}{0pt}
\tikzstyle{smap}=[draw,shape=sNEbox,fill=white]
\tikzstyle{smapdag}=[draw,shape=sSEbox,fill=white]
\tikzstyle{smapadj}=[draw,shape=sSEbox,fill=white]
\tikzstyle{smaptrans}=[draw,shape=sSWbox,fill=white]
\tikzstyle{smapconj}=[draw,shape=sNWbox,fill=white]

\tikzstyle{dsmap}=[draw,dashed,shape=sNEbox,fill=white]
\tikzstyle{dsmapdag}=[draw,dashed,shape=sSEbox,fill=white]
\tikzstyle{dsmaptrans}=[draw,dashed,shape=sSWbox,fill=white]
\tikzstyle{dsmapconj}=[draw,dashed,shape=sNWbox,fill=white]

\boxshape{mNEbox}{0pt}{10pt}
\boxshape{mSEbox}{0pt}{-10pt}
\boxshape{mNWbox}{10pt}{0pt}
\boxshape{mSWbox}{-10pt}{0pt}
\tikzstyle{mmap}=[draw,shape=mNEbox]
\tikzstyle{mmapdag}=[draw,shape=mSEbox]
\tikzstyle{mmaptrans}=[draw,shape=mSWbox]
\tikzstyle{mmapconj}=[draw,shape=mNWbox]

\tikzstyle{mmapgray}=[draw,fill=gray!40!white,shape=mNEbox]
\tikzstyle{smapgray}=[draw,fill=gray!40!white,shape=sNEbox]

\makeatletter

\pgfdeclareshape{cornerpoint}{
\inheritsavedanchors[from=rectangle] 
\inheritanchorborder[from=rectangle]
\inheritanchor[from=rectangle]{center}
\inheritanchor[from=rectangle]{north}
\inheritanchor[from=rectangle]{south}
\inheritanchor[from=rectangle]{west}
\inheritanchor[from=rectangle]{east}
\backgroundpath{
\southwest \pgf@xa=\pgf@x \pgf@ya=\pgf@y
\northeast \pgf@xb=\pgf@x \pgf@yb=\pgf@y

\pgfmathsetmacro{\pgf@shorten@left}{\pgfkeysvalueof{/tikz/shorten left}}
\pgfmathsetmacro{\pgf@shorten@right}{\pgfkeysvalueof{/tikz/shorten right}}

\pgfpathmoveto{\pgfpoint{0.5 * (\pgf@xa + \pgf@xb)}{\pgf@ya - 5pt}}
\pgfpathlineto{\pgfpoint{\pgf@xa - 8pt + \pgf@shorten@left}{\pgf@yb - 1.5 * \pgf@shorten@left}}
\pgfpathlineto{\pgfpoint{\pgf@xa - 8pt + \pgf@shorten@left}{\pgf@yb}}
\pgfpathlineto{\pgfpoint{\pgf@xb + 8pt - \pgf@shorten@right}{\pgf@yb}}
\pgfpathlineto{\pgfpoint{\pgf@xb + 8pt - \pgf@shorten@right}{\pgf@yb - 1.5 * \pgf@shorten@right}}
\pgfpathclose
}
}

\pgfdeclareshape{cornercopoint}{
\inheritsavedanchors[from=rectangle] 
\inheritanchorborder[from=rectangle]
\inheritanchor[from=rectangle]{center}
\inheritanchor[from=rectangle]{north}
\inheritanchor[from=rectangle]{south}
\inheritanchor[from=rectangle]{west}
\inheritanchor[from=rectangle]{east}
\backgroundpath{
\southwest \pgf@xa=\pgf@x \pgf@ya=\pgf@y
\northeast \pgf@xb=\pgf@x \pgf@yb=\pgf@y

\pgfmathsetmacro{\pgf@shorten@left}{\pgfkeysvalueof{/tikz/shorten left}}
\pgfmathsetmacro{\pgf@shorten@right}{\pgfkeysvalueof{/tikz/shorten right}}

\pgfpathmoveto{\pgfpoint{0.5 * (\pgf@xa + \pgf@xb)}{\pgf@yb + 5pt}}
\pgfpathlineto{\pgfpoint{\pgf@xa - 8pt + \pgf@shorten@left}{\pgf@ya + 1.5 * \pgf@shorten@left}}
\pgfpathlineto{\pgfpoint{\pgf@xa - 8pt + \pgf@shorten@left}{\pgf@ya}}
\pgfpathlineto{\pgfpoint{\pgf@xb + 8pt - \pgf@shorten@right}{\pgf@ya}}
\pgfpathlineto{\pgfpoint{\pgf@xb + 8pt - \pgf@shorten@right}{\pgf@ya + 1.5 * \pgf@shorten@right}}
\pgfpathclose
}
}

\makeatother

\pgfkeyssetvalue{/tikz/shorten left}{0pt}
\pgfkeyssetvalue{/tikz/shorten right}{0pt}

\tikzstyle{kpoint common}=[draw,fill=white,inner sep=1pt,minimum height=4mm]
\tikzstyle{kpoint sc}=[shape=cornerpoint,kpoint common]
\tikzstyle{kpoint adjoint sc}=[shape=cornercopoint,kpoint common]
\tikzstyle{kpoint}=[shape=cornerpoint,shorten left=5pt,kpoint common]
\tikzstyle{kpoint adjoint}=[shape=cornercopoint,shorten left=5pt,kpoint common]
\tikzstyle{kpoint conjugate}=[shape=cornerpoint,shorten right=5pt,kpoint common]
\tikzstyle{kpoint transpose}=[shape=cornercopoint,shorten right=5pt,kpoint common]
\tikzstyle{kpoint symm}=[shape=cornerpoint,shorten left=5pt,shorten right=5pt,kpoint common]

\tikzstyle{black kpoint}=[shape=cornerpoint,shorten left=5pt,kpoint common,fill=black,font=\color{white}]
\tikzstyle{black kpoint adjoint}=[shape=cornercopoint,shorten left=5pt,kpoint common,fill=black,font=\color{white}]
\tikzstyle{black kpointadj}=[shape=cornercopoint,shorten left=5pt,kpoint common,fill=black,font=\color{white}]

\tikzstyle{black dkpoint}=[shape=cornerpoint,shorten left=5pt,kpoint common,fill=black, doubled,font=\color{white}]
\tikzstyle{black dkpoint adjoint}=[shape=cornercopoint,shorten left=5pt,kpoint common,fill=black, doubled,font=\color{white}]
\tikzstyle{black dkpointadj}=[shape=cornercopoint,shorten left=5pt,kpoint common,fill=black, doubled,font=\color{white}] 

\tikzstyle{kpointdag}=[kpoint adjoint]
\tikzstyle{kpointadj}=[kpoint adjoint]
\tikzstyle{kpointconj}=[kpoint conjugate]
\tikzstyle{kpointtrans}=[kpoint transpose]

\tikzstyle{big kpoint}=[kpoint, minimum width=1.2 cm, minimum height=8mm, inner sep=4pt, text depth=3mm]

\tikzstyle{wide kpoint}=[kpoint, minimum width=1 cm, inner sep=2pt]
\tikzstyle{wide kpointdag}=[kpointdag, minimum width=1 cm, inner sep=2pt]
\tikzstyle{wide kpointconj}=[kpointconj, minimum width=1 cm, inner sep=2pt]
\tikzstyle{wide kpointtrans}=[kpointtrans, minimum width=1 cm, inner sep=2pt]

\tikzstyle{gray kpoint}=[kpoint,fill=gray!50!white]
\tikzstyle{gray kpointdag}=[kpointdag,fill=gray!50!white]
\tikzstyle{gray kpointadj}=[kpointadj,fill=gray!50!white]
\tikzstyle{gray kpointconj}=[kpointconj,fill=gray!50!white]
\tikzstyle{gray kpointtrans}=[kpointtrans,fill=gray!50!white]

\tikzstyle{gray dkpoint}=[kpoint,fill=gray!50!white,doubled]
\tikzstyle{gray dkpointdag}=[kpointdag,fill=gray!50!white,doubled]
\tikzstyle{gray dkpointadj}=[kpointadj,fill=gray!50!white,doubled]
\tikzstyle{gray dkpointconj}=[kpointconj,fill=gray!50!white,doubled]
\tikzstyle{gray dkpointtrans}=[kpointtrans,fill=gray!50!white,doubled]

\tikzstyle{white label}=[draw,fill=white,rectangle,inner sep=0.7 mm]
\tikzstyle{gray label}=[draw,fill=gray!50!white,rectangle,inner sep=0.7 mm]
\tikzstyle{black label}=[draw,fill=black,rectangle,inner sep=0.7 mm]

\tikzstyle{dkpoint}=[kpoint,doubled]
\tikzstyle{wide dkpoint}=[wide kpoint,doubled]
\tikzstyle{dkpointdag}=[kpoint adjoint,doubled]
\tikzstyle{wide dkpointdag}=[wide kpointdag,doubled]
\tikzstyle{dkcopoint}=[kpoint adjoint,doubled]
\tikzstyle{dkpointadj}=[kpoint adjoint,doubled]
\tikzstyle{dkpointconj}=[kpoint conjugate,doubled]
\tikzstyle{dkpointtrans}=[kpoint transpose,doubled]

\tikzstyle{kscalar}=[kpoint common, shape=EBox, inner xsep=-1pt, inner ysep=3pt,font=\small]
\tikzstyle{kscalarconj}=[kpoint common, shape=WBox, inner xsep=-1pt, inner ysep=3pt,font=\small]

\tikzstyle{spekpoint}=[kpoint sc,minimum height=5mm,inner sep=3pt]
\tikzstyle{spekcopoint}=[kpoint adjoint sc,minimum height=5mm,inner sep=3pt]

\tikzstyle{dspekpoint}=[spekpoint,doubled]
\tikzstyle{dspekcopoint}=[spekcopoint,doubled]

 \tikzstyle{discard}=[circuit ee IEC, ground,rotate=180,scale=1.5,inner sep=-2mm]
 \tikzstyle{downground}=[circuit ee IEC,thick,ground,rotate=-90,scale=1.5,inner sep=-2mm]

\tikzstyle{maxmix}=[regular polygon,regular polygon sides=3,draw=black,xscale=0.4,yscale=0.3,inner sep=-0.5pt,minimum width=10mm,fill=gray,regular polygon rotate=180]

 \tikzstyle{bigground}=[regular polygon,regular polygon sides=3,draw=gray,scale=0.50,inner sep=-0.5pt,minimum width=10mm,fill=gray]

\tikzstyle{arrs}=[-latex,font=\small,auto]
\tikzstyle{arrow plain}=[arrs]
\tikzstyle{arrow dashed}=[dashed,arrs]
\tikzstyle{arrow bold}=[very thick,arrs]
\tikzstyle{arrow hide}=[draw=white!0,-]
\tikzstyle{arrow reverse}=[latex-]
\tikzstyle{cdnode}=[]

\tikzstyle{green dashed arrow}=[green, arrow dashed]
\tikzstyle{dashed blue}=[blue, dashed]
\tikzstyle{red dashed arrow}=[red, arrow dashed]
\tikzstyle{orange arrow}=[orange, arrs]
\tikzstyle{blue arrow}=[blue, arrs]
\tikzstyle{magenta arrow}=[magenta, arrs]

\tikzstyle{dotted line}=[-, style=dotted, tikzit draw=brown]
\tikzstyle{dashed line}=[-, style=dashed, tikzit draw=cyan]
\tikzstyle{green fill line}=[-, fill={green!90!black}, tikzit draw=green]
\tikzstyle{blue fill}=[-, fill=blue, tikzit fill=blue, tikzit draw={rgb,255: red,102; green,117; blue,255}]
\tikzstyle{red}=[-, draw=red, tikzit draw=red]
\tikzstyle{blue}=[-, draw=blue, tikzit draw=blue]
\tikzstyle{thick black}=[-, draw=black, tikzit draw=black, line width=1pt]
\tikzstyle{dotted red}=[-, draw=red, style=dotted, tikzit draw=red]
\tikzstyle{dotted blue}=[-, draw=blue, tikzit draw=blue, style=dotted]
\tikzstyle{dashed thick blue}=[-, draw={rgb,255: red,28; green,176; blue,255}, tikzit draw={rgb,255: red,83; green,19; blue,156}, line width=1pt, style=dashed]
\tikzstyle{dashed thick red}=[-, draw=red, tikzit draw={rgb,255: red,255; green,100; blue,10}, line width=1pt, style=dashed]
\tikzstyle{green}=[-, draw=green, tikzit draw=green]
\tikzstyle{dotted green}=[-, draw=green, tikzit draw=green, style=dotted]
\tikzstyle{arrow}=[->]
\tikzstyle{arrow green dashed}=[draw=green, ->, tikzit draw=green, style=dashed]
\tikzstyle{arrow dashed red}=[draw=red, ->, style=dashed, tikzit draw=red]
\tikzstyle{dashed green}=[-, tikzit draw=green, draw=green, style=dashed]

\usepackage{physics}
\usepackage{amssymb}
\usepackage{amsmath}
\usepackage{amsfonts}
\usepackage{braket}
\usepackage{mathtools}
\usepackage{amsthm}
\usepackage{ulem}
\usepackage{bbold}
\usepackage{subcaption}
\usepackage{hyperref}
\usepackage{stmaryrd}
\usepackage{xcolor}
\usepackage[symbol]{footmisc}

\def\be{\begin{equation}}
\def\ee{\end{equation}}
\def\ba{\begin{align}}
\def\ea{\end{align}}
\newcommand{\cat}[1]{\ensuremath{\mathbf{#1}}}

\newcommand{\Rel}{\cat{Rel}}

\newcommand{\id}[1][]{\ensuremath{1_{#1}}}

\DeclareMathOperator{\Lin}{Lin}

\newtheorem{theorem}{Theorem}[section]

\newtheorem{corollary}{Corollary}[section]
\newtheorem{definition}{Definition}[section]
\newtheorem{lemma}{Lemma}[section]

\newtheorem{hyp}{Induction Hypothesis}
\newtheorem{principle}{Principle}
\newtheorem{conjecture}{Conjecture}

\DeclareTextFontCommand{\texttt}{\ttfamily\upshape}
\DeclareTextFontCommand{\textrm}{\rmfamily\upshape}

\newcommand{\changes}[1]{#1}

\renewcommand{\id}{\mathbb{1}}

\newcommand{\cc}{\mathcal C}

\newcommand{\cf}{\mathcal F}

\newcommand{\ch}{\mathcal H}
\newcommand{\ci}{\mathcal I}

\newcommand{\cl}{\mathcal L}
\newcommand{\cm}{\mathcal M}
\newcommand{\cn}{\mathcal N}

\renewcommand{\cp}{\mathcal P}

\newcommand{\cs}{\mathcal S}

\newcommand{\cu}{\mathcal U}
\newcommand{\cv}{\mathcal V}

\newcommand{\al}{\alpha}
\newcommand{\bet}{\beta}
\newcommand{\ga}{\gamma}
\newcommand{\Ga}{\Gamma}
\newcommand{\Gatop}{{\Gamma^\top}}
\newcommand{\la}{\lambda}
\newcommand{\laN}{{\la_N}}
\newcommand{\laNaug}{\la_N^\aug}
\newcommand{\laNsec}{\la_N^\sec}

\newcommand{\GalaN}{{\left( \Ga, {(\laN)}_N \right)}}
\newcommand{\GaBran}{\Ga^\Bran}
\newcommand{\La}{\Lambda}
\newcommand{\ze}{\zeta}
\newcommand{\Th}{\Theta}

\newcommand{\veck}{{\Vec{k}}}
\newcommand{\vecl}{{\Vec{l}}}
\newcommand{\muNveck}{{\mu_{N} \left(\veck \right)}}
\newcommand{\muNvecl}{{\mu_{N} \left(\vecl \right)}}

\newcommand{\kA}{{k_A}}

\newcommand{\inn}{\textrm{\upshape in}}
\newcommand{\out}{\textrm{\upshape out}}
\renewcommand{\int}{\textrm{\upshape int}}
\newcommand{\pad}{\textrm{\upshape pad}}
\newcommand{\inv}{^{-1}}
\newcommand{\aug}{\textrm{\upshape aug}}
\renewcommand{\sec}{\textrm{\upshape sec}}
\newcommand{\anc}{\textrm{\upshape anc}}
\newcommand{\str}{\textrm{\upshape str}}

\newcommand{\cpst}{{\cp^\str}}
\newcommand{\cfst}{{\cf^\str}}
\newcommand{\cpi}{{\cp_i}}
\newcommand{\cfi}{{\cf_i}}

\newcommand{\cpist}{{\cp_i^\str}}
\newcommand{\cfist}{{\cf_i^\str}}

\newcommand{\Nodes}{\texttt{Nodes}_\Ga}
\newcommand{\Ind}{\texttt{Ind}}
\newcommand{\Indin}{\Ind^\inn}
\newcommand{\Indout}{\Ind^\out}
\newcommand{\IndinNal}{\Indin_{\Nal}}
\newcommand{\IndoutNal}{\Indout_{\Nal}}
\newcommand{\IndinMbe}{\Indin_{\Mbe}}

\newcommand{\Bran}{\texttt{Bran}}
\newcommand{\BranGa}{\Bran_\Ga}
\newcommand{\Arr}{\texttt{Arr}_\Ga}
\newcommand{\head}{\texttt{head}}
\newcommand{\tail}{\texttt{tail}}
\renewcommand{\dim}{\texttt{dim}}
\newcommand{\Link}{\texttt{Link}}
\newcommand{\LinkVal}{\texttt{LinkVal}}
\newcommand{\PossVal}{\texttt{PossVal}_\Ga}

\newcommand{\chin}{\ch^\inn}
\newcommand{\chout}{\ch^\out}
\newcommand{\chinN}{\chin_N}
\newcommand{\choutN}{\chout_N}
\newcommand{\chinNal}{\chin_\Nal}
\newcommand{\choutNal}{\chout_\Nal}
\newcommand{\Happ}{\texttt{Happens}}
\newcommand{\HappNal}{{\texttt{Happens}_\Nal}}
\newcommand{\HappMbe}{{\texttt{Happens}_\Mbe}}

\newcommand{\Nal}{{N^\al}}
\newcommand{\Nbe}{{N^\bet}}
\newcommand{\Nalin}{{N^\al_\inn}}
\newcommand{\Nalout}{{N^\al_\out}}
\newcommand{\Mbe}{{M^\bet}}
\newcommand{\Mbein}{{M^\bet_\inn}}
\newcommand{\Mbeout}{{M^\bet_\out}}
\newcommand{\Oga}{{O^\ga}}
\newcommand{\Ogain}{{O^\ga_\inn}}
\newcommand{\Ogaout}{{O^\ga_\out}}

\newcommand{\zeiin}{{\ze_i^\inn}}
\newcommand{\zeipadin}{{\ze_{i,\pad}^\inn}}
\newcommand{\zeiout}{{\ze_i^\out}}
\newcommand{\zeipadout}{{\ze_{i,\pad}^\out}}
\newcommand{\zejin}{{\ze_j^\inn}}

\newcommand{\barzeiin}{{\Bar{\ze}_i^\inn}}
\newcommand{\barzeipadin}{{\Bar{\ze}_{i,\pad}^\inn}}
\newcommand{\barzeiout}{{\Bar{\ze}_i^\out}}
\newcommand{\barzeipadout}{{\Bar{\ze}_{i,\pad}^\out}}
\newcommand{\zeiplin}{{\ze_{i+1}^\inn}}
\newcommand{\zeiplpadin}{{\ze_{i+1,\pad}^\inn}}
\newcommand{\zeiplout}{{\ze_{i+1}^\out}}
\newcommand{\zeiplpadout}{{\ze_{i+1,\pad}^\out}}

\newcommand{\barzeiplpadout}{{\Bar{\ze}_{i+1,\pad}^\out}}
\newcommand{\zein}{{\ze^\inn}}
\newcommand{\zepadin}{{\ze_{\pad}^\inn}}
\newcommand{\zeout}{{\ze^\out}}
\newcommand{\zepadout}{{\ze_{\pad}^\out}}

\newcommand{\barzepadout}{{\Bar{\ze}_{\pad}^\out}}

\newcommand{\Spad}{\cs_\pad}
\newcommand{\SGa}{\cs_{\GalaN}}
\newcommand{\SGalaN}{\cs_{\GalaN}}

\newcommand{\SRelGa}{\cs^\Rel_{\Ga}}
\newcommand{\SRelGapad}{\cs^\Rel_{\Ga,\pad}}
\newcommand{\ex}{\textsc{exch}}
\newcommand{\swap}{\textsc{swap}}
\newcommand{\copyy}{\textsc{copy}}
\newcommand{\witness}{\textsc{witness}}

\newcommand{\iNalpad}{i_{N,\pad}^\al}
\newcommand{\pNalpad}{p_{N,\pad}^\al}
\newcommand{\UNal}{{U_N^\al}}

\newcommand{\av}[1]{{\color{olive}#1}}

\begin{document}

\renewcommand{\thefootnote}{\fnsymbol{footnote}}
\makeatletter
    \renewcommand\@makefnmark{\hbox{\@textsuperscript{\normalfont\color{white}\@thefnmark}}}
    \renewcommand\@makefntext[1]{%
    \parindent 1em\noindent
            \hb@xt@1.8em{%
    \hss\@textsuperscript{\normalfont\@thefnmark}}#1}
\makeatother
  \title{Consistent circuits for indefinite causal order}


\author{Augustin Vanrietvelde}
\orcid{0000-0001-9022-8655}
\email{vanrietvelde@telecom-paris.fr}
\thanks{\\ Current address: Télécom Paris, Palaiseau}
\affiliation{Quantum Group, Department of Computer Science, University of Oxford}
\affiliation{Department of Physics, Imperial College London}
\affiliation{HKU-Oxford Joint Laboratory for Quantum Information and Computation}

\author{Nick Ormrod$^*$}
\orcid{0000-0003-2717-8709}
\affiliation{Quantum Group, Department of Computer Science, University of Oxford}
\affiliation{HKU-Oxford Joint Laboratory for Quantum Information and Computation}
\email{normrod@perimeterinstitute.ca}
\thanks{\\ Current address: Perimeter Institute for Theoretical Physics, Waterloo}

\author{Hl\'er Kristj\'ansson$^*$}
\orcid{0000-0003-4465-2863}
\email{hler.kristjansson@umontreal.ca}
\thanks{\\ Current address: Department of Computer Science and Operations Research,  Universit\'e de Montr\'eal}
\affiliation{Quantum Group, Department of Computer Science, University of Oxford}
\affiliation{HKU-Oxford Joint Laboratory for Quantum Information and Computation}

\author{Jonathan Barrett}
\orcid{0000-0002-2222-0579}
\email{jonathan.barrett@cs.ox.ac.uk}
\affiliation{Quantum Group, Department of Computer Science, University of Oxford}

\begin{abstract}
    
    Over the past decade, a number of quantum processes have been proposed which are logically consistent,  yet feature a cyclic causal structure.
    However, there is no general formal method to construct a process with an exotic causal structure in a way that ensures, and makes clear why, it is consistent. 
    Here we provide such a method, given by an extended circuit formalism.
    This only requires directed graphs endowed with Boolean matrices, which encode basic constraints on operations.
    Our framework (a) defines a set of elementary rules for checking the validity of any such graph, (b) provides a way of constructing consistent processes as a circuit from valid graphs, and (c) yields an intuitive interpretation of the causal relations within a process and an explanation of why they do not lead to inconsistencies.
    We display how several standard examples of exotic processes, including ones that violate causal inequalities, are among the class of processes that can be generated in this way.
    We conjecture that this class in fact includes all unitarily extendible processes. {\footnote{These two authors made equal contributions to this work.}}
\end{abstract}

\maketitle

\tableofcontents
\renewcommand{\thefootnote}{\arabic{footnote}}
\makeatletter
    \renewcommand\@makefnmark{\hbox{\@textsuperscript{\normalfont\color{black}\@thefnmark}}}
\makeatother

\section{Introduction}

\changes{As traditionally understood, quantum theory is \textit{compositional}. On the one hand, one can apply two channels, one after the other, resulting in an overall process which is itself a quantum channel. On the other hand, two quantum channels can be applied at the same time, each to a different system, defining a quantum channel on the bipartite system. In other words, two quantum channels can be \textit{composed}, in sequence or in parallel, to form another quantum channel. This notion of compositionality is generalized by the quantum circuit formalism \cite{deutsch1989quantum,aharonov1998quantum,abramsky2004categorical,nielsen2000quantum,coecke_kissinger_2017}, in which arbitrary quantum channels can composed by `wiring them up' into a circuit, and any circuit defines a valid quantum channel. 

All of this is very convenient, because compositionality turns out to be a very useful and illuminating feature of the theory. Most obviously, compositionality provides a convenient way of discovering new quantum channels: begin with a set of known quantum channels, and compose them together to form a circuit of your choosing. Once this is done, there is no need to check that the circuit defines a valid quantum channel; the circuit construction itself is a proof that it does. It follows that circuit constructions can be used not only to discover new channels, but also to verify that a given process is in fact a valid quantum channel.

A circuit can be `tweaked' by replacing the individual quantum channels in the circuit without changing the way in which they are composed to form a circuit. Given a circuit construction of a channel, a vast class of related processes is obtained by tweaking, and every one of these process must be a quantum channel. Thus a circuit construction provides a \textit{compositional explanation} of why a process is a valid quantum channel: it is valid not because of any specific facts about the individual channels in the circuit, but solely as a result of the way in which they are composed.

Thus compositionality helps one discover, verify, and explain quantum channels. In addition, the circuit formalism facilitates the development of graphical calculi that help speed up calculations, prove theorems, and better appreciate symmetries of processes that are suppressed by more conventional, one-dimensional notation \cite{coecke_kissinger_2017, chiribella2008transforming}.

However, in the last two decades, a class of quantum processes have been proposed that call into question whether compositionality is a completely general feature of quantum theory \cite{hardy2005probability,chiribella2009beyond,chiribella2013quantum,oreshkov2012quantum, baumeler2014maximal, baumeler2016space, wechs2021quantum}. These \textit{higher-order} processes act not on states, but on quantum channels. They are typically represented by supermaps \cite{chiribella2009beyond, chiribella2013quantum}, or, equivalently, process matrices \cite{oreshkov2012quantum}. Many of them place their input channels in an `indefinite causal order', meaning that they cannot be constructed as a circuit without \textit{feedback loops}. Feedback loops — forbidden by the standard circuit formalism — are formed when a wire coming out of the circuit is plugged into a wire going into the circuit. They are problematic because when quantum channels are wired together to form a `loopy circuit' it is \textit{not} guaranteed that the circuit itself defines a quantum channel. Introducing loops into the circuit formalism means sacrificing one of its most appealing features: that a circuit construction is a proof of the validity of the constructed process.

Without this crucial feature, the circuit framework loses much of its appeal. Circuit construction is no longer a safe way of discovering new processes, since a loopy circuit of channels might fail to define a valid process. A circuit that does define a valid process cannot be safely tweaked, and thus does not provide a compositional explanation of the process it defines. It is doubtful that the loopy circuit formalism can facilitate the development of useful graphical calculi.

Since a circuit construction of a process with indefinite causal order is not a proof of its consistency, one needs other proof methods. Existing methods for demonstrating the consistency of a process often involve brute-force, numerical analysis on its \textit{process matrix} \cite{oreshkov2012quantum}. Such methods become impractical as the numbers of dimensions and parties increase. Even when these methods work, they provide little or no conceptual insight. 

We are therefore forced to confront the question of whether compositionality is a general feature of quantum theory, or a notion that applies only to those quantum processes that lack indefinite causal order. Should we infer from the issues with feedback loops that processes with indefinite causal order are not compositional? Or do we simply lack the right formalism to see the compositional properties that these processes have? The stakes are high, since our answer to this question will determine whether the many advantages that compositionality brings to quantum processes with definite causal order carry to the indefinite case.

This paper shows that the notion of compositionality in quantum theory generalizes beyond processes with definite causal order. The main result is a theorem showing that whenever a circuit construction follows certian rules, it defines a valid process, even if the construction uses feedback loops and the circuit has indefinite causal order. This leads to a refined notion of compositionality, according to which processes can be composed only when these conditions are met. A circuit formalism based on this notion of compositionality recovers the key feature that is lost when loops are added to standard quantum circuits: that a circuit construction of a processes is a proof of its validity. As a result, other benefits of compositionality mentioned above are also restored, including the possibility of `tweaking' processes and providing compositional explanations of the validity of a process. 

While we do not show that all higher-order processes can be constructed using the framework, we do construct a number of prominent examples in the literature, including the quantum switch \cite{chiribella2009beyond, chiribella2013quantum} and 3-switch \cite{colnaghi2012quantum}, the recently proposed Grenoble process \cite{wechs2021quantum}, and the Lugano process (also called Baumeler-Wolf or Ara\'ujo-Feix) \cite{baumeler2014maximal, baumeler2016space, araujo2017purification}.
We are led to conjecture that all `unitarily extendible' \cite{araujo2017purification} processes can be constructed using the framework. We thus show that many significant processes with indefinite causal order admit a compositional explanation, and conjecture that all unitarily extendible processes do too.

Let us now describe our circuit framework in a little more detail. The circuits in the framework provide more fine-grained information about the causal structure than is available using standard circuits with loops. They are based on an extension to the standard quantum circuit framework known as routed quantum circuits, introduced in Ref.\ \cite{vanrietvelde2021routed} (see also Refs.\ \cite{barrett2019, lorenz2020, vanrietvelde2021coherent, wilson2021composable}), in which not only tensor product but also direct sum structures are graphically represented. This matters for us because our rules of compositionality refer to those direct sum structures.

The rules are articulated with the help of a \textit{routed graph}, which provides an abstract description of how processes are to be composed. Not all routed graphs are valid, in the sense that when one composes quantum channels in accordance with the routed graph, the resulting circuit defines a consistent process. Our main theorem is that any routed graph that has two properties, called \textit{bi-univocality} and \textit{weak loops}, is valid. We thus propose that in the context of indefinite causal order, `composing quantum channels' is taken to mean combining them in accordance with a routed graph that satisfies bi-univocality and weak loops. 

We emphasise that this paper does not aim to address the question of the \textit{physical realisability} of processes (though its ideas might help tackle it in the future). Rather, we are interested in understanding the abstract, logical structure lying at the heart of valid processes with indefinite causal order. Given a valid process, we do not ask whether it could be implemented in practice, or even in principle, given the laws of physics that govern our particular universe. We ask only whether a purported quantum process is \textit{logically consistent} (i.e.\ whether it implies any contradiction), what makes it so, and how it can be mathematically constructed in a way that makes this obvious. 

In summary, this paper restores a notion of compositionality, along with many of its benefits, to a significant class of processes with indefinite causal order. The notion of compositionality is subtler than the one familiar from the definite order case, due to the additional rules that must be followed when composing processes. Note, however, that even in the definite order case there is a rule constraining composition: that individual compositions between processes do not amount to forming a global loop.
Therefore, the introduction of restrictions on composition in this paper is not a novelty; we only propose a refinement of the commonly admitted ones.

The paper is structured as follows. We begin by introducing our framework using the example of a reconstruction of the quantum switch, in order to provide a pedagogical introduction to its main notions with a simple example. We then present our framework in full generality, which describes how to construct processes from elementary operations and their connectivity, in a way that guarantees logical consistency. Following this, we display how our framework allows us to reconstruct other processes prominent in the literature. We explain how the route structure displays the core behaviour of the processes in a compact way. We embed these examples into larger families of similar processes, and thus highlight the conceptual intuitions for their validity. We conclude with a short discussion and outlook, in which we spell out a conjecture that all unitarily extendible processes can be built using our method.}
\section{Reconstructing the quantum switch} \label{sec:2switch}

The quantum switch \cite{chiribella2009beyond, chiribella2013quantum} is a process with indefinite causal order whose logical consistency is relatively easy to understand. However, the obvious ways of grasping its consistency -- such as recognising it as the coherent control of processes with definite causal orders -- do not generalise to more exotic instances of indefinite causal order, such as the Grenoble \cite{wechs2021quantum} or Lugano \cite{baumeler2014maximal, baumeler2016space, araujo2017purification} processes (which we will discuss later in Section \ref{sec: examples}). In this section, we reconstruct the quantum switch in a way that guarantees that it is a valid supermap\footnote{Note that in this paper, we will adopt the \textit{supermap} representation of higher-order processes \cite{chiribella2008transforming, chiribella2009beyond, chiribella2013quantum}, as opposed to the (mathematically equivalent) \textit{process matrix} representation \cite{oreshkov2012quantum, araujo2017purification} also used in the literature. Because the relationship and equivalence between the two pictures, both at the conceptual and mathematical levels, can be a source of confusion, and in the interest of readers more accustomed to process matrices, we spell out the connection between the representations in detail in Appendix \ref{app: process matrices and supermaps}.}, using a method that does generalise to more exotic unitarily extendible processes. In so doing, we sketch out the main ingredients of our framework before it is formally introduced in Section \ref{sec: theorem}.

We start this section by defining supermaps, and, in particular, the supermap describing the quantum switch. We then briefly summarise some key concepts of the routed circuits framework \cite{vanrietvelde2021routed, vanrietvelde2021coherent, wilson2021composable}, which we will use for the reconstruction. Finally, we perform the reconstruction in a pedagogical way. 

\subsection{Supermaps and the quantum switch}

The quantum switch is a process in which the order of application of two transformations is coherently controlled.
Mathematically, it can be represented by  a \textit{supermap}, called \texttt{SWITCH}.
In the literature, a supermap is typically defined as a linear map that transforms quantum channels to quantum channels \cite{chiribella2008transforming, chiribella2009beyond, chiribella2013quantum}.
However, since we are only interested in unitarily extendible supermaps in this paper, we shall define a supermap as a linear map that transforms linear operators on a Hilbert space to linear operators on a Hilbert space.
Given two input operators $U$ and $V$ of the same dimension $d$, \texttt{SWITCH} returns an operator of the form
\begin{equation} \label{switch def}
    \textnormal{\texttt{SWITCH}}(U, V) = \ket{0}\bra{0} \otimes VU + \ket{1}\bra{1} \otimes UV \, .
\end{equation}
More generally, $U$ and $V$ could be acting on their own local ancillary systems, $X$ and $Y$ respectively. Then \texttt{SWITCH} is defined as follows:
\begin{equation} \label{switch def 2}
    \texttt{SWITCH}(U, V) = \ket{0}\bra{0} \otimes (I_X \otimes V) (U \otimes I_Y)  + \ket{1}\bra{1} \otimes (U \otimes I_Y)  (I_X \otimes V)
\end{equation}

If $U$ and $V$ are both unitary operators, then $\texttt{SWITCH}(U, V)$ is also a unitary operator. We call supermaps like this, that always map unitary operators to unitary operators, \textit{superunitaries}. A little reflection reveals that any superunitary \changes{(defined for all input unitaries that can act on their own local ancillary systems)} uniquely defines a supermap in the traditional sense as a map on channels via the Stinespring dilation of the channels. 


An equally formal but more intuitive representation of supermaps can be provided using diagrams \cite{coecke_kissinger_2017, Kissinger2019}. The idea, illustrated for a monopartite supermap in Figure \ref{fig:monopartite smap}, is to represent unitary operators using boxes, and supermaps as shapes that give another box once one inserts a box into each of its `nodes'.
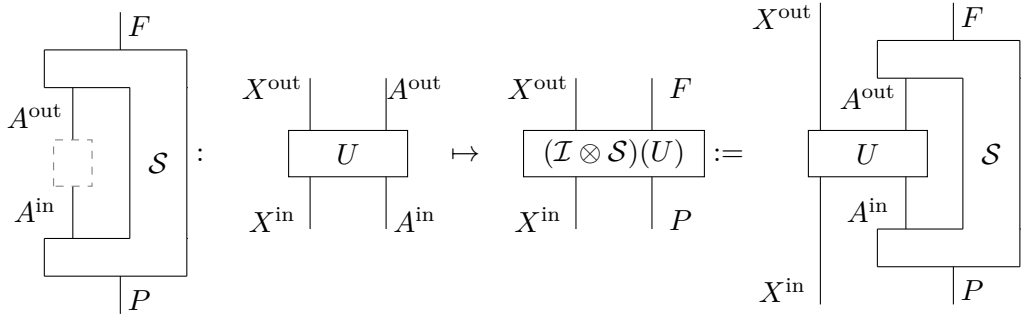
\begin{figure*}
    \centering
    $\begin{tikzpicture}
	\begin{pgfonlayer}{nodelayer}
		\node [style=none] (0) at (4.75, 2.75) {};
		\node [style=none] (1) at (4.75, 1.75) {};
		\node [style=none] (2) at (4.75, 1.75) {};
		\node [style=none] (3) at (4.75, -2.25) {};
		\node [style=none] (4) at (4.75, -2.25) {};
		\node [style=none] (5) at (4.75, -3.25) {};
		\node [style=none] (6) at (1, -3.25) {};
		\node [style=none] (7) at (1, 2.75) {};
		\node [style=none] (12) at (4, -0.25) {$\cs$};
		\node [style=none] (13) at (3.25, 1.75) {};
		\node [style=none] (14) at (3.25, -2.25) {};
		\node [style=none] (15) at (1.75, 1.75) {};
		\node [style=none] (18) at (1, 1.75) {};
		\node [style=none] (19) at (1, -2.25) {};
		\node [style=none] (20) at (1.75, -2.25) {};
		\node [style=none] (21) at (3, -3.25) {};
		\node [style=none] (22) at (3, -4.25) {};
		\node [style=none] (26) at (3.5, -3.875) {$P$};
		\node [style=none] (27) at (3, 2.75) {};
		\node [style=none] (28) at (3, 3.75) {};
		\node [style=none] (33) at (0.75, 1) {${A^\out}$};
		\node [style=none] (34) at (0.75, -1.5) {${A^\inn}$};
		\node [style=none] (36) at (3.5, 3.375) {$F$};
		\node [style=dashedmap] (37) at (1.75, -0.25) {};
	\end{pgfonlayer}
	\begin{pgfonlayer}{edgelayer}
		\draw (0.center) to (1.center);
		\draw (1.center) to (2.center);
		\draw (2.center) to (3.center);
		\draw (3.center) to (4.center);
		\draw (4.center) to (5.center);
		\draw (5.center) to (6.center);
		\draw (7.center) to (0.center);
		\draw (18.center) to (13.center);
		\draw (13.center) to (14.center);
		\draw (14.center) to (19.center);
		\draw (15.center) to (20.center);
		\draw (18.center) to (7.center);
		\draw (19.center) to (6.center);
		\draw (21.center) to (22.center);
		\draw (27.center) to (28.center);
	\end{pgfonlayer}
\end{tikzpicture}:  \ \ \    \begin{tikzpicture}
	\begin{pgfonlayer}{nodelayer}
		\node [style=large map] (0) at (0, 0) {$U$};
		\node [style=none] (1) at (-1, -2) {};
		\node [style=none] (2) at (-1, 2) {};
		\node [style=none] (3) at (1, 2) {};
		\node [style=none] (4) at (1, -2) {};
		\node [style=none] (5) at (-2, 1.75) {$X^\out$};
		\node [style=none] (6) at (-2, -1.75) {$X^\inn$};
		\node [style=none] (7) at (1.75, 1.75) {$A^\out$};
		\node [style=none] (8) at (1.75, -1.75) {$A^\inn$};
	\end{pgfonlayer}
	\begin{pgfonlayer}{edgelayer}
		\draw (1.center) to (2.center);
		\draw (4.center) to (3.center);
	\end{pgfonlayer}
\end{tikzpicture} \mapsto \ \ 
       \begin{tikzpicture}
	\begin{pgfonlayer}{nodelayer}
		\node [style=large map] (0) at (0, 0) {$(\ci \otimes \cs)(U)$};
		\node [style=none] (1) at (-1, -2) {};
		\node [style=none] (2) at (-1, 2) {};
		\node [style=none] (3) at (1, 2) {};
		\node [style=none] (4) at (1, -2) {};
		\node [style=none] (5) at (-2, 1.75) {$X^\out$};
		\node [style=none] (6) at (-2, -1.75) {$X^\inn$};
		\node [style=none] (7) at (1.75, 1.75) {$F$};
		\node [style=none] (8) at (1.75, -1.75) {$P$};
	\end{pgfonlayer}
	\begin{pgfonlayer}{edgelayer}
		\draw (1.center) to (2.center);
		\draw (4.center) to (3.center);
	\end{pgfonlayer}
\end{tikzpicture} := \begin{tikzpicture}
	\begin{pgfonlayer}{nodelayer}
		\node [style=none] (0) at (4.75, 3) {};
		\node [style=none] (1) at (4.75, 2) {};
		\node [style=none] (2) at (4.75, 2) {};
		\node [style=none] (3) at (4.75, -2) {};
		\node [style=none] (4) at (4.75, -2) {};
		\node [style=none] (5) at (4.75, -3) {};
		\node [style=none] (6) at (1, -3) {};
		\node [style=none] (7) at (1, 3) {};
		\node [style=none] (12) at (4, 0) {$\cs$};
		\node [style=none] (13) at (3.25, 2) {};
		\node [style=none] (14) at (3.25, -2) {};
		\node [style=none] (15) at (1.75, 2) {};
		\node [style=none] (18) at (1, 2) {};
		\node [style=none] (19) at (1, -2) {};
		\node [style=none] (20) at (1.75, -2) {};
		\node [style=none] (21) at (3, -3) {};
		\node [style=none] (22) at (3, -4) {};
		\node [style=none] (26) at (3.5, -3.625) {$P$};
		\node [style=none] (27) at (3, 3) {};
		\node [style=none] (28) at (3, 4) {};
		\node [style=none] (33) at (0.75, 1.5) {${A^\out}$};
		\node [style=none] (34) at (0.75, -1.5) {${A^\inn}$};
		\node [style=none] (36) at (3.5, 3.625) {$F$};
		\node [style=large map] (38) at (0.75, 0) {$U$};
		\node [style=none] (39) at (-0.5, -4) {};
		\node [style=none] (40) at (-0.5, 4) {};
		\node [style=none] (41) at (-1.5, -3.625) {$X^\inn$};
		\node [style=none] (42) at (-1.5, 3.625) {$X^\out$};
	\end{pgfonlayer}
	\begin{pgfonlayer}{edgelayer}
		\draw (0.center) to (1.center);
		\draw (1.center) to (2.center);
		\draw (2.center) to (3.center);
		\draw (3.center) to (4.center);
		\draw (4.center) to (5.center);
		\draw (5.center) to (6.center);
		\draw (7.center) to (0.center);
		\draw (18.center) to (13.center);
		\draw (13.center) to (14.center);
		\draw (14.center) to (19.center);
		\draw (15.center) to (20.center);
		\draw (18.center) to (7.center);
		\draw (19.center) to (6.center);
		\draw (21.center) to (22.center);
		\draw (27.center) to (28.center);
		\draw (39.center) to (40.center);
	\end{pgfonlayer}
\end{tikzpicture}$

    \caption{Diagrammatic representation of a monopartite superunitary. $\cs$ is a linear map from unitaries of the form $U: \ch_{X^\inn} \otimes \ch_{A^\inn} \rightarrow \ch_{X^\out} \otimes \ch_{A^\out}$ that maps the ingoing space for any pair of ancillary systems $X^\inn$ and $X^\out$ to unitaries of the form $(\ci \otimes \cs)(U): \ch_{X^\inn} \otimes \ch_P \rightarrow \ch_{X^\out} \otimes \ch_F$. This definition can be extended in an obvious way for supermaps acting on an arbitrary finite number of unitary operators.}
    \label{fig:monopartite smap}
\end{figure*}

In general, one can use a circuit decomposition with sums to represent a supermap. Such a decomposition is provided for the switch in Figure \ref{fig:switch pic}. The meaning of this circuit decomposition is that the  action of the supermap is given by Figure \ref{fig:switch output}, the right-hand side of which has precisely the same formal meaning as (\ref{switch def 2}).
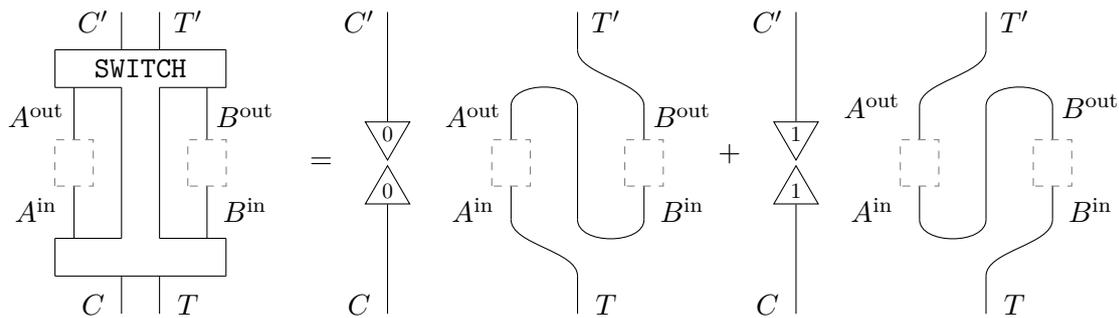
\begin{figure*} 
    \begin{tikzpicture}
	\begin{pgfonlayer}{nodelayer}
		\node [style=none] (0) at (2.5, 3) {};
		\node [style=none] (1) at (2.5, 2) {};
		\node [style=none] (2) at (0.75, 2) {};
		\node [style=none] (3) at (0.75, -2) {};
		\node [style=none] (4) at (2.5, -2) {};
		\node [style=none] (5) at (2.5, -3) {};
		\node [style=none] (6) at (-2, -3) {};
		\node [style=none] (7) at (-2, 3) {};
		\node [style=none] (8) at (2, 2) {};
		\node [style=none] (9) at (2, -2) {};
		\node [style=none] (12) at (0.25, 2.5) {\texttt{SWITCH}};
		\node [style=none] (13) at (-0.25, 2) {};
		\node [style=none] (14) at (-0.25, -2) {};
		\node [style=none] (15) at (-1.5, 2) {};
		\node [style=none] (18) at (-2, 2) {};
		\node [style=none] (19) at (-2, -2) {};
		\node [style=none] (20) at (-1.5, -2) {};
		\node [style=none] (21) at (0.75, -3) {};
		\node [style=none] (22) at (0.75, -4) {};
		\node [style=none] (23) at (-0.25, -3) {};
		\node [style=none] (24) at (-0.25, -4) {};
		\node [style=none] (31) at (0.75, 3) {};
		\node [style=none] (32) at (0.75, 4) {};
		\node [style=none] (33) at (-0.25, 3) {};
		\node [style=none] (34) at (-0.25, 4) {};
		\node [style=none] (37) at (3, 1.25) {${ B^\out}$};
		\node [style=none] (38) at (3, -1.25) {${ B^\inn}$};
		\node [style=none] (39) at (-2.5, 1.25) {${A^\out}$};
		\node [style=none] (40) at (-2.5, -1.25) {${A^\inn}$};
		\node [style=none] (43) at (5, 0) {$=$};
		\node [style=none] (45) at (24.25, 0.5) {};
		\node [style=none] (46) at (20.75, 1.5) {};
		\node [style=none] (47) at (20.75, 0.5) {};
		\node [style=none] (48) at (24.25, -1.5) {};
		\node [style=none] (49) at (22.5, -4) {};
		\node [style=none] (50) at (22.5, 4) {};
		\node [style=none] (51) at (20.75, -1.5) {};
		\node [style=none] (52) at (22.5, 3) {};
		\node [style=none] (53) at (24.25, 1.5) {};
		\node [style=none] (54) at (22.5, -3) {};
		\node [style=none] (55) at (24.25, 0.5) {};
		\node [style=none] (56) at (20.75, 0.5) {};
		\node [style=none] (57) at (22.5, 1.5) {};
		\node [style=none] (58) at (22.5, -1.5) {};
		\node [style=none] (59) at (23.25, 3.75) {$T'$};
		\node [style=none] (64) at (23.25, -3.75) {$T$};
		\node [style=point] (74) at (17.5, 0.75) {$1$};
		\node [style=copoint] (75) at (17.5, -0.75) {$1$};
		\node [style=none] (76) at (17.5, -4) {};
		\node [style=none] (77) at (17.5, 4) {};
		\node [style=none] (83) at (25.25, 1.5) {${ B^\out}$};
		\node [style=none] (84) at (25.25, -1.25) {${ B^\inn}$};
		\node [style=none] (85) at (19.5, 1.5) {${A^\out}$};
		\node [style=none] (86) at (19.5, -1.25) {${A^\inn}$};
		\node [style=none] (89) at (10, 0.5) {};
		\node [style=none] (90) at (13.5, 1.5) {};
		\node [style=none] (91) at (13.5, 0.5) {};
		\node [style=none] (92) at (10, -1.5) {};
		\node [style=none] (93) at (11.75, -4) {};
		\node [style=none] (94) at (11.75, 4) {};
		\node [style=none] (95) at (13.5, -1.5) {};
		\node [style=none] (96) at (11.75, 3) {};
		\node [style=none] (97) at (10, 1.5) {};
		\node [style=none] (98) at (11.75, -3) {};
		\node [style=none] (99) at (10, 0.5) {};
		\node [style=none] (100) at (13.5, 0.5) {};
		\node [style=none] (101) at (11.75, 1.5) {};
		\node [style=none] (102) at (11.75, -1.5) {};
		\node [style=none] (104) at (12.5, -3.75) {$T$};
		\node [style=point] (105) at (6.75, 0.75) {$0$};
		\node [style=copoint] (106) at (6.75, -0.75) {$0$};
		\node [style=none] (107) at (6.75, -4) {};
		\node [style=none] (108) at (6.75, 4) {};
		\node [style=none] (109) at (14.5, 1.25) {${ B^\out}$};
		\node [style=none] (110) at (14.5, -1.25) {${ B^\inn}$};
		\node [style=none] (111) at (9, 1.25) {${A^\out}$};
		\node [style=none] (112) at (9, -1.25) {${A^\inn}$};
		\node [style=none] (115) at (15.75, 0.25) {$+$};
		\node [style=dashedmap] (116) at (-1.5, 0) {};
		\node [style=dashedmap] (117) at (2, 0) {};
		\node [style=dashedmap] (118) at (10, 0) {};
		\node [style=dashedmap] (119) at (13.5, 0) {};
		\node [style=dashedmap] (120) at (20.75, 0) {};
		\node [style=dashedmap] (121) at (24.25, 0) {};
		\node [style=none] (122) at (6, -3.75) {$C$};
		\node [style=none] (123) at (16.75, -3.75) {$C$};
		\node [style=none] (124) at (-1, 3.75) {$C'$};
		\node [style=none] (125) at (1.5, 3.75) {$T'$};
		\node [style=none] (126) at (6, 3.75) {$C'$};
		\node [style=none] (127) at (12.5, 3.75) {$T'$};
		\node [style=none] (128) at (16.75, 3.75) {$C'$};
		\node [style=none] (129) at (-1, -3.75) {$C$};
		\node [style=none] (130) at (1.5, -3.75) {$T$};
		\node [style=none] (131) at (13.5, 1.5) {};
		\node [style=none] (132) at (17.5, 1.5) {};
	\end{pgfonlayer}
	\begin{pgfonlayer}{edgelayer}
		\draw (0.center) to (1.center);
		\draw (1.center) to (2.center);
		\draw (2.center) to (3.center);
		\draw (3.center) to (4.center);
		\draw (4.center) to (5.center);
		\draw (5.center) to (6.center);
		\draw (7.center) to (0.center);
		\draw (8.center) to (9.center);
		\draw (18.center) to (13.center);
		\draw (13.center) to (14.center);
		\draw (14.center) to (19.center);
		\draw (15.center) to (20.center);
		\draw (18.center) to (7.center);
		\draw (19.center) to (6.center);
		\draw (21.center) to (22.center);
		\draw (24.center) to (23.center);
		\draw (31.center) to (32.center);
		\draw (34.center) to (33.center);
		\draw (46.center) to (56.center);
		\draw (47.center) to (51.center);
		\draw (53.center) to (55.center);
		\draw (45.center) to (48.center);
		\draw (54.center) to (49.center);
		\draw (50.center) to (52.center);
		\draw [in=-90, out=90, looseness=0.75] (54.center) to (48.center);
		\draw [in=90, out=90] (53.center) to (57.center);
		\draw (57.center) to (58.center);
		\draw [in=-90, out=-90] (58.center) to (51.center);
		\draw [in=-90, out=90, looseness=0.75] (46.center) to (52.center);
		\draw (76.center) to (75);
		\draw (74) to (77.center);
		\draw (90.center) to (100.center);
		\draw (91.center) to (95.center);
		\draw (97.center) to (99.center);
		\draw (89.center) to (92.center);
		\draw (98.center) to (93.center);
		\draw (94.center) to (96.center);
		\draw [in=-90, out=90, looseness=0.75] (98.center) to (92.center);
		\draw [in=90, out=90] (97.center) to (101.center);
		\draw (101.center) to (102.center);
		\draw [in=-90, out=-90] (102.center) to (95.center);
		\draw [in=-90, out=90, looseness=0.75] (90.center) to (96.center);
		\draw (107.center) to (106);
		\draw (105) to (108.center);
	\end{pgfonlayer}
\end{tikzpicture}
\caption{Circuit decomposition of the switch. The left term in the sum projects onto the $\ket{0}$ state of the control and implements Alice's transformation before Bob's. The right term has a similar interpretation. Formally, the wire bent into a `U' shape can be interpreted as the unnormalised Bell ket $\ket{00}+\ket{11}$, and the upside-down `U' as the corresponding bra.} \label{fig:switch pic}
\end{figure*}
\begin{figure*} 
\scalebox{.9}{\begin{tikzpicture}
	\begin{pgfonlayer}{nodelayer}
		\node [style=none] (99) at (41.75, 3) {};
		\node [style=none] (100) at (41.75, 2) {};
		\node [style=none] (101) at (40, 2) {};
		\node [style=none] (102) at (40, -2) {};
		\node [style=none] (103) at (41.75, -2) {};
		\node [style=none] (104) at (41.75, -3) {};
		\node [style=none] (105) at (37.25, -3) {};
		\node [style=none] (106) at (37.25, 3) {};
		\node [style=none] (107) at (41.25, 2) {};
		\node [style=none] (108) at (41.25, -2) {};
		\node [style=none] (109) at (39.5, 2.5) {\texttt{SWITCH}};
		\node [style=none] (110) at (39, 2) {};
		\node [style=none] (111) at (39, -2) {};
		\node [style=none] (112) at (37.75, 2) {};
		\node [style=none] (113) at (37.25, 2) {};
		\node [style=none] (114) at (37.25, -2) {};
		\node [style=none] (115) at (37.75, -2) {};
		\node [style=none] (116) at (40, -3) {};
		\node [style=none] (117) at (40, -4) {};
		\node [style=none] (118) at (39, -3) {};
		\node [style=none] (119) at (39, -4) {};
		\node [style=none] (120) at (40, 3) {};
		\node [style=none] (121) at (40, 4) {};
		\node [style=none] (122) at (39, 3) {};
		\node [style=none] (123) at (39, 4) {};
		\node [style=none] (128) at (44.75, 0) {$=$};
		\node [style=none] (129) at (65.5, 0.5) {};
		\node [style=none] (130) at (62, 1.5) {};
		\node [style=none] (131) at (62, 0.5) {};
		\node [style=none] (132) at (65.5, -1.5) {};
		\node [style=none] (133) at (63.75, -4) {};
		\node [style=none] (134) at (63.75, 4) {};
		\node [style=none] (135) at (62, -1.5) {};
		\node [style=none] (136) at (63.75, 3) {};
		\node [style=none] (137) at (65.5, 1.5) {};
		\node [style=none] (138) at (63.75, -3) {};
		\node [style=none] (139) at (65.5, 0.5) {};
		\node [style=none] (140) at (62, 0.5) {};
		\node [style=none] (141) at (63.75, 1.5) {};
		\node [style=none] (142) at (63.75, -1.5) {};
		\node [style=none] (143) at (64.5, 3.75) {$T'$};
		\node [style=none] (144) at (64.5, -3.75) {$T$};
		\node [style=point] (145) at (58.75, 0.75) {$1$};
		\node [style=copoint] (146) at (58.75, -0.75) {$1$};
		\node [style=none] (147) at (58.75, -4) {};
		\node [style=none] (148) at (58.75, 4) {};
		\node [style=none] (153) at (50, 0.5) {};
		\node [style=none] (154) at (53.5, 1.5) {};
		\node [style=none] (155) at (53.5, 0.5) {};
		\node [style=none] (156) at (50, -1.5) {};
		\node [style=none] (157) at (51.75, -4) {};
		\node [style=none] (158) at (51.75, 4) {};
		\node [style=none] (159) at (53.5, -1.5) {};
		\node [style=none] (160) at (51.75, 3) {};
		\node [style=none] (161) at (50, 1.5) {};
		\node [style=none] (162) at (51.75, -3) {};
		\node [style=none] (163) at (50, 0.5) {};
		\node [style=none] (164) at (53.5, 0.5) {};
		\node [style=none] (165) at (51.75, 1.5) {};
		\node [style=none] (166) at (51.75, -1.5) {};
		\node [style=none] (167) at (52.5, -3.75) {$T$};
		\node [style=point] (168) at (46.5, 0.75) {$0$};
		\node [style=copoint] (169) at (46.5, -0.75) {$0$};
		\node [style=none] (170) at (46.5, -4) {};
		\node [style=none] (171) at (46.5, 4) {};
		\node [style=none] (176) at (56.75, 0.25) {$+$};
		\node [style=none] (183) at (45.75, -3.75) {$C$};
		\node [style=none] (184) at (58, -3.75) {$C$};
		\node [style=none] (185) at (38.25, 3.75) {$C'$};
		\node [style=none] (186) at (40.75, 3.75) {$T'$};
		\node [style=none] (187) at (45.75, 3.75) {$C'$};
		\node [style=none] (188) at (52.5, 3.75) {$T'$};
		\node [style=none] (189) at (58, 3.75) {$C'$};
		\node [style=none] (190) at (38.25, -3.75) {$C$};
		\node [style=none] (191) at (40.75, -3.75) {$T$};
		\node [style=none] (192) at (53.5, 1.5) {};
		\node [style=none] (193) at (58.75, 1.5) {};
		\node [style=semilarge map] (194) at (37, 0) {$U$};
		\node [style=none] (195) at (36.25, -4) {};
		\node [style=none] (196) at (36.25, 4) {};
		\node [style=none] (197) at (42.75, -4) {};
		\node [style=none] (198) at (42.75, 4) {};
		\node [style=semilarge map] (199) at (42, 0) {$V$};
		\node [style=none] (200) at (35.5, -3.75) {$X^\inn$};
		\node [style=none] (201) at (35.5, 3.75) {$X^\out$};
		\node [style=semilarge map] (204) at (61.25, 0) {$U$};
		\node [style=none] (205) at (60.5, -4) {};
		\node [style=none] (206) at (60.5, 4) {};
		\node [style=none] (207) at (59.75, -3.75) {$X^\inn$};
		\node [style=none] (208) at (59.75, 3.75) {$X^\out$};
		\node [style=semilarge map] (209) at (49.25, 0) {$U$};
		\node [style=none] (210) at (48.5, -4) {};
		\node [style=none] (211) at (48.5, 4) {};
		\node [style=none] (212) at (47.75, -3.75) {$X^\inn$};
		\node [style=none] (213) at (47.75, 3.75) {$X^\out$};
		\node [style=none] (214) at (43.5, -3.75) {$Y^\inn$};
		\node [style=none] (215) at (43.5, 3.75) {$Y^\out$};
		\node [style=none] (216) at (67, -4) {};
		\node [style=none] (217) at (67, 4) {};
		\node [style=semilarge map] (218) at (66.25, 0) {$V$};
		\node [style=none] (219) at (67.75, -3.75) {$Y^\inn$};
		\node [style=none] (220) at (67.75, 3.75) {$Y^\out$};
		\node [style=none] (221) at (55, -4) {};
		\node [style=none] (222) at (55, 4) {};
		\node [style=semilarge map] (223) at (54.25, 0) {$V$};
		\node [style=none] (224) at (55.75, -3.75) {$Y^\inn$};
		\node [style=none] (225) at (55.75, 3.75) {$Y^\out$};
	\end{pgfonlayer}
	\begin{pgfonlayer}{edgelayer}
		\draw (99.center) to (100.center);
		\draw (100.center) to (101.center);
		\draw (101.center) to (102.center);
		\draw (102.center) to (103.center);
		\draw (103.center) to (104.center);
		\draw (104.center) to (105.center);
		\draw (106.center) to (99.center);
		\draw (107.center) to (108.center);
		\draw (113.center) to (110.center);
		\draw (110.center) to (111.center);
		\draw (111.center) to (114.center);
		\draw (112.center) to (115.center);
		\draw (113.center) to (106.center);
		\draw (114.center) to (105.center);
		\draw (116.center) to (117.center);
		\draw (119.center) to (118.center);
		\draw (120.center) to (121.center);
		\draw (123.center) to (122.center);
		\draw (130.center) to (140.center);
		\draw (131.center) to (135.center);
		\draw (137.center) to (139.center);
		\draw (129.center) to (132.center);
		\draw (138.center) to (133.center);
		\draw (134.center) to (136.center);
		\draw [in=-90, out=90, looseness=0.75] (138.center) to (132.center);
		\draw [in=90, out=90] (137.center) to (141.center);
		\draw (141.center) to (142.center);
		\draw [in=-90, out=-90] (142.center) to (135.center);
		\draw [in=-90, out=90, looseness=0.75] (130.center) to (136.center);
		\draw (147.center) to (146);
		\draw (145) to (148.center);
		\draw (154.center) to (164.center);
		\draw (155.center) to (159.center);
		\draw (161.center) to (163.center);
		\draw (153.center) to (156.center);
		\draw (162.center) to (157.center);
		\draw (158.center) to (160.center);
		\draw [in=-90, out=90, looseness=0.75] (162.center) to (156.center);
		\draw [in=90, out=90] (161.center) to (165.center);
		\draw (165.center) to (166.center);
		\draw [in=-90, out=-90] (166.center) to (159.center);
		\draw [in=-90, out=90, looseness=0.75] (154.center) to (160.center);
		\draw (170.center) to (169);
		\draw (168) to (171.center);
		\draw (195.center) to (196.center);
		\draw (197.center) to (198.center);
		\draw (205.center) to (206.center);
		\draw (210.center) to (211.center);
		\draw (216.center) to (217.center);
		\draw (221.center) to (222.center);
	\end{pgfonlayer}
\end{tikzpicture} }
\caption{The action of the switch on a pair of unitary operators. The order of implementation of unitaries on the target system and ancillas is coherently controlled.}
\label{fig:switch output}
\end{figure*}
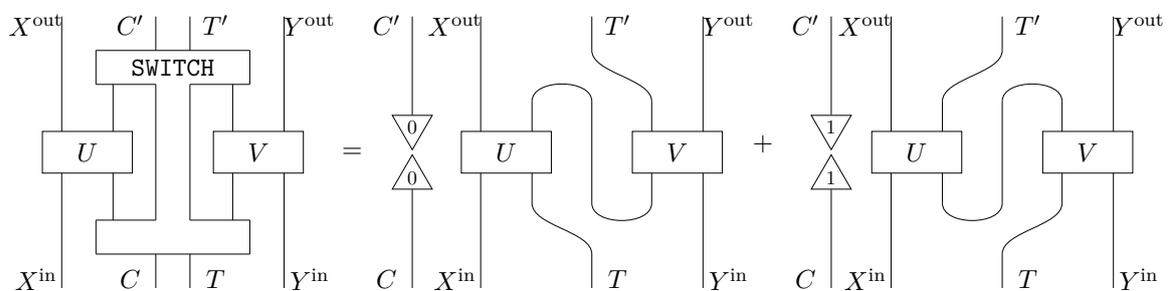
However, if we do not allow for sums,\footnote{This is good practice because 1) sums lead to an exponential multiplication of the number of diagrams to consider, and 2) an intuitive presentation as a sum will not be available at all in more involved cases, like the Grenoble process or the Lugano process.} then it is impossible to draw $\texttt{SWITCH}(U, V)$ as a standard circuit in which both $U$ and $V$ appear exactly once, unless we allow feedback loops \cite{chiribella2009beyond, chiribella2013quantum, Oreshkov_2019}. This means that we cannot write $\texttt{SWITCH}(U, V)$ as a standard circuit whose form makes it immediate that the switch is a consistent process.

This feature of the switch, while somewhat inconvenient, is also a hallmark of indefinite causal order. 
Such processes combine their input operations in a way that cannot be understood as wiring them up to form a circuit, while avoiding feedback loops. 
While this is precisely what makes these processes interesting, it also can leave us guessing as to what makes them consistent.
If their consistency is not guaranteed by the possibility of constructing them as a standard, acyclic circuit, then what is it guaranteed by?

We will show that the consistency of the switch is guaranteed through a presentation of it as a \textit{routed} quantum circuit \cite{vanrietvelde2021routed, vanrietvelde2021coherent, wilson2021composable}. Our method involves writing down a decorated directed graph, called the 
\textit{routed graph}, which captures the basic compositional structure of a large class of circuits.
To check that all circuits that display this compositional structure are valid ones, it suffices to show that the graph satisfies two principles.
Moreover, this method will generalise to other prominent unitarily extendible processes, such as Grenoble and Lugano.

All of this is in spite of the fact that the routed graph may contain feedback loops. Thus while the usual framework of quantum circuits shows how one can validly combine quantum processes in a definite causal order, our framework shows how quantum processes can be validly combined in an indefinite order.

At the heart of what makes a routed graph succeed or fail to generate consistent processes are the \textit{sectorial constraints} it enforces at its nodes. These restrict the transformations that they accept, and, in valid graphs, they can play a vital role in outlawing those transformations that would lead to inconsistencies. But these constraints are not captured by standard quantum circuits. To handle sectorial constraints properly, we need to understand the basics of routed quantum circuits, to which we now turn.

\subsection{Routed quantum circuits}

We start our summary of the framework of routed quantum circuits by considering a simple physical scenario that is conveniently represented with routed circuits. Suppose a photon with an internal degree of freedom is sent to a superposition of two paths, either going to the left or to the right depending on the logical value of a control qubit $C$.

This can be represented with the following transformation $W: \ch_C \otimes \ch_T \rightarrow \ch_L \otimes \ch_R$, where $T$ is the internal degree of freedom of the photon.
\begin{equation}
    \begin{split} \label{sup channels def}
       & W(\ket{0} \otimes\ket{\psi}) = \ket{\psi} \otimes \ket{\textnormal{vac}} \\
        &  W(\ket{1} \otimes\ket{\psi}) = \ket{\textnormal{vac}} \otimes \ket{\psi}
    \end{split}
\end{equation}
Each of the output spaces $\ch_L$ and $\ch_R$ is the direct sum of a vacuum sector spanned by $\ket{\textnormal{vac}}$ and a single-particle sector with the dimension of the target space $\ch_T$ \cite{chiribella2019shannon}: $\ch_L = \ch_L^\textnormal{vac} \oplus \ch_L^\textnormal{par}$ and $\ch_R = \ch_R^\textnormal{vac} \oplus \ch_R^\textnormal{par}$ .

The physical evolution here is clearly reversible; yet, $W$ is not a unitary operator, because its output space, $\ch_L \otimes \ch_R = \bigoplus_{i, j \in \{\textrm{vac}, \textrm{par}\}} \ch_L^i \otimes \ch_R^j$, is `too large'. To fix this, we could consider an extension of $W$ to a unitary acting on a larger input space. However, this would require us to consider additional, possibly physically irrelevant degrees of freedom just to represent our original scenario, even though this original scenario was already reversible. 

Another solution is to restrict the definition of $W$ to the subspace of its output space that can actually be populated -- this is the one-particle subspace  $\ch^{\rm prac}:=(\ch_L^{\rm par} \otimes \ch_R^\textnormal{vac}) \oplus (\ch_L^\textnormal{vac} \otimes \ch_R^\textnormal{par})$.  The resulting operator $\tilde{W} :  \ch_C \otimes \ch_T \rightarrow \ch^\textnormal{prac}$ is indeed unitary, but it cannot be represented as a standard circuit with one output for $L$ and another for $R$. The reason is that the formal meaning of putting two wires next to each other in a standard circuit is taking the tensor product of the associated Hilbert spaces, while our $\ch^\textnormal{prac}$ is composed from $\ch_L$ and $\ch_R$ via direct sums as well as tensor products.


In order to achieve a better circuit representation, we can label the output subspaces in the following way.
\begin{equation}
    \begin{split}
        & \ch_{L^\textnormal{par}} := \ch_L^0 \\
        & \ch_{L^\textnormal{vac}} := \ch_L^1 \\
        & \ch_{R^\textnormal {par}} := \ch_R^1  \\
        & \ch_{R^\textnormal{vac}} := \ch_R^0
    \end{split}
\end{equation}
Then the physically relevant output space $\ch^\textit{\rm prac}$ is the direct sum of the subspaces $\ch_L^i \otimes \ch_R^j$ for which the indices match. This is equivalent to saying that $W$ \textit{follows the route} given by the Kroenecker delta $\delta^{ij}$. This just means that it respects the equation
\begin{equation} \label{wroute}
    W = \sum_{ij} \delta^{ij} \cdot (\pi_L^i \otimes \pi_R^j) \circ W \, ,
\end{equation}
where the $\pi_L^i$ and $\pi_R^j$ are projectors onto $\ch_L^i$ and $\ch_R^j$ respectively. Now, we can formally represent our transformation as a unitary by adding indices to the outputs that represent the sectorisation of the Hilbert spaces, and decorating $W$ with the route matrix $\delta^{ij}$. This gives the following diagram.

\begin{equation} \label{sup channels 1}
    \begin{tikzpicture}
	\begin{pgfonlayer}{nodelayer}
		\node [style=large map] (0) at (0, 0) {$W$};
		\node [style=none] (1) at (-1, -2.25) {};
		\node [style=none] (2) at (-1, 2) {};
		\node [style=none] (3) at (0.75, -2.25) {};
		\node [style=none] (4) at (0.75, 2) {};
		\node [style=none] (5) at (-4, 0) {$\delta^{ij}$};
		\node [style=none] (16) at (-2, 1.5) {$L^{i}$};
		\node [style=none] (18) at (1.75, 1.5) {$R^{j}$};
		\node [style=none] (22) at (1.75, -1.75) {$T$};
		\node [style=none] (23) at (-2, -1.75) {$C$};
	\end{pgfonlayer}
	\begin{pgfonlayer}{edgelayer}
		\draw (2.center) to (1.center);
		\draw (4.center) to (3.center);
	\end{pgfonlayer}
\end{tikzpicture}
\end{equation}
The interpretation of this diagram is that we have a transformation $W$ which follows the route $\delta$. The route matrix tells us that $W$ only maps states to its so-called \textit{practical output space}, which is now defined via the route as $\ch_{\textrm{prac}} := \bigoplus_{ij} \delta^{ij} \, \ch_L^i \otimes \ch_R^j$. $W$ is unitary with respect to this output subspace, so the routed map $(\delta, W)$ is called a \textit{routed unitary} transformation, even though $W$ is not strictly speaking a unitary operator.

We can simplify the diagram (\ref{sup channels 1}) by introducing a shorthand called \textit{index-matching}. Since the effect of $\delta^{ij}$ is just to `match up' the value of the output indices in the practical output space, we can avoid writing the matrix explicitly and instead just match the output indices directly:
\begin{equation} \label{sup channels 2}
    \begin{tikzpicture}
	\begin{pgfonlayer}{nodelayer}
		\node [style=large map] (0) at (0, 0) {$W$};
		\node [style=none] (1) at (-1, -2.25) {};
		\node [style=none] (2) at (-1, 2) {};
		\node [style=none] (3) at (0.75, -2.25) {};
		\node [style=none] (4) at (0.75, 2) {};
		\node [style=none] (16) at (-2, 1.5) {$L^{i}$};
		\node [style=none] (18) at (1.75, 1.5) {$R^{i}$};
		\node [style=none] (22) at (1.75, -1.75) {$T$};
		\node [style=none] (23) at (-2, -1.75) {$C$};
	\end{pgfonlayer}
	\begin{pgfonlayer}{edgelayer}
		\draw (2.center) to (1.center);
		\draw (4.center) to (3.center);
	\end{pgfonlayer}
\end{tikzpicture} \, .
\end{equation}

We now explain the notion of a routed linear map in full generality. Given a transformation $U: \ch_A \rightarrow \ch_B$, we can construct a routed transformation $(\lambda, U)$ by first sectorising $U$'s input and output Hilbert spaces into a set of orthogonal subspaces:
\begin{equation}
    \begin{split}
        & \ch_A = \bigoplus_i \ch_A^i \\
        & \ch_B = \bigoplus_j \ch_B^j
    \end{split}
\end{equation}
We then specify a Boolean matrix $\lambda$, which is there to tell us which $\ch_A^i$ may be mapped to which $\ch_B^j$. We call this sort of restriction on the form of $U$ a \textit{sectorial constraint}.
Specifically, $U$ \textit{follows} $\lambda$ if it satisfies
\begin{equation} \label{route}
    U = \sum \lambda_i^j \cdot \pi^j_B \circ U \circ \mu_A^i
\end{equation}
where the $\pi_B^j$'s and $\mu_A^i$'s project onto sectors $\ch_B^j$ and $\ch_A^i$ respectively. For $(\lambda, U)$ to be a valid routed transformation, we need $U$ to follow $\lambda$. 

The routes can be composed in parallel and in sequence by the Cartesian product and matrix multiplication, respectively. Routed maps can be composed by composing their elements pairwise: the composition of the linear maps will necessarily follow the composition of the routes. This enables us to build up large routed circuits using elementary routed maps as building blocks.

The route $\lambda$ defines $U$'s practical input and output spaces $\bigoplus_i (\sum_j \lambda_i^j) \ch_A^i$ and \\ $\bigoplus_j (\sum_i \lambda_i^j) \ch_B^j$ respectively. $U$ is a routed unitary if the transformation is unitary when we restrict its definition to these spaces.

The last thing to introduce is the notion of a \textit{routed supermap} \cite{vanrietvelde2021coherent}. In a standard superunitary, any unitary operator that maps the ingoing space $\ch_{A^\inn}$ to the outgoing space $\ch_{A^\out}$ of one of the nodes is considered a valid input to that node. For a routed supermap, a node might be equipped with some route $\lambda$ given a sectorisation of its input and output spaces. Then, the only valid unitary transformations for that node are those that follow $\lambda$. Formally, the valid inputs are those that respect $U = \sum_{ij} \lambda_i^j \cdot \mu^j_{B} \circ U \circ \pi^i_{A}$, or  $U=\sum_{ij} \lambda_i^j \cdot (\mu^j_{B} \otimes I_X) \circ U \circ (\pi^i_{A} \otimes  I_X)$ if the unitary also acts on some ancillary system $X$.

\subsection{Extracting the relevant structure: routed circuit decomposition, skeletal supermap, routed graph}

Armed with an understanding of routed maps, we can now give the promised routed circuit decomposition of the switch. Luckily, all we really need is a routed unitary of the form $(\delta, W)$, represented in (\ref{sup channels 2}). 

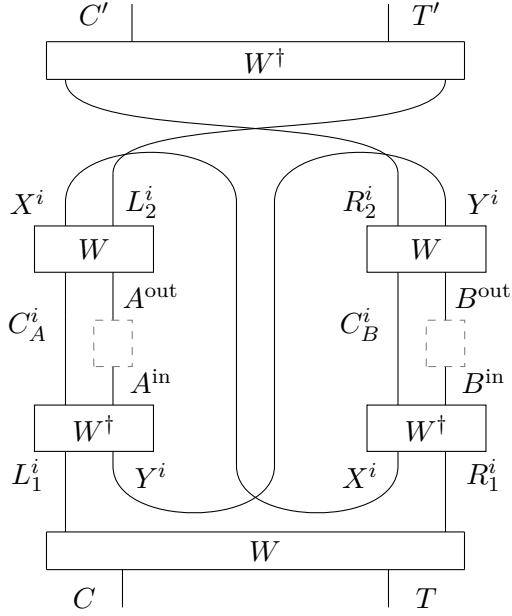
\begin{figure*}
    \centering
    \begin{tikzpicture}
	\begin{pgfonlayer}{nodelayer}
		\node [style=none] (1) at (5, 0.5) {};
		\node [style=none] (2) at (5, -0.5) {};
		\node [style=none] (3) at (-6, -0.5) {};
		\node [style=none] (4) at (-0.25, 0) {$W$};
		\node [style=none] (6) at (-5.5, 0.5) {};
		\node [style=none] (7) at (-5.5, 3) {};
		\node [style=none] (8) at (-4.25, 3.5) {};
		\node [style=none] (9) at (-4.25, 6.25) {};
		\node [style=none] (11) at (-4.25, 6.25) {};
		\node [style=none] (12) at (-4.25, 7.5) {};
		\node [style=none] (13) at (-4, -0.5) {};
		\node [style=none] (14) at (-4, -1.5) {};
		\node [style=none] (15) at (-5, -1.25) {$C$};
		\node [style=none] (21) at (4, 3.5) {};
		\node [style=none] (25) at (4.5, 0.5) {};
		\node [style=none] (26) at (4.5, 3) {};
		\node [style=none] (27) at (-6.5, 2) {$L_1^i$};
		\node [style=none] (28) at (5.5, 2) {$R_1^i$};
		\node [style=none] (29) at (-3.25, 4.5) {${A^\inn}$};
		\node [style=none] (30) at (-3.25, 6.75) {${A^\out}$};
		\node [style=none] (37) at (-6, 12.5) {};
		\node [style=none] (38) at (5, 12.5) {};
		\node [style=none] (39) at (5, 13.5) {};
		\node [style=none] (40) at (-6, 13.5) {};
		\node [style=none] (41) at (-0.25, 13) {$W^\dag$};
		\node [style=none] (42) at (-5.5, 12.5) {};
		\node [style=none] (44) at (3, 13.5) {};
		\node [style=none] (46) at (-4.25, 8.25) {};
		\node [style=none] (47) at (-4.25, 10) {};
		\node [style=none] (48) at (3.25, 10) {};
		\node [style=none] (50) at (-3.5, 9.25) {$L_2^i$};
		\node [style=none] (51) at (2.25, 9.25) {$R_2^i$};
		\node [style=none] (53) at (-5.5, 9.25) {};
		\node [style=none] (54) at (-1, 9.25) {};
		\node [style=none] (55) at (-1, 2.25) {};
		\node [style=none] (56) at (3.25, 2.25) {};
		\node [style=none] (57) at (3.25, 3) {};
		\node [style=none] (58) at (-5.5, 8.25) {};
		\node [style=none] (59) at (4.5, 8.25) {};
		\node [style=none] (60) at (4.5, 9.25) {};
		\node [style=none] (61) at (0, 9.25) {};
		\node [style=none] (62) at (0, 2.25) {};
		\node [style=none] (63) at (-4.25, 2.25) {};
		\node [style=none] (64) at (-4.25, 3) {};
		\node [style=none] (65) at (-3.25, 2) {$Y^i$};
		\node [style=none] (66) at (2.25, 2) {$X^i$};
		\node [style=none] (67) at (-6.5, 9.25) {$X^i$};
		\node [style=none] (68) at (5.5, 9.25) {$Y^i$};
		\node [style=none] (75) at (4.5, 12.5) {};
		\node [style=none] (77) at (3, -0.5) {};
		\node [style=none] (78) at (3, -1.5) {};
		\node [style=none] (79) at (4, -1.25) {$T$};
		\node [style=none] (80) at (-3.75, 14.5) {};
		\node [style=none] (81) at (-3.75, 13.5) {};
		\node [style=none] (82) at (-4.75, 14.25) {$C'$};
		\node [style=none] (83) at (3, 14.5) {};
		\node [style=none] (84) at (3, 13.5) {};
		\node [style=none] (85) at (4, 14.25) {$T'$};
		\node [style=large map] (91) at (-4.75, 8) {$W$};
		\node [style=large map] (92) at (-4.75, 3.25) {$W^\dag$};
		\node [style=large map] (93) at (4, 3.25) {$W^\dag$};
		\node [style=large map] (94) at (4, 8) {$W$};
		\node [style=none] (95) at (-5.5, 3.5) {};
		\node [style=none] (96) at (-5.5, 6.25) {};
		\node [style=none] (97) at (-5.5, 6.25) {};
		\node [style=none] (98) at (-5.5, 7.5) {};
		\node [style=none] (100) at (-6.5, 6) {$C^i_A$};
		\node [style=none] (101) at (4.5, 3.5) {};
		\node [style=none] (102) at (4.5, 6.25) {};
		\node [style=none] (103) at (4.5, 6.25) {};
		\node [style=none] (104) at (4.5, 7.5) {};
		\node [style=none] (105) at (3.25, 3.5) {};
		\node [style=none] (106) at (3.25, 6.25) {};
		\node [style=none] (107) at (3.25, 6.25) {};
		\node [style=none] (108) at (3.25, 7.5) {};
		\node [style=none] (110) at (2.25, 6) {$C^i_B$};
		\node [style=none] (111) at (5.5, 4.5) {${B^\inn}$};
		\node [style=none] (112) at (5.5, 6.75) {${B^\out}$};
		\node [style=none] (113) at (3.25, 8.25) {};
		\node [style=dashed map] (114) at (-4.25, 5.5) {};
		\node [style=dashed map] (115) at (4.5, 5.5) {};
		\node [style=none] (116) at (-6, 0.5) {};
	\end{pgfonlayer}
	\begin{pgfonlayer}{edgelayer}
		\draw (1.center) to (2.center);
		\draw (2.center) to (3.center);
		\draw (7.center) to (6.center);
		\draw (9.center) to (8.center);
		\draw (12.center) to (11.center);
		\draw (13.center) to (14.center);
		\draw (26.center) to (25.center);
		\draw (40.center) to (37.center);
		\draw (37.center) to (38.center);
		\draw (38.center) to (39.center);
		\draw (39.center) to (40.center);
		\draw (47.center) to (46.center);
		\draw [in=90, out=90] (53.center) to (54.center);
		\draw (54.center) to (55.center);
		\draw [in=-90, out=-90] (55.center) to (56.center);
		\draw (56.center) to (57.center);
		\draw (58.center) to (53.center);
		\draw (59.center) to (60.center);
		\draw [in=90, out=90] (60.center) to (61.center);
		\draw (61.center) to (62.center);
		\draw [in=-90, out=-90] (62.center) to (63.center);
		\draw (63.center) to (64.center);
		\draw (77.center) to (78.center);
		\draw (80.center) to (81.center);
		\draw (83.center) to (84.center);
		\draw (96.center) to (95.center);
		\draw (98.center) to (97.center);
		\draw (102.center) to (101.center);
		\draw (104.center) to (103.center);
		\draw (106.center) to (105.center);
		\draw (108.center) to (107.center);
		\draw (113.center) to (48.center);
		\draw (116.center) to (1.center);
		\draw (116.center) to (3.center);
		\draw [in=90, out=-90, looseness=0.50] (42.center) to (48.center);
		\draw [in=-90, out=90, looseness=0.50] (47.center) to (75.center);
	\end{pgfonlayer}
\end{tikzpicture}
    \caption{Routed circuit decomposition of the switch, using index matching. $W$ is the routed unitary defined in (\ref{sup channels def}). The wires bent into `cup' or `cap' shapes represent the (unnormalised) perfectly correlated entangled states \cite{coecke_kissinger_2017}. Overall, the `$i=0$' sectors correspond to the branch where Alice's intervention is implemented before Bob's, while `$i=1$' corresponds to the branch where Bob's intervention is implemented before Alice's. Thus the cycle is constructed from two acyclic components corresponding to definite orders of implementation.}
    \label{fig:switch routed circuit}
\end{figure*}

Our decomposition of the switch is presented in Figure \ref{fig:switch routed circuit}. The basic intuition behind the diagram lies in the following interpretation. When we prepare the control qubit in $\ket{0}$, the target system may only enter the sectors of the wires inside the diagram corresponding to $i=0$. Recalling that $\ch_L^0$ was a $d$-dimensional sector, and that $\ch_R^0$ was a trivial, one-dimensional `dummy sector', this means that the particle will exit via the left output port of every $W$ it enters. Meanwhile, the right output will receive a one-dimensional dummy system, analogous to the  vacuum in the interferometric example above (although this is merely a formal analogy – we are not committing to any
particular physical interpretation of this dummy system).
 This means that the particle will go through Alice's node first, then Bob's, then out to the future. The opposite is true when we prepare the control in $\ket{1}$.


We now want to nail down how the route structure in Figure \ref{fig:switch routed circuit} can be leveraged to certify that the supermap is a valid one. To do this, we first need to to consider a further pruned version of the circuit, in which only the essential information appears. This is given by what we call a \textit{skeletal supermap}: a supermap that includes nothing other than wires, without any boxes representing non-identity unitary transformations.
The idea is that we can obtain the original supermap from the skeletal supermap by `fleshing it out', i.e.\ inserting some unitary transformations into the nodes.
If we can show that this skeletal supermap is a valid superunitary, then it follows immediately that our original supermap is a valid superunitary.

\begin{figure}
    \centering
    \begin{tikzpicture}
	\begin{pgfonlayer}{nodelayer}
		\node [style=none] (0) at (-6, 1.25) {};
		\node [style=none] (6) at (-5.5, 1.25) {};
		\node [style=none] (7) at (-5.5, 4.25) {};
		\node [style=none] (13) at (-4, 0.25) {};
		\node [style=none] (14) at (-4, -0.75) {};
		\node [style=none] (15) at (-4.5, -0.5) {$C$};
		\node [style=none] (25) at (4.5, 1.25) {};
		\node [style=none] (26) at (4.5, 4.25) {};
		\node [style=none] (27) at (-6.5, 3.25) {$L_1^i$};
		\node [style=none] (28) at (5.5, 3.25) {$R_1^i$};
		\node [style=none] (37) at (-6, 12.5) {};
		\node [style=none] (39) at (5, 13.5) {};
		\node [style=none] (42) at (-5.5, 12.5) {};
		\node [style=none] (44) at (3, 13.5) {};
		\node [style=none] (46) at (-4.25, 7.5) {};
		\node [style=none] (47) at (-4.25, 10) {};
		\node [style=none] (48) at (3.25, 10) {};
		\node [style=none] (50) at (-3.25, 8.75) {$L_2^i$};
		\node [style=none] (51) at (2.25, 8.75) {$R_2^i$};
		\node [style=none] (53) at (-5.5, 8.75) {};
		\node [style=none] (54) at (-1, 8.75) {};
		\node [style=none] (55) at (-1, 3) {};
		\node [style=none] (56) at (3.25, 3) {};
		\node [style=none] (57) at (3.25, 4.25) {};
		\node [style=none] (58) at (-5.5, 7.5) {};
		\node [style=none] (59) at (4.5, 7.5) {};
		\node [style=none] (60) at (4.5, 8.75) {};
		\node [style=none] (61) at (0, 8.75) {};
		\node [style=none] (62) at (0, 3) {};
		\node [style=none] (63) at (-4.25, 3) {};
		\node [style=none] (64) at (-4.25, 4.25) {};
		\node [style=none] (65) at (-3.25, 3.25) {$Y^i$};
		\node [style=none] (66) at (2.25, 3.25) {$X^i$};
		\node [style=none] (67) at (-6.5, 8.75) {$X^i$};
		\node [style=none] (68) at (5.5, 8.75) {$Y^i$};
		\node [style=none] (75) at (4.5, 12.5) {};
		\node [style=none] (77) at (3, 0.25) {};
		\node [style=none] (78) at (3, -0.75) {};
		\node [style=none] (79) at (3.5, -0.5) {$T$};
		\node [style=none] (80) at (-4, 14.5) {};
		\node [style=none] (81) at (-4, 13.5) {};
		\node [style=none] (82) at (-4.5, 14.25) {$C'$};
		\node [style=none] (83) at (3, 14.5) {};
		\node [style=none] (84) at (3, 13.5) {};
		\node [style=none] (85) at (3.5, 14.25) {$T'$};
		\node [style=none] (113) at (3.25, 7.5) {};
		\node [style=none] (114) at (-6, 4.25) {};
		\node [style=none] (115) at (-3.75, 4.25) {};
		\node [style=none] (116) at (-3.75, 7.5) {};
		\node [style=none] (117) at (-6, 7.5) {};
		\node [style=none] (118) at (2.75, 4.25) {};
		\node [style=none] (119) at (2.75, 7.5) {};
		\node [style=none] (120) at (5, 4.25) {};
		\node [style=none] (121) at (5, 7.5) {};
		\node [style=none] (122) at (-5, 6) {$A$};
		\node [style=none] (123) at (3.75, 6) {$B$};
		\node [style=none] (124) at (-6, 0.25) {};
		\node [style=none] (125) at (5, 0.25) {};
		\node [style=none] (126) at (5, 1.25) {};
		\node [style=none] (127) at (-6, 1.25) {};
		\node [style=none] (128) at (-6, 13.5) {};
		\node [style=none] (129) at (-6, 12.5) {};
		\node [style=none] (130) at (-6, 13.5) {};
		\node [style=none] (131) at (5, 12.5) {};
		\node [style=none] (132) at (-0.5, 0.75) {$P$};
		\node [style=none] (133) at (-0.5, 0.75) {};
		\node [style=none] (134) at (-0.5, 13) {$F$};
	\end{pgfonlayer}
	\begin{pgfonlayer}{edgelayer}
		\draw (7.center) to (6.center);
		\draw (13.center) to (14.center);
		\draw (26.center) to (25.center);
		\draw (47.center) to (46.center);
		\draw [bend left=90] (53.center) to (54.center);
		\draw (54.center) to (55.center);
		\draw [bend right=90] (55.center) to (56.center);
		\draw (56.center) to (57.center);
		\draw (58.center) to (53.center);
		\draw (59.center) to (60.center);
		\draw [bend right=90] (60.center) to (61.center);
		\draw (61.center) to (62.center);
		\draw [bend left=90] (62.center) to (63.center);
		\draw (63.center) to (64.center);
		\draw (77.center) to (78.center);
		\draw (80.center) to (81.center);
		\draw (83.center) to (84.center);
		\draw (113.center) to (48.center);
		\draw [style=dashed, gray] (117.center) to (114.center);
		\draw [style=dashed, gray] (114.center) to (115.center);
		\draw [style=dashed, gray] (115.center) to (116.center);
		\draw [style=dashed, gray] (116.center) to (117.center);
		\draw [style=dashed, gray] (119.center) to (118.center);
		\draw [style=dashed, gray] (121.center) to (120.center);
		\draw [style=dashed, gray] (121.center) to (119.center);
		\draw [style=dashed, gray] (118.center) to (120.center);
		\draw [style=dashed, gray] (127.center) to (126.center);
		\draw [style=dashed, gray] (126.center) to (125.center);
		\draw [style=dashed, gray] (125.center) to (124.center);
		\draw [style=dashed, gray] (127.center) to (124.center);
		\draw [style=dashed, gray] (130.center) to (129.center);
		\draw [style=dashed, gray] (131.center) to (39.center);
		\draw [style=dashed, gray] (39.center) to (130.center);
		\draw [style=dashed, gray] (129.center) to (131.center);
		\draw [in=270, out=90, looseness=0.50] (48.center) to (42.center);
		\draw [in=-90, out=90, looseness=0.50] (47.center) to (75.center);
	\end{pgfonlayer}
\end{tikzpicture}
    \caption{Skeletal supermap for the switch. The nodes suffer sectorial constraints represented by the index-matching of the input and output wires, making this a \textit{routed} supermap. The routed circuit for the switch in Figure \ref{fig:switch routed circuit} is obtained by inserting unitary transformations into the nodes $P$ and $F$, and the monopartite superunitary (\ref{comb}) into $A$ and $B$.}
    \label{fig:switch skeletal}
\end{figure}
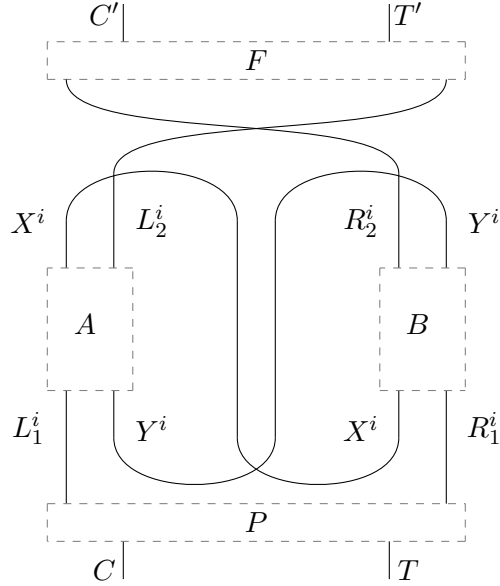

A suitable skeletal supermap for the switch is represented in Figure \ref{fig:switch skeletal}. Inserting $W$ and $W^\dag$ into the bottom and top nodes respectively, and inserting the superunitary 
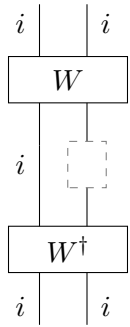
\begin{equation}\label{comb}
    \begin{tikzpicture}
	\begin{pgfonlayer}{nodelayer}
		\node [style=none] (0) at (1.5, 3.25) {};
		\node [style=none] (1) at (2.75, 1.5) {};
		\node [style=none] (2) at (2.75, 5.25) {};
		\node [style=none] (3) at (2.75, 6.25) {};
		\node [style=none] (4) at (2.75, 10) {};
		\node [style=none] (6) at (2.75, 8.25) {};
		\node [style=none] (7) at (1.5, 8.25) {};
		\node [style=none] (8) at (2.75, 3.25) {};
		\node [style=large map] (9) at (2.25, 8) {$W$};
		\node [style=large map] (10) at (2.25, 3.5) {$W^\dag$};
		\node [style=none] (11) at (1.5, 1.5) {};
		\node [style=none] (12) at (1.5, 6.25) {};
		\node [style=none] (13) at (1.5, 6.25) {};
		\node [style=none] (14) at (1.5, 10) {};
		\node [style=none] (15) at (1, 5.75) {$i$};
		\node [style=none] (16) at (1, 9.5) {$i$};
		\node [style=none] (17) at (3.25, 9.5) {$i$};
		\node [style=none] (18) at (1, 2) {$i$};
		\node [style=none] (19) at (3.25, 2) {$i$};
		\node [style=dashedmap] (20) at (2.75, 5.75) {};
	\end{pgfonlayer}
	\begin{pgfonlayer}{edgelayer}
		\draw (2.center) to (1.center);
		\draw (4.center) to (3.center);
		\draw (12.center) to (11.center);
		\draw (14.center) to (13.center);
	\end{pgfonlayer}
\end{tikzpicture}
\end{equation}
into each of the middle nodes yields \texttt{SWITCH}.

We can represent the skeletal supermap using an even simpler graph.
All we need to consider is the connectivity between the nodes, the routes and indices, and the specific index values that represent one-dimensional sectors. 
A representation of all this information, and nothing more, is provided by the \textit{routed graph}. This consists of 
\begin{itemize}
\item a vertex for each node in the skeletal supermap, decorated with its route;
\item arrows representing the wires connecting the nodes in the skeletal supermap;
\item next to each arrow, the index of the corresponding wire;
\item next to each arrow, the specific values of its index that corresponds to a one-dimensional sector.
\end{itemize}

When a node has a `delta route' -- that is, a route that is equal to 1 if and only if all the indices take the same value -- we can adopt the shorthand index-matching representation where we decorate each of its ingoing and outgoing arrows with the same index. 

The routed graph for the switch's skeletal counterpart is given in Figure \ref{fig:rswitch routed graph}, with and without the index-matching shorthand. Remarkably, this elementary object contains all the information we need to confirm that the switch is a valid superunitary, or in other words, that it is consistent.

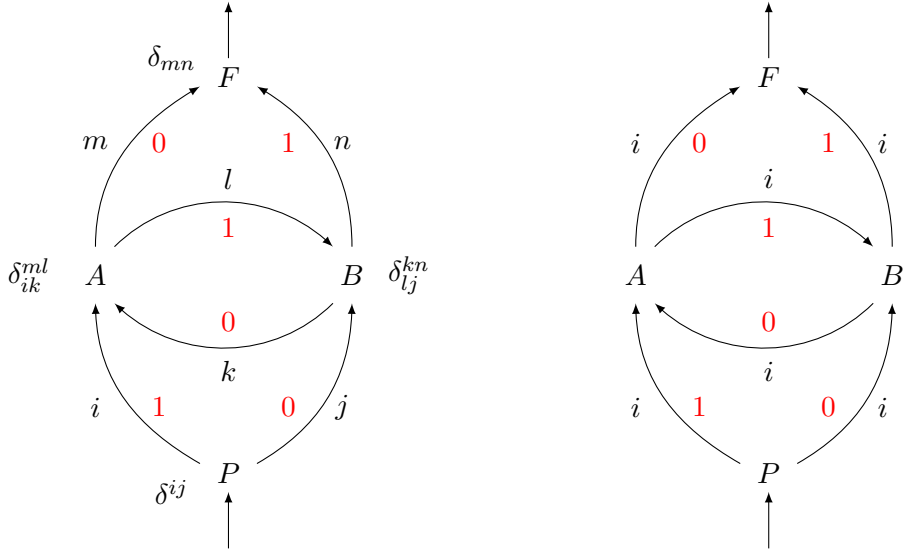
\begin{figure} 
\centering
    \subfloat{\begin{tikzpicture}
	\begin{pgfonlayer}{nodelayer}
		\node [style=none] (0) at (-3.5, 0) {$A$};
		\node [style=none] (1) at (3.25, 0) {$B$};
		\node [style=none] (2) at (0, -5.25) {$P$};
		\node [style=none] (3) at (0, -7.25) {};
		\node [style=none] (4) at (0, 5.25) {$F$};
		\node [style=none] (5) at (0, 7.25) {};
		\node [style=none] (6) at (0, -5.75) {};
		\node [style=none] (7) at (0.75, -5) {};
		\node [style=none] (8) at (-0.75, -5) {};
		\node [style=none] (9) at (-0.75, 5) {};
		\node [style=none] (10) at (0.75, 5) {};
		\node [style=none] (11) at (-3.5, 0.75) {};
		\node [style=none] (12) at (-3, 0.75) {};
		\node [style=none] (13) at (-3.5, -0.75) {};
		\node [style=none] (14) at (-3, -0.75) {};
		\node [style=none] (15) at (3.25, -0.75) {};
		\node [style=none] (16) at (3.25, 0.75) {};
		\node [style=none] (17) at (2.75, 0.75) {};
		\node [style=none] (18) at (2.75, -0.75) {};
		\node [style=none] (19) at (3, -3.5) {$j$};
		\node [style=none] (20) at (-3.5, -3.5) {$i$};
		\node [style=none] (21) at (3, 3.5) {$n$};
		\node [style=none] (22) at (-3.5, 3.5) {$m$};
		\node [style=none] (23) at (0, 5.75) {};
		\node [style=none] (24) at (1.575, 3.5) {$\textcolor{red}{1}$};
		\node [style=none] (25) at (-1.825, 3.5) {$\textcolor{red}{0}$};
		\node [style=none] (26) at (-1.825, -3.5) {$\textcolor{red}{1}$};
		\node [style=none] (27) at (1.575, -3.5) {$\textcolor{red}{0}$};
		\node [style=none] (28) at (0, 1.25) {$\textcolor{red}{1}$};
		\node [style=none] (29) at (0, -1.25) {$\textcolor{red}{0}$};
		\node [style=none] (30) at (0, 2.55) {$l$};
		\node [style=none] (31) at (0, -2.55) {$k$};
		\node [style=none] (32) at (-1.5, -5.75) {$\delta^{ij}$};
		\node [style=none] (33) at (-5.25, 0) {$\delta_{ik}^{ml}$};
		\node [style=none] (34) at (4.75, 0) {$\delta_{lj}^{kn}$};
		\node [style=none] (35) at (-1.5, 5.75) {$\delta_{mn}$};
	\end{pgfonlayer}
	\begin{pgfonlayer}{edgelayer}
		\draw [style=arrow plain] (3.center) to (6.center);
		\draw [style=arrow plain, in=-90, out=150] (8.center) to (13.center);
		\draw [style=arrow plain, in=-90, out=30] (7.center) to (15.center);
		\draw [style=arrow plain, in=-30, out=90] (16.center) to (10.center);
		\draw [style=arrow plain, in=-150, out=90] (11.center) to (9.center);
		\draw [style=arrow plain, bend left=45] (12.center) to (17.center);
		\draw [style=arrow plain, bend left=45] (18.center) to (14.center);
		\draw [style=arrow plain] (23.center) to (5.center);
	\end{pgfonlayer}
\end{tikzpicture}}
    \qquad
    \qquad
    \qquad
    \subfloat{\begin{tikzpicture}
	\begin{pgfonlayer}{nodelayer}
		\node [style=none] (0) at (-3.5, 0) {$A$};
		\node [style=none] (1) at (3.25, 0) {$B$};
		\node [style=none] (2) at (0, -5.25) {$P$};
		\node [style=none] (3) at (0, -7.25) {};
		\node [style=none] (4) at (0, 5.25) {$F$};
		\node [style=none] (5) at (0, 7.25) {};
		\node [style=none] (6) at (0, -5.75) {};
		\node [style=none] (7) at (0.75, -5) {};
		\node [style=none] (8) at (-0.75, -5) {};
		\node [style=none] (9) at (-0.75, 5) {};
		\node [style=none] (10) at (0.75, 5) {};
		\node [style=none] (11) at (-3.5, 0.75) {};
		\node [style=none] (12) at (-3, 0.75) {};
		\node [style=none] (13) at (-3.5, -0.75) {};
		\node [style=none] (14) at (-3, -0.75) {};
		\node [style=none] (15) at (3.25, -0.75) {};
		\node [style=none] (16) at (3.25, 0.75) {};
		\node [style=none] (17) at (2.75, 0.75) {};
		\node [style=none] (18) at (2.75, -0.75) {};
		\node [style=none] (19) at (3, -3.5) {$i$};
		\node [style=none] (20) at (-3.5, -3.5) {$i$};
		\node [style=none] (21) at (3, 3.5) {$i$};
		\node [style=none] (22) at (-3.5, 3.5) {$i$};
		\node [style=none] (23) at (0, 5.75) {};
		\node [style=none] (24) at (1.575, 3.5) {$\textcolor{red}{1}$};
		\node [style=none] (25) at (-1.825, 3.5) {$\textcolor{red}{0}$};
		\node [style=none] (26) at (-1.825, -3.5) {$\textcolor{red}{1}$};
		\node [style=none] (27) at (1.575, -3.5) {$\textcolor{red}{0}$};
		\node [style=none] (28) at (0, 1.25) {$\textcolor{red}{1}$};
		\node [style=none] (29) at (0, -1.25) {$\textcolor{red}{0}$};
		\node [style=none] (30) at (0, 2.55) {$i$};
		\node [style=none] (31) at (0, -2.55) {$i$};
	\end{pgfonlayer}
	\begin{pgfonlayer}{edgelayer}
		\draw [style=arrow plain] (3.center) to (6.center);
		\draw [style=arrow plain, in=-90, out=150] (8.center) to (13.center);
		\draw [style=arrow plain, in=-90, out=30] (7.center) to (15.center);
		\draw [style=arrow plain, in=-30, out=90] (16.center) to (10.center);
		\draw [style=arrow plain, in=-150, out=90] (11.center) to (9.center);
		\draw [style=arrow plain, bend left=45] (12.center) to (17.center);
		\draw [style=arrow plain, bend left=45] (18.center) to (14.center);
		\draw [style=arrow plain] (23.center) to (5.center);
	\end{pgfonlayer}
\end{tikzpicture}}

\caption{Routed graph for the switch, with and without index-matching. The vertices represent nodes from the skeletal supermap. On the left, each wire is equipped with its own index, and the numbers in red denote the values of the index that correspond to a one-dimensional sector. Each node is decorated with a $\delta$ matrix representing the route, which is equal to 1 if and only if all of its arguments are equal. Lower indices refer to input wires, upper indices refer to output wires. Since all the routes are `delta' routes, we can use the convenient shorthand of index-matching to produce a simpler diagram with the same meaning, as on the right.}
\label{fig:rswitch routed graph}
\end{figure}

\subsection{Checking for validity}\label{sec:checking_validity}

In our framework, one can just consider the routed graph depicting the connectivity of the supermap, and infer from it that the supermap is valid. This amounts to checking that the routed graph conforms to a couple of principles. Here we shall present these principles and the way to check them in a pedagogical manner, taking advantage of the relative simplicity of the switch's case.

To motivate these principles, a good place to start is with the intuition that in a self-consistent protocol, information should not genuinely be able to flow in a circle. This is because, if it did, then at any point on the circle we could control the outgoing information on incoming information that is inconsistent with it. This happens in the grandfather paradox, where Alice's grandfather is killed if Alice exists, even though Alice's existence is incompatible with his murder.

Yet from the present routed graph it seems as if information does flow in a circle between $A$ and $B$. What we need to do is use the information in the graph to obtain a more fine-grained perspective from which the cycle disappears (or is at least shown to be harmless).

We start by fine-graining each node into a number of \textit{branches}. If the route of a node dictates that there are exactly $n$ disjoint subspaces of the input space that must be mapped one-to-one to $n$ disjoint regions of the output space, we say that there are $n$ branches. 

To make this clear, we can represent the route matrices as diagrams with arrows from the input sectors $\ch^i_\inn$ to the output sectors $\ch^j_\out$ being present when the corresponding route matrix element $\lambda_i^j$ is equal to $1$. For the node $A$, we have a route of the form $\delta_i^j$, which is represented in Figure \ref{fig:branches}. In this sort of diagram, each disconnected `island', circled in red, corresponds to a distinct branch. Thus $A$ has two distinct branches, which we label $A^i$ in correspondence with the value of $i$. On the other hand, although the node $P$'s outgoing space has two sectors, $P$ only has one branch, since its graph is fully connected, as represented in Figure \ref{fig:branches2}.

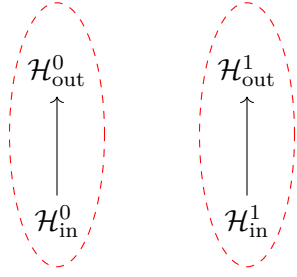
\begin{figure}
    \centering
\begin{tikzpicture} 
\node (1) at (0, 0) {$\ch_\inn^0$}; 
\node (2) at (5, 0) {$\ch_\inn^1$};
\node (3) at (0, 4) {$\ch_\out^0$}; 
\node (4) at (5, 4) {$\ch_\out^1$}; 
\draw[->] (1) to (3);
\draw[->] (2) to (4);
\draw[red, dashed] (0cm,2.3cm) ellipse[x radius=1.25,y radius=3.5];
\draw[red, dashed] (5cm,2.3cm) ellipse[x radius=1.25,y radius=3.5];
\end{tikzpicture}
    \caption{The route for the $A$ node of the skeletal supermap.}
    \label{fig:branches}
\end{figure}

 \begin{figure}
    \centering
\begin{tikzpicture} 
\node (1) at (2.5, 0) {$\ch_\inn$}; 
\node (3) at (0, 4) {$\ch_\out^0$}; 
\node (4) at (5, 4) {$\ch_\out^1$}; 
\draw[->] (1) to (3);
\draw[->] (1) to (4);
\draw[red, dashed] (2.5, 2.7) ellipse[x radius=3.9, y radius=3.9];
\end{tikzpicture}
    \caption{The route for the $P$ node of the skeletal supermap.}
    \label{fig:branches2}
\end{figure}
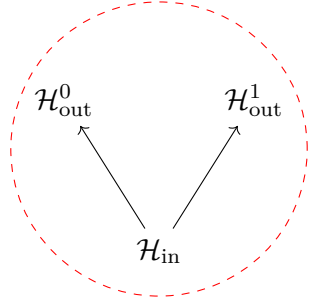

Intuitively, branches correspond to \textit{alternatives} in a node: for example, in node $A$, either the branch $A^0$ or the branch $A^1$ will happen.\footnote{Of course, because we are in quantum theory, both could happen in a superposition. But a remarkable feature of our framework is that, in order to check the validity of the routed graph, we do not have to consider superpositions: it is sufficient to reason as if the branches were mutually exclusive. Therefore this is what we will do through this section.} It is these branches, rather than the original nodes, that we will study. In particular, we will check whether the branches themselves form  informational loops; this will be the subject of our second principle.

\subsubsection{Bifurcation choices and bi-univocality}
Before proceeding to investigate informational loops, we first need to check another, more fundamental, property. As we remarked, branches correspond to alternatives. Thus, roughly speaking, we should make sure that the route structure is detailed enough to specify that exactly one branch happens at each node. This is captured by the notion of a \textit{bifurcation choice}.

Let us introduce it with an example, resorting to agents for intuition. In the route for the $P$ node in Figure \ref{fig:branches2}, the ingoing space may be mapped to two different sectors of the outgoing space, meaning that an agent can `choose' to send information to just one of these sectors. More generally, an agent at a node can make a `bifurcation choice' for each branch of that node (for the branches that contain only one output value, the choice is trivial). 

In the routed graph of Figure \ref{fig:rswitch routed graph}, only $P$ features a bifurcation choice. Furthermore, this bifurcation choice amounts to picking the value of the index $i$ through the graph; thus, if the agent at $P$ picks, say, $i=0$, this leads (through the other routes) to that value getting instantiated through the graph, and therefore to the branches $A^0$ and $B^0$ `happening'. The symmetric situation happens for the $i=1$ choice. In other words, each possible bifurcation choice determines exactly one branch to happen at every node.

It is this behaviour, rather elementary in the case of the switch, that we want to ask for in general. This leads to a principle that we will call \textit{univocality}:\footnote{This is a shameless gallicism. `Univocal' means `speaking with one voice', i.e., yielding exactly one output. For instance, functions are univocal, while relations (represented here by Boolean matrices) are generically equivocal.} any tuple of choices made at every branch leads to exactly one branch happening at every node. In other words, once the agents at the nodes of our skeletal supermap make all their bifurcation choices, there is a determinate fact, for each branch, about whether the quantum state will pass through it. More formally, this will be defined as the fact that the routed graph defines a function (as opposed to a relation) from bifurcation choices to `branch statuses', where branch statuses are bits representing whether a given branch has happened or not. (Section \ref{sec: theorem} describes how this function is defined.) This can be seen as forbidding situations where bifurcation choices would either underdetermine branch statuses (i.e.\ lead to several possible branch assignments) or overdetermine them (i.e.\ lead to no possible assignment at all).\footnote{On the relationship between underdetermination and overdetermination in cyclic processes, see Ref.\ \cite{Baumeler_2021}.}

For the switch, this is satisfied because the bifurcation choice at the $P$ node of the skeletal supermap determines which branches of $A$ and $B$ we end up in. This corresponds to the fact that in \texttt{SWITCH}, the logical state of the control qubit fixes the causal order (recalling the fact that the causal order is what defined the different branches of $A$ and $B$).

We also require that the `time-reversed' routed graph, obtained by reversing the direction of the arrows on the original routed graph, satisfies univocality as well. 
This is satisfied by the switch, corresponding to the fact that the information about which causal order took place ends up recorded in the control qubit at the end of the protocol. 
If both the routed graph and its time-reversed version satisfy univocality, we say that the routed graph satisfies \textbf{bi-univocality}.
Thus the entire bi-univocality condition is satisfied by \texttt{SWITCH}. We summarise the condition as follows:
\begin{quote}
\textit{Bi-univocality:} The routed graph and the time-reversed routed graph define functions from bifurcation choices to branch statuses.
\end{quote}

\subsubsection{The branch graph and weak loops} \label{sec: switch branch graph}

We now turn to our second principle, which deals with whether influences between branches flow in a circle. To check this, we construct a directed `branch' graph representing the flow of information between different branches in the routed graph, depicted in Figure \ref{fig:branch graph}. Causal/informational loops among the branches will be understood as loops in this graph.

The branch graph contains solid, dashed green, and dashed red arrows. The solid arrows represent the flow of quantum information along `paths' in the routed graph permitted by the routes, while the dashed arrows represent the flow of information via choices of which path to follow, when multiple paths are permitted by the routes. We explain each of these in turn at an intuitive level; the general formal procedure for constructing the branch graph from a routed graph is described in Section \ref{sec: theorem}.

To understand the solid arrows, note that there are two possible joint value assignments to all of the indices in the routed graph: either $i=0$ everywhere, or $i=1.$ For the $i=0$ assignment, the arrows $P \rightarrow B$ and $B \rightarrow A$ in Figure \ref{fig:rswitch routed graph} correspond to one-dimensional sectors, as indicated by the red zeroes. What this shows is that no quantum information flows from $P$ to the branch $B^0$ or from $B^0$ to $A^0$. For this reason, there are no solid arrows $P \rightarrow B^0$ or $B^0 \rightarrow A^0$ in the branch graph. On the other hand, quantum information does flow from $P$ into the branch $A^0$, then into $B^0$, and then finally into $F$. Thus we have the path $P \rightarrow A^0 \rightarrow B^0 \rightarrow F$ of solid arrows in the branch graph.  By following precisely analogous reasoning for the $i=1$ assignment, we arrive at the solid arrows in Figure \ref{fig:branch graph}. Evidently, the solid arrows in the branch graph do not form a loop\footnote{We note that this corresponds to an observation from Ref.\ \cite{barrett2021cyclic}, that, although the switch has a cyclic causal structure, it can still be written as a direct sum of (pure) processes with a definite causal order. We want to stress however that such an observation is in general \textit{not} sufficient to ensure the consistency of the process, as it overlooks the need to 1) check bi-univocality, and 2) also represent dashed arrows in the branch graph.}.

\begin{figure}
    \centering
\begin{tikzpicture}
	\begin{pgfonlayer}{nodelayer}
		\node [style=none] (0) at (-4, -2.75) {};
		\node [style=none] (34) at (-1, -5) {};
		\node [style=none] (51) at (-3.5, -2.5) {};
		\node [style=none] (64) at (-0.75, -4.5) {};
		\node [style=none] (65) at (0, -5) {$P$};
		\node [style=none] (66) at (-4.25, -2) {$A^0$};
		\node [style=none] (67) at (4.25, -2) {$B^1$};
		\node [style=none] (68) at (-4.25, 2) {$B^0$};
		\node [style=none] (69) at (4.25, 2) {$A^1$};
		\node [style=none] (70) at (0, 5) {$F$};
		\node [style=none] (71) at (4, -2.75) {};
		\node [style=none] (72) at (1, -5) {};
		\node [style=none] (73) at (3.5, -2.5) {};
		\node [style=none] (74) at (0.75, -4.5) {};
		\node [style=none] (75) at (-4.25, 1.25) {};
		\node [style=none] (76) at (-4.25, -1.25) {};
		\node [style=none] (77) at (4.25, 1.25) {};
		\node [style=none] (78) at (4.25, -1.25) {};
		\node [style=none] (79) at (-4, 2.75) {};
		\node [style=none] (80) at (-1, 5) {};
		\node [style=none] (81) at (-3.5, 2.5) {};
		\node [style=none] (82) at (-0.75, 4.5) {};
		\node [style=none] (83) at (4, 2.75) {};
		\node [style=none] (84) at (1, 5) {};
		\node [style=none] (85) at (3.5, 2.5) {};
		\node [style=none] (86) at (0.75, 4.5) {};
		\node [style=none] (87) at (-3.5, 1.5) {};
		\node [style=none] (88) at (-0.25, -4.25) {};
		\node [style=none] (89) at (3.5, 1.5) {};
		\node [style=none] (90) at (0.25, -4.25) {};
		\node [style=none] (91) at (0.25, 4.25) {};
		\node [style=none] (92) at (3.5, -1.5) {};
		\node [style=none] (93) at (-0.25, 4.25) {};
		\node [style=none] (94) at (-3.5, -1.5) {};
	\end{pgfonlayer}
	\begin{pgfonlayer}{edgelayer}
		\draw [style=green dashed arrow, bend left=15] (34.center) to (0.center);
		\draw [style=arrow] (64.center) to (51.center);
		\draw [style=green dashed arrow, bend right=15] (72.center) to (71.center);
		\draw [style=arrow] (74.center) to (73.center);
		\draw [style=arrow] (76.center) to (75.center);
		\draw [style=arrow] (78.center) to (77.center);
		\draw [style=red dashed arrow, bend left=15] (79.center) to (80.center);
		\draw [style=arrow] (81.center) to (82.center);
		\draw [style=red dashed arrow, bend right=15] (83.center) to (84.center);
		\draw [style=arrow] (85.center) to (86.center);
		\draw [style=green dashed arrow, bend right=15] (88.center) to (87.center);
		\draw [style=green dashed arrow, bend left=15] (90.center) to (89.center);
		\draw [style=red dashed arrow, bend left=15] (92.center) to (91.center);
		\draw [style=red dashed arrow, bend right=15] (94.center) to (93.center);
	\end{pgfonlayer}
\end{tikzpicture}
    \caption{The branch graph for the routed graph in Figure \ref{fig:rswitch routed graph}. Each vertex represents the branch of some node in the routed graph. The branches of $A$ and $B$ are labelled with superscripts  corresponding to the relevant value of $i$ in the index-matching routed graph. The other nodes have only one branch, and we denote this branch with the same letter we used for the original nodes. The solid arrows are attributed by considering the connection between the branches encoded in the routed graph; the dashed green and red arrows represent relations of functional dependence in the functions from bifurcation choices to branch statuses required by bi-univocality. The graph contains no cycles of any kind, so it trivially satisfies the weak loops condition.}
    \label{fig:branch graph}
\end{figure}
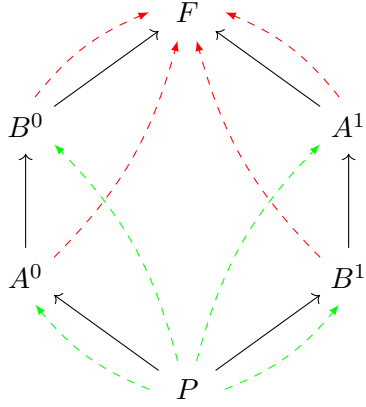

To rule out informational loops, it is necessary that the solid arrows do not form a loop, but it is not sufficient. What the lack of this kind of loop shows is that the quantum information confined within particular branches by the routes does not flow in a circle. But there is another type of information flowing in the routed circuit: the information that determines \textit{which branch happens}. This information is represented by the dashed arrows in Figure \ref{fig:branch graph}. It is possibilistic in nature, and can therefore be captured entirely using routes, based on the theory of finite relations (i.e.\ Boolean matrices).

Fortunately, if univocality is satisfied, then we already know that the routed graph defines a function from bifurcation choices to the statuses of the branches (i.e.\ the binary variables encoding whether each branch happened). We can thus define the green dashed arrows as representing the functional dependencies within this function. Namely, there is a green dashed arrow from $X^\alpha$ to $Y^\beta$ just in case the branch status at $Y^\beta$ depends on the bifurcation choice at $X^\alpha$. For example, there is a green dashed arrow from $P$ to $A^0$ because one can choose whether or not $A^0$ happens by choosing which logical state to prepare the control qubit in at $P$.  If a similar influence relation holds from $Y^\beta$ to $X^\alpha$ in the time-reversed version of the protocol, then we draw a red dashed arrow from $X^\alpha$ to $Y^\beta$.  Doing this for all the branches gives us the dashed arrows in Figure \ref{fig:branch graph}. 

The full branch graph gives a complete account of the flow of information in the skeletal supermap of the switch. It represents both the quantum information that flows within the branches with the solid arrows, and the `which-branch' information that is affected by bifurcation choices. This second sort of information can be thought of classically, since it corresponds to preferred sectorisations of the state spaces. We also call it possibilistic, since it is purely about the binary question of whether a branch does or does not happen given certain bifurcation choices, and can accordingly be represented by the routes using the theory of relations, represented by Boolean matrices. 

From this fine-grained perspective, it is clear that no information actually flows in a loop in \texttt{SWITCH}, since the branch graph of Figure \ref{fig:branch graph} satisfies

\begin{quote}
    \textit{No loops:} There are no directed loops in the branch graph.
\end{quote}

According to the upcoming Theorem \ref{thm: Main} -- the main theorem of this paper --, the satisfaction of bi-univocality and no loops is enough to demonstrate the validity of the skeletal supermap, and hence of the switch itself. Thus the logical consistency of the switch is a consequence of the satisfaction of these principles.

Remarkably though, a principle logically weaker than no loops is enough to ensure the validity of the supermap. We did not need to show that the routed graph contains no loops at all, but only loops of a weak, and harmless, type. Specifically, we needed to show that the graph satisfies the following principle, which we call \textbf{weak loops}.
\begin{quote}
\textit{Weak loops:} Any given loop in the branch graph is entirely made up of dashed arrows of a single colour.
\end{quote}

Our main theorem states that bi-univocality together with weak loops implies  that a skeletal supermap is valid, and hence that any associated protocol is self-consistent. While all protocols we have studied that do not violate causal inequalities satisfy no loops, Section \ref{sec:lugano} will show that the Lugano process has green loops (see Figure \ref{fig: Lugano branch graph}). This will lead us to conjecture that  \textit{the presence of monochromatic loops is precisely what enables the violation of causal inequalities}.




\changes{
Some readers might be familiar with another routed circuit for the quantum switch, namely the one in Figure 7 of \cite{barrett2021cyclic}. Let us briefly explain why we have used a different circuit here. Suppose we started with Figure 7 of \cite{barrett2021cyclic} and followed the same procedure that we have followed in this section, removing the unitaries in the circuit to obtain a skeletal graph, and then using the skeletal supermap to write down a routed graph. That routed graph is \textit{not} a valid one. 

The fastest way to see this is to note that in Figure 7 of \cite{barrett2021cyclic}, there is a loop made up entirely of wires that have no index, and hence a trivial sectorization. This will result a loop made up entirely of black wires in the routed graph, and thus a violation of the weak loops condition. Of course, the routed circuit in \cite{barrett2021cyclic} describes a perfectly valid process -- namely the switch -- but the process is only valid because of the \textit{specific} unitaries used in the circuit, and not as an inevitable consequence of the way that those unitaries are combined.}

\subsection{Why do we need bi-univocality?}


Naively, one might imagine that the lack of a causal/informational loop among the branches is enough to guarantee that a protocol is consistent. In this subsection, we explain why this intuition fails. 

\changes{Let us first explain why univocality is important.}
To this end, consider the supermap in Figure \ref{fig:grandfather smap}. A single wire is bent round in a loop, and serves both as input and as output to a node. The wire represents a qubit partitioned into sectors spanned by $\ket{0}$ and $\ket{1}$ respectively. We impose a delta-route on the node so that the transformations we insert must map each sector to itself.

\begin{figure}
    \centering
    \begin{tikzpicture}
	\begin{pgfonlayer}{nodelayer}
		\node [style=none] (0) at (1, -1.75) {};
		\node [style=none] (1) at (1, 1.75) {};
		\node [style=none] (2) at (-2, 1.75) {};
		\node [style=none] (3) at (-2, -1.75) {};
		\node [style=none] (4) at (1, 0.5) {};
		\node [style=none] (5) at (1, -0.5) {};
		\node [style=none] (6) at (2, -1.5) {$A^i_\inn$};
		\node [style=none] (7) at (2, 1.5) {$A^i_\out$};
		\node [style=none] (8) at (-3, 0.25) {$X^i$};
		\node [style=dashedmap] (9) at (1, 0) {};
	\end{pgfonlayer}
	\begin{pgfonlayer}{edgelayer}
		\draw [in=90, out=90, looseness=1.25] (1.center) to (2.center);
		\draw (3.center) to (2.center);
		\draw [in=-90, out=-90, looseness=1.25] (3.center) to (0.center);
		\draw (0.center) to (5.center);
		\draw (4.center) to (1.center);
	\end{pgfonlayer}
\end{tikzpicture}
    \caption{A routed skeletal supermap that leads to a paradox, where a qubit is sent back to itself. Formally, the wires bent into `cup' and `cap' shapes can be thought of as the unnormalised perfectly correlated entangled ket and bra respectively. The qubit is partitioned into sectors spanned by the logical states $\ket{0}$ and $\ket{1}$. The index-matching means that the agent at the node must map the logical states $\ket{i}$ to themselves, up to dephasing.}
    \label{fig:grandfather smap}
\end{figure}
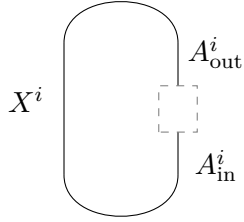

The node has two branches, each corresponding to a one-dimensional sector. Hence, no information flows between the branches.  Clearly, then, there can be no question of an informational loop within a branch. Nevertheless, the supermap is invalid, and fails to represent any logically consistent process. This is a consequence of the fact that \textit{nothing can determine the value of the index} $i$. For if it did, then univocality would be respected, and our main theorem would imply the validity of the supermap.

The key problem here is one of underdetermination. The fact that $i$ is not fixed means that there is no point in the circuit where an agent could make a bifurcation choice determining which branch happens. So, even though there is no \textit{over}determination \textit{within} a branch (resulting in a grandfather paradox), there is still \textit{under}determination \textit{of} the branch. To see how this leads to concrete problems, one needs only to notice that inserting a CNOT gate into the supermap (where the NOT part acts on an ancillary system)  leads to a trace-decreasing output transformation, and that this is a direct consequence of the underdetermination of $i$.\footnote{The reader might also notice that inserting a $Z$ Pauli matrix into the node results in a sort of grandfather paradox associated with the Fourier basis. But the foregoing considerations show that even in the classical case, where there is no Fourier basis, we still need bi-univocality.} 

\begin{figure}
    \centering
    \begin{tikzpicture}
	\begin{pgfonlayer}{nodelayer}
		\node [style=none] (0) at (1, -1.75) {};
		\node [style=none] (1) at (1, 1.75) {};
		\node [style=none] (2) at (-2, 1.75) {};
		\node [style=none] (3) at (-2, -1.75) {};
		\node [style=none] (4) at (1, 0.5) {};
		\node [style=none] (5) at (1, -0.5) {};
		\node [style=none] (6) at (2, -1.5) {$A^i_\inn$};
		\node [style=none] (7) at (2, 1.5) {$A^i_\out$};
		\node [style=none] (8) at (-3, 0.25) {$X^i$};
		\node [style=very large dashedmap] (9) at (2.5, 0) {};
		\node [style=none] (10) at (4, -2.5) {};
		\node [style=none] (11) at (4, -0.5) {};
		\node [style=none] (12) at (5, -1.5) {$S^i$};
		\node [style=dashedmap] (13) at (4, -3) {};
		\node [style=none] (14) at (4, -5) {};
		\node [style=none] (15) at (4, -3.5) {};
		\node [style=none] (16) at (5, -4.75) {$P$};
	\end{pgfonlayer}
	\begin{pgfonlayer}{edgelayer}
		\draw [in=90, out=90, looseness=1.25] (1.center) to (2.center);
		\draw (3.center) to (2.center);
		\draw [in=-90, out=-90, looseness=1.25] (3.center) to (0.center);
		\draw (0.center) to (5.center);
		\draw (4.center) to (1.center);
		\draw (10.center) to (11.center);
		\draw (14.center) to (15.center);
	\end{pgfonlayer}
\end{tikzpicture}
    \caption{\changes{A non-valid routed skeletal supermap, leading to a paradox. This supermap is similar to that of Figure \ref{fig:grandfather smap}, but now includes an $S$ qubit system, partitioned into two one-dimensional sectors, and an unsectorised $P$ qubit system.}}
    \label{fig:grandfather smap 2}
\end{figure}

\changes{Let us now illustrate why not only univocality, but \textit{bi}-univocality is necessary for the consistency of a protocol. To this end, consider the supermap in Figure \ref{fig:grandfather smap 2}. In contrast with the previous example, this supermap includes an additional node in which `$i$ is created'; the corresponding routed graph therefore satisfies the univocality principle. However, this supermap is not a valid one: plugging an identity in the $P \to S^i$ node and the routed unitary $\ket{ii} \mapsto \ket{i}$ in the other node yields the $\bra{0} + \bra{1}$ effect, which is of course not a unitary. This failure can be ascribed to the fact that the routed graph corresponding to this supermap does not satisfy \textit{bi}-univocality: its adjoint does not satisfy univocality. In broad terms, not only should the indices `come from somewhere', they should also `go somewhere'.\footnote{\changes{One might naively believe that only requiring univocality would lead to super\textit{isometries} that would not necessarily be superunitaries, but the above example also shows this is not the case, since the supermap in Figure \ref{fig:grandfather smap 2} is not even a superisometry.}}
}




Before moving on to lay out our framework in the next section, let us first summarize this one, in which we have shown how to represent the switch as a routed circuit and how to certify the validity of this circuit despite its feedback loops.

We wrote the switch as a routed circuit (Figure \ref{fig:switch routed circuit}). We captured this circuit's basic structure by trimming it down to a `skeletal' routed supermap (Figure \ref{fig:switch skeletal}), from which the switch can be constructed, by inserting unitary transformations into the nodes. We represented the structure of the skeletal supermap as an equivalent routed graph (Figure \ref{fig:rswitch routed graph}). We then showed that this routed graph satisfies two conditions, bi-univocality and weak loops, which, by our main theorem, imply that the skeletal supermap is valid (i.e.\ takes unitaries to unitaries), which in turn implies that any routed circuit with its connectivity is valid as well.

Bi-univocality requires that choices of bifurcation in the routed graph lead to a definite fact about the branch that happens at each node. It also requires a similar statement to hold about the time-reversed version of the routed graph, obtained by reversing the direction of the arrows.

If bi-univocality holds, then we can ask whether the routed graph satisfies the weak loops condition. To evaluate this condition, we form a branch graph, in which solid arrows represent the ability of quantum information to flow between the different branches. Green dashed arrows indicate that bifurcation choices at one branch can influence whether another branch happens in the routed graph. Red dashed arrows represent the same thing for a time-reversed version of the routed graph. The weak loops condition states that any given loop in the branch graph must be formed entirely of dashed arrows of the same colour. The switch satisfies this trivially since its branch graph contains no loops all.

Our constructions do not only provide a technical way to certify the consistency of a process, but also make its inner structure evident. Indeed, the routed graph gives an intuition of the crucial structural features of the switch. Furthermore, its branch graph displays the order in which its branches happen, and tells us which branches control which other branches happen. This will be particularly valuable when we perform the same reconstruction for more elaborate processes in Section \ref{sec: examples}.

\section{The framework} \label{sec: theorem}

In this section, we present our framework in detail and state our main theorem, which says that any routed graph satisfying bi-univocality and weak loops defines a valid superunitary. To keep things readable, we will give definitions at a semi-formal level; a fully formal account is given in Appendix \ref{app: Theorem}.

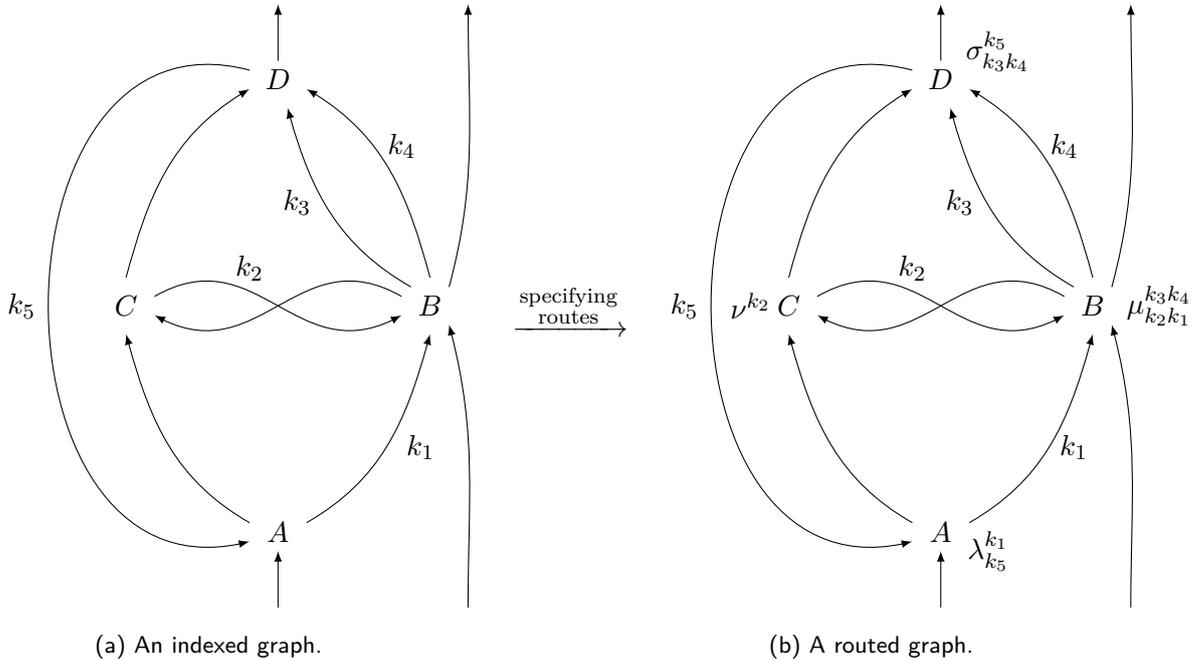
\begin{figure}
    \centering
    \begin{subfigure}[c]{0.4\textwidth}
         \centering
         \begin{tikzpicture}
	\begin{pgfonlayer}{nodelayer}
		\node [style=none] (0) at (4, 0) {$B$};
		\node [style=none] (1) at (0, -6) {$A$};
		\node [style=none] (2) at (-4, 0) {$C$};
		\node [style=none] (3) at (0, 6) {$D$};
		\node [style=none] (4) at (0, 8) {};
		\node [style=none] (5) at (5, 8) {};
		\node [style=none] (6) at (0, -8) {};
		\node [style=none] (8) at (4, 0.75) {};
		\node [style=none] (9) at (0.75, 5.75) {};
		\node [style=none] (10) at (4, -0.75) {};
		\node [style=none] (11) at (0.75, -5.75) {};
		\node [style=none] (12) at (-4, 0.75) {};
		\node [style=none] (13) at (-0.75, 5.75) {};
		\node [style=none] (16) at (0.25, 5.25) {};
		\node [style=none] (17) at (3.5, 0.5) {};
		\node [style=none] (18) at (3.25, 0.25) {};
		\node [style=none] (19) at (-3.25, -0.25) {};
		\node [style=none] (20) at (-3.25, 0.25) {};
		\node [style=none] (21) at (3.25, -0.25) {};
		\node [style=none] (22) at (4.5, 0.5) {};
		\node [style=none] (24) at (5, -8) {};
		\node [style=none] (25) at (4.5, -0.5) {};
		\node [style=none] (26) at (0, 6.5) {};
		\node [style=none] (27) at (0, -6.5) {};
		\node [style=none] (28) at (-4, -0.75) {};
		\node [style=none] (29) at (-0.75, -5.75) {};
		\node [style=none] (30) at (-0.75, 6.25) {};
		\node [style=none] (31) at (-0.75, -6.25) {};
		\node [style=none] (33) at (-0.75, 1) { $k_2$};
		\node [style=none] (34) at (0.5, 2.75) { $k_3$};
		\node [style=none] (35) at (3.25, 4.25) { $k_4$};
		\node [style=none] (37) at (-6.75, 0) { $k_5$};
		\node [style=none] (38) at (3.75, -3.75) { $k_1$};
	\end{pgfonlayer}
	\begin{pgfonlayer}{edgelayer}
		\draw [style=arrow plain, in=-30, out=105] (8.center) to (9.center);
		\draw [style=arrow plain, in=-150, out=75] (12.center) to (13.center);
		\draw [style=arrow plain, in=-75, out=150] (17.center) to (16.center);
		\draw [style=arrow plain, in=-30, out=150, looseness=1.25] (18.center) to (19.center);
		\draw [style=arrow plain, in=-150, out=30, looseness=1.25] (20.center) to (21.center);
		\draw [style=arrow plain, in=-90, out=75] (22.center) to (5.center);
		\draw [style=arrow plain] (26.center) to (4.center);
		\draw [style=arrow plain] (6.center) to (27.center);
		\draw [style=arrow plain, in=-105, out=30] (11.center) to (10.center);
		\draw [style=arrow plain, in=-75, out=150] (29.center) to (28.center);
		\draw [style=arrow plain, in=-165, out=165, looseness=1.50] (30.center) to (31.center);
		\draw [style=arrow plain, in=-75, out=90] (24.center) to (25.center);
	\end{pgfonlayer}
\end{tikzpicture}
         \caption{An indexed graph.}
         \label{fig:Indexed Graph}
     \end{subfigure}%
     \hspace{1cm}$\xrightarrow{\substack{\textrm{specifying}\\ \text{routes}}}$\hspace{0.2cm}%
     \begin{subfigure}[c]{0.4\textwidth}
         \centering
         \begin{tikzpicture}
	\begin{pgfonlayer}{nodelayer}
		\node [style=none] (0) at (3, 0) {$B$};
		\node [style=none] (1) at (-1, -6) {$A$};
		\node [style=none] (2) at (-5, 0) {$C$};
		\node [style=none] (3) at (-1, 6) {$D$};
		\node [style=none] (4) at (-1, 8) {};
		\node [style=none] (5) at (4, 8) {};
		\node [style=none] (6) at (-1, -8) {};
		\node [style=none] (8) at (3, 0.75) {};
		\node [style=none] (9) at (-0.25, 5.75) {};
		\node [style=none] (10) at (3, -0.75) {};
		\node [style=none] (11) at (-0.25, -5.75) {};
		\node [style=none] (12) at (-5, 0.75) {};
		\node [style=none] (13) at (-1.75, 5.75) {};
		\node [style=none] (16) at (-0.75, 5.25) {};
		\node [style=none] (17) at (2.5, 0.5) {};
		\node [style=none] (18) at (2.25, 0.25) {};
		\node [style=none] (19) at (-4.25, -0.25) {};
		\node [style=none] (20) at (-4.25, 0.25) {};
		\node [style=none] (21) at (2.25, -0.25) {};
		\node [style=none] (22) at (3.5, 0.5) {};
		\node [style=none] (24) at (4, -8) {};
		\node [style=none] (25) at (3.5, -0.5) {};
		\node [style=none] (26) at (-1, 6.5) {};
		\node [style=none] (27) at (-1, -6.5) {};
		\node [style=none] (28) at (-5, -0.75) {};
		\node [style=none] (29) at (-1.75, -5.75) {};
		\node [style=none] (30) at (-1.75, 6.25) {};
		\node [style=none] (31) at (-1.75, -6.25) {};
		\node [style=none] (33) at (-1.75, 1) { $k_2$};
		\node [style=none] (34) at (-0.5, 2.75) { $k_3$};
		\node [style=none] (35) at (2.25, 4.25) { $k_4$};
		\node [style=none] (37) at (-7.75, 0) { $k_5$};
		\node [style=none] (38) at (2.5, -3.75) { $k_1$};
		\node [style=none] (39) at (4.75, 0) { $\mu_{k_2 k_1}^{k_3 k_4}$};
		\node [style=none] (40) at (-6, 0) { $\nu^{k_2}$};
		\node [style=none] (41) at (0.5, 6.75) { $\sigma_{k_3 k_4}^{k_5}$};
		\node [style=none] (42) at (0.25, -6.5) { $\lambda_{k_5}^{k_1}$};
	\end{pgfonlayer}
	\begin{pgfonlayer}{edgelayer}
		\draw [style=arrow plain, in=-30, out=105] (8.center) to (9.center);
		\draw [style=arrow plain, in=-150, out=75] (12.center) to (13.center);
		\draw [style=arrow plain, in=-75, out=150] (17.center) to (16.center);
		\draw [style=arrow plain, in=-30, out=150, looseness=1.25] (18.center) to (19.center);
		\draw [style=arrow plain, in=-150, out=30, looseness=1.25] (20.center) to (21.center);
		\draw [style=arrow plain, in=-90, out=75] (22.center) to (5.center);
		\draw [style=arrow plain] (26.center) to (4.center);
		\draw [style=arrow plain] (6.center) to (27.center);
		\draw [style=arrow plain, in=-105, out=30] (11.center) to (10.center);
		\draw [style=arrow plain, in=-75, out=150] (29.center) to (28.center);
		\draw [style=arrow plain, in=-165, out=165, looseness=1.50] (30.center) to (31.center);
		\draw [style=arrow plain, in=-75, out=90] (24.center) to (25.center);
	\end{pgfonlayer}
\end{tikzpicture}
         \caption{A routed graph.}
         \label{fig:Routed Graph}
     \end{subfigure}
    \caption{Examples of an indexed graph and of a routed graph; the latter is obtained from the former by specifying a branched route at every node. The arrows not bearing indices have a trivial (i.e.\ a singleton) set of index values.}
    \label{fig:Indexed and Routed Graphs}
\end{figure}

The most basic notion we need is that of a routed graph: this is a directed multi\footnote{A multigraph is a graph in which there can be several arrows between two given nodes. In the interest of generality, we will allow them, even though for the purposes of the certification of supermaps' validity, any multigraph could just be turned into an equivalent graph by merging wires.}-graph with decorated nodes and arrows. The nodes are decorated with routes, and the arrows are decorated with indices that are in turn equipped with a `dimension' for each index value.
A routed graph with its routes still unspecified will be called an indexed graph. Examples are given in Figure \ref{fig:Indexed and Routed Graphs}.

\begin{definition}[Indexed and routed graphs]
An \emph{indexed graph} $\Gamma$ is a directed multigraph in which each arrow is attributed a non-empty set of index values. Each of these values is furthermore attributed a non-zero natural number, called its dimension\footnote{This will be the dimension of the corresponding sector in the interpretation of the graph as a supermap. Note that for our theorem, all we need to know is which sectors are one-dimensional.}.

A \emph{routed graph} $(\Ga, (\laN)_{N \in \Nodes})$ is an indexed graph for which a relation (or `route') has been specified at every node. The route $\laN$ at node $N$ goes from the Cartesian product of the sets of indices of the arrows going into $N$, to that of the sets of indices of the arrows going out of $N$.
\end{definition}

We also allow these graphs to feature arrows `coming from nowhere' (resp.\ `going nowhere'): these will be interpreted as global inputs (resp.\ global outputs) of the supermap. We ask for these not to be indexed, that is, to have trivial (i.e.\ singleton) sets of index values.\footnote{This requirement is there only to make the statement of univocality simpler, as otherwise one would have to distinguish several cases. Any routed graph with indexed input and output arrows can be turned into one without, by adjoining to it a global input node and a global output node.}

\begin{figure}
    \centering
    $\begin{tikzpicture}
	\begin{pgfonlayer}{nodelayer}
		\node [style=large map] (0) at (0, 0) {$\lambda_N$};
		\node [style=none] (1) at (-0.75, -0.25) {};
		\node [style=none] (2) at (0.75, -0.25) {};
		\node [style=none] (3) at (-1, 0.25) {};
		\node [style=none] (4) at (1, 0.25) {};
		\node [style=none] (5) at (0.75, -2) {};
		\node [style=none] (6) at (-0.75, -2) {};
		\node [style=none] (7) at (-1, 2) {};
		\node [style=none] (8) at (1, 2) {};
		\node [style=none] (11) at (0, 0.25) {};
		\node [style=none] (12) at (0, 2) {};
		\node [style=none] (13) at (-1.5, 2) {};
		\node [style=none] (14) at (-1.25, 2.5) {};
		\node [style=none] (15) at (1.25, 2.5) {};
		\node [style=none] (16) at (1.5, 2) {};
		\node [style=none] (17) at (0, 2.5) {};
		\node [style=none] (18) at (0, 3) {};
		\node [align=center, text width=3.5cm] (19) at (0, 4.25) {\color{blue} indices of the arrows coming out of $N$};
		\node [style=none] (20) at (-1.25, -2) {};
		\node [style=none] (21) at (-1, -2.5) {};
		\node [style=none] (22) at (1, -2.5) {};
		\node [style=none] (23) at (1.25, -2) {};
		\node [style=none] (24) at (0, -2.5) {};
		\node [style=none] (25) at (0, -3) {};
		\node [align=center, text width=3.5cm] (26) at (0, -4.25) {\color{blue} indices of the arrows\\going into $N$};
	\end{pgfonlayer}
	\begin{pgfonlayer}{edgelayer}
		\draw (3.center) to (7.center);
		\draw (4.center) to (8.center);
		\draw (1.center) to (6.center);
		\draw (2.center) to (5.center);
		\draw (11.center) to (12.center);
		\draw [style=blue] (13.center) to (14.center);
		\draw [style=blue] (14.center) to (17.center);
		\draw [style=blue] (17.center) to (15.center);
		\draw [style=blue] (15.center) to (16.center);
		\draw [style=blue] (17.center) to (18.center);
		\draw [style=blue] (20.center) to (21.center);
		\draw [style=blue] (21.center) to (24.center);
		\draw [style=blue] (24.center) to (22.center);
		\draw [style=blue] (22.center) to (23.center);
		\draw [style=blue] (24.center) to (25.center);
	\end{pgfonlayer}
\end{tikzpicture} {\LARGE \cong} \quad \quad \quad  \begin{tikzpicture}
	\begin{pgfonlayer}{nodelayer}
		\node [style=very small black dot] (0) at (7.5, -1.5) {};
		\node [style=very small black dot] (1) at (9, -1.5) {};
		\node [style=very small black dot] (2) at (-3.5, -1.5) {};
		\node [style=very small black dot] (3) at (-1.25, 1.5) {};
		\node [style=very small black dot] (4) at (-1.25, -1.5) {};
		\node [style=very small black dot] (5) at (0.5, -1.5) {};
		\node [style=very small black dot] (6) at (0.5, 1.5) {};
		\node [style=very small black dot] (7) at (2, -1.5) {};
		\node [style=very small black dot] (8) at (-4.25, 1.5) {};
		\node [style=very small black dot] (9) at (2, 1.5) {};
		\node [style=very small black dot] (10) at (3.5, -1.5) {};
		\node [style=very small black dot] (11) at (-2.75, 1.5) {};
		\node [style=very small black dot] (12) at (5, -1.5) {};
		\node [style=very small black dot] (13) at (4.25, 1.5) {};
		\node [style=very small black dot] (14) at (8, 1.5) {};
		\node [style=none] (15) at (-4.125, 1) {};
		\node [style=none] (16) at (-2.875, 1) {};
		\node [style=none] (17) at (-1.25, 1) {};
		\node [style=none] (18) at (-1.25, -1) {};
		\node [style=none] (19) at (0.5, 1) {};
		\node [style=none] (20) at (0.75, 1) {};
		\node [style=none] (21) at (1.75, 1) {};
		\node [style=none] (22) at (2, 1) {};
		\node [style=none] (23) at (2, -1) {};
		\node [style=none] (24) at (1.75, -1) {};
		\node [style=none] (25) at (0.5, -1) {};
		\node [style=none] (26) at (0.75, -1) {};
		\node [style=none] (27) at (4, 1) {};
		\node [style=none] (28) at (4.5, 1) {};
		\node [style=none] (30) at (3.5, -1) {};
		\node [style=none] (31) at (5, -1) {};
		\node [style=none] (32) at (-3.75, -1) {};
		\node [style=none] (33) at (-3.25, -1) {};
		\node [style=none] (34) at (-4.225, 2) {};
		\node [style=none] (35) at (-2.75, 2) {};
		\node [style=none] (37) at (-4.75, 1.5) {};
		\node [style=none] (38) at (-3.5, -2) {};
		\node [style=none] (39) at (-2.25, 1.5) {};
		\node [style=none] (40) at (-3, -1.75) {};
		\node [style=none] (41) at (-4, -1.75) {};
		\node [style=none] (42) at (-1.25, 2) {};
		\node [style=none] (43) at (-1.5, 2) {};
		\node [style=none] (44) at (-1.75, 1.75) {};
		\node [style=none] (45) at (-0.75, 1.75) {};
		\node [style=none] (46) at (-1, 2) {};
		\node [style=none] (47) at (-1.25, -2) {};
		\node [style=none] (48) at (-1, -2) {};
		\node [style=none] (49) at (-1.5, -2) {};
		\node [style=none] (50) at (-1.75, -1.75) {};
		\node [style=none] (51) at (-0.75, -1.75) {};
		\node [style=none] (52) at (1.25, 2) {};
		\node [style=none] (53) at (0.25, 2) {};
		\node [style=none] (54) at (0, 1.75) {};
		\node [style=none] (55) at (2.5, 1.75) {};
		\node [style=none] (56) at (2.25, 2) {};
		\node [style=none] (57) at (1.25, -2) {};
		\node [style=none] (58) at (2.25, -2) {};
		\node [style=none] (59) at (0.25, -2) {};
		\node [style=none] (60) at (0, -1.75) {};
		\node [style=none] (61) at (2.5, -1.75) {};
		\node [style=none] (62) at (3.525, -2) {};
		\node [style=none] (63) at (5, -2) {};
		\node [style=none] (64) at (3, -1.5) {};
		\node [style=none] (65) at (4.25, 2) {};
		\node [style=none] (66) at (5.5, -1.5) {};
		\node [style=none] (67) at (4.75, 1.75) {};
		\node [style=none] (68) at (3.75, 1.75) {};
		\node [style=none] (69) at (-2, 4) {};
		\node [style=none] (70) at (-3.25, 2.5) {};
		\node [style=none] (71) at (-1.25, 2.5) {};
		\node [style=none] (72) at (1.25, 2.5) {};
		\node [style=none] (73) at (3.75, 2.5) {};
		\node [style=none] (74) at (0.25, 4.75) {\color{blue} $\lambda_N$'s branches};
		\node [style=none] (75) at (-1, 3.75) {};
		\node [style=none] (76) at (1, 3.75) {};
		\node [style=none] (77) at (2.5, 4) {};
		\node [align=center, text width=4cm] (84) at (8.25, -4.75) {\color{red}values outside $\lambda_N$'s practical inputs};
		\node [style=none] (85) at (11.25, 0) {\Large $\lambda_N$};
		\node [align=center, text width=4cm] (92) at (8, 5) {\color{red}values outside $\lambda_N$'s practical outputs};
		\node [style=none] (93) at (-3.5, -2.75) {\color{blue} $N^\mathrm{I}$};
		\node [style=none] (94) at (-1.25, -2.75) {\color{blue} $N^\mathrm{II}$};
		\node [style=none] (95) at (1.25, -2.75) {\color{blue} $N^\mathrm{III}$};
		\node [style=none] (96) at (4.25, -2.75) {\color{blue} $N^\mathrm{IV}$};
		\node [style=none] (97) at (-8, -1.5) {inputs};
		\node [style=none] (98) at (-8, 1.5) {outputs};
		\node [style=none] (103) at (8.25, -1) {};
		\node [style=none] (104) at (7.25, -1) {};
		\node [style=none] (105) at (7, -1.25) {};
		\node [style=none] (106) at (9.5, -1.25) {};
		\node [style=none] (107) at (9.25, -1) {};
		\node [style=none] (108) at (8.25, -2) {};
		\node [style=none] (109) at (9.25, -2) {};
		\node [style=none] (110) at (7.25, -2) {};
		\node [style=none] (111) at (7, -1.75) {};
		\node [style=none] (112) at (9.5, -1.75) {};
		\node [style=none] (115) at (8, 2) {};
		\node [style=none] (116) at (7.75, 2) {};
		\node [style=none] (117) at (7.5, 1.75) {};
		\node [style=none] (118) at (8.5, 1.75) {};
		\node [style=none] (119) at (8.25, 2) {};
		\node [style=none] (120) at (8, 1) {};
		\node [style=none] (121) at (8.25, 1) {};
		\node [style=none] (122) at (7.75, 1) {};
		\node [style=none] (123) at (7.5, 1.25) {};
		\node [style=none] (124) at (8.5, 1.25) {};
		\node [style=none] (125) at (8, 2.5) {};
		\node [style=none] (126) at (8, 3.75) {};
		\node [style=none] (127) at (8.25, -3.5) {};
		\node [style=none] (128) at (8.25, -2.5) {};
	\end{pgfonlayer}
	\begin{pgfonlayer}{edgelayer}
		\draw [style=arrow plain] (32.center) to (15.center);
		\draw [style=arrow plain] (33.center) to (16.center);
		\draw [style=arrow plain] (18.center) to (17.center);
		\draw [style=arrow plain] (25.center) to (19.center);
		\draw [style=arrow plain] (26.center) to (21.center);
		\draw [style=arrow plain] (24.center) to (20.center);
		\draw [style=arrow plain] (23.center) to (22.center);
		\draw [style=arrow plain] (30.center) to (27.center);
		\draw [style=arrow plain] (31.center) to (28.center);
		\draw [style=dashed, blue] (39.center)
			 to [in=75, out=-105] (40.center)
			 to [in=0, out=-105] (38.center)
			 to [in=-75, out=180] (41.center)
			 to [in=285, out=105] (37.center)
			 to [in=180, out=105, looseness=1.25] (34.center)
			 to [in=180, out=0] (35.center)
			 to [in=75, out=0] cycle;
		\draw [style=dashed, blue] (45.center)
			 to [in=90, out=-90] (51.center)
			 to [in=0, out=-90] (48.center)
			 to [in=360, out=180] (47.center)
			 to [in=360, out=180] (49.center)
			 to [in=270, out=180] (50.center)
			 to [in=270, out=90] (44.center)
			 to [in=180, out=90] (43.center)
			 to [in=180, out=0] (42.center)
			 to [in=180, out=0] (46.center)
			 to [in=90, out=0] cycle;
		\draw [style=dashed, blue] (55.center)
			 to [in=90, out=-90] (61.center)
			 to [in=0, out=-90] (58.center)
			 to [in=360, out=180] (57.center)
			 to [in=360, out=180] (59.center)
			 to [in=270, out=180] (60.center)
			 to [in=270, out=90] (54.center)
			 to [in=180, out=90] (53.center)
			 to [in=180, out=0] (52.center)
			 to [in=180, out=0] (56.center)
			 to [in=90, out=0] cycle;
		\draw [style=dashed, blue] (66.center)
			 to [in=-75, out=105] (67.center)
			 to [in=0, out=105] (65.center)
			 to [in=75, out=-180] (68.center)
			 to [in=-285, out=-105] (64.center)
			 to [in=-180, out=-105, looseness=1.25] (62.center)
			 to [in=-180, out=0] (63.center)
			 to [in=-75, out=0] cycle;
		\draw [style=blue] (70.center) to (69.center);
		\draw [style=blue] (71.center) to (75.center);
		\draw [style=blue] (72.center) to (76.center);
		\draw [style=blue] (73.center) to (77.center);
		\draw [style=dashed, red] (106.center)
			 to [in=90, out=-90] (112.center)
			 to [in=0, out=-90] (109.center)
			 to [in=360, out=180] (108.center)
			 to [in=360, out=180] (110.center)
			 to [in=270, out=180] (111.center)
			 to [in=270, out=90] (105.center)
			 to [in=180, out=90] (104.center)
			 to [in=180, out=0] (103.center)
			 to [in=180, out=0] (107.center)
			 to [in=90, out=0] cycle;
		\draw [style=dashed, red] (118.center)
			 to [in=90, out=-90] (124.center)
			 to [in=0, out=-90] (121.center)
			 to [in=360, out=180] (120.center)
			 to [in=360, out=180] (122.center)
			 to [in=270, out=180] (123.center)
			 to [in=270, out=90] (117.center)
			 to [in=180, out=90] (116.center)
			 to [in=180, out=0] (115.center)
			 to [in=180, out=0] (119.center)
			 to [in=90, out=0] cycle;
		\draw [style=red] (125.center) to (126.center);
		\draw [style=red] (127.center) to (128.center);
	\end{pgfonlayer}
\end{tikzpicture}$
    \caption{`Looking inside' a branched route $\lambda_N$ (which is to be used for the node $N$ of an indexed graph). On the left, we see $\lambda_N$ as a box: each of its input (resp.\ output) wires is the set of index values of one of the arrows going into $N$ (resp.\ coming out of it). On the right, we see its `unfolded' structure, specifying how $\lambda_N$ connects input values to output values; each of the input black dots corresponds to a possible value (more precisely, a possible tuple of values) of the input indices, and similarly for the outputs.
    }
    \label{fig: Looking inside}
\end{figure}
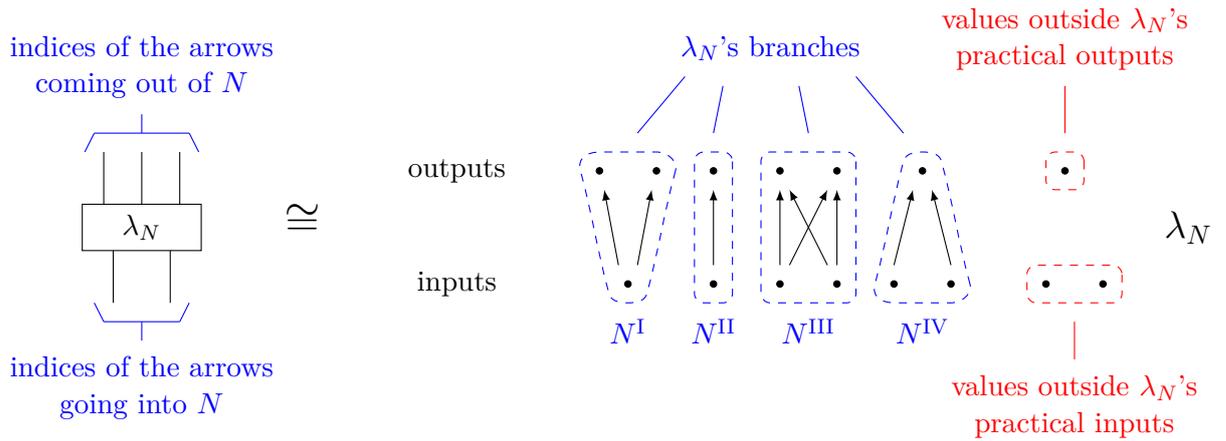

We will in fact not need all types of relations: we will restrict ourselves to considering \textit{branched} ones.

\begin{definition}[Branched routes]
A route $\laN$ is \emph{branched} if any two of its input values are either connected to the exact same output values, or have no output values in common.
\end{definition}

An example is given in Figure \ref{fig: Looking inside}. As seen in this figure, a branched relation $\laN$  defines compatible (partial) partitions of its input and output sets, which we call $\laN$'s \textit{branches} (or, in a slight abuse of notation, $N$'s branches, which will be called the $\Nal$ with $\al$ varying), with each input value of a branch being connected to all output values of this branch and vice versa.

There can also be input (resp.\ output) values that are not connected to anything by $\lambda_N$; these will be said to be outside its practical inputs (resp.\ outputs), and are considered to be part of no branch at all. These values correspond to sectors which are just there for formal purposes and will never be used in practice -- part of the role of bi-univocality will be to ensure that this does not lead to any inconsistencies.

A skeletal routed supermap can be naturally defined from a routed graph.

\begin{definition}[Skeletal supermap associated to a routed graph]
Given a routed graph $\GalaN$, its associated skeletal (routed) supermap is obtained by interpreting each wire as a sectorised Hilbert space, whose sectors are labelled by the set of index values of this wire, with each sector having the dimension that was assigned to its corresponding index value; and interpreting each node as a slot for a linear map, going from the tensor product of the Hilbert spaces associated to its incoming arrows, to that of the Hilbert spaces associated to its outgoing arrows, and following the route associated to that node. The supermap acts on linear maps by connecting them along the graph of $\Ga$\footnote{Note that this procedure has an unambiguous meaning, despite the cycles in $\Ga$, due to the fact that finite-dimensional complex linear maps form a traced monoidal category \cite{joyal_street_verity_1996}.}. 
\end{definition}

Our goal is to define structural requirements on routed graphs ensuring that their associated supermap is a (routed) \textit{superunitary}; i.e.,\ that it yields a unitary map when arbitrary unitary maps, following the routes, are plugged at each of its nodes. Note that a map being unitary, in this context, means that it is unitary when restricted to act only on its practical input space, consisting of the input sectors whose indices are practical inputs of the route, and to map to its similarly defined practical output space. 

Our first principle will be univocality. The idea is that some branches feature \textit{bifurcations}, i.e.\ include several output values (e.g.\ branches $N^\textrm{I}$ and $N^\textrm{III}$ in Figure \ref{fig: Looking inside}). `Bifurcation choices', in a branch at a node -- i.e.\ choosing a single output value for this branch, and erasing the arrows to the other output values -- will in general lead to some branches at other nodes `not happening' -- i.e.\ to none of their input values being instantiated. Univocality tells us that \textit{any tuple of bifurcation choices} throughout the graph should lead to \textit{one and exactly one} branch happening at every node. To make this requirement formal, we will `augment' our relations, i.e.\ supplement them with ancillary wires: ancillary input wires with which bifurcation choices in each branch can be specified; and ancillary output wires which record, in a binary variable, whether each branch happened or not.

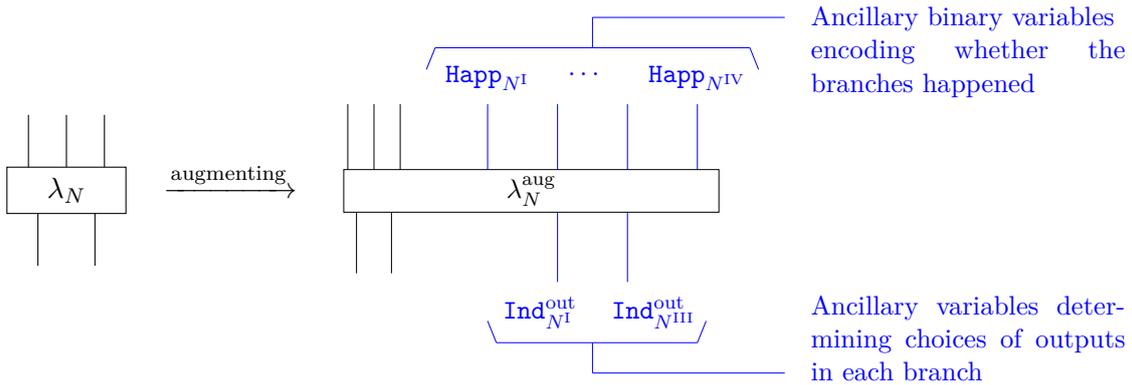
\begin{figure}
    \centering
    $ \begin{tikzpicture}
	\begin{pgfonlayer}{nodelayer}
		\node [style=large map] (0) at (0, 0) {$\lambda_N$};
		\node [style=none] (1) at (-0.75, -0.25) {};
		\node [style=none] (2) at (0.75, -0.25) {};
		\node [style=none] (3) at (-1, 0.25) {};
		\node [style=none] (4) at (1, 0.25) {};
		\node [style=none] (5) at (0.75, -2) {};
		\node [style=none] (6) at (-0.75, -2) {};
		\node [style=none] (7) at (-1, 2) {};
		\node [style=none] (8) at (1, 2) {};
		\node [style=none] (9) at (0, 0.25) {};
		\node [style=none] (10) at (0, 2) {};
	\end{pgfonlayer}
	\begin{pgfonlayer}{edgelayer}
		\draw (3.center) to (7.center);
		\draw (4.center) to (8.center);
		\draw (1.center) to (6.center);
		\draw (2.center) to (5.center);
		\draw (9.center) to (10.center);
	\end{pgfonlayer}
\end{tikzpicture} \quad \xrightarrow{\textrm{augmenting}} \quad \scalebox{.92}{\begin{tikzpicture}[tikzfig]
	\begin{pgfonlayer}{nodelayer}
		\node [style=very very large map] (0) at (0, 0) {$\lambda_N^\aug$};
		\node [style=none] (1) at (-10, 0) {};
		\node [style=none] (2) at (-8, 0) {};
		\node [style=none] (5) at (-8, -4.75) {};
		\node [style=none] (6) at (-10, -4.75) {};
		\node [style=none] (13) at (1.5, 0.25) {};
		\node [style=none] (14) at (9.5, 0.25) {};
		\node [style=none] (15) at (1.5, 5) {};
		\node [style=none] (16) at (9.5, 5) {};
		\node [style=none] (17) at (5.5, 0.25) {};
		\node [style=none] (18) at (5.5, 5) {};
		\node [style=none] (19) at (-2.5, 0.25) {};
		\node [style=none] (20) at (-2.5, 5) {};
		\node [style=none] (21) at (1.5, -0.5) {};
		\node [style=none] (22) at (5.5, -0.5) {};
		\node [style=none] (23) at (5.5, -5.25) {};
		\node [style=none] (24) at (1.5, -5.25) {};
		\node [style=none] (25) at (-10.5, 0.25) {};
		\node [style=none] (26) at (-7.5, 0.25) {};
		\node [style=none] (27) at (-10.5, 5) {};
		\node [style=none] (28) at (-7.5, 5) {};
		\node [style=none] (29) at (-9, 0.25) {};
		\node [style=none] (30) at (-9, 5) {};
		\node [style=none] (31) at (-2.5, 6.5) {\color{blue} $\texttt{Happ}_{N^{\mathrm{I}}}$};
		\node [style=none] (32) at (3, 6.75) {\color{blue} $\ldots$};
		\node [style=none] (33) at (9.5, 6.5) {\color{blue} $\texttt{Happ}_{N^{\mathrm{IV}}}$};
		\node [style=none] (34) at (0.5, -7) {\color{blue} $\Ind^\out_{N^{\mathrm{I}}}$};
		\node [style=none] (35) at (7, -7) {\color{blue} $\Ind^\out_{N^{\mathrm{III}}}$};
		\node [style=none] (36) at (-6, 7) {};
		\node [style=none] (37) at (-5.5, 8.25) {};
		\node [style=none] (38) at (12.5, 8.25) {};
		\node [style=none] (39) at (13, 7) {};
		\node [style=none] (40) at (3.5, 8.25) {};
		\node [style=none] (41) at (3.5, 10) {};
		\node [style=none] (42) at (14.5, 10) {};
		\node [style=none] (43) at (-2.5, -7.5) {};
		\node [style=none] (44) at (-2, -8.75) {};
		\node [style=none] (45) at (9.5, -8.75) {};
		\node [style=none] (46) at (10, -7.5) {};
		\node [style=none] (47) at (3.5, -8.75) {};
		\node [style=none] (48) at (3.5, -10.5) {};
		\node [style=none] (49) at (14.5, -10.5) {};
		\node [align=justify, text width=4.5cm] (50) at (25, 8) {\color{blue} Ancillary binary variables\\encoding whether the branches happened};
		\node [align=justify, text width=4.5cm] (51) at (25, -8.5) {\color{blue} Ancillary variables determining choices of outputs in each branch};
	\end{pgfonlayer}
	\begin{pgfonlayer}{edgelayer}
		\draw (1.center) to (6.center);
		\draw (2.center) to (5.center);
		\draw [style=blue] (13.center) to (15.center);
		\draw [style=blue] (14.center) to (16.center);
		\draw [style=blue] (17.center) to (18.center);
		\draw [style=blue] (19.center) to (20.center);
		\draw [style=blue] (21.center) to (24.center);
		\draw [style=blue] (22.center) to (23.center);
		\draw (25.center) to (27.center);
		\draw (26.center) to (28.center);
		\draw (29.center) to (30.center);
		\draw [style=blue] (36.center) to (37.center);
		\draw [style=blue] (37.center) to (40.center);
		\draw [style=blue] (40.center) to (38.center);
		\draw [style=blue] (38.center) to (39.center);
		\draw [style=blue] (40.center) to (41.center);
		\draw [style=blue] (41.center) to (42.center);
		\draw [style=blue] (43.center) to (44.center);
		\draw [style=blue] (44.center) to (47.center);
		\draw [style=blue] (47.center) to (45.center);
		\draw [style=blue] (45.center) to (46.center);
		\draw [style=blue] (47.center) to (48.center);
		\draw [style=blue] (48.center) to (49.center);
	\end{pgfonlayer}
\end{tikzpicture}}$
    \caption{The `augmented' version of the branched route $\lambda_N$ described in Figure \ref{fig: Looking inside}. The ancillary input wires for branches $N^\mathrm{II}$ and $N^\mathrm{IV}$ are not written, as they are trivial: each of these branches has only one output value.}
    \label{fig: Augmenting}
\end{figure}

\begin{definition}[Augmenting]
We take a branched route $\laN$. For each of its branches $\Nal$ we denote the set of output values of this branch as $\Indout_\Nal$, and define a binary set $\Happ_\Nal \cong \{0,1\}$.

The \emph{augmented version} $\la_N^\aug$ of $\laN$ is the partial function going from $\laN$'s input values and from the $\Indout_\Nal$'s, to $\laN$'s output values and the $\Happ_\Nal$'s, defined in the following way:  
\begin{itemize}
    \item if its argument from $\laN$'s input values is among the input values of a branch $\Nal$, then it returns its $\Indout_\Nal$'s argument, value $1$ in $\Happ_\Nal$, and value $0$ in $\Happ_{N^{\al'}}$ for $\al' \neq \al$;
    \item if its argument from $\laN$'s input values is not among the input values of any branch -- i.e.\ if it is outside of $\laN$'s practical input values --, then the output is undefined. 
\end{itemize}
\end{definition}

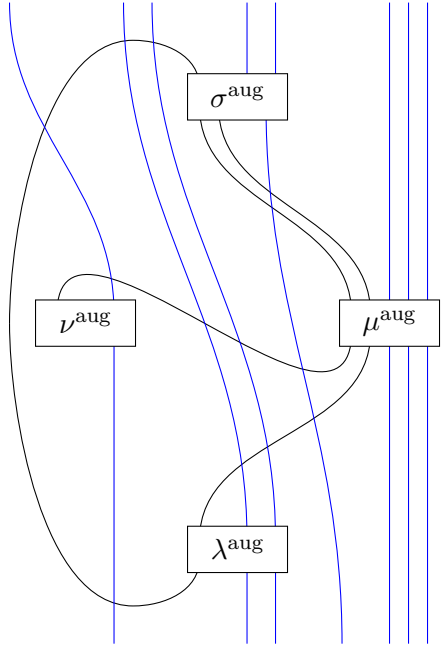
\begin{figure}
    \centering
    \begin{tikzpicture}
	\begin{pgfonlayer}{nodelayer}
		\node [style=semilarge map] (0) at (4, 0) {$\mu^\aug$};
		\node [style=semilarge map] (1) at (0, -6) {$\lambda^\aug$};
		\node [style=semilarge map] (2) at (-4, 0) {$\nu^\aug$};
		\node [style=semilarge map] (3) at (0, 6) {$\sigma^\aug$};
		\node [style=none] (4) at (0.25, 8.5) {};
		\node [style=none] (7) at (4, 0.5) {};
		\node [style=none] (9) at (3.5, -0.25) {};
		\node [style=none] (10) at (-1, -5.75) {};
		\node [style=none] (14) at (3.5, 0.25) {};
		\node [style=none] (15) at (3, 0.25) {};
		\node [style=none] (16) at (-3.25, -0.25) {};
		\node [style=none] (17) at (-4.75, 0.25) {};
		\node [style=none] (18) at (3, -0.25) {};
		\node [style=none] (19) at (4.5, 0.5) {};
		\node [style=none] (21) at (5, -0.5) {};
		\node [style=none] (24) at (-4.75, -0.25) {};
		\node [style=none] (26) at (-1, 6.25) {};
		\node [style=none] (28) at (-6, 0) {};
		\node [style=none] (29) at (-2.75, 7.5) {};
		\node [style=none] (31) at (-1, 5.75) {};
		\node [style=none] (32) at (-0.5, 5.75) {};
		\node [style=none] (33) at (-1, -6.25) {};
		\node [style=none] (34) at (-6, 0) {};
		\node [style=none] (35) at (-2.75, -7.5) {};
		\node [style=none] (36) at (-6, 8.5) {};
		\node [style=none] (37) at (-3.25, 0.25) {};
		\node [style=none] (38) at (-6, 8.5) {};
		\node [style=none] (39) at (-3.25, -8.5) {};
		\node [style=none] (40) at (0.25, -8.5) {};
		\node [style=none] (41) at (1, -8.5) {};
		\node [style=none] (42) at (0.25, -6.25) {};
		\node [style=none] (43) at (1, -6.25) {};
		\node [style=none] (44) at (4.5, -0.5) {};
		\node [style=none] (45) at (4, -0.5) {};
		\node [style=none] (46) at (4, -8.5) {};
		\node [style=none] (47) at (4.5, -8.5) {};
		\node [style=none] (48) at (5, -8.5) {};
		\node [style=none] (49) at (5, 0.5) {};
		\node [style=none] (50) at (4, 8.5) {};
		\node [style=none] (51) at (4.5, 8.5) {};
		\node [style=none] (52) at (5, 8.5) {};
		\node [style=none] (53) at (0.75, 5.75) {};
		\node [style=none] (54) at (1, 6.5) {};
		\node [style=none] (55) at (0.25, 6.5) {};
		\node [style=none] (56) at (1, 8.5) {};
		\node [style=none] (57) at (2.75, -8.5) {};
		\node [style=none] (58) at (0.25, -5.75) {};
		\node [style=none] (59) at (1, -5.75) {};
		\node [style=none] (60) at (-3, 8.5) {};
		\node [style=none] (61) at (-2.25, 8.5) {};
	\end{pgfonlayer}
	\begin{pgfonlayer}{edgelayer}
		\draw [in=-90, out=90] (10.center) to (9.center);
		\draw [in=0, out=90] (26.center) to (29.center);
		\draw [in=90, out=180, looseness=0.75] (29.center) to (28.center);
		\draw [in=0, out=-90] (33.center) to (35.center);
		\draw [in=-90, out=-180, looseness=0.75] (35.center) to (34.center);
		\draw [in=90, out=-90] (31.center) to (15.center);
		\draw [in=90, out=-90] (32.center) to (14.center);
		\draw [in=-90, out=90, looseness=1.25] (17.center) to (18.center);
		\draw [style=blue, in=90, out=-90] (38.center) to (37.center);
		\draw [style=blue] (55.center) to (4.center);
		\draw [style=blue] (56.center) to (54.center);
		\draw [style=blue] (50.center) to (7.center);
		\draw [style=blue] (51.center) to (19.center);
		\draw [style=blue] (52.center) to (49.center);
		\draw [style=blue] (45.center) to (46.center);
		\draw [style=blue] (47.center) to (44.center);
		\draw [style=blue] (21.center) to (48.center);
		\draw [style=blue, in=90, out=-90] (53.center) to (57.center);
		\draw [style=blue] (43.center) to (41.center);
		\draw [style=blue] (42.center) to (40.center);
		\draw [style=blue] (16.center) to (39.center);
		\draw [style=blue, in=90, out=-90] (60.center) to (58.center);
		\draw [style=blue, in=90, out=-90] (61.center) to (59.center);
	\end{pgfonlayer}
\end{tikzpicture}
    \caption{The `choice relation' for the routed graph of Figure \ref{fig:Routed Graph}. Here, we are assuming that $\nu$ has one branch, that $\lambda$ and $\sigma$ have two (with one of $\sigma$'s branches having trivial bifurcation choices), and $\mu$ has three. For better readability, trivial wires are left implicit and ancillary wires are written in blue.
    }
    \label{fig: Univocality}
\end{figure}

To illustrate this definition, let us display what it gives in the case of the routed graph of the switch, depicted in Figure \ref{fig:rswitch routed graph}. The augmented version of $P$'s route has one extra binary input (encoding the bifurcation choice that can be made at this node), and no extra output (as $P$ only features one branch); this augmented version is just the identity function from the extra input to the original output (remembering that the original input of the route was trivial). As for the route of node $A$, its augmented version features no extra input (as $A$ features no bifurcation choice), and two extra outputs: the first one encodes whether branch $A^0$ happened, the second one encodes whether branch $A^1$ happened. This augmented version (which has two inputs and four outputs) is the partial function defined by $(0,0) \mapsto (0,0,1,0)$ and $(1,1) \mapsto (1,1,0,1)$ and undefined on inputs $(0,1)$ and $(1,0)$ (which are outside the practical input values). The same goes for $B$. Finally, the augmented version of $F$'s route is identical to the original version (there are no bifurcation choices and only one branch).

As represented in Figure \ref{fig: Augmenting}, the augmented version of a route features extra ancillary wires. One can then form a relation by connecting the non-ancillary wires of the $\la_N^\aug$'s according to the indexed graph $\Ga$ (see Figure \ref{fig: Univocality} for an example)\footnote{Similarly to before, this procedure makes sense because relations form a traced monoidal category. The $\la_N^\aug$ are here viewed as relations, as any partial function can be.}.  We call this the `choice relation', which we will write $\Lambda_{\GalaN}$. $\Lambda_{\GalaN}$ goes from the bifurcation choices to the $\Happ$ binary variables that tell us whether each branch happened. The requirement that the former unambiguously determine the latter then takes a natural form.

\begin{principle}[Univocality and bi-univocality]
A routed graph $\GalaN$ satisfies the principle of \emph{univocality} if its choice relation $\Lambda_{\GalaN}$ is a function.

$\GalaN$ satisfies the principle of \emph{bi-univocality} if both it and its adjoint $(\Ga^\top, (\la_N^\top)_N)$ satisfy univocality.
\end{principle}

The adjoint of a routed graph is simply the routed graph obtained by reversing the direction of its arrows, and taking the adjoints of its routes: it can be interpreted as its time-reversed version. Being bi-univocal thus means being `univocal both ways'.

When univocality is satisfied, the choice relation -- which is then a choice function -- plays another role: its causal structure (defined by functional dependence) tells us which bifurcation choices can affect the status of which branch. This will define the green dashed arrows in the branch graph, whereas the analogous information in the choice function of the adjoint graph will define the (reverse of the) red dashed arrows.

Our last job is to define the solid arrows in the branch graph. The idea is that the `$N^\al$' branch of node $N$ has a direct influence on the `$\Mbe$' branch of node $M$ if there is an arrow from $N$ to $M$ that doesn't become either inconsistent or trivial (i.e.\ reduce to either zero sectors or to a single one-dimensional one) when one fixes $N$ to be in branch $\al$ and $M$ to be in branch $\beta$. To capture this, we will have to talk about \textit{consistent assignments of values} to the indices of all arrows in the graph.

\begin{definition}[Consistent assignment]
A \emph{consistent assignment} of values to $\GalaN$'s indices is an assignment of a value to the arrows' indices, such that for any node $N$, the tuple of values for $N$'s inputs is related by $\laN$ to the tuple of values for $N$'s outputs.
\end{definition}

Note that (as proven in Appendix \ref{app: Theorem}) an assignment is consistent if and only if at every node, the tuple of input values and the tuple of output values that it yields are in the same branch (and in particular are not outside the practical inputs/outputs). In that sense, one can talk about this consistent assignment of values as, in particular, assigning a given branch to every node.

The idea of solid arrows, embodied by the following definition, is then that one draws a solid arrow from $\Nal$ to $\Mbe$ if there is an arrow  $A$ joining $N$ to $M$, except if $\Nal$ and $\Mbe$ can never happen jointly, or if there is a single value of $A$'s index compatible with both of them happening, and this value makes $A$ trivial.

\begin{definition}
Taking a branch $\Nal$ of node $N$ and a branch $\Mbe$ of node $M$, we say that there is a solid arrow $\Nal \to \Mbe$ if there exists an arrow from $N$ to $M$, except if:
\begin{itemize}
    \item there are no consistent assignment of values that assign branch $\al$ to $N$ and branch $\bet$ to $M$;
    \item or if all such assignments assign the same value to the index of the arrow $N \to M$, and this value has dimension 1 (i.e.\ corresponds to a one-dimensional sector).
\end{itemize}
If there are several arrows from $N$ to $M$, then we say that there is a solid arrow $\Nal \to \Mbe$ unless the above applies to all of them.
\end{definition}

With this in our toolbox, we can define the branch graph.\footnote{\changes{Note that the branch graph is not defined if bi-univocality is not satisfied. Indeed, if either $\Lambda_{\GalaN}$ or $\Lambda_{\left(\Ga^\top, \left( \la_N^\top \right)_N \right)}$ is not a function but merely a relation, we cannot talk about its causal structure and therefore we cannot define either the green or the red dashed arrows. This has the important consequence that bi-univocality and weak loops are not logically independent principles; making sense of the second requires the first to hold.}}

\begin{definition}[Branch graph]
The \emph{branch graph} of a routed graph $\GalaN$ that satisfies bi-univocality is the graph in which:

\begin{itemize}
    \item the nodes are given by the branches of $\GalaN$'s nodes;
    \item solid arrows are given by the previous definition;
    \item there is a green dashed arrow $N^\al \to \Mbe$ if the choice function $\Lambda_{\GalaN}$ features causal influence (i.e.\ functional dependence) from $\IndoutNal$ to $\HappMbe$;
    \item there is a red dashed arrow $N^\al \to \Mbe$ if the choice function of the adjoint graph, $\Lambda_{\left(\Ga^\top, \left( \la_N^\top \right)_N \right)}$, features causal influence (i.e.\ functional dependence) from $\IndinMbe$ to $\HappNal$.
\end{itemize}
\end{definition}

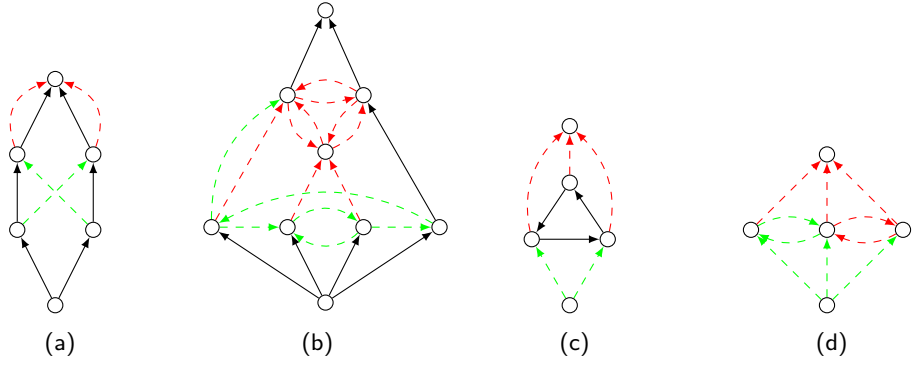
\begin{figure}
    \centering
    \begin{subfigure}[b]{0.2\textwidth}
         \centering
         \begin{tikzpicture}
	\begin{pgfonlayer}{nodelayer}
		\node [style=white dot] (0) at (0, -3) {};
		\node [style=white dot] (1) at (-1, -1) {};
		\node [style=white dot] (2) at (1, -1) {};
		\node [style=white dot] (3) at (-1, 1) {};
		\node [style=white dot] (4) at (1, 1) {};
		\node [style=white dot] (5) at (0, 3) {};
	\end{pgfonlayer}
	\begin{pgfonlayer}{edgelayer}
		\draw [style=arrow plain] (0) to (1);
		\draw [style=arrow plain] (0) to (2);
		\draw [style=arrow plain] (1) to (3);
		\draw [style=arrow plain] (2) to (4);
		\draw [style=arrow plain] (3) to (5);
		\draw [style=arrow plain] (4) to (5);
		\draw [style=green dashed arrow] (2) to (3);
		\draw [style=green dashed arrow] (1) to (4);
		\draw [style=red dashed arrow, bend right=45, looseness=1.25] (4) to (5);
		\draw [style=red dashed arrow, bend left=45, looseness=1.25] (3) to (5);
	\end{pgfonlayer}
\end{tikzpicture}
         \caption{}
     \end{subfigure}%
     \quad%
     \begin{subfigure}[b]{0.2\textwidth}
         \centering
         \begin{tikzpicture}
	\begin{pgfonlayer}{nodelayer}
		\node [style=white dot] (0) at (0, -4) {};
		\node [style=white dot] (1) at (-3, -2) {};
		\node [style=white dot] (2) at (-1, -2) {};
		\node [style=white dot] (3) at (1, -2) {};
		\node [style=white dot] (4) at (3, -2) {};
		\node [style=white dot] (5) at (0, 0) {};
		\node [style=white dot] (6) at (-1, 1.5) {};
		\node [style=white dot] (7) at (1, 1.5) {};
		\node [style=white dot] (8) at (0, 3.75) {};
	\end{pgfonlayer}
	\begin{pgfonlayer}{edgelayer}
		\draw [style=arrow plain] (0) to (1);
		\draw [style=arrow plain] (0) to (2);
		\draw [style=arrow plain] (0) to (3);
		\draw [style=arrow plain] (0) to (4);
		\draw [style=green dashed arrow] (1) to (2);
		\draw [style=green dashed arrow] (3) to (4);
		\draw [style=green dashed arrow, bend right] (4) to (1);
		\draw [style=green dashed arrow, bend left=45] (2) to (3);
		\draw [style=green dashed arrow, bend left=45] (3) to (2);
		\draw [style=red dashed arrow, bend right, looseness=1.25] (5) to (7);
		\draw [style=red dashed arrow, bend right=15, looseness=1.25] (7) to (5);
		\draw [style=red dashed arrow, bend right, looseness=1.25] (7) to (6);
		\draw [style=red dashed arrow, bend right, looseness=1.25] (6) to (5);
		\draw [style=red dashed arrow, bend right=15, looseness=1.25] (6) to (7);
		\draw [style=red dashed arrow, bend right=15, looseness=1.25] (5) to (6);
		\draw [style=arrow plain] (7) to (8);
		\draw [style=red dashed arrow] (2) to (5);
		\draw [style=red dashed arrow] (3) to (5);
		\draw [style=arrow plain] (4) to (7);
		\draw [style=arrow plain] (6) to (8);
		\draw [style=green dashed arrow, bend left] (1) to (6);
		\draw [style=red dashed arrow] (1) to (6);
	\end{pgfonlayer}
\end{tikzpicture}
         \caption{}
     \end{subfigure}%
     \quad%
     \begin{subfigure}[b]{0.2\textwidth}
         \centering
         \begin{tikzpicture}
	\begin{pgfonlayer}{nodelayer}
		\node [style=white dot] (0) at (-1, -0.5) {};
		\node [style=white dot] (1) at (1, -0.5) {};
		\node [style=white dot] (2) at (0, 1) {};
		\node [style=white dot] (3) at (0, -2.25) {};
		\node [style=white dot] (4) at (0, 2.5) {};
	\end{pgfonlayer}
	\begin{pgfonlayer}{edgelayer}
		\draw [style=green dashed arrow] (3) to (0);
		\draw [style=green dashed arrow] (3) to (1);
		\draw [style=red dashed arrow] (2) to (4);
		\draw [style=red dashed arrow, bend left] (0) to (4);
		\draw [style=red dashed arrow, bend right] (1) to (4);
		\draw [style=arrow plain] (0) to (1);
		\draw [style=arrow plain] (1) to (2);
		\draw [style=arrow plain] (2) to (0);
	\end{pgfonlayer}
\end{tikzpicture}
         \caption{}
     \end{subfigure}%
     \quad%
     \begin{subfigure}[b]{0.2\textwidth}
         \centering
         \begin{tikzpicture}
	\begin{pgfonlayer}{nodelayer}
		\node [style=white dot] (0) at (-2, 0) {};
		\node [style=white dot] (1) at (0, 0) {};
		\node [style=white dot] (2) at (2, 0) {};
		\node [style=white dot] (3) at (0, -2) {};
		\node [style=white dot] (4) at (0, 2) {};
	\end{pgfonlayer}
	\begin{pgfonlayer}{edgelayer}
		\draw [style=green dashed arrow, bend left] (0) to (1);
		\draw [style=green dashed arrow, bend left] (1) to (0);
		\draw [style=red dashed arrow, bend left] (1) to (2);
		\draw [style=red dashed arrow, bend left] (2) to (1);
		\draw [style=green dashed arrow] (3) to (0);
		\draw [style=green dashed arrow] (3) to (1);
		\draw [style=green dashed arrow] (3) to (2);
		\draw [style=red dashed arrow] (0) to (4);
		\draw [style=red dashed arrow] (1) to (4);
		\draw [style=red dashed arrow] (2) to (4);
	\end{pgfonlayer}
\end{tikzpicture}
         \caption{}
     \end{subfigure}%
    \caption{Examples of branch graphs. (a) and (b) satisfy the weak loops principle, but (c) and (d) do not. For (d), this is due to the presence of a bi-coloured $\infty$-shaped loop in the central layer.}
    \label{fig: Branch Graphs}
\end{figure}

Examples of branch graphs are shown in Figure \ref{fig: Branch Graphs}. Now that the branch graph is defined, we can check whether it satisfies our second principle.

\begin{principle}[Weak loops]
We say that a loop in a branch graph is \emph{weak} if it is entirely made of dashed arrows of the same colour.

A routed graph satisfies the principle of weak loops if every loop in its branch graph is weak.
\end{principle}

Note that, as a particular case, any routed graph whose branch graph features no loops trivially satisfies this principle. This will be sufficient to check the consistency of processes featuring (possibly dynamical) coherent control of causal order. We will conjecture that the more exotic processes, which violate causal inequalities, are characterised by the existence of weak loops in their branch graph.

Finally, we can display our main theorem.

\begin{theorem} \label{thm: Main}
Let $\GalaN$ be a routed graph satisfying the principles of bi-univocality and weak loops. (We then say that it is valid.) Then its associated skeletal supermap is a routed superunitary.
\end{theorem}

The proof of Theorem \ref{thm: Main} is given in Appendix \ref{app: Theorem}.

The next corollary, which is direct, stresses the fact that there are then many supermaps which can be obtained from this skeletal supermap, and that the validity of the latter implies that they are valid as well.

\begin{corollary}
Let $\GalaN$ be a valid routed graph. Then, any supermap built from its associated skeletal supermap by plugging in unitaries at some of its nodes and unitary monopartite supermaps at other nodes is a superunitary.
\end{corollary}

\section{Examples of constructing processes with indefinite causal order} \label{sec: examples}

In this section, we reconstruct three further examples of processes with indefinite causal order from valid routed graphs, namely the quantum 3-switch, the Grenoble process and the Lugano process. 
This will enable us to see each of these processes as a member of a large family of processes that can be constructed `in the same way' -- i.e. from the same routed graph.
This in turn will allow us to distinguish between those features of the process that are `accidental', and those that are essential for the consistency of the process.

What results is reminiscent of the situation for processes without indefinite causal order. 
Such processes can be represented as circuits, in which it is immediate that changing the particular transformations will preserve the consistency of the process, so long as the connectivity of the circuit is maintained.
In our reconstructions of processes with indefinite causal order, it is immediate that changing the particular transformations in a routed circuit decomposition of a process with indefinite causal order preserves consistency, so long as the resulting routed circuit is still `fleshing out' the same (valid) routed graph.

\bigskip

Before turning to the examples, we briefly explain a shorthand way of presenting the routed graphs. 
Rather than directly stating the route associated with each node, it is sometimes simpler to specify a \textit{global index constraint}, from which the individual routes can be derived. 
This global index constraint specifies the allowed joint value-assignments for all of the indices in the graph.
Formally, it can be represented as a Boolean tensor $G$ over the Cartesian product of all the indices in the graph. We set the coefficients $G_i=1$ to be equal to 1 for allowed joint value assignments $i$ (note that $i$ here denotes a list of indices), and $G_i=0$ for those that are not allowed.
Then, we can calculate the route $\lambda_N$ at some specific node $N$ as the least restrictive route consistent with the global index constraint. 
Assuming that the set of indices $\texttt{Ind}_N^\textnormal{in}$ going into the node and the set of outgoing indices $\texttt{Ind}_N^\textnormal{out}$ are disjoint, so that no arrows start and finish at that same node,\footnote{Note that the requirements of bi-univocality and weak loops imply that any  indices on these `self-loops' would have to have values that either never happen, or else correspond to one-dimensional branches, and that the value of the indices are fixed by the other ingoing and outgoing arrows. This means that one gets exactly the same unitary processes from the routed graph if one removes the self-loops. Accordingly, the assumption of no self-loops does not sacrifice any generality.} we can calculate this by marginalising over the indices $\texttt{Ind}\backslash \{\texttt{Ind}_N^\textnormal{in} \sqcup \texttt{Ind}_N^\textnormal{out}\}$ that do not come out of or go into the node.
Writing the indices as $i=(i_\textnormal{in}i_\textnormal{out}i')$, where the $i' \in \{\texttt{Ind}_N^\textnormal{in} \sqcup \texttt{Ind}_N^\textnormal{out}\}$ denote the joint value-assignments of those `irrelevant indices', the marginalisation is performed by taking the Boolean sum over $i'$,
\begin{equation}
    (\lambda_N)_{i_\textnormal{in}}^{i_\textnormal{out}} := \sum_{i'} G_{i_\textnormal{in}i_\textnormal{out}i'} \, .
\end{equation}

In the examples below, we represent the global index by using a combination of index-matching on the routed graph and `floating' equations relating the indices written beside the graph. The idea with index-matching is that when indices on two different arrows are matched, the global index constraint must be 0 for all joint value assignments in which they are not equal. Similarly, the global index constraint is 0 for all joint value assignments not satisfying the floating equations. The global index constraint of the routed graph is then the most general Boolean matrix compatible with the index-matching and the equations.

Similarly, we will also present routed \textit{circuits} using a global index constraint. In that case, we derive the routes associated with the individual \textit{transformations} (i.e.\ the boxes) that make up the circuit by marginalising over the global index constraint of the circuit.

We want to stress, once again, that the use of index-matching and global index constraints is only a graphical shorthand: in order to study the graphs and check the principles, they have to be formally translated into routes for the nodes.

\subsection{The quantum 3-switch}

The quantum 3-switch \cite{colnaghi2012quantum} is a unitary process defined analogously to the quantum switch, but with three intermediate agents: Alice ($A$), Bob ($B$) and Charlie ($C$). The Past ($P$) consists of a 6-dimensional control qudit $P_C$ and a $d$-dimensional target qudit $P_T$. Depending on the initial state of the control qudit, the three agents receive the target qudit in a different order, outlined in Table \ref{tab:3switch}. At the end, the target qudit is sent to the Future ($F$).
\begin{table}[ht]
    \centering
    \begin{tabular}{c|c}
    Control state     & Order \\
    \hline
      $\ket{1}$   & $A-B-C$ \\
      $\ket{2}$   & $A-C-B$ \\
      $\ket{3}$   & $B-C-A$ \\
      $\ket{4}$   & $B-A-C$ \\
      $\ket{5}$   & $C-A-B$ \\
      $\ket{6}$   & $C-B-A$ \\
    \end{tabular}
    \caption{The relative order of the agents Alice ($A$), Bob ($B$) and Charlie ($C$) depending on the value of the control state.}
    \label{tab:3switch}
\end{table}

\subsubsection{The routed graph}

We start by drawing a routed graph from which the quantum 3-switch, amongst other processes, can be constructed. This routed graph is given in Figure \ref{fig:3switch_routedgraph}. The global index constraint is represented by matching the indices on different arrows, and by the floating equation $l+m+n+p+q+r=1$. This equation enforces that precisely one of the six summed over indices is equal to one. Thus the global index constraint is the Boolean matrix that ensures that matched indices take the same value, and that exactly one of the six distinct values is 1.

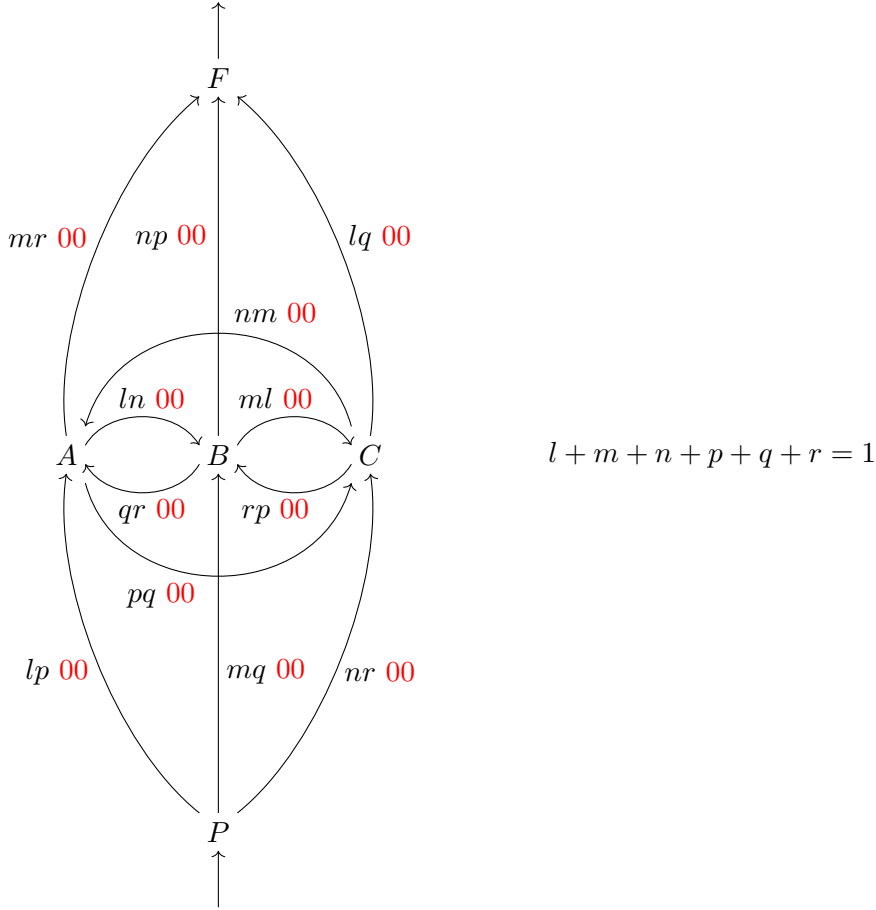
\begin{figure}[ht]
	\centering
	\begin{tikzpicture}
	\begin{pgfonlayer}{nodelayer}
		\node [style=none] (1) at (0, -10) {$P$};
		\node [style=none] (2) at (0, 0) {$B$};
		\node [style=none] (3) at (-4, 0) {$A$};
		\node [style=none] (4) at (4, 0) {$C$};
		\node [style=none] (6) at (0, -12) {};
		\node [style=none] (8) at (-0.5, -9.5) {};
		\node [style=none] (9) at (0.5, -9.5) {};
		\node [style=none] (10) at (-4, -0.5) {};
		\node [style=none] (11) at (4, -0.5) {};
		\node [style=none] (12) at (-3.5, -0.75) {};
		\node [style=none] (13) at (3.5, -0.25) {};
		\node [style=none] (14) at (0, -0.5) {};
		\node [style=none] (15) at (0, -9.5) {};
		\node [style=none] (16) at (0, -10.5) {};
		\node [style=none] (17) at (-0.5, -0.25) {};
		\node [style=none] (18) at (0.5, -0.25) {};
		\node [style=none] (19) at (-3.5, -0.25) {};
		\node [style=none] (20) at (3.5, -0.75) {};
		\node [style=none] (23) at (-4.25, -5.75) {$lp~{\textcolor{red}{00}}$};
		\node [style=none] (24) at (4.25, -5.75) {$nr~{\textcolor{red}{00}}$};
		\node [style=none] (25) at (1.25, -5.75) {$mq~{\textcolor{red}{00}}$};
		\node [style=none] (26) at (-1.75, -1.5) {$qr~{\textcolor{red}{00}}$};
		\node [style=none] (27) at (1.5, -1.5) {$rp~{\textcolor{red}{00}}$};
		\node [style=none] (28) at (-1.5, -3.75) {$pq~{\textcolor{red}{00}}$};
		\node [style=none] (29) at (0, 10) {$F$};
		\node [style=none] (33) at (0, 12) {};
		\node [style=none] (34) at (-0.5, 9.5) {};
		\node [style=none] (35) at (0.5, 9.5) {};
		\node [style=none] (36) at (-4, 0.5) {};
		\node [style=none] (37) at (4, 0.5) {};
		\node [style=none] (38) at (3.5, 0.75) {};
		\node [style=none] (39) at (3.5, 0.25) {};
		\node [style=none] (40) at (0, 0.5) {};
		\node [style=none] (41) at (0, 9.5) {};
		\node [style=none] (42) at (0, 10.5) {};
		\node [style=none] (43) at (-0.5, 0.25) {};
		\node [style=none] (44) at (0.5, 0.25) {};
		\node [style=none] (45) at (-3.5, 0.25) {};
		\node [style=none] (46) at (-3.5, 0.75) {};
		\node [style=none] (47) at (-4.5, 5.75) {$mr~{\textcolor{red}{00}}$};
		\node [style=none] (48) at (4.25, 5.75) {$lq~{\textcolor{red}{00}}$};
		\node [style=none] (49) at (-1.25, 5.75) {$np~{\textcolor{red}{00}}$};
		\node [style=none] (50) at (-1.75, 1.5) {$ln~{\textcolor{red}{00}}$};
		\node [style=none] (51) at (1.5, 1.5) {$ml~{\textcolor{red}{00}}$};
		\node [style=none] (52) at (1.5, 3.75) {$nm~{\textcolor{red}{00}}$};
		\node [style=none] (55) at (13, 0) {$l+m+n+p+q+r=1$};
	\end{pgfonlayer}
	\begin{pgfonlayer}{edgelayer}
		\draw [style=arrow, bend right, looseness=0.75] (9.center) to (11.center);
		\draw [style=arrow] (15.center) to (14.center);
		\draw [style=arrow, bend left, looseness=0.75] (8.center) to (10.center);
		\draw [style=arrow, bend left=60] (17.center) to (19.center);
		\draw [style=arrow, bend left=60] (13.center) to (18.center);
		\draw [style=arrow, bend right=75, looseness=1.25] (12.center) to (20.center);
		\draw [style=arrow] (42.center) to (33.center);
		\draw [style=arrow, bend right=75, looseness=1.25] (38.center) to (46.center);
		\draw [style=arrow] (6.center) to (16.center);
		\draw [style=arrow, bend left, looseness=0.75] (36.center) to (34.center);
		\draw [style=arrow, bend right, looseness=0.75] (37.center) to (35.center);
		\draw [style=arrow, bend left=60] (45.center) to (43.center);
		\draw [style=arrow, bend left=60] (44.center) to (39.center);
		\draw [style=arrow] (40.center) to (41.center);
	\end{pgfonlayer}
\end{tikzpicture}
    	\caption{The routed graph for the quantum 3-switch, using a global index constraint.}
	\label{fig:3switch_routedgraph}
\end{figure} 

The route at node $P$ (which we denote $\eta$) is, by definition, the most liberal route compatible with the global index constraint. This is the route that forces exactly one of its indices to be equal to $1$:
\begin{equation}
	\begin{cases}
		\eta^{100000} = \eta^{010000} = \eta^{001000}=\eta^{000100} =\eta^{000010}=\eta^{000001} = 1 \, ; \\
		 \eta^{lmnpqr} =0 \quad {\rm otherwise}.
	\end{cases}
\end{equation}
The route $\eta$ also has a convenient graphical representation, depicted in Figure \ref{fig:3switchpast}. $\eta$ has a single branch with a bifurcation choice between six options, each corresponding to one of the indices $lmnpqr$ being equal to 1. 
Each option enforces one of the six possible causal orders.

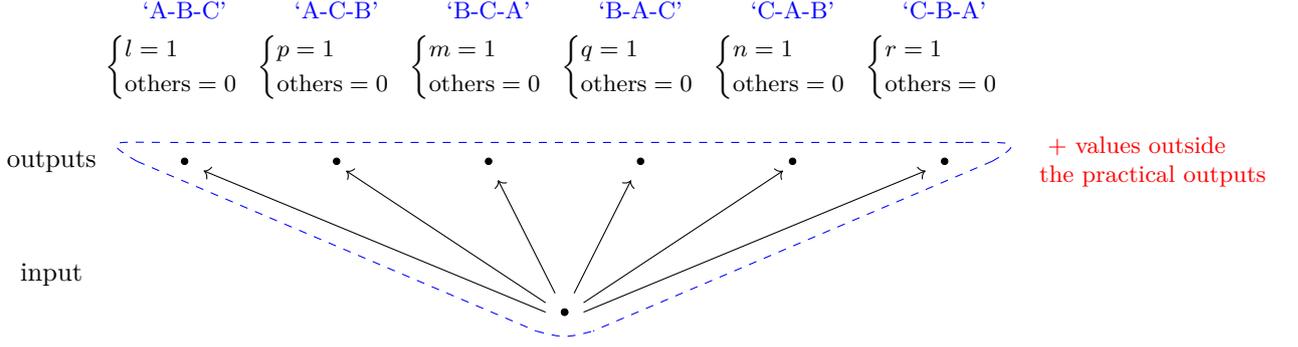
\begin{figure}
    \centering
    \begin{tikzpicture}
	\begin{pgfonlayer}{nodelayer}
		\node [style=none] (0) at (-16, -2) {};
		\node [style=none] (1) at (-25, 1.75) {};
		\node [style=none] (2) at (-21.25, 1.75) {};
		\node [style=none] (3) at (-16, -1.75) {};
		\node [style=none] (4) at (-26.75, 2) {};
		\node [style=none] (5) at (-16.25, -2.5) {};
		\node [style=none] (6) at (-14.75, -2.5) {};
		\node [style=none] (7) at (-4.25, 2) {};
		\node [style=very small black dot] (10) at (-15.5, -2) {};
		\node [style=very small black dot] (11) at (-15.5, -2) {};
		\node [style=very small black dot] (12) at (-25.5, 2) {};
		\node [style=none] (14) at (-26, 2.5) {};
		\node [style=none] (15) at (-5, 2.5) {};
		\node [style=none] (16) at (-26.75, 2) {};
		\node [style=none] (17) at (-26, 2.5) {};
		\node [style=none] (26) at (-29, 2) {\small{outputs}};
		\node [style=none] (27) at (-29, -1) {\small{input}};
		\node [style=none] (29) at (0.5, -4.75) {};
		\node [style=none, font={\footnotesize}] (32) at (-25.5, 4.5) {\footnotesize{$\begin{cases} l= 1 \\ {\rm others =0} \end{cases}$}};
		\node [style=very small black dot] (68) at (-21.5, 2) {};
		\node [style=very small black dot] (69) at (-17.5, 2) {};
		\node [style=very small black dot] (70) at (-13.5, 2) {};
		\node [style=very small black dot] (71) at (-9.5, 2) {};
		\node [style=very small black dot] (72) at (-5.5, 2) {};
		\node [style=none, font={\footnotesize}] (74) at (-21.5, 4.5) {$\begin{cases} p=1 \\ {\rm others =0} \end{cases}$};
		\node [style=none, font={\footnotesize}] (75) at (-17.5, 4.5) {$\begin{cases} m= 1 \\ {\rm others =0} \end{cases}$};
		\node [style=none, font={\footnotesize}] (77) at (-13.5, 4.5) {$\begin{cases} q= 1 \\ {\rm others =0} \end{cases}$};
		\node [style=none, font={\footnotesize}] (78) at (-9.5, 4.5) {$\begin{cases} n= 1 \\ {\rm others =0} \end{cases}$};
		\node [style=none, font={\footnotesize}] (79) at (-5.5, 4.5) {$\begin{cases} r= 1 \\ {\rm others =0} \end{cases}$};
		\node [style=none] (80) at (-17.25, 1.5) {};
		\node [style=none] (81) at (-13.75, 1.5) {};
		\node [style=none] (82) at (-9.75, 1.75) {};
		\node [style=none] (83) at (-6, 1.75) {};
		\node [style=none] (84) at (-15.75, -1.5) {};
		\node [style=none] (85) at (-15.25, -1.5) {};
		\node [style=none] (86) at (-15, -1.75) {};
		\node [style=none] (87) at (-15, -2) {};
		\node [style=none] (88) at (-25.5, 6) {\footnotesize{\color{blue}`A-B-C'}};
		\node [style=none] (89) at (-21.5, 6) {\footnotesize{\color{blue}`A-C-B'}};
		\node [style=none] (90) at (-17.5, 6) {\footnotesize{\color{blue}`B-C-A'}};
		\node [style=none] (91) at (-13.5, 6) {\footnotesize{\color{blue}`B-A-C'}};
		\node [style=none] (92) at (-9.5, 6) {\footnotesize{\color{blue}`C-A-B'}};
		\node [style=none] (93) at (-5.5, 6) {\footnotesize{\color{blue}`C-B-A'}};
		\node [style=none, align=justify, font={\footnotesize}, text width=4.5cm] (94) at (1.5, 2) { \footnotesize{ \color{red} + values outside} \\  {\color{red}the practical outputs}};
	\end{pgfonlayer}
	\begin{pgfonlayer}{edgelayer}
		\draw [style=arrow] (0.center) to (1.center);
		\draw [style=arrow] (3.center) to (2.center);
		\draw [style=dashed blue] (4.center) to (5.center);
		\draw [style=dashed blue, in=-165, out=-15, looseness=1.25] (5.center) to (6.center);
		\draw [style=dashed blue, in=0, out=30, looseness=3.25] (7.center) to (15.center);
		\draw [style=dashed blue] (15.center) to (14.center);
		\draw [style=dashed blue, in=180, out=150, looseness=3.50] (16.center) to (17.center);
		\draw [style=dashed blue] (7.center) to (6.center);
		\draw [style=arrow] (84.center) to (80.center);
		\draw [style=arrow] (85.center) to (81.center);
		\draw [style=arrow] (86.center) to (82.center);
		\draw [style=arrow] (87.center) to (83.center);
	\end{pgfonlayer}
\end{tikzpicture}
    \caption{The route and branch structure of the node $P$ of the quantum $3$-switch. There is one unique branch with a bifurcation choice between six options, each of which enforces a causal order.}
    \label{fig:3switchpast}
\end{figure}


Let us explain how this works in detail. In the routed graph for the standard switch, the arrow $P\rightarrow A$ came with two index values, corresponding to whether or not Alice received the message first. But for the 3-switch, if Alice does receive the message first, then there are two further possibilities: either she comes first and the causal order is clockwise ($A-B-C$), or she comes first and the order is anticlockwise ($A-C-B$). For this reason, the arrow from $P$ to $A$ has three index values overall. The sectors where she gets the message first correspond to $(l=1, p=0)$ and $(l=0, p=1)$; while $(l=0, p=0)$ corresponds to a one-dimensional `dummy' sector. Likewise, all internal wires are associated with three sectors; two non-`dummy' sectors for when one of their indices equals one, and a `dummy' sector for when both are equal to zero\footnote{Note that the sectors with both indices equal to $1$, although formally present, are irrelevant: they correspond to impossible joint assignments of values.}.

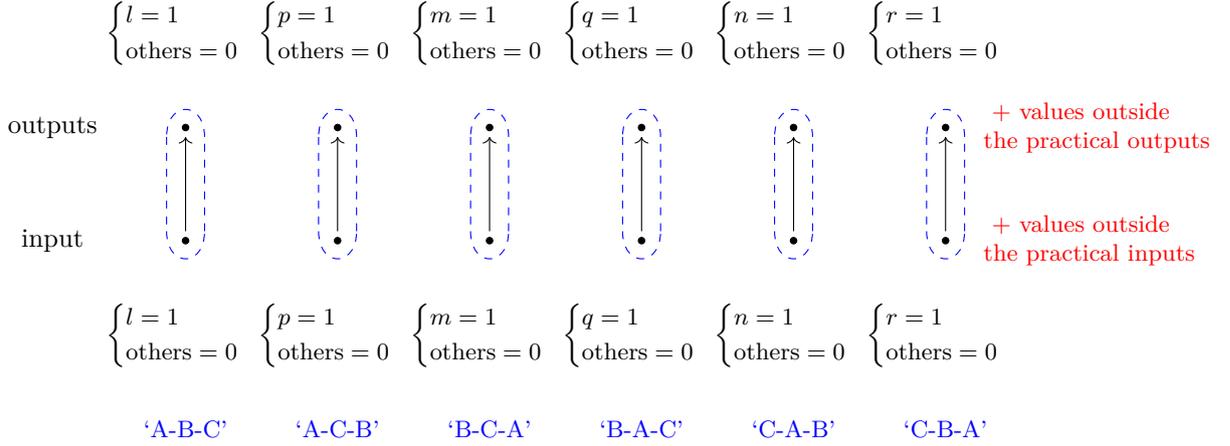
\begin{figure}
    \centering
    \begin{tikzpicture}
	\begin{pgfonlayer}{nodelayer}
		\node [style=none] (26) at (-29, 2) {\small{outputs}};
		\node [style=none] (27) at (-29, -1) {\small{input}};
		\node [style=none] (29) at (0.5, -4.75) {};
		\node [style=none, font={\footnotesize}] (32) at (-25.5, 4.5) {\footnotesize{$\begin{cases} l= 1 \\ {\rm others =0} \end{cases}$}};
		\node [style=none, font={\footnotesize}] (74) at (-21.5, 4.5) {$\begin{cases} p=1 \\ {\rm others =0} \end{cases}$};
		\node [style=none, font={\footnotesize}] (75) at (-17.5, 4.5) {$\begin{cases} m= 1 \\ {\rm others =0} \end{cases}$};
		\node [style=none, font={\footnotesize}] (77) at (-13.5, 4.5) {$\begin{cases} q= 1 \\ {\rm others =0} \end{cases}$};
		\node [style=none, font={\footnotesize}] (78) at (-9.5, 4.5) {$\begin{cases} n= 1 \\ {\rm others =0} \end{cases}$};
		\node [style=none, font={\footnotesize}] (79) at (-5.5, 4.5) {$\begin{cases} r= 1 \\ {\rm others =0} \end{cases}$};
		\node [style=none] (88) at (-25.5, -6) {\color{blue} \footnotesize{`A-B-C'}};
		\node [style=none] (89) at (-21.5, -6) {\footnotesize{\color{blue} `A-C-B'}};
		\node [style=none] (90) at (-17.5, -6) {\footnotesize{\color{blue}`B-C-A'}};
		\node [style=none] (91) at (-13.5, -6) {\footnotesize{\color{blue}`B-A-C'}};
		\node [style=none] (92) at (-9.5, -6) {\footnotesize{\color{blue}`C-A-B'}};
		\node [style=none] (93) at (-5.5, -6) {\footnotesize{\color{blue}`C-B-A'}};
		\node [style=none, align=justify, font={\footnotesize}, text width=4.5cm] (94) at (0, 2) { \footnotesize{ \color{red} + values outside} \\  {\color{red}the practical outputs}};
		\node [style=very small black dot] (108) at (-25.5, -1) {};
		\node [style=very small black dot] (109) at (-25.5, 2) {};
		\node [style=none] (110) at (-25.5, -0.75) {};
		\node [style=none] (111) at (-25.5, 1.75) {};
		\node [style=none] (112) at (-26, -0.75) {};
		\node [style=none] (113) at (-26, 1.75) {};
		\node [style=none] (114) at (-25, -0.75) {};
		\node [style=none] (115) at (-25, 1.75) {};
		\node [style=very small black dot] (116) at (-21.5, -1) {};
		\node [style=very small black dot] (117) at (-21.5, 2) {};
		\node [style=none] (118) at (-21.5, -0.75) {};
		\node [style=none] (119) at (-21.5, 1.75) {};
		\node [style=none] (120) at (-22, -0.75) {};
		\node [style=none] (121) at (-22, 1.75) {};
		\node [style=none] (122) at (-21, -0.75) {};
		\node [style=none] (123) at (-21, 1.75) {};
		\node [style=very small black dot] (124) at (-17.5, -1) {};
		\node [style=very small black dot] (125) at (-17.5, 2) {};
		\node [style=none] (126) at (-17.5, -0.75) {};
		\node [style=none] (127) at (-17.5, 1.75) {};
		\node [style=none] (128) at (-18, -0.75) {};
		\node [style=none] (129) at (-18, 1.75) {};
		\node [style=none] (130) at (-17, -0.75) {};
		\node [style=none] (131) at (-17, 1.75) {};
		\node [style=very small black dot] (132) at (-13.5, -1) {};
		\node [style=very small black dot] (133) at (-13.5, 2) {};
		\node [style=none] (134) at (-13.5, -0.75) {};
		\node [style=none] (135) at (-13.5, 1.75) {};
		\node [style=none] (136) at (-14, -0.75) {};
		\node [style=none] (137) at (-14, 1.75) {};
		\node [style=none] (138) at (-13, -0.75) {};
		\node [style=none] (139) at (-13, 1.75) {};
		\node [style=very small black dot] (140) at (-5.5, -1) {};
		\node [style=very small black dot] (141) at (-5.5, 2) {};
		\node [style=none] (142) at (-5.5, -0.75) {};
		\node [style=none] (143) at (-5.5, 1.75) {};
		\node [style=none] (144) at (-6, -0.75) {};
		\node [style=none] (145) at (-6, 1.75) {};
		\node [style=none] (146) at (-5, -0.75) {};
		\node [style=none] (147) at (-5, 1.75) {};
		\node [style=very small black dot] (148) at (-9.5, -1) {};
		\node [style=very small black dot] (149) at (-9.5, 2) {};
		\node [style=none] (150) at (-9.5, -0.75) {};
		\node [style=none] (151) at (-9.5, 1.75) {};
		\node [style=none] (152) at (-10, -0.75) {};
		\node [style=none] (153) at (-10, 1.75) {};
		\node [style=none] (154) at (-9, -0.75) {};
		\node [style=none] (155) at (-9, 1.75) {};
		\node [style=none, align=justify, font={\footnotesize}, text width=4.5cm] (156) at (0, -1) { \footnotesize{ \color{red} + values outside} \\  {\color{red}the practical inputs}};
		\node [style=none, font={\footnotesize}] (157) at (-25.5, -3.5) {\footnotesize{$\begin{cases} l= 1 \\ {\rm others =0} \end{cases}$}};
		\node [style=none, font={\footnotesize}] (158) at (-21.5, -3.5) {$\begin{cases} p=1 \\ {\rm others =0} \end{cases}$};
		\node [style=none, font={\footnotesize}] (159) at (-17.5, -3.5) {$\begin{cases} m= 1 \\ {\rm others =0} \end{cases}$};
		\node [style=none, font={\footnotesize}] (160) at (-13.5, -3.5) {$\begin{cases} q= 1 \\ {\rm others =0} \end{cases}$};
		\node [style=none, font={\footnotesize}] (161) at (-9.5, -3.5) {$\begin{cases} n= 1 \\ {\rm others =0} \end{cases}$};
		\node [style=none, font={\footnotesize}] (162) at (-5.5, -3.5) {$\begin{cases} r= 1 \\ {\rm others =0} \end{cases}$};
	\end{pgfonlayer}
	\begin{pgfonlayer}{edgelayer}
		\draw [style=arrow] (110.center) to (111.center);
		\draw [style=dashed blue] (113.center) to (112.center);
		\draw [style=dashed blue, bend right=90, looseness=2.50] (112.center) to (114.center);
		\draw [style=dashed blue] (114.center) to (115.center);
		\draw [style=dashed blue, in=90, out=90, looseness=2.50] (115.center) to (113.center);
		\draw [style=arrow] (118.center) to (119.center);
		\draw [style=dashed blue] (121.center) to (120.center);
		\draw [style=dashed blue, bend right=90, looseness=2.50] (120.center) to (122.center);
		\draw [style=dashed blue] (122.center) to (123.center);
		\draw [style=dashed blue, in=90, out=90, looseness=2.50] (123.center) to (121.center);
		\draw [style=arrow] (126.center) to (127.center);
		\draw [style=dashed blue] (129.center) to (128.center);
		\draw [style=dashed blue, bend right=90, looseness=2.50] (128.center) to (130.center);
		\draw [style=dashed blue] (130.center) to (131.center);
		\draw [style=dashed blue, in=90, out=90, looseness=2.50] (131.center) to (129.center);
		\draw [style=arrow] (134.center) to (135.center);
		\draw [style=dashed blue] (137.center) to (136.center);
		\draw [style=dashed blue, bend right=90, looseness=2.50] (136.center) to (138.center);
		\draw [style=dashed blue] (138.center) to (139.center);
		\draw [style=dashed blue, in=90, out=90, looseness=2.50] (139.center) to (137.center);
		\draw [style=arrow] (142.center) to (143.center);
		\draw [style=dashed blue] (145.center) to (144.center);
		\draw [style=dashed blue, bend right=90, looseness=2.50] (144.center) to (146.center);
		\draw [style=dashed blue] (146.center) to (147.center);
		\draw [style=dashed blue, in=90, out=90, looseness=2.50] (147.center) to (145.center);
		\draw [style=arrow] (150.center) to (151.center);
		\draw [style=dashed blue] (153.center) to (152.center);
		\draw [style=dashed blue, bend right=90, looseness=2.50] (152.center) to (154.center);
		\draw [style=dashed blue] (154.center) to (155.center);
		\draw [style=dashed blue, in=90, out=90, looseness=2.50] (155.center) to (153.center);
	\end{pgfonlayer}
\end{tikzpicture}
    \caption{The route and branch structure for the intermediate agents in the $3$-switch. There are six branches, each corresponding to a causal order.}
    \label{fig:3switchalice}
\end{figure}

Now suppose that the agent at $P$ makes the `$l=1$' bifurcation choice, so that the message is sent to Alice.
The global index constraint then enforces the route at the nodes $A$, $B$, $C$ depicted in Figure \ref{fig:3switchalice}.
Thus Alice's route implies that she has no choice but to preserve the value of $l$, meaning that she must send the message along the arrow from $A$ to $B$, since this is her only outgoing arrow that does not correspond to a dummy sector when $l=1$. Then Bob similarly has no choice but to pass the message to Charlie, and finally Charlie is forced to send the message into the Future. The net result is that the message moves inexorably along the path $P \rightarrow A \rightarrow B \rightarrow C \rightarrow F$ of arrows decorated with an $l$ index, giving the causal order $A-B-C$. Thus, if an agent at $P$ makes the bifurcation choice that $l=1$, they pick out this causal order. 

Similarly, any option from the bifurcation choice enforces one of the six possible causal orders. In this sense, the bifurcation choice at $P$ is a choice between causal orders, just as in the case of the original quantum switch. This state of affairs -- that the causal order is determined by a bifurcation choice at the Past node -- is characteristic of the (non-dynamical) coherent control of causal orders.

Now let us show that the routed graph satisfies our two principles. It is clear that the bifurcation choice at $P$, picking which index is equal to 1, determines the status of all branches of the intermediate nodes, since these branches are all defined by a certain index equalling 1 (see Figure \ref{fig:3switchalice}). This bifurcation choice is the only one in the routed graph, and $P$ and $F$ each have just one branch (the route at $F$ is just the time-reversed version of the one at $P$, obtained by reversing the direction of the arrows in Figure \ref{fig:3switchpast}). Thus the sole bifurcation choice in the routed graph leads to a single branch happening at each node; formally speaking, we have a function from bifurcation choices to branch statuses.
That is, the routed graph satisfies univocality. 

\changes{

Recall that bi-univocality requires that the time-reverse of the routed graph also satisfies univocality, where the time-reverse is obtained by reversing the arrows. Taking the time-reverse of Figure \ref{fig:3switch_routedgraph} and then relabelling
\begin{equation}
    \begin{split}
        P &\longleftrightarrow F \\
         l &\longleftrightarrow r \\
         m &\longleftrightarrow p \\
         n &\longleftrightarrow q \\
    \end{split}
\end{equation}
results in exactly the same routed graph as Figure \ref{fig:3switch_routedgraph} itself. In other words, the routed graph is \textit{time-symmetric} (up to relabelling). If a routed graph is univocal and time-symmetric, then it must also satisfy bi-univocality. Hence Figure \ref{fig:3switch_routedgraph} satisfies bi-univocality.

Since the routed graph satisfies bi-univocality, we can draw its branch graph, following the rules in Section \ref{sec:checking_validity}: we display it in Figure \ref{fig:3switch_branchgraph}.} In this graph, the six branches for each of the nodes $A$, $B$ and $C$ are denoted by the specification of which index is equal to 1 (with all the others equal to 0), e.g.\ $A^{l=1}, A^{p=1}$, etc. There are no loops in the branch graph, meaning that the routed graph trivially satisfies weak loops. We can thus invoke Theorem \ref{thm: Main} to conclude that any process that can be obtained from the routed graph of Figure \ref{fig:3switch_routedgraph}, including the quantum 3-switch, is a consistent quantum process.

\begin{figure}[ht]
	\centering
	\input{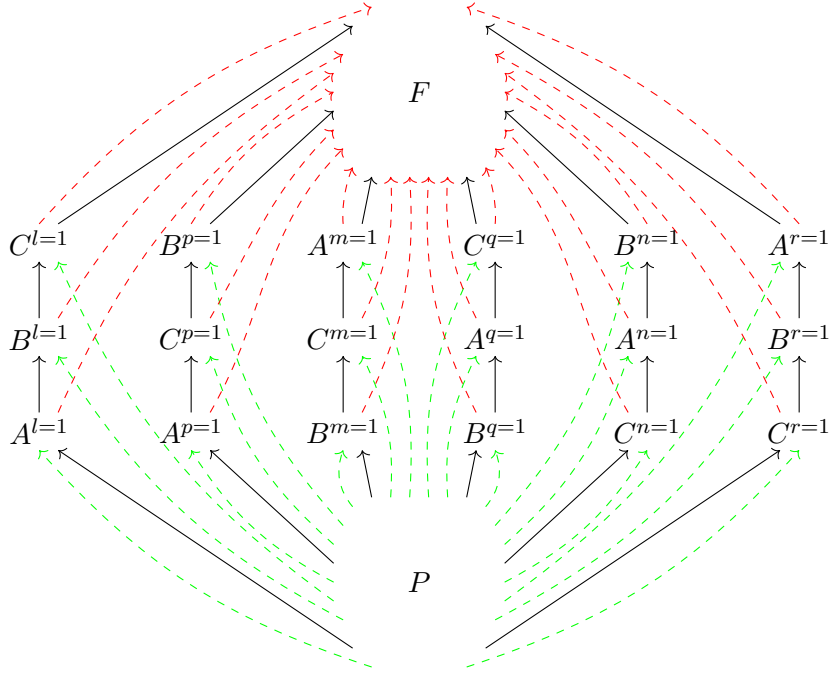}
	\caption{The branch graph for the quantum 3-switch.}
	\label{fig:3switch_branchgraph}
\end{figure}

\subsubsection{The routed circuit}

\begin{figure}[ht]
	\centering
	\input{figures/3switch.tikz}
	\caption{A routed circuit diagram for the quantum 3-switch, using a global index constraint. To avoid too much clutter, instead of explicitly drawing loops, output lines that end with dots are to be interpreted as being looped back to join the corresponding input lines with dots (with the same system labels, including indices). Some wires are coloured for readability.}
	\label{fig:3switch}
\end{figure} 

The routed circuit for the quantum 3-switch can be constructed from the routed graph in Figure \ref{fig:3switch_routedgraph} by inserting unitary transformations into the corresponding skeletal supermap. This is displayed in Figure \ref{fig:3switch}, where we have again used the shorthand of global index constraints. The routes of the transformations can be derived from the global index constraint, just like the routes of the nodes in the routed graph. 

The systems in the routed circuit have the following properties:
\begin{itemize}
	\item 
	The systems $P_T,F_T,A_{\rm in}, A_{\rm out}, B_{\rm in}, B_{\rm out}, C_{\rm in}, C_{\rm out}$ are all isomorphic, and correspond to a $d$-dimensional space. 
	\item $P_C, F_C$ are 6-dimensional control systems.
	\item The routed system $C_k$ is also a 6-dimensional control system,  with an explicit partition into six one-dimensional sectors. 
	\item The routed systems $R^{lp}, S^{mq}, T^{nr}, X^{ln}, Y^{pq}$ are all $(2d+1)$-dimensional systems, this time partitioned into \textit{two} $d$-dimensional sectors and a single 1-dimensional `dummy' sector. For example, $R^{lp} = R^{00}\oplus R^{10} \oplus R^{01}$, where $R^{00}$ is the 1-dimensional sector. The presence of two separate $d$-dimensional sectors corresponds to the fact that each of these wires can carry the message in two separate causal orders. We denote the unique state in the 1-dimensional sectors by $\ket{\rm dum}$. 
\end{itemize}

The unitary $U_P$ at the bottom of the diagram is given by  the isomorphism:
\begin{equation}
	U_P :
	\begin{cases}
		\ket{1}_{P_C} \otimes \ket{\psi}_{P_T} \mapsto \ket{\psi}_{R^{10}} \otimes \ket{\rm dum}_{S^{00}} \otimes \ket{\rm dum}_{T^{00}}  \\
		 \ket{2}_{P_C} \otimes \ket{\psi}_{P_T} \mapsto\ket{\psi}_{R^{01}}  \otimes \ket{\rm dum}_{S^{00}} \otimes \ket{\rm dum}_{T^{00} } \\
		 \ket{3}_{P_C} \otimes \ket{\psi}_{P_T} \mapsto\ket{\rm dum}_{R^{00}} \otimes \ket{\psi}_{S^{10}} \otimes \ket{\rm dum}_{T^{00}}  \\
		 \ket{4}_{P_C} \otimes \ket{\psi}_{P_T} \mapsto \ket{\rm dum}_{R^{00}} \otimes\ket{\psi}_{S^{01}}  \otimes  \ket{\rm dum}_{T^{00}}  \\
		 \ket{5}_{P_C} \otimes \ket{\psi}_{P_T} \mapsto \ket{\rm dum}_{R^{00}} \otimes \ket{\rm dum}_{S^{00}} \otimes \ket{\psi}_{T^{10}}   \\
		 \ket{6}_{P_C} \otimes \ket{\psi}_{P_T} \mapsto \ket{\rm dum}_{R^{00}} \otimes \ket{\rm dum}_{S^{00}} \otimes \ket{\psi}_{T^{01}}   \\
	\end{cases}
\label{eq:U_3switch}
\end{equation}
between the non-routed system $P_C \otimes P_T$ ($6 d$-dimensional) and the routed system \\ $\bigoplus_{lmnpqr} \eta^{lmnpqr} R^{lp} \otimes S^{mq} \otimes T^{nr}$ [also of dimension $2(d\times1\times 1) +2(1\times d\times 1) +2(1\times 1\times d)=6d$].

$U_F$ has the same form as $U_P$, where the $\ket{1}_{F_C}$ state of the control qubit is again mapped to the $l=1$ sector,  $\ket{2}_{F_C}$ is again mapped to $p=1$ sector, and so on. The other unitaries denoted by $U$ are the unique unitaries of the form above that respect the index-matching.



\subsection{The Grenoble process}

In their 2021 paper \cite{wechs2021quantum}, Wechs and co-authors from Grenoble presented a new tripartite process with \textit{dynamical} indefinite causal order, that is, where the causal order is not predetermined at the start of the process, but can be influenced by the intermediate agents themselves. In the present work, we shall call this process the \textit{Grenoble process}. 

Like the 3-switch, the Grenoble process involves three intermediate agents, who receive information from the global Past and ultimately send information into the global Future.\footnote{Note, that in the original formulation in Ref.\ \cite{wechs2021quantum}, the Future is split into more than one party, whilst in this work, to simplify the presentation we consider only one Future party.}
The Past (P) consists of a 3-dimensional control qutrit $P_C$ and a 2-dimensional target qubit $P_T$.
As with the previous processes we have studied, the logical state of the control system determines which of the intermediate agents will receive the message first. 
However, unlike the previous processes, this control system does not enforce a single causal order.
This is because the agent who receives the message first is free to choose which agent will receive it second. 
In particular, the logical state of the target qubit after it passes through the first node will determine who gets it second: $\ket{0}$ means it will be sent in clockwise order (for example, to Bob if Alice was first), while $\ket{1}$ means it will be sent in anticlockwise order (for example, to Charlie if Alice was first). Finally, before the action of the third and final agent, the information about the relative order of the first two agents is scrambled onto an ancillary qubit, which is transferred directly to the Future (F).

In the Grenoble process, the emergent causal order depends not only on the global Past, but also on the actions of the intermediate agents. This is the hallmark of \textit{dynamical} coherent control of causal order. In our terms, this will correspond to the fact that a causal order (and the branch statuses that fix it) is determined not only by a bifuraction choice at the Past, but also by bifurcation choices of the intermediate agents.


\subsubsection{The routed graph}

To begin with, we write down a routed graph from which the Grenoble process, amongst others, can be constructed. This graph is given in Figure \ref{fig:grenoble_routedgraph}, again using global index constraints. 

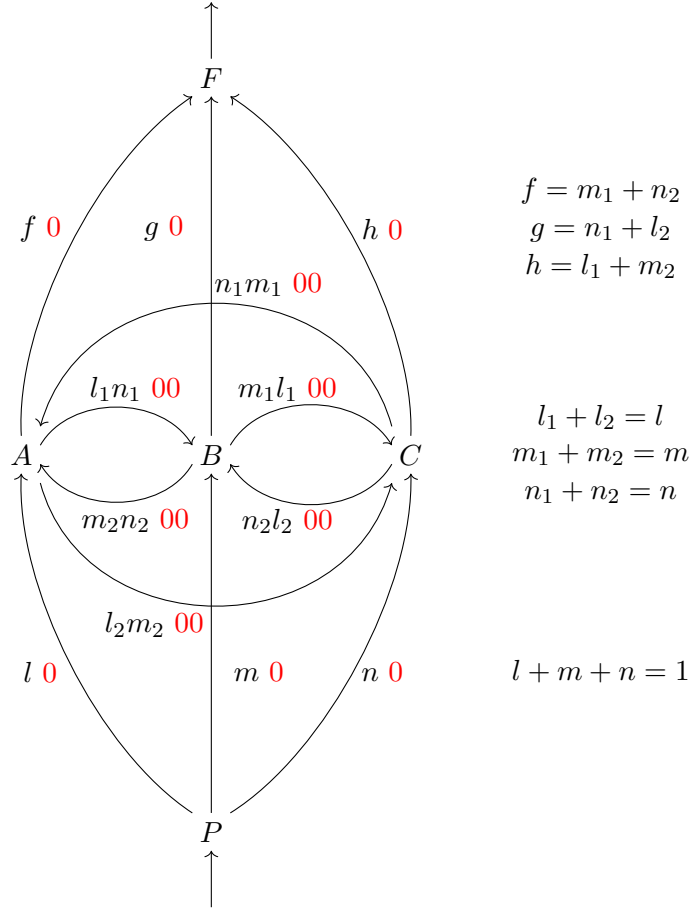
\begin{figure}[ht]
	\centering
\scalebox{1}{\begin{tikzpicture}
	\begin{pgfonlayer}{nodelayer}
		\node [style=none] (0) at (7.75, -10) {$P$};
		\node [style=none] (1) at (7.75, 0) {$B$};
		\node [style=none] (2) at (2.75, 0) {$A$};
		\node [style=none] (3) at (13, 0) {$C$};
		\node [style=none] (4) at (7.75, -12) {};
		\node [style=none] (5) at (7.25, -9.5) {};
		\node [style=none] (6) at (8.25, -9.5) {};
		\node [style=none] (7) at (2.75, -0.5) {};
		\node [style=none] (8) at (13, -0.5) {};
		\node [style=none] (9) at (3.25, -0.75) {};
		\node [style=none] (10) at (12.5, -0.25) {};
		\node [style=none] (11) at (7.75, -0.5) {};
		\node [style=none] (12) at (7.75, -9.5) {};
		\node [style=none] (13) at (7.75, -10.5) {};
		\node [style=none] (14) at (7.25, -0.25) {};
		\node [style=none] (15) at (8.25, -0.25) {};
		\node [style=none] (16) at (3.25, -0.25) {};
		\node [style=none] (17) at (12.5, -0.75) {};
		\node [style=none] (18) at (3.25, -5.75) {$l~{\textcolor{red}{0}}$};
		\node [style=none] (19) at (12.25, -5.75) {$n~{\textcolor{red}{0}}$};
		\node [style=none] (20) at (9, -5.75) {$m~{\textcolor{red}{0}}$};
		\node [style=none] (21) at (5.75, -1.75) {$m_2 n_2 ~{\textcolor{red}{00}}$};
		\node [style=none] (22) at (9.75, -1.75) {$n_2 l_2~{\textcolor{red}{00}}$};
		\node [style=none] (23) at (6.25, -4.5) {$l_2 m_2~{\textcolor{red}{00}}$};
		\node [style=none] (24) at (7.75, 10) {$F$};
		\node [style=none] (25) at (7.75, 12) {};
		\node [style=none] (26) at (7.25, 9.5) {};
		\node [style=none] (27) at (8.25, 9.5) {};
		\node [style=none] (28) at (2.75, 0.5) {};
		\node [style=none] (29) at (13, 0.5) {};
		\node [style=none] (30) at (12.5, 0.75) {};
		\node [style=none] (31) at (12.5, 0.25) {};
		\node [style=none] (32) at (7.75, 0.5) {};
		\node [style=none] (33) at (7.75, 9.5) {};
		\node [style=none] (34) at (7.75, 10.5) {};
		\node [style=none] (35) at (7.25, 0.25) {};
		\node [style=none] (36) at (8.25, 0.25) {};
		\node [style=none] (37) at (3.25, 0.25) {};
		\node [style=none] (38) at (3.25, 0.75) {};
		\node [style=none] (39) at (3.25, 6) {$f~{\textcolor{red}{0}}$};
		\node [style=none] (40) at (12.25, 6) {$h~{\textcolor{red}{0}}$};
		\node [style=none] (41) at (6.5, 6) {$g~{\textcolor{red}{0}}$};
		\node [style=none] (42) at (5.75, 1.75) {$l_1 n_1 ~{\textcolor{red}{00}}$};
		\node [style=none] (43) at (9.75, 1.75) {$m_1 l_1~{\textcolor{red}{00}}$};
		\node [style=none] (44) at (9.25, 4.5) {$n_1 m_1~{\textcolor{red}{00}}$};
		\node [style=none] (47) at (18, -5.75) {$l+m+n=1$};
		\node [style=none] (48) at (18, 7) {$f=m_1 + n_2$};
		\node [style=none] (49) at (18, 6) {$g=n_1 + l_2$};
		\node [style=none] (50) at (18, 5) {$h=l_1 + m_2$};
		\node [style=none] (51) at (18, 0) {$m_1 + m_2=m$};
		\node [style=none] (52) at (18, 1) {$l_1 + l_2=l$};
		\node [style=none] (53) at (18, -1) {$n_1 + n_2=n$};
		\node [style=none] (54) at (18, 1.75) {};
	\end{pgfonlayer}
	\begin{pgfonlayer}{edgelayer}
		\draw [style=arrow, bend right, looseness=0.75] (6.center) to (8.center);
		\draw [style=arrow] (12.center) to (11.center);
		\draw [style=arrow, bend left, looseness=0.75] (5.center) to (7.center);
		\draw [style=arrow, bend left=60] (14.center) to (16.center);
		\draw [style=arrow, bend left=60] (10.center) to (15.center);
		\draw [style=arrow, bend right=75, looseness=1.25] (9.center) to (17.center);
		\draw [style=arrow] (34.center) to (25.center);
		\draw [style=arrow, bend right=75, looseness=1.25] (30.center) to (38.center);
		\draw [style=arrow] (4.center) to (13.center);
		\draw [style=arrow, bend left, looseness=0.75] (28.center) to (26.center);
		\draw [style=arrow, bend right, looseness=0.75] (29.center) to (27.center);
		\draw [style=arrow, bend left=60] (37.center) to (35.center);
		\draw [style=arrow, bend left=60] (36.center) to (31.center);
		\draw [style=arrow] (32.center) to (33.center);
	\end{pgfonlayer}
\end{tikzpicture}}
\caption{The routed graph for the Grenoble process, using a global index constraint.}
	\label{fig:grenoble_routedgraph}
\end{figure} 

For each arrow, the sector corresponding to all of its indices being equal to zero is a one-dimensional sector. The global index constraint (in particular, the floating equation $l+m+n=1$), imposes a route at the node $P$ that forces exactly one of the outgoing indices to equal 1, depicted in Figure \ref{fig:grenoblepast}. The route at the $F$ node is just the time-reverse of the route at the Past. The global index constraint also gives rise to a route at $A$ depicted in Figure \ref{fig:grenoble route}. The routes at $B$ and $C$ are closely analogous.

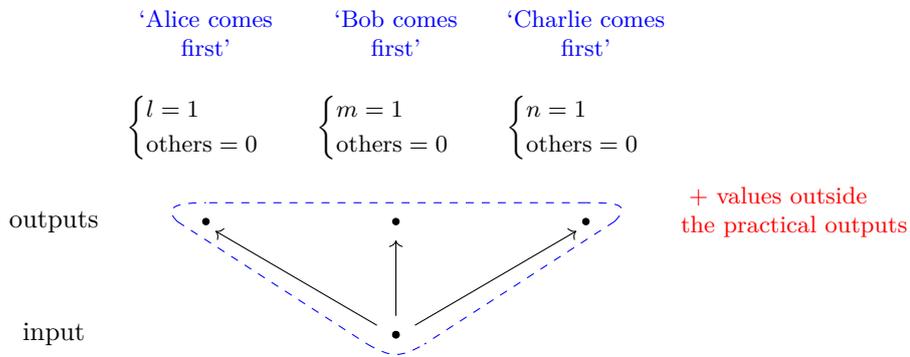
\begin{figure}
    \centering
    \begin{tikzpicture}
	\begin{pgfonlayer}{nodelayer}
		\node [style=none] (0) at (6.5, -1.75) {};
		\node [style=none] (1) at (2.25, 0.75) {};
		\node [style=none] (2) at (7, 0.5) {};
		\node [style=none] (3) at (7, -1.5) {};
		\node [style=none] (4) at (1.25, 1) {};
		\node [style=none] (5) at (6.25, -2.25) {};
		\node [style=none] (6) at (7.75, -2.25) {};
		\node [style=none] (7) at (12.75, 1) {};
		\node [style=very small black dot] (10) at (7, -2) {};
		\node [style=very small black dot] (11) at (7, -2) {};
		\node [style=very small black dot] (12) at (2, 1) {};
		\node [style=very small black dot] (13) at (7, 1) {};
		\node [style=none] (14) at (2, 1.5) {};
		\node [style=none] (15) at (12, 1.5) {};
		\node [style=none] (16) at (1.25, 1) {};
		\node [style=none] (17) at (2, 1.5) {};
		\node [style=none] (26) at (-2, 1) {\small{outputs}};
		\node [style=none] (27) at (-2, -2) {\small{input}};
		\node [style=none, align=justify, font={\footnotesize}, text width=4.5cm] (28) at (19, 1.25) {{ \color{red} + values outside} \\  {\color{red}the practical outputs}};
		\node [style=none, font={\footnotesize}] (31) at (7, 3.5) {$\begin{cases} m=1 \\ {\rm others=0} \end{cases}$};
		\node [style=none, font={\footnotesize}] (32) at (2, 3.5) {$\begin{cases} l= 1 \\ {\rm others =0} \end{cases}$};
		\node [style=very small black dot] (63) at (12, 1) {};
		\node [style=none] (65) at (7.5, -1.75) {};
		\node [style=none] (66) at (11.75, 0.75) {};
		\node [style=none, font={\footnotesize}] (67) at (12, 3.5) {$\begin{cases} n=1 \\ {\rm others=0} \end{cases}$};
		\node [align=center, font={\footnotesize}, text width=4.5cm] (68) at (2, 6) {\color{blue} `Alice comes \\ first'};
		\node [align=center, font={\footnotesize}, text width=4.5cm] (69) at (7, 6) {\color{blue}`Bob comes \\ first'};
		\node [align=center, font={\footnotesize}, text width=4.5cm] (70) at (12, 6) {\color{blue}`Charlie comes \\ first'};
	\end{pgfonlayer}
	\begin{pgfonlayer}{edgelayer}
		\draw [style=arrow] (0.center) to (1.center);
		\draw [style=arrow] (3.center) to (2.center);
		\draw [style=dashed blue] (4.center) to (5.center);
		\draw [style=dashed blue, in=-150, out=-30, looseness=1.25] (5.center) to (6.center);
		\draw [style=dashed blue, in=0, out=45, looseness=2.00] (7.center) to (15.center);
		\draw [style=dashed blue] (15.center) to (14.center);
		\draw [style=dashed blue, in=180, out=135, looseness=1.75] (16.center) to (17.center);
		\draw [style=dashed blue] (7.center) to (6.center);
		\draw [style=arrow] (65.center) to (66.center);
	\end{pgfonlayer}
\end{tikzpicture}
    \caption{The branch structure for the node $P$ of the Grenoble process. The bifurcation choice in the unique branch determines which agent comes first.}
    \label{fig:grenoblepast}
\end{figure}

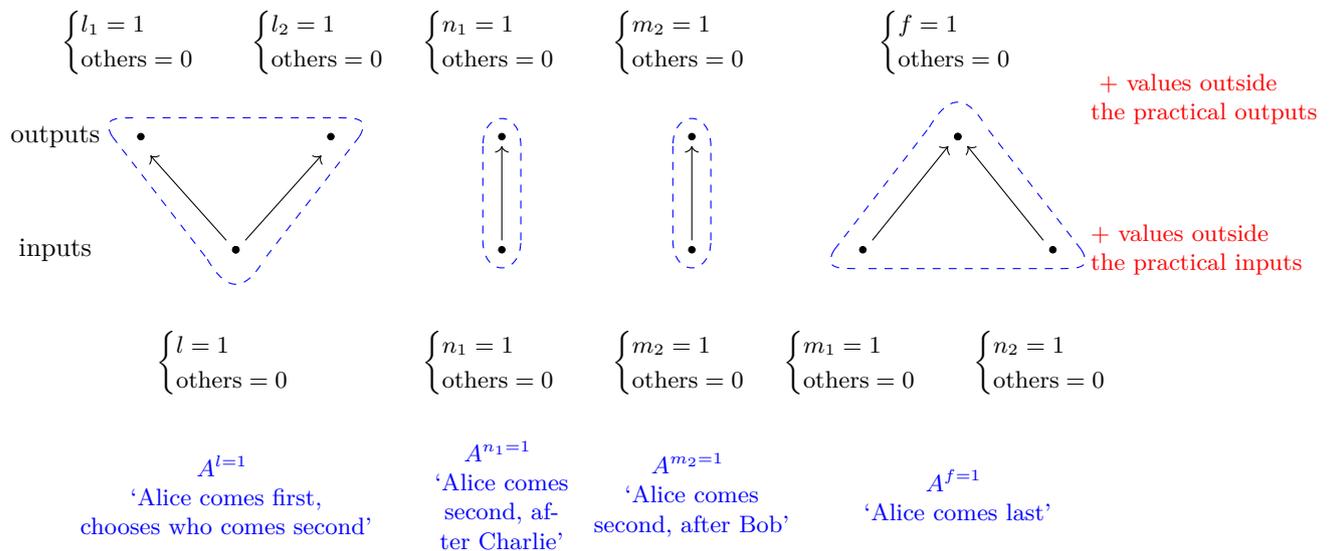
\begin{figure}
    \centering
    \begin{tikzpicture}
	\begin{pgfonlayer}{nodelayer}
		\node [style=none] (10) at (-9.25, -1.75) {};
		\node [style=none] (11) at (-11.25, 0.5) {};
		\node [style=none] (12) at (-6.75, 0.5) {};
		\node [style=none] (13) at (-8.75, -1.75) {};
		\node [style=none] (22) at (-12.25, 1) {};
		\node [style=none] (23) at (-9.75, -2.25) {};
		\node [style=none] (24) at (-8.25, -2.25) {};
		\node [style=none] (25) at (-5.75, 1) {};
		\node [align=center, font={\footnotesize}, text width=4.5cm] (26) at (-9.25, -8.5) {{\color{blue} $A^{l=1}$ }\\ {\color{blue} `Alice comes first, \\ chooses who comes second'}};
		\node [style=none] (43) at (29.75, 7.5) {\color{blue} \small{Alice comes last}};
		\node [style=very small black dot] (45) at (-9, -2) {};
		\node [style=very small black dot] (46) at (-9, -2) {};
		\node [style=very small black dot] (47) at (-11.5, 1) {};
		\node [style=very small black dot] (48) at (-6.5, 1) {};
		\node [style=none] (49) at (-11.5, 1.5) {};
		\node [style=none] (50) at (-6.5, 1.5) {};
		\node [style=none] (53) at (-12.25, 1) {};
		\node [style=none] (55) at (-11.5, 1.5) {};
		\node [style=very small black dot] (80) at (3, -2) {};
		\node [style=very small black dot] (81) at (3, 1) {};
		\node [style=none] (82) at (3, -1.75) {};
		\node [style=none] (83) at (3, 0.75) {};
		\node [style=none] (84) at (2.5, -1.75) {};
		\node [style=none] (85) at (2.5, 0.75) {};
		\node [style=none] (86) at (3.5, -1.75) {};
		\node [style=none] (87) at (3.5, 0.75) {};
		\node [style=none] (88) at (-13.75, 1) {\small{outputs}};
		\node [style=none] (89) at (-13.75, -2) {\small{inputs}};
		\node [style=none, align=justify, font={\footnotesize}, text width=4.5cm] (90) at (18, 2) {{ \color{red} + values outside} \\  {\color{red}the practical outputs}};
		\node [style=none] (91) at (-15.25, -4.75) {};
		\node [style=none, font={\footnotesize}] (93) at (-2, 3.5) {$\begin{cases} n_1=1 \\ {\rm others=0}\end{cases}$};
		\node [style=none, font={\footnotesize}] (94) at (-6.5, 3.5) {$\begin{cases} l_2=1 \\ {\rm others=0} \end{cases}$};
		\node [style=none, font={\footnotesize}] (95) at (-11.5, 3.5) {$\begin{cases} l_1 = 1 \\ {\rm others =0} \end{cases}$};
		\node [style=very small black dot] (96) at (-2, -2) {};
		\node [style=very small black dot] (97) at (-2, 1) {};
		\node [style=none] (98) at (-2, -1.75) {};
		\node [style=none] (99) at (-2, 0.75) {};
		\node [style=none] (100) at (-2.5, -1.75) {};
		\node [style=none] (101) at (-2.5, 0.75) {};
		\node [style=none] (102) at (-1.5, -1.75) {};
		\node [style=none] (103) at (-1.5, 0.75) {};
		\node [style=none, font={\footnotesize}] (104) at (3, 3.5) {$\begin{cases} m_2 = 1 \\ {\rm others=0} \end{cases}$};
		\node [style=none, font={\footnotesize}] (105) at (10, 3.5) {$\begin{cases}  f=1 \\ {\rm others=0} \end{cases}$};
		\node [style=none] (107) at (7.75, -1.75) {};
		\node [style=none] (108) at (9.75, 0.75) {};
		\node [style=none] (109) at (10.25, 0.75) {};
		\node [style=none] (110) at (12.25, -1.75) {};
		\node [style=none] (111) at (6.75, -2) {};
		\node [style=none] (112) at (9.25, 1.25) {};
		\node [style=none] (113) at (10.75, 1.25) {};
		\node [style=none] (114) at (13.25, -2) {};
		\node [style=very small black dot] (115) at (10, 1) {};
		\node [style=very small black dot] (116) at (10, 1) {};
		\node [style=very small black dot] (117) at (7.5, -2) {};
		\node [style=very small black dot] (118) at (12.5, -2) {};
		\node [style=none] (119) at (7.5, -2.5) {};
		\node [style=none] (120) at (12.5, -2.5) {};
		\node [style=none] (121) at (6.75, -2) {};
		\node [style=none] (122) at (7.5, -2.5) {};
		\node [style=none, align=justify, font={\footnotesize}, text width=4.5cm] (123) at (18, -2) {\color{red} + values outside \\ the practical inputs};
		\node [style=none, font={\footnotesize}] (124) at (3, -5) {$\begin{cases} m_2=1 \\ {\rm others=0}\end{cases}$};
		\node [style=none, font={\footnotesize}] (125) at (-2, -5) {$\begin{cases} n_1=1 \\ {\rm others=0} \end{cases}$};
		\node [style=none, font={\footnotesize}] (126) at (-9, -5) {$\begin{cases} l = 1 \\ {\rm others =0} \end{cases}$};
		\node [style=none, font={\footnotesize}] (127) at (7.5, -5) {$\begin{cases} m_1 = 1 \\ {\rm others=0} \end{cases}$};
		\node [style=none, font={\footnotesize}] (128) at (12.5, -5) {$\begin{cases}  n_2=1 \\ {\rm others=0} \end{cases}$};
		\node [align=center, font={\footnotesize}, text width=3cm] (129) at (-2, -8.5) {{\color{blue} $A^{n_1=1}$ }\\ {\color{blue} `Alice comes \\ second,  after Charlie'}};
		\node [align=center, font={\footnotesize}, text width=3cm] (130) at (3, -8.5) {{\color{blue} $A^{m_2=1}$ }\\ {\color{blue} `Alice comes \\ second, after Bob'}};
		\node [align=center, font={\footnotesize}] (131) at (10, -8.5) {{\color{blue} $A^{f=1}$ }\\ {\color{blue} `Alice comes last'}};
	\end{pgfonlayer}
	\begin{pgfonlayer}{edgelayer}
		\draw [style=arrow] (10.center) to (11.center);
		\draw [style=arrow] (13.center) to (12.center);
		\draw [style=dashed blue] (22.center) to (23.center);
		\draw [style=dashed blue, in=-120, out=-60, looseness=1.75] (23.center) to (24.center);
		\draw [style=dashed blue, in=0, out=60, looseness=1.75] (25.center) to (50.center);
		\draw [style=dashed blue] (50.center) to (49.center);
		\draw [style=dashed blue, in=180, out=120, looseness=1.75] (53.center) to (55.center);
		\draw [style=dashed blue] (25.center) to (24.center);
		\draw [style=arrow] (82.center) to (83.center);
		\draw [style=dashed blue] (85.center) to (84.center);
		\draw [style=dashed blue, bend right=90, looseness=2.50] (84.center) to (86.center);
		\draw [style=dashed blue] (86.center) to (87.center);
		\draw [style=dashed blue, in=90, out=90, looseness=2.50] (87.center) to (85.center);
		\draw [style=arrow] (98.center) to (99.center);
		\draw [style=dashed blue] (101.center) to (100.center);
		\draw [style=dashed blue, bend right=90, looseness=2.50] (100.center) to (102.center);
		\draw [style=dashed blue] (102.center) to (103.center);
		\draw [style=dashed blue, in=90, out=90, looseness=2.50] (103.center) to (101.center);
		\draw [style=arrow] (107.center) to (108.center);
		\draw [style=arrow] (110.center) to (109.center);
		\draw [style=dashed blue] (111.center) to (112.center);
		\draw [style=dashed blue, in=120, out=60, looseness=1.75] (112.center) to (113.center);
		\draw [style=dashed blue, in=0, out=-60, looseness=1.75] (114.center) to (120.center);
		\draw [style=dashed blue] (120.center) to (119.center);
		\draw [style=dashed blue, in=-180, out=-120, looseness=1.75] (121.center) to (122.center);
		\draw [style=dashed blue] (114.center) to (113.center);
	\end{pgfonlayer}
\end{tikzpicture}
    \caption{The branch structure for the node $A$ of the Grenoble process.}
    \label{fig:grenoble route}
\end{figure}

Just like the routed graph for the 3-switch, the bifurcation choice at $P$ determines which agent comes first. 
But unlike the 3-switch, this bifurcation does not enforce an entire causal order. 
Rather, it is left to the first intermediate agent to decide which one should come second.
For example, suppose that an agent at $P$ makes the bifurcation choice $l=1$.
This sends the message to Alice (since the only outgoing arrow from $P$ that is not associated with a trivial dummy sector in this case is $P \rightarrow A$).
This leads to Alice having the binary bifurcation choice associated with the branch $A^{l=1}$, depicted in Figure \ref{fig:grenoble route}.
This bifurcation choice determines which agent comes second.
For example, suppose Alice chooses the bifurcation option $l_1=1$. 
Then the message is passed along the `$l_1n_2$'-indexed arrow to Bob. 
Then Bob's route implies he has no such choice: he is forced to preserve the value of $l_1=1$, and is thereby compelled to send the message along the $m_1l_1$ arrow to Charlie (he is confined to a branch $B^{l_1=1}$, analogous to $A^{n_1=1}$ in Figure \ref{fig:grenoble route}). 
Finally, Charlie, confined to a branch $C^{h=1}$ analogous to $A^{f=1}$ in Figure 1, is forced to send the message off into the Future.
Thus Alice's choice $l_1=1$ enforces the clockwise causal order $A-B-C$.
On the other hand, choosing $l_2=1$ leads to the anticlockwise order $A-C-B$.

The situation is analogous if another one of the agents comes first. 
If Bob comes first, he makes a bifurcation choice between $m_1=1$ and $m_2=1$ that enforces either the clockwise order $B-C-A$ or the anticlockiwse order $B-A-C$, respectively.
Finally, if Charlie comes first, he chooses between $n_1=1$ and the clockwise order order $C-A-B$, or $n_2=1$ and the anticlockwise order $C-B-A$ . 

This scenario also allows for the \changes{potential} disappearance of the information about the order of agents that acted already. Indeed, suppose that Alice comes last. 
This means she has either received the message coming clockwise from Charlie, or anticlockwise from Bob: i.e., either $m_1=1$ or $n_2=1$, respectively.
In both cases, the floating equation $f=m_1+n_2$ guarantees that $f=1$, meaning that the information about which agent came first and which came second \changes{may be deleted.}\footnote{\changes{This does not mean, however, that this information is \textit{necessarily} lost; Alice might plug in an operation that carries it into the future. The point is that that the route structure makes no requirement for this information to be preserved in a specific structural way. Alice might therefore also plug in an operation that unitarily rotates, degrades, or outright discards this information.

This point may be better understood by appealing to the symmetry between this reverse, `backward-facing' bifurcation (visible in Figure \ref{fig:grenoble route} ,in the $A^{f=1}$ branch in which Alice comes last), and a `forward-facing' bifurcation (as present e.g.\ in the branch $A^{l=1}$ in which Alice comes first). In the latter, we colloquially say that Alice `chooses' who comes second, which she could do e.g.\ by controlling her output state on an ancilla; but she might as well just plug in a unitary, so that this information corresponds to some information present in the state of her main input, coming from the past. Either way, what the route structure indicates is that \textit{from this point on, this information is tracked by the route structure}. The situation with a backwards-facing bifurcation, as present in the $A^{f=1}$ branch, is to be understood symmetrically: from that point on, the information about who came first might be processed in an arbitrary way (including being deleted) by Alice or by later agents, depending on their choices of operation.}} This can be seen in the structure of the `$f=1$' branch in Figure \ref{fig:grenoble route}.
Again, the situation is analogous if another agent comes last.

To  construct the branch graph, consider the following. The node $P$ consists of a single branch with a bifurcation choice between three options, each corresponding to the case when one of the three indices $l,m,n$   equals 1. In the time-reversed version of the routed graph, the node $F$ has a bifurcation with three options, each corresponding to the case when one of the three indices $f,g,h$   equals 1. 
 
\begin{figure}[ht]
	\centering
	\input{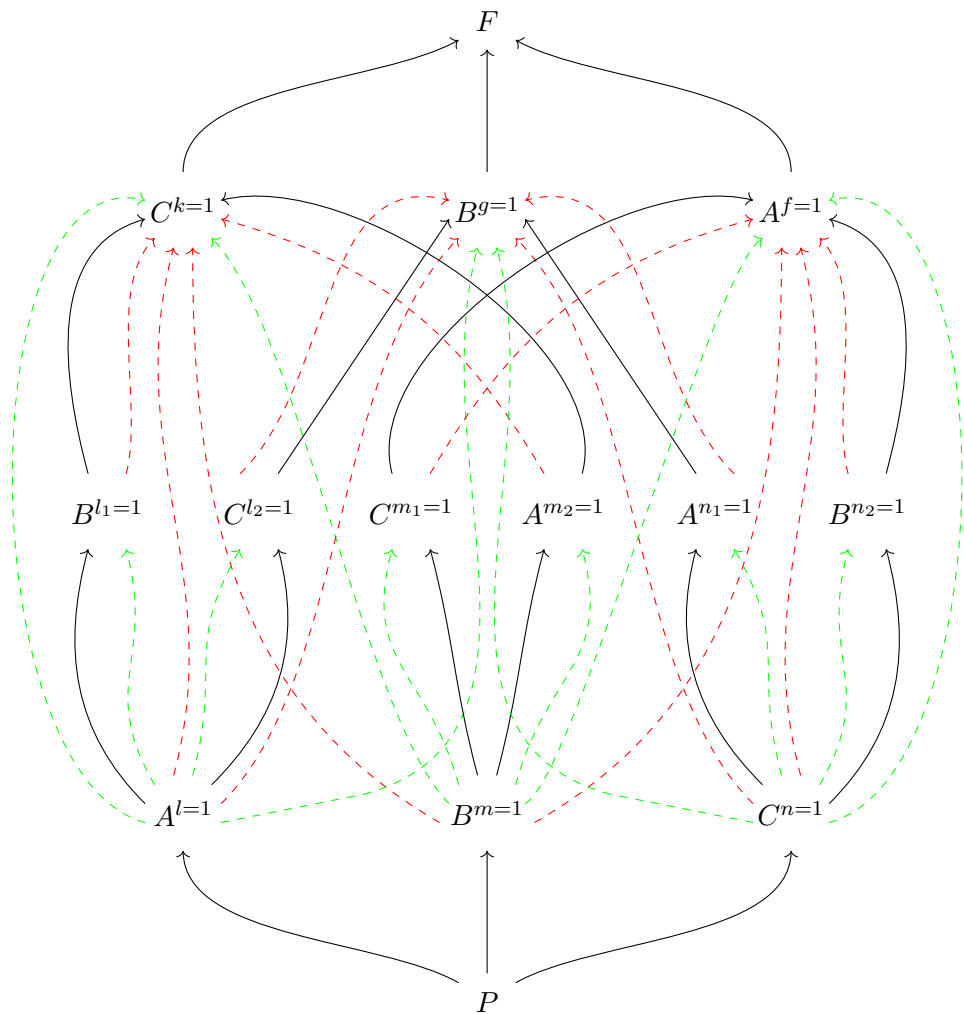}
	\caption{The branch graph for the Grenoble process. For clarity, we have omitted the green and red dashed arrows from $P$ and to $F$, respectively: they simply point upwards in the diagram to/from all of the other branched nodes.}
	\label{fig:grenoble_branchgraph}
\end{figure} 

The routes at the nodes $A,B,C$ consist of four branches, as illustrated in Figure \ref{fig:grenoble route}. One of these branches corresponds to the case when the index on the wire coming directly from $P$ equals 1, with a bifurcation between two options splitting this index into an index of the same name with subscript either 1 or 2 (corresponding to whether the message is sent clockwise or anticlockwise). Another branch corresponds to the case when the index on the wire going directly to $F$ equals 1, with a  bifurcation in the time-reversed routed graph combining the second index of each of the two wires coming in from the other agents (corresponding to whether the message came from the clockwise or anticlockwise direction). The final two branches correspond to the cases when one of the two indices that appear on both the input and output wires of the nodes are equal to 1. Following Figure \ref{fig:grenoble route}, we shall denote these branches by superscripts labelling which index is equal to 1.

\changes{
Like the one for the 3-switch, the routed graph for the Grenoble process is time symmetric. To see this, note that reversing the direction of the arrows and relabelling
\begin{equation}
    \begin{split}
        P &\longleftrightarrow F \\
        l &\longleftrightarrow f \\
        m &\longleftrightarrow g  \\
          n  &\longleftrightarrow h  \\
            l_1 &\longleftrightarrow n_2  \\
              l_2 &\longleftrightarrow  m_1 \\
                  m_2 &\longleftrightarrow n_1  \\
    \end{split}
\end{equation}
results in exactly the same graph. 

The fact that the routed graph satisfies univocality is implicit in the above explanation of how the bifurcation choices pick out a causal order. For they do so precisely by determining which branch happens at each node. For example, the bifurcation choice $l=1$ at $P$ leads to the branch $A^{l=1}$ happening, corresponding to Alice coming first. Then Alice's bifurcation choice determines which branches happen at $B$ and $C$: $B^{l_1=1}$ and $C^{h=1}$ if she chooses $l_1=1$; $C^{l_2=1}$ and $B^{g=1}$ if she chooses $l_2=1$. In general, the bifurcation choices at $P$ and at the resulting first intermediate node always determine which branches happen. Thus, the routed graph satisfies univocality. Since it is time-symmetric, it immediately follows that it also satisfies bi-univocality.}

This allows us to draw the branch graph, which is shown in Figure \ref{fig:grenoble_branchgraph}. Since this branch graph has no loops, it trivially satisfies our weak loops condition. Thus the routed graph is valid. It follows that the Grenoble process -- and any other process constructed from this routed graph -- is consistent.

\subsubsection{The routed circuit}

\begin{figure}[ht]
	\centering
	\input{figures/grenoble.tikz}
	\caption{A routed circuit diagram for the Grenoble process, using a global index constraint. To avoid graphical clutter, we have avoided explicitly drawing loops. Instead, output lines that end with dots are to be interpreted as being looped back to join the corresponding input lines with dots (with the same system labels, including indices). Wires are coloured for better readability.}
	\label{fig:grenoble}
\end{figure} 

By inserting suitable unitary transformations into the skeletal supermap associated with Figure \ref{fig:grenoble_routedgraph}, we can now construct the Grenoble process. A routed circuit for the Grenoble process is given in Figure \ref{fig:grenoble}:
\begin{itemize}
\item The systems $P_T,F_T,A_{\rm in}, A_{\rm out}, B_{\rm in}, B_{\rm out}, C_{\rm in}, C_{\rm out}$ are all isomorphic, and correspond to a 2-dimensional target Hilbert space (encoding the message). 
\item $P_C$ is a 3-dimensional control system; $F_C$ is a 3-dimensional control system, $F_A$ is a 2-dimensional ancillary system. 
\item The routed systems $Q^{fl n_1 m_2}, Q^{gm l_1 n_2}, Q^{hn m_1 l_2}$  are 4-dimensional control systems,  with an explicit partition into four 1-dimensional sectors, each corresponding to exactly one of their four indices being equal to one.
\item The routed systems $R^{l}, R^{f}, S^m, S^g, T^n, T^h$  are 3-dimensional systems,  with an explicit partition into one 1-dimensional `dummy' sector and one 2-dimensional `message' sector, for example, $R^l = R^0 \oplus R^1$, where $R^{0}$ is the 1-dimensional sector.
\item The routed systems $D^f, E^g, K^h$  are  similarly 3-dimensional systems,  with an explicit partition into one 1-dimensional `dummy' sector and one 2-dimensional `ancillary' sector, for example, $D^l = D^0 \oplus D^1$, where $D^{0}$ is the 1-dimensional sector. The 2-dimensional `ancillary' system will be used to store the information about whether the message was sent clockwise or anticlockwise (or in a superposition of the two) after the first agent, conditional on the state of the qubit before the action of the third agent.
\item The routed systems $X^{n_1 m_1}, Y^{m_2 n_2}, X^{l_1 n_1}, Y^{n_2 l_1}, X^{m_1 l_1}, Y^{l_2 m_2}$ are all $4$-dimensional systems, partitioned into \textit{one} $2$-dimensional `message' sector (corresponding to the message travelling from the second to the third agent), \textit{one} 1-dimensional `message' sector (corresponding to the message travelling from the first to the second agent, in which case the space is only one-dimensional because the state of the message itself determines to whom it is sent next), and \textit{one} 1-dimensional `dummy' sector. For example, $X^{n_1 m_1} = X^{00}\oplus X^{10} \oplus X^{01}$, where $X^{00}$ is the 1-dimensional `dummy' sector and $X^{10}$ is the 1-dimensional `message' sector.
\end{itemize}

The global index imposes a route $\delta^{(l+m+n), 1}$ on $U_P$ that forces exactly one of its output indices to be equal to 1. In other words, its practical output space is $\bigoplus_{lmn} \delta^{(l+m+n), 1} R^{l} \otimes S^{m} \otimes T^{n}$. 

$U_P$ is a three-party generalisation of the superposition-of-paths unitary (\ref{sup channels def}) from Section \ref{sec:2switch}. Its action is given by the following, where
we label the kets by individual sectors, rather than by systems:
\begin{equation}
	U_P :
	\begin{cases}
		\ket{1}_{P_C} \otimes \ket{\psi}_{P_T} \mapsto \ket{\psi}_{R^{1}} \otimes \ket{\rm dum}_{S^{0}} \otimes \ket{\rm dum}_{T^{0}}  \\
		\ket{2}_{P_C} \otimes \ket{\psi}_{P_T} \mapsto\ket{\rm dum}_{R^{0}} \otimes \ket{\psi}_{S^{1}} \otimes \ket{\rm dum}_{T^{0}}  \\
		\ket{3}_{P_C} \otimes \ket{\psi}_{P_T} \mapsto \ket{\rm dum}_{R^{0}} \otimes \ket{\rm dum}_{S^{0}} \otimes \ket{\psi}_{T^{1}}   \\
	\end{cases}
	\label{eq:U_grenoble}
\end{equation}
Thus $U_P$ defines a unitary transformation from $P_C \otimes P_T$ to $\bigoplus_{lmn} \delta^{(l+m+n), 1} R^{l} \otimes S^{m} \otimes T^{n}$. In fact, the global index constraint (in particular, the floating equation $l+m+n=1$) restricts $U_P$'s practical output space to $\bigoplus_{lmn} \delta^{(l+m+n), 1} R^{l} \otimes S^{m} \otimes T^{n}$, meaning that it defines a routed unitary transformation.

$V_A$ is defined below.
Note that here the labelling by sectors is necessary to distinguish between states belonging to different sectors that we label with the same ket, e.g.\ $\ket{0}_{X^{01}}$ and $\ket{0}_{X^{10}}$.

\begin{equation} \label{eq:v_grenoble}
V_A:
\begin{cases}
  \ket{0}_{R^1} \otimes \ket{{\rm dum}}_{X^{00}} \otimes \ket{{\rm dum}}_{Y^{00}} \mapsto \ket{0}_{A_{\rm in}} \otimes \ket{0}_Q \otimes \ket{{\rm dum}}_{D^0} \\  
  \ket{1}_{R^1} \otimes \ket{{\rm dum}}_{X^{00}} \otimes \ket{{\rm dum}}_{Y^{00}} \mapsto \ket{1}_{A_{\rm in}} \otimes \ket{0}_Q \otimes \ket{{\rm dum}}_{D^0} \\  
  \ket{{\rm dum}}_{R^0} \otimes \ket{0}_{X^{10}} \otimes \ket{{\rm dum}}_{Y^{00}} \mapsto \ket{0}_{A_{\rm in}} \otimes \ket{1}_Q \otimes \ket{{\rm dum}}_{D^0} \\  
  \ket{{\rm dum}}_{R^0} \otimes \ket{{\rm dum}}_{X^{00}} \otimes  \ket{1}_{Y^{10}} \mapsto \ket{1}_{A_{\rm in}} \otimes \ket{2}_Q \otimes \ket{{\rm dum}}_{D^0} \\  
  
  \ket{{\rm dum}}_{R^0} \otimes \ket{0}_{X^{01}} \otimes  \ket{{\rm dum}}_{Y^{00}} \mapsto \ket{0}_{A_{\rm in}} \otimes \ket{3}_Q \otimes \ket{0}_{D^1} \\  
  \ket{{\rm dum}}_{R^0} \otimes \ket{1}_{X^{01}} \otimes  \ket{\rm{dum}}_{Y^{00}} \mapsto \ket{1}_{A_{\rm in}} \otimes \ket{3}_Q \otimes \ket{1}_{D^1} \\  
  \ket{{\rm dum}}_{R^0} \otimes \ket{{\rm dum}}_{X^{00}} \otimes  \ket{0}_{Y^{01}} \mapsto \ket{0}_{A_{\rm in}} \otimes \ket{3}_Q \otimes \ket{1}_{D^1} \\  
  \ket{{\rm dum}}_{R^0} \otimes \ket{{\rm dum}}_{X^{00}} \otimes  \ket{1}_{Y^{01}} \mapsto \ket{1}_{A_{\rm in}} \otimes \ket{3}_Q \otimes \ket{0}_{D^1} \\  
\end{cases}
\end{equation}
Since the global index constraint restricts $V_A$'s practical input and output spaces to those sectors where exactly one index is equal to 1, it also defines a routed isometry\footnote{The notion of a routed isometry is defined similarly to that of a routed unitary \cite{vanrietvelde2021routed}.}. $V_B$ and $V_C$ are defined similarly. 

The routed unitary $W_A$ is defined as follows:

\begin{equation} \label{eq:w_grenoble}
    W_A:
    \begin{cases}
        \ket{0}_{A_{\rm out}} \otimes \ket{0}_Q \mapsto \ket{{\rm dum}}_{R^0} \otimes \ket{0}_{X^{10}} \otimes \ket{{\rm dum}}_{Y^{00}} \\
        \ket{1}_{A_{\rm out}} \otimes \ket{0}_Q \mapsto \ket{{\rm dum}}_{R^0} \otimes \ket{{\rm dum}}_{X^{00}} \otimes \ket{{\rm 1}}_{Y^{10}} \\
        \ket{0}_{A_{\rm out}} \otimes \ket{1}_Q \mapsto \ket{{\rm dum}}_{R^0} \otimes \ket{0}_{X^{01}} \otimes \ket{{\rm dum}}_{Y^{00}} \\
        \ket{1}_{A_{\rm out}} \otimes \ket{1}_Q \mapsto \ket{{\rm dum}}_{R^0} \otimes \ket{1}_{X^{01}} \otimes \ket{{\rm dum}}_{Y^{00}} \\

        \ket{0}_{A_{\rm out}} \otimes \ket{2}_Q \mapsto \ket{{\rm dum}}_{R^0} \otimes \ket{{\rm dum}}_{X^{00}} \otimes \ket{0}_{Y^{01}} \\
        \ket{1}_{A_{\rm out}} \otimes \ket{2}_Q \mapsto \ket{{\rm dum}}_{R^0} \otimes \ket{{\rm dum}}_{X^{00}} \otimes \ket{1}_{Y^{01}} \\
        \ket{0}_{A_{\rm out}} \otimes \ket{3}_Q \mapsto \ket{0}_{R^1} \otimes \ket{{\rm dum}}_{X^{00}} \otimes \ket{{\rm dum}}_{Y^{00}} \\
        \ket{0}_{A_{\rm out}} \otimes \ket{3}_Q \mapsto \ket{1}_{R^1} \otimes \ket{{\rm dum}}_{X^{00}} \otimes \ket{{\rm dum}}_{Y^{00}} \\
    \end{cases}
\end{equation}
The routed unitaries $W_B$ and $W_C$ are defined in a similar way. Finally, the routed unitary \textit{\AE} is given by the following: 
\begin{equation}    \label{eq:ae_grenoble}
\textit{\AE}_F:
    \begin{cases}
        \ket{\psi}_{R^1} \otimes \ket{\xi}_{D^1} \otimes \ket{{\rm dum}}_{S^0} \otimes \ket{{\rm dum}}_{E^0} \otimes \ket{{\rm dum}}_{T^0} \otimes \ket{{\rm dum}}_{K^0} \mapsto \ket{\xi}_{F_A} \otimes \ket{1}_{F_C} \otimes \ket{\psi}_{F_T}  \\
        \ket{{\rm dum}}_{R^0} \otimes \ket{{\rm dum}}_{D^0} \otimes \ket{\psi}_{S^1} \otimes \ket{\xi}_{E^1} \otimes \ket{{\rm dum}}_{T^0} \otimes \ket{{\rm dum}}_{K^0} \mapsto \ket{\xi}_{F_A} \otimes \ket{2}_{F_C} \otimes \ket{\psi}_{F_T}  \\
        \ket{{\rm dum}}_{R^0} \otimes \ket{{\rm dum}}_{D^0} \otimes \ket{{\rm dum}}_{S^0} \otimes \ket{{\rm dum}}_{E^0} \otimes \ket{\psi}_{T^1} \otimes \ket{\xi}_{K^1} \mapsto \ket{\xi}_{F_A} \otimes \ket{3}_{F_C} \otimes \ket{\psi}_{F_T}  \\
    \end{cases}
\end{equation}

Note that the Grenoble process is an isometric process, with the overall output dimension greater than the overall input dimension 
(in particular, $V$ is a routed isometry). The process can be made unitary in a natural way, by adding an extra 2-dimensional ancillary qubit to the input of the Past and adding routed wires of dimension $1+2$ from the Past to  each of the routed unitaries $W$, bearing the same index as the wire from the Past to the corresponding $V $. This makes the process symmetric in time. As a result, this increases the dimension of the Hilbert space of the sector carrying the message between the first and second agents from 1 to 2. In turn, this increases the dimensionality of the input space to the unitaries $V$, making the entire process unitary. 

Note also that the Future cannot necessarily determine the relative order of the first two agents from their control and ancillary qubits $F_C, F_A$, if the third agent performs a non-unitary operation (because the order information encoded in the ancillary qubit relied on knowledge of the state of the message before the action of the third agent).

One peculiar feature of the Grenoble process is that the qubit that we have called the `target qubit' -- that is, the system that passes between the intermediate agents -- plays a dual role. On the one hand, it is the `message' that the agents receive. On the other hand, it also plays a role in determining the causal order. In particular, after it passes through the first agent, its logical state determines which agent receives it next. Thus if Alice comes first and wants to send the target qubit to Bob, she must send him the $\ket{0}$ state, but if she wants to send it to Charlie, she must send him $\ket{1}$. 

Our reconstruction of the Grenoble process makes it obvious that this feature is not necessary to make the process consistent. Starting from the same routed graph, one can easily define a variation on the Grenoble process, in which Alice is also given a second, `control' qubit. This control qubit determines which agent comes second, leaving Alice free to send that agent whatever state on the target qubit she likes. Bob and Charlie can also be given their own qubits. Since this process can be obtained by fleshing out a routed graph whose validity we have already checked, it is immediate that this new process is also consistent. This illustrates a useful feature of our framework for constructing processes with indefinite causal order; namely, that variations on a process can be defined in a straightforward way, leading to a clearer understanding of which features of the original process were essential for its logical consistency, and which other features can be changed at will.

\subsection{The Lugano process} \label{sec:lugano}

The Lugano process, discovered by Ara\'ujo and Feix and then published and further studied by Baumeler and Wolf \cite{baumeler2014maximal, baumeler2016space} (and therefore sometimes also called BW or AF/BW), is the seminal example of a unitary process violating causal inequalities. It was first presented as a classical process, whose unitary extension to quantum theory can be derived in a straightforward way \cite{araujo2017purification}. As we place ourselves in a general quantum framework here, we will primarily focus on this quantum version of the process; we note that the classical version can be obtained from the quantum one by feeding it specific input states and introducing decoherence in each of the wires of the circuit. This shows, more generally, that at least some exotic classical processes are also part of the class of processes that can be built through our procedure.

Indeed, we will show here how the (unitary) Lugano process can be constructed from a valid routed graph; this will provide an example of a process violating causal inequalities that can also be accommodated by our framework. In fact, we will derive a larger family of processes, defined by a same valid routed graph, and display how the Lugano process can be obtained as the simplest instance of this family. The other processes in this family share the basic behaviour of the Lugano process, but can feature, on top of it, arbitrarily large dimensions and arbitrarily complex operations.

\subsubsection{The logical structure}

Before we present the routed graph for the construction of the Lugano process, let us start with an intuitive account of the logical structure lying at the heart of it. This logical structure can be presented as a voting protocol involving three agents, in which each of the agents receives part of the result of the vote before having even cast their vote. Why this is possible without leading to any logical paradox, of the grandfather type, is the central point to understand.

In this voting protocol, each agent casts a vote for which of the other two agents they would like to see come last in the causal order. Alice, for instance, can either vote for Bob or for Charlie to come last.
If there is a majority, then the winning agent can both i) learn that they won the vote, and ii) receive (arbitrarily large) messages from each of the two losing agents. As for the two losers, each of them can only learn that they lost (i.e.\ no majority was obtained in favour of them), and they cannot receive any messages from the other agents. If no majority is obtained, then all agents learn that they lost, and none of them can signal to any other.


This voting protocol would have nothing surprising if it assumed that the winner learns of their victory and receives messages from the losers `after' all the votes are cast.
Yet in the Lugano process, the crucial fact is that Alice, for instance, learns whether she won, and (if she won) receives Bob and Charlie's messages, \textit{before} she casts her own vote; and the same goes for Bob and Charlie.
This sounds dangerously close to a grandfather paradox, since each agent contributes to an outcome that they might become aware of before they make their contribution. It seems likely that the agents could somehow take advantage of this system to send messages back to their own past, and decide what they do based on those messages, leading to logical inconsistencies.


Why this never happens -- why, more precisely, the agents still have no way to send information back to themselves -- can be figured out with a bit of analysis of the voting system. Indeed, Alice, for instance, finds herself in either of two cases. The first one is that she won: a majority `was' obtained in favour of her. This means that she cannot send messages to either of the agents, since only the winner can receive messages. Nor can she signal to other agents by casting her vote: her victory implies that both Bob and Charlie voted for her, in which case her own vote is irrelevant to the outcome. Therefore, if Alice wins, then she cannot send any information back to herself via the other agents.

Alternatively, Alice could lose the vote. If so, then she cannot receive any messages from the other agents, so she has no hope of sending information back to herself through their messages to her. Therefore, if she wants to send information to herself, she will have to try to change the outcome of the vote in her favour (thus creating a grandfather-type paradox). But she cannot do this by simply changing her own vote, as there being a majority in her favour only depends on how the other agents vote. Nor can she make herself win by encouraging the other agents to vote differently. Alice can only send a message to (say) Bob if Charlie voted for Bob as well. So, whatever Bob does, there will never be a majority in favour of Alice. Therefore, if Alice loses, she cannot send any information back to herself. For this reason, the Lugano process, despite conflicting with intuitions about causal and temporal structure, does not lead to any logical paradoxes, after all.

\subsubsection{The routed graph}

\begin{figure}
	\centering
\scalebox{1.3}{\begin{tikzpicture}
	\begin{pgfonlayer}{nodelayer}
		\node [style=none] (0) at (-7, -4) {$A$};
		\node [style=none] (1) at (0, -4) {$B$};
		\node [style=none] (2) at (7, -4) {$C$};
		\node [style=none] (3) at (0, 2) {$X$};
		\node [style=none] (6) at (0, 7.5) {$F$};
		\node [style=none] (8) at (-7, -3.5) {};
		\node [style=none] (9) at (-6.5, -3.5) {};
		\node [style=none] (11) at (0.75, -3.5) {};
		\node [style=none] (12) at (0, -3.5) {};
		\node [style=none] (13) at (-6.75, -4.5) {};
		\node [style=none] (15) at (-7.25, -4.5) {};
		\node [style=none] (16) at (-0.25, -4.5) {};
		\node [style=none] (19) at (6.5, -3.5) {};
		\node [style=none] (24) at (6.75, -4.5) {};
		\node [style=none] (29) at (-0.5, 2.5) {};
		\node [style=none] (30) at (0, 2.5) {};
		\node [style=none] (31) at (-0.75, 1.5) {};
		\node [style=none] (32) at (0.75, 1.5) {};
		\node [style=none] (33) at (0, 1.5) {};
		\node [style=none] (38) at (-0.75, 7.25) {};
		\node [style=none] (39) at (0, 6.75) {};
		\node [style=none] (40) at (0.25, 6.75) {};
		\node [style=none] (48) at (-6.25, -6.5) {};
		\node [style=none] (62) at (-0.75, 3.25) {};
		\node [style=none] (63) at (-0.75, 4.5) {};
		\node [style=none] (64) at (-0.75, 3.25) {};
		\node [style=none] (65) at (-0.75, 3.5) {};
		\node [style=none] (68) at (-7.5, -5.75) {};
		\node [style=none] (69) at (-7.5, -6.5) {};
		\node [style=none] (71) at (-0.5, -5.75) {};
		\node [style=none] (73) at (-0.5, -6.5) {};
		\node [style=none] (74) at (6.5, -5.75) {};
		\node [style=none] (75) at (6.5, -6.5) {};
		\node [style=none] (80) at (-3.75, 0) {\color{blue} \footnotesize $i l$};
		\node [style=none] (81) at (-5.75, 5) {\color{blue} \footnotesize $l$};
		\node [style=none] (82) at (3.75, -0.25) {\color{orange} \footnotesize $k n$};
		\node [style=none] (85) at (-1, -1) {\color{magenta} \footnotesize $j m$};
		\node [style=none] (93) at (0.25, -4.5) {};
		\node [style=none] (94) at (0.75, -6.5) {};
		\node [style=none] (95) at (7.25, -4.5) {};
		\node [style=none] (96) at (7.75, -6.5) {};
		\node [style=none] (97) at (-8, -4.75) {\color{blue} \footnotesize $l$};
		\node [style=none] (98) at (-1, -4.75) {\color{magenta} \footnotesize $m$};
		\node [style=none] (99) at (6, -4.75) {\color{orange} \footnotesize $n$};
		\node [style=none] (125) at (10.5, 4) {\footnotesize  $l=\bar{j} \cdot k$};
		\node [style=none] (126) at (10.5, 3) {\footnotesize  $m=\bar{k} \cdot i$};
		\node [style=none] (127) at (10.5, 2) {\footnotesize  $n =\bar{i} \cdot j$};
		\node [style=none] (128) at (7, -3.5) {};
		\node [style=none] (129) at (0.75, 7.25) {};
		\node [style=none] (130) at (-0.75, 2.25) {};
		\node [style=none] (131) at (-2.25, 3.25) {};
		\node [style=none] (132) at (-2.25, 4.5) {};
		\node [style=none] (133) at (-2.25, 3.25) {};
		\node [style=none] (134) at (-2.25, 3.5) {};
		\node [style=none] (136) at (-8, -5.5) {\footnotesize \color{red}$0$};
		\node [style=none] (137) at (0.5, 2.5) {};
		\node [style=none] (138) at (0.75, 3.25) {};
		\node [style=none] (139) at (0.75, 4.5) {};
		\node [style=none] (140) at (0.75, 3.25) {};
		\node [style=none] (141) at (0.75, 3.5) {};
		\node [style=none] (142) at (5.75, 5) {\color{orange} \footnotesize $n$};
		\node [style=none] (143) at (3.5, 2.25) {\color{magenta} \footnotesize $m$};
		\node [style=none] (144) at (-2.75, 4) {\color{blue} \footnotesize $l$};
		\node [style=none] (145) at (-1.25, 4) {\color{magenta} \footnotesize $m$};
		\node [style=none] (146) at (1.25, 4) {\color{orange} \footnotesize $n$};
		\node [style=none] (147) at (-0.75, 5.75) {\footnotesize $i j k$};
		\node [style=none] (148) at (-1, -5.5) {\footnotesize \color{red}$0$};
		\node [style=none] (149) at (6, -5.5) {\footnotesize \color{red}$0$};
		\node [style=none] (150) at (-5, -0.75) {\footnotesize \color{red}$(0,1)$};
		\node [style=none] (153) at (10.5, 0) {\footnotesize  $i \cdot l = 0$};
		\node [style=none] (154) at (10.5, -1) {\footnotesize  $j \cdot m = 0$};
		\node [style=none] (155) at (10.5, -2) {\footnotesize  $k \cdot n = 0$};
		\node [style=none] (156) at (-1, -1.75) {\footnotesize \color{red}$(0,1)$};
		\node [style=none] (157) at (4.75, -1) {\footnotesize \color{red}$(0,1)$};
	\end{pgfonlayer}
	\begin{pgfonlayer}{edgelayer}
		\draw [style=blue arrow, in=-165, out=90, looseness=1.25] (8.center) to (38.center);
		\draw [style=blue arrow] (9.center) to (31.center);
		\draw [style=orange arrow] (19.center) to (32.center);
		\draw [style=magenta arrow, in=-45, out=60, looseness=1.25] (11.center) to (40.center);
		\draw [style=magenta arrow] (12.center) to (33.center);
		\draw [style=arrow plain] (30.center) to (39.center);
		\draw [style=arrow plain, in=-60, out=90] (48.center) to (13.center);
		\draw [style=magenta, in=-90, out=120] (29.center) to (65.center);
		\draw [style=magenta, dotted] (65.center) to (63.center);
		\draw [style=blue, dotted] (69.center) to (68.center);
		\draw [style=magenta, dotted] (73.center) to (71.center);
		\draw [style=orange, dotted] (75.center) to (74.center);
		\draw [style=blue arrow, in=-120, out=90] (68.center) to (15.center);
		\draw [style=magenta arrow, in=-120, out=90] (71.center) to (16.center);
		\draw [style=orange arrow, in=-120, out=90, looseness=1.25] (74.center) to (24.center);
		\draw [style=arrow plain, in=-60, out=90] (94.center) to (93.center);
		\draw [style=arrow plain, in=-60, out=90] (96.center) to (95.center);
		\draw [style=orange arrow, in=-15, out=90, looseness=1.25] (128.center) to (129.center);
		\draw [style=blue, in=-90, out=150] (130.center) to (134.center);
		\draw [style=blue, dotted] (134.center) to (132.center);
		\draw [style=orange, in=-90, out=60] (137.center) to (141.center);
		\draw [style=orange, dotted] (141.center) to (139.center);
	\end{pgfonlayer}
\end{tikzpicture}}
\caption{\changes{The routed graph for the Lugano process. To help intuition, we used different colours to denote the arrows that pertain to some particular agent (i.e.\ the ones whose indices encode the `votes' or the `vote result' for that agent). Each of the indices has only two possible values, $0$ or $1$. To reduce clutter, we have used arrows with dotted ends to avoid drawing all the arrows explicitly; pairs of dotted arrows with the same index are shorthand for a single unbroken arrow. Conditions in red next to an arrow indicate which sectors of that arrow have to be trivial; for instance, in the blue arrow from $A$ to $X$, the sector $(i=0, l =1)$ is trivial.}}
	\label{fig: Lugano routed graph}
\end{figure}
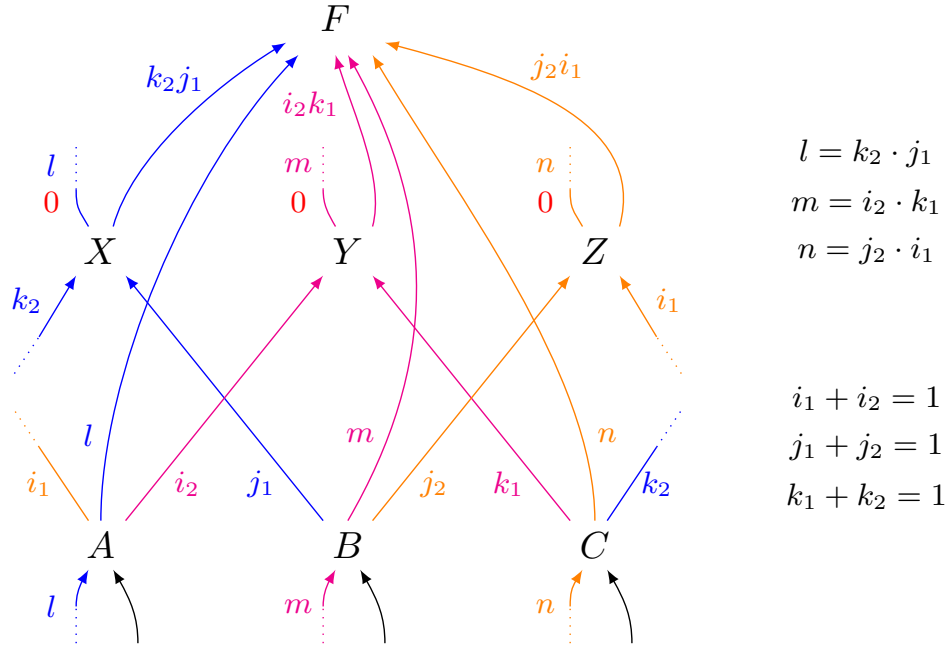

Of course, our description of the Lugano process so far has only been pitched at an intuitive level. The point of the routed graph that we will now present is precisely to formalise this intuitive description; while the validity of this graph -- defined as the satisfaction of our two principles -- will provide a formal counterpart to our argument that no logical paradoxes should arise from this protocol. Our routed graph is depicted in Figure \ref{fig: Lugano routed graph}.

In this routed graph, the nodes $A$, $B$, and $C$, representing the three agents, \changes{are supplemented with an additional node $X$ that can be thought of -- continuing with our metaphor -- as a `vote-counting station', in which the votes for each of the agents are centralised and counted.}\footnote{\changes{We stress that the inclusion of this auxiliary node is an \textit{ad hoc} procedure, since our framework provides no way to infer the routed circuit decomposition of a given process. An informal intuition for this inclusion is that, when framed as a skeletal supermap, the Lugano process seems to require a node at which the information is processed and dispatched to other nodes. As we elaborate upon in the conclusion, deriving formally the routed circuit decomposition of a given process is an important subject for future research.}}


Let us determine the routes in the graph and explain their meaning, starting with node $A$. The index $l$ of the arrow going into $A$ from its vote-counting station $X$ indicates whether Alice won the vote (it has value $1$ if she wins, and $0$ if she loses). \changes{The index $i$ encodes whether Alice votes for Bob ($i=1$) or Charlie ($i=0$).  The constraint $i \cdot l = 0$ enforces that, if Alice wins, her vote (which is then irrelevant anyway) is recorded as a vote for Charlie.\footnote{This quirk is necessary to avoid Alice sending information through the vote-counting station if she wins.}} This leads to the route for $A$ depicted in Figure \ref{fig: LuganoNodes}. The route consists of two branches corresponding to Alice's victory or defeat, with a binary bifurcation choice representing her own vote in the losing branch. The routes for $B$ and $C$ are fully analogous.

\begin{figure}
	\centering
	\begin{tikzpicture}
	\begin{pgfonlayer}{nodelayer}
		\node [style=very small black dot] (0) at (-4, -1.25) {};
		\node [style=very small black dot] (1) at (-5.5, 1.75) {};
		\node [style=very small black dot] (2) at (-2.5, 1.75) {};
		\node [style=none] (3) at (-5.3, 1.25) {};
		\node [style=none] (4) at (-2.7, 1.25) {};
		\node [style=none] (5) at (-4.3, -0.75) {};
		\node [style=none] (6) at (-3.7, -0.75) {};
		\node [style=none] (7) at (-5.475, 2.25) {};
		\node [style=none] (8) at (-2.5, 2.25) {};
		\node [style=none] (9) at (-6, 1.75) {};
		\node [style=none] (10) at (-4, -1.75) {};
		\node [style=none] (11) at (-2, 1.75) {};
		\node [style=none] (12) at (-3.5, -1.5) {};
		\node [style=none] (13) at (-4.5, -1.5) {};
		\node [align=center, text width=4cm, font={\footnotesize}] (14) at (11, 1.75) {\color{red} (+ values outside the practical outputs)};
		\node [align=center] (15) at (-4, -3.25) {\color{blue} $A^0$ \\ \color{blue} `Alice loses'};
		\node [style=none] (16) at (-10, -1.25) {inputs};
		\node [style=none] (17) at (-10, 1.75) {outputs};
		\node [font={\footnotesize}] (18) at (-2, -1.25) {$l=0$};
		\node [align=center, font={\footnotesize}] (19) at (-6, 3.75) { $ \begin{cases} l =0 \\ i=0 \end{cases} $};
		\node [align=center, font={\footnotesize}] (20) at (-2, 3.75) { $ \begin{cases} l =0 \\ i=1 \end{cases} $};
		\node [align=center] (35) at (4, -3.25) {\color{blue} $A^1$ \\ \color{blue} `Alice wins'};
		\node [font={\footnotesize}] (36) at (6, -1.25) {$l=1$};
		\node [align=center, font={\footnotesize}] (38) at (4, 3.75) { $ \begin{cases} l =1 \\ i = 0 \end{cases} $};
		\node [style=very small black dot] (39) at (4, -1.25) {};
		\node [style=very small black dot] (40) at (4, 1.75) {};
		\node [style=none] (41) at (4, -1) {};
		\node [style=none] (42) at (4, 1.5) {};
		\node [style=none] (43) at (3.5, -1) {};
		\node [style=none] (44) at (3.5, 1.5) {};
		\node [style=none] (45) at (4.5, -1) {};
		\node [style=none] (46) at (4.5, 1.5) {};
	\end{pgfonlayer}
	\begin{pgfonlayer}{edgelayer}
		\draw [style=arrow plain] (5.center) to (3.center);
		\draw [style=arrow plain] (6.center) to (4.center);
		\draw [style=dashed, blue] (11.center)
			 to (12.center)
			 to [in=0, out=-120] (10.center)
			 to [in=-60, out=180] (13.center)
			 to (9.center)
			 to [in=180, out=120, looseness=1.25] (7.center)
			 to [in=180, out=0] (8.center)
			 to [in=60, out=0] cycle;
		\draw [style=arrow] (41.center) to (42.center);
		\draw [style=dashed blue] (44.center) to (43.center);
		\draw [style=dashed blue, bend right=90, looseness=2.50] (43.center) to (45.center);
		\draw [style=dashed blue] (45.center) to (46.center);
		\draw [style=dashed blue, in=90, out=90, looseness=2.50] (46.center) to (44.center);
	\end{pgfonlayer}
\end{tikzpicture}
	\caption{\changes{The branch structure of the route for node $A$. $B$ and $C$ are fully analogous.}}
	\label{fig: LuganoNodes}
\end{figure}
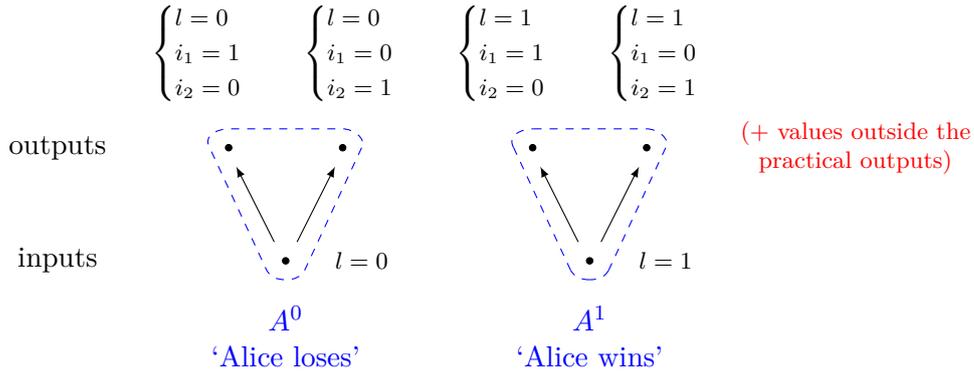

\changes{Now let us consider $X$, the vote-counting station. This node features all of the graph's indices in both its inputs and its outputs. The values of $l,m,n$ are uniquely determined by those of $i,j,k$ through the global index constraints, and it can be checked that there are 5 allowed triples of values for the latter ($000$, $111$, $001$, $100$, and $010$), so $X$ features $5$ branches and none of these branches feature bifurcations.}

Finally, the `global Future' node $F$ just serves to channel out the remaining information. Since it receives all the distinct indices in the graph, its route is just given by the global index constraint. In other words, its practical input set of values is just the set of values permitted by the index-matching and the floating equations.

\changes{For the arrows $X \to A$, $X \to B$, and $X \to C$, the $0$ value of their index corresponds to a one-dimensional `dummy' sector.
The interpretation of this is once again natural: the messages are sent to an agent only if this agent won. For the arrows $A \to X$, $B \to X$, and $C \to X$, the sectors for which $l=1$ (resp.\ $m=1$, $n=1$) are trivial as well; this serves to prevent the winner from sending messages to the other agents through the counting station.}


We can now check that the routed graph of Figure \ref{fig: Lugano routed graph} satisfies our two principles. We start with univocality. The choice relation for this graph can be checked to be a function from the six binary bifurcation choices to the statuses of the branches. This function can be meaningfully presented in the following algorithmic way:
\begin{itemize}
    \item \changes{Look at the votes of the losing branches ($A^0$, $B^0$ and $C^0$). If a majority is found in these votes (say, in favour of Alice), then set the `result' indices accordingly (in this case, $l=1$, $m=n=0$) and use the bifurcation choices of the losing branches ($B^0$ and $C^0$) to set the value of the votes of `losers' ($j, k$); set the value of the winner's vote to $0$ ($i=0$);}
    \item If no majority is found, define $l=m=n=0$ and use the bifurcation choices of the losing branches to set all votes.
    \item Now that the values of all indices in the graph have been fixed, derive which branches happened and which didn't.
\end{itemize}

\begin{figure}
	\centering
	\begin{tikzpicture}
	\begin{pgfonlayer}{nodelayer}
		\node [style=none] (0) at (-4, -0.25) {$A^0$};
		\node [style=none] (1) at (0, -0.25) {$B^0$};
		\node [style=none] (2) at (4, -0.25) {$C^0$};
		\node [style=none] (3) at (-3.25, 0.25) {};
		\node [style=none] (4) at (-3.5, 0.5) {};
		\node [style=none] (5) at (-0.25, 0.25) {};
		\node [style=none] (6) at (0.25, 0.25) {};
		\node [style=none] (7) at (3.5, 0.5) {};
		\node [style=none] (8) at (3.25, 0.25) {};
		\node [style=none] (9) at (3.5, -1.25) {};
		\node [style=none] (10) at (3.25, -1) {};
		\node [style=none] (11) at (0.25, -1) {};
		\node [style=none] (12) at (-0.25, -1) {};
		\node [style=none] (13) at (-3.25, -1) {};
		\node [style=none] (14) at (-3.5, -1.25) {};
		\node [style=none] (15) at (-2, 5) {$X^{111}$};
		\node [style=none] (18) at (-4, 8) {$A^1$};
		\node [style=none] (19) at (0, 8) {$B^1$};
		\node [style=none] (20) at (4, 8) {$C^1$};
		\node [style=none] (21) at (-4.15, 1) {};
		\node [style=none] (22) at (-3.4, 1) {};
		\node [style=none] (23) at (-0.125, 1) {};
		\node [style=none] (24) at (0.125, 1) {};
		\node [style=none] (25) at (3.4, 1) {};
		\node [style=none] (26) at (4.15, 1) {};
		\node [style=none] (27) at (-4.15, 2.75) {};
		\node [style=none] (28) at (-2.25, 2.75) {};
		\node [style=none] (29) at (-1, 2.75) {};
		\node [style=none] (31) at (1, 2.75) {};
		\node [style=none] (32) at (4.15, 2.75) {};
		\node [style=none] (33) at (2.25, 2.75) {};
		\node [style=none] (34) at (-3.25, 2.25) {$\ldots$};
		\node [style=none] (35) at (0, 2.25) {$\ldots$};
		\node [style=none] (36) at (3.25, 2.25) {$\ldots$};
		\node [style=none] (37) at (-7, 3.5) {};
		\node [style=none] (38) at (7, 3.5) {};
		\node [style=none] (39) at (-7, 6.5) {};
		\node [style=none] (40) at (7, 6.5) {};
		\node [style=none] (75) at (-7, 9.25) {};
		\node [style=none] (76) at (7, 9.25) {};
		\node [style=none] (77) at (0, 10.25) {$F$};
		\node [style=none] (78) at (-8.5, 10.75) {};
		\node [style=none] (79) at (-8.5, -1) {};
		\node [style=none] (80) at (-4, 5) {$X^{000}$};
		\node [style=none] (81) at (0, 5) {$X^{001}$};
		\node [style=none] (82) at (2, 5) {$X^{100}$};
		\node [style=none] (83) at (4, 5) {$X^{010}$};
	\end{pgfonlayer}
	\begin{pgfonlayer}{edgelayer}
		\draw [style=green dashed arrow, bend left] (3.center) to (5.center);
		\draw [style=green dashed arrow, bend left] (12.center) to (13.center);
		\draw [style=green dashed arrow, bend left] (6.center) to (8.center);
		\draw [style=green dashed arrow, bend left, looseness=1.25] (10.center) to (11.center);
		\draw [style=green dashed arrow, bend left] (4.center) to (7.center);
		\draw [style=green dashed arrow, bend left] (9.center) to (14.center);
		\draw [style=arrow plain] (21.center) to (27.center);
		\draw [style=arrow plain] (22.center) to (28.center);
		\draw [style=arrow plain] (23.center) to (29.center);
		\draw [style=arrow plain] (24.center) to (31.center);
		\draw [style=arrow plain] (26.center) to (32.center);
		\draw [style=arrow plain] (25.center) to (33.center);
		\draw [blue, dashed] (37.center) to (38.center);
		\draw [blue, dashed] (39.center) to (40.center);
		\draw [blue, dashed] (75.center) to (76.center);
		\draw [blue, ultra thick, -Triangle] (79.center) to (78.center);
	\end{pgfonlayer}
\end{tikzpicture}
	\caption{\changes{A simplified version of the branch graph for the Lugano process. We do not draw all the arrows, as this would create a lot of clutter and would be superfluous for our purposes of checking for cycles; we rather just organise the branches in layers, such that all unspecified  arrows only ever go `up' with respect to this partition.}}
	\label{fig: Lugano branch graph}
\end{figure}
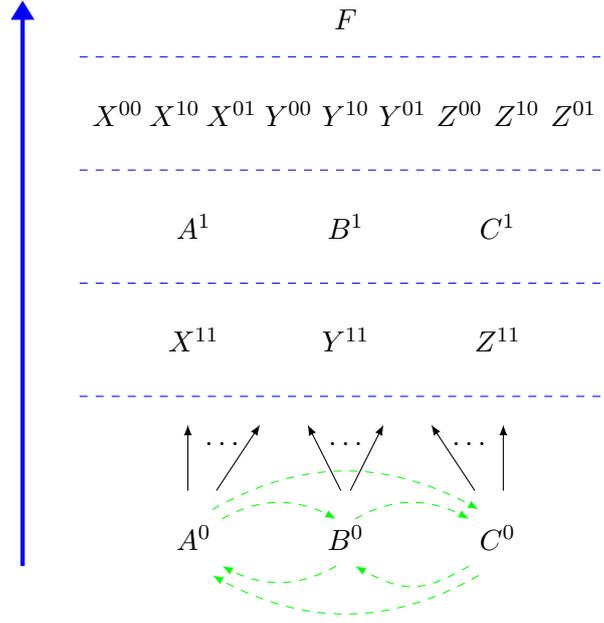 

Univocality is thus satisfied. Its time-reversed version can be checked to be satisfied as well: all bifurcation choices in the reverse graph are located in $F$, and they have the effect of fixing all indices to consistent joint values.


A simplified version of the branch graph is presented in Figure \ref{fig: Lugano branch graph}. We see that there \textit{are} loops in the branch graph, specifically in its bottom layer; yet they are only composed of green dashed arrows. This entails that the routed graph of Figure \ref{fig: Lugano routed graph} satisfies the weak loops principle, and is thus valid. In Section \ref{sec: conclusion}, we will formulate the conjecture that this presence of weak loops is a signature of its causal inequalities violating nature.

\subsubsection{The routed circuit}

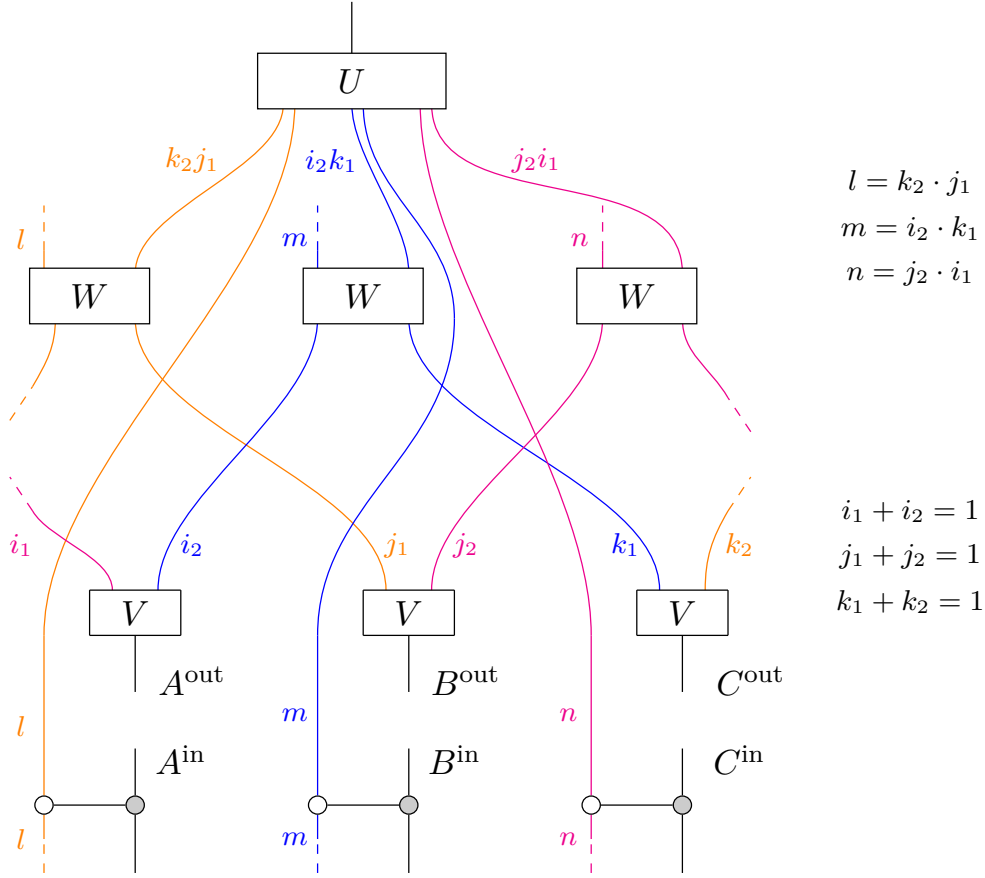
\begin{figure}[ht]
	\centering
	\scalebox{1.2}{\begin{tikzpicture}
	\begin{pgfonlayer}{nodelayer}
		\node [style=semilarge map] (4) at (0, 3) {$W$};
		\node [style=very large map] (6) at (0, 8.5) {$U$};
		\node [style=none] (8) at (-7, -2.75) {};
		\node [style=none] (9) at (-5, -1.75) {};
		\node [style=none] (11) at (-1, -2.75) {};
		\node [style=gray dot] (13) at (-5, -6.5) {};
		\node [style=white dot] (14) at (-7, -6.5) {};
		\node [style=white dot] (15) at (-1, -6.5) {};
		\node [style=none] (16) at (7, -1.75) {};
		\node [style=none] (17) at (5, -2.75) {};
		\node [style=white dot] (19) at (5, -6.5) {};
		\node [style=none] (26) at (-1, 2.5) {};
		\node [style=none] (27) at (1, 2.5) {};
		\node [style=none] (42) at (-5, -8) {};
		\node [style=none] (55) at (-7, -7.25) {};
		\node [style=none] (56) at (-7, -8) {};
		\node [style=none] (57) at (-1, -7.25) {};
		\node [style=none] (58) at (-1, -8) {};
		\node [style=none] (59) at (5, -7.25) {};
		\node [style=none] (60) at (5, -8) {};
		\node [style=gray dot] (76) at (1, -6.5) {};
		\node [style=none] (77) at (1, -8) {};
		\node [style=gray dot] (78) at (7, -6.5) {};
		\node [style=none] (79) at (7, -8) {};
		\node [style=none] (81) at (-1.5, -7.25) {\color{magenta} \footnotesize $m$};
		\node [style=none] (82) at (4.5, -7.25) {\color{orange} \footnotesize $n$};
		\node [style=none] (90) at (-5, -5.25) {};
		\node [style=none] (91) at (1, -5.25) {};
		\node [style=none] (92) at (7, -5.25) {};
		\node [style=none] (93) at (1, -4) {};
		\node [style=none] (94) at (-5, -4) {};
		\node [style=none] (95) at (7, -4) {};
		\node [style=none] (96) at (1, -2.75) {};
		\node [style=none] (97) at (7, -2.75) {};
		\node [style=none] (98) at (-5, -2.75) {};
		\node [style=none] (99) at (-7, -2.75) {};
		\node [style=none] (100) at (-1, -2.75) {};
		\node [style=none] (101) at (5, -2.75) {};
		\node [style=none] (102) at (0, 9) {};
		\node [style=none] (103) at (0, 10.25) {};
		\node [style=none] (106) at (-1.5, -4.5) {\color{magenta} \footnotesize $m$};
		\node [style=none] (107) at (4.5, -4.5) {\color{orange} \footnotesize $n$};
		\node [style=none] (108) at (-4, -5.5) {$A^\inn$};
		\node [style=none] (109) at (-3.75, -3.75) {$A^\out$};
		\node [style=none] (110) at (2, -5.5) {$B^\inn$};
		\node [style=none] (111) at (2.25, -3.75) {$B^\out$};
		\node [style=none] (112) at (8.25, -5.5) {$C^\inn$};
		\node [style=none] (113) at (8.5, -3.75) {$C^\out$};
		\node [style=none] (176) at (-4, 0.5) {\color{blue} \footnotesize $i l$};
		\node [style=none] (177) at (-7.5, 4) {\color{blue} \footnotesize $l$};
		\node [style=none] (178) at (3.5, 1) {\color{orange} \footnotesize $k n$};
		\node [style=none] (180) at (1.5, -0.25) {\color{magenta} \footnotesize $j m$};
		\node [style=none] (185) at (-7.5, -7.75) {\color{blue} \footnotesize $l$};
		\node [style=none] (188) at (11, 3) {\footnotesize  $l=\bar{j} \cdot k$};
		\node [style=none] (189) at (11, 2) {\footnotesize  $m=\bar{k} \cdot i$};
		\node [style=none] (190) at (11, 1) {\footnotesize  $n =\bar{i} \cdot j$};
		\node [style=none] (209) at (-1, 6.5) {\footnotesize $i j k$};
		\node [style=none] (213) at (11, -1) {\footnotesize  $i \cdot l = 0$};
		\node [style=none] (214) at (11, -2) {\footnotesize  $j \cdot m = 0$};
		\node [style=none] (215) at (11, -3) {\footnotesize  $k \cdot n = 0$};
		\node [style=semilarge map] (218) at (-6, -2.25) {$V$};
		\node [style=semilarge map] (219) at (0, -2.25) {$V$};
		\node [style=semilarge map] (220) at (6, -2.25) {$V$};
		\node [style=none] (221) at (1, -1.75) {};
		\node [style=none] (222) at (0, 2.5) {};
		\node [style=none] (223) at (-0.5, 3.5) {};
		\node [style=none] (224) at (0, 3.5) {};
		\node [style=none] (225) at (-0.75, 4.25) {};
		\node [style=none] (226) at (-0.75, 5.5) {};
		\node [style=none] (227) at (-0.75, 4.25) {};
		\node [style=none] (228) at (-0.75, 4.5) {};
		\node [style=none] (229) at (-0.75, 3.25) {};
		\node [style=none] (230) at (-2.25, 4.25) {};
		\node [style=none] (231) at (-2.25, 5.5) {};
		\node [style=none] (232) at (-2.25, 4.25) {};
		\node [style=none] (233) at (-2.25, 4.5) {};
		\node [style=none] (234) at (0.5, 3.5) {};
		\node [style=none] (235) at (0.75, 4.25) {};
		\node [style=none] (236) at (0.75, 5.5) {};
		\node [style=none] (237) at (0.75, 4.25) {};
		\node [style=none] (238) at (0.75, 4.5) {};
		\node [style=none] (239) at (-2.75, 5) {\color{blue} \footnotesize $l$};
		\node [style=none] (240) at (-1.25, 5) {\color{magenta} \footnotesize $m$};
		\node [style=none] (241) at (1.25, 5) {\color{orange} \footnotesize $n$};
		\node [style=none] (242) at (0, 8) {};
		\node [style=none] (243) at (0.5, 8) {};
		\node [style=none] (244) at (1, 8) {};
		\node [style=none] (245) at (-0.75, 8) {};
		\node [style=none] (246) at (-7, 5) {};
		\node [style=none] (247) at (5, 5) {};
		\node [style=none] (248) at (2.75, 4.25) {};
		\node [style=none] (249) at (-1, -1.75) {};
		\node [style=none] (250) at (-7.5, -5) {\color{blue} \footnotesize $l$};
		\node [style=none] (251) at (5.75, 4) {\color{orange} \footnotesize $n$};
		\node [style=none] (252) at (3.5, 4) {\color{magenta} \footnotesize $m$};
	\end{pgfonlayer}
	\begin{pgfonlayer}{edgelayer}
		\draw [blue, in=-90, out=90, looseness=0.75] (9.center) to (26.center);
		\draw [orange, in=-90, out=90, looseness=0.75] (16.center) to (27.center);
		\draw (42.center) to (13);
		\draw [style=blue, dashed] (56.center) to (55.center);
		\draw [style=magenta, dashed] (58.center) to (57.center);
		\draw [style=orange, dashed] (60.center) to (59.center);
		\draw [blue] (55.center) to (14);
		\draw [magenta] (57.center) to (15);
		\draw [orange] (59.center) to (19);
		\draw (77.center) to (76);
		\draw (79.center) to (78);
		\draw (14) to (13);
		\draw (15) to (76);
		\draw (19) to (78);
		\draw (13) to (90.center);
		\draw (76) to (91.center);
		\draw (78) to (92.center);
		\draw (94.center) to (98.center);
		\draw (93.center) to (96.center);
		\draw (95.center) to (97.center);
		\draw [blue] (14) to (99.center);
		\draw [magenta] (15) to (100.center);
		\draw [orange] (19) to (101.center);
		\draw (102.center) to (103.center);
		\draw [magenta, in=-90, out=90, looseness=0.75] (221.center) to (222.center);
		\draw [style=magenta, in=-90, out=120] (223.center) to (228.center);
		\draw [style=magenta, dotted] (228.center) to (226.center);
		\draw [style=blue, in=-90, out=150] (229.center) to (233.center);
		\draw [style=blue, dotted] (233.center) to (231.center);
		\draw [style=orange, in=-90, out=60] (234.center) to (238.center);
		\draw [style=orange, dotted] (238.center) to (236.center);
		\draw (224.center) to (242.center);
		\draw [style=blue] (99.center) to (246.center);
		\draw [style=blue, in=-90, out=90, looseness=0.75] (246.center) to (245.center);
		\draw [style=orange] (247.center) to (101.center);
		\draw [style=orange, in=90, out=-90, looseness=0.75] (244.center) to (247.center);
		\draw [style=magenta, in=90, out=-90] (243.center) to (248.center);
		\draw [style=magenta, in=-90, out=90, looseness=1.25] (249.center) to (248.center);
	\end{pgfonlayer}
\end{tikzpicture}}
	\caption{A routed circuit diagram for the Lugano process. We follow the same graphical conventions as in Figure \ref{fig: Lugano routed graph}. The gates at the bottom are CNOTs, controlled on the coloured wires.}
	\label{fig: Lugano routed circuit}
\end{figure} 

\changes{We proved that our routed graph is a valid one, and therefore that any routed unitary circuit built from it defines a valid process. In particular, the Lugano process (as defined e.g.\ in equation (27) of Ref.\ \cite{araujo2017quantum}) is obtained by taking all sectors in all wires to be one-dimensional, except for the $l=1$ sector of the $A \to F$ wire, the $m=1$ sector of the $B \to F$ wire, and the $n=1$ sector of the $C \to F$ wire which we take to have dimension 2; and by fleshing out the circuit as depicted in Figure \ref{fig: Lugano routed circuit}. In this figure, $V$ serves to encode a losing agent's vote in the values of the outgoing indices; for example, the $V$ above Alice's node can be written as\footnote{We denote $z$ for the value of an arbitrary basis of the $l=1$ two-dimensional sector of $A \to F$.}
\begin{equation}
\begin{split}
    &V:= \ket{l=0}\ket{i=0,l=0} \bra{l=0} \bra{0} + \ket{l=0}\ket{i=1,l=0}\bra{l=0} \bra{1} \\
    &+ \ket{l=1, z=0}\ket{i=0,l=1}\bra{l=1} \bra{0} 
    + \ket{l=1, z=1}\ket{i=0,l=1}\bra{l=1} \bra{1}
    \, ,
\end{split}
\end{equation}
$W$ sends the information about the value of its incoming indices to the Future, while also sending the information about the product of those values to the wire that loops back around to the Past:\footnote{We denote $C$ for the set of tuples $(i,j,k,l,m,n)$ satisfying the graph's global index constraints.}

\begin{equation}
    W := \sum_{(i,j,k,l,m,n) \in C} \ket{l} \ket{m} \ket{ijk} \ket{n} \bra{il} \bra{jm} \bra{kn} \, .
\end{equation}
}
Finally, $U$ simply embeds its practical input space (defined by the global index constraint) into the global Future. Its precise form is irrelevant to our concerns, so we leave it out.

This shows how a paradigmatic unitary process that violates causal inequalities can be rebuilt using our method. We emphasise, however, that the Lugano process is merely the simplest example of a process obtained from fleshing out a routed circuit of the form of Figure \ref{fig: Lugano routed graph}; one could instead take this routed circuit to feature arbitrarily large dimensions (as long as the crucial sectors we specified remain one-dimensional), and fill it up with arbitrary operations (as long as they follow the routes). In other words, we have in fact defined a large \textit{family} of processes that all rely on the same core behaviour as the Lugano process.

It is particularly worth noting that, while in the Lugano process the message sent to the winner is trivial (it is necessarily the $\ket{1}$ state), this family of consistent processes includes those where each losing agent can send arbitrarily large messages to the winner. Thus, the routed graph makes clear that the triviality of the messages in the original Lugano process is an arbitrary feature, that is not essential to its consistency.

\changes{
Readers familiar with \cite{barrett2021cyclic} might note that the routed circuit we have constructed for the Lugano process is different from the one in Figure 10 of that work. This comes as no surprise, since the most obvious routed graph from which Figure 10 of \cite{barrett2021cyclic} can be obtained is not a valid one (c.f.\ the remark at the end of Section \ref{sec: switch branch graph}).}

\section{Discussion and outlook} \label{sec: conclusion}

In this paper, we presented a method for constructing processes with indefinite causal order, based on a decorated directed graph called the routed graph. 
The method of construction ensures, and hence explains, the consistency of the process, and makes its logical structure evident. 
In particular, we proved that any process constructed from a routed graph satisfying our two principles, bi-univocality and weak loops, is necessarily consistent.

Our method can be used to construct a number of unitarily extendible processes. We explicitly constructed the quantum switch, the 3-switch, the Grenoble process, and the Lugano process.  For each of these processes, our method can also construct a large family of processes obtainable from the same routed graph.  
We expect that the other currently known examples of unitary processes that are built from classical processes analogous to Lugano can also be constructed using our method. 
Ultimately, we are led to the following conjecture.

\begin{conjecture}
Any unitary process -- and therefore any unitarily extendible process -- can be obtained by `fleshing out' a valid routed graph.
\end{conjecture}

Another fact pointing towards this conjecture is that bipartite unitarily extendible processes were recently proven \cite{barrett2021cyclic, yokojima2021consequences} to reduce to the coherent control of causal orders analogous to the quantum switch, which can therefore be written as valid routed circuits. Our conjecture can be thought of as a tentative generalisation of this result to $\geq 3$-partite processes. We expect that significant progress in this direction could be obtained if one were to prove another conjecture, that of the existence of causal decompositions of unitary channels in the general case \cite{lorenz2020}. 

This leads us to a limitation of our current results: they offer no systematic way to \textit{decompose} a known process into a routed circuit (except, in some cases, through a careful conceptual analysis of it). An important subject for future work, deeply related to the above conjecture, would be to come up with ways to supplement the bottom-up procedure presented in this paper with a top-down procedure, in which one would start with a `black-box' unknown process and extract a way of writing it as the fleshing-out of a valid routed graph.

Another limitation is that this paper had no concern for the \textit{physicality} of processes, i.e.\ for the question of whether and how they could be implemented in practice, using either standard or exotic physics. This was a conscious choice on our part, as we wanted to rather focus on the question of their logical \textit{conceivability}. However, we expect that our way of dealing with the latter question might, through the clarifications and the diagrammatic method it provides, pave the way for work on possible implementations or on physical principles constraining them.

An important consequence of our work is that it shows how at least a large class of valid quantum processes can be derived from the sole study of \textit{possibilistic} structures, encapsulated by routes. These possibilistic structures impose constraints on quantum operations, but there is nothing specifically quantum about them; they could be interpreted as constraints on classical operations as well\footnote{\changes{Note that \textit{possibilistic} structures also occur in the works of Abramsky and Brandenburger \cite{abramsky2011sheaf}, where they can be viewed as the support of \textit{probabilistic} structures concerning measurement outcomes -- see also \cite{fritz2009possibilistic}. In these works, the possibilistic structures pertain to fundamentally quantum phenomena, namely the possibilities for joint measurement outcomes of multipartite quantum states. This is in contrast to the possibilistic structures in our work, which simply denote classical possibilities in the flow of information.}}. This adds to the idea, already conveyed by the discovery of classical exotic processes, that the logical possibility for indefinite causal order does not always arise from the specifics of quantum structures. If our above conjecture turned out to be true, this would warrant this conclusion for any unitary and unitarily extendible process, whose quantum nature is nothing more than coherence between the branches of an equally admissible classical process.

By contrast, some non-unitarily extendible processes, such as the OCB process \cite{oreshkov2012quantum, araujo2017purification}, appear to feature a more starkly quantum behaviour in their display of indefinite causal order. This can be seen for example in the fact that the violation of causal inequalities by the OCB process relies on a choice between the use of maximally incompatible bases on the part of one agent. A more quantitative clue is the fact that the OCB process saturates a Tsirelson-like bound on non-causal correlations \cite{Brukner2014, Liu2024}. It is therefore unlikely that such processes could be built using our method, as routes do not capture any specifically quantum (i.e.\ linear algebraic) behaviour. In particular, the display of a unitary process with OCB-like features would probably provide a counter-example to our conjecture.

In the course of the presentation of the framework and of the main examples, we commented on the fact that the presence of (necessarily weak) loops in the branch graph were associated with the violation of causal inequalities: processes showcasing (possibly dynamical) coherent control of causal order, and therefore incapable of violating causal inequalities \cite{wechs2021quantum} -- such as the switch, the 3-switch and the Grenoble processes -- featured no such loops; while the Lugano process, which does violate causal inequalities, had loops in its branch graph. This leads us to the following conjecture.

\begin{conjecture}
The skeletal supermap corresponding to a routed graph violates causal inequalities if and only if its branch graph features (necessarily weak) loops.
\end{conjecture}

Proving this conjecture would unlock a remarkable correspondence between, on the one hand, the structural features of processes, and, on the other hand, their operational properties. An interesting question is how this would connect to (partial) characterisations of causal inequalities-violating processes via their causal structure \cite{TselentisInPrep}.

Our work facilitates a transition from a paradigm of defining processes with indefinite causal order one by one and checking their consistency by hand, to one of generating large classes of such processes from the study of elementary graphs, with their consistency baked in. In that, it follows the spirit of Ref.\ \cite{wechs2021quantum}, with more emphasis on the connectivity of processes and on the formal language with which one can describe the consistent ones. Another difference is that the framework presented here also allows us to build at least some of the unitary processes that violate causal inequalities \cite{branciard2015simplest}. \changes{It is important to note that of course, given the matrix form of a purported process in a particular basis, it may or may not be easier to verify if it is a valid process compared to being given the routed graph. However, it is the ability to check the validity for large classes of related processes at once, including processes with arbitrary large dimensions, which lies at the heart of the advantages provided by our method.}

A natural application would be to build and study new exotic processes using our framework; we leave this for future work.\footnote{This has in fact been achieved since the first appearance of this paper as a preprint: in \cite{mejdoub2025unilateraldeterminationcausalorder}, a process built using this paper's methods was displayed, featuring a novel unilateral determination of causal order.} More generally, the fact that our rules for validity only rely on the study of graphs decorated with Boolean matrices opens the way for a systematic algorithmic search for instances, using numerical methods.

A final feature of our framework is how, through the use of graphical methods and meaningful principles, it makes more intelligible, and more amenable to intuition, the reasons why a process can be both cyclic and consistent -- a notoriously obscure behaviour, especially in the case of processes violating causal inequalities. Our two rules for validity, however, are still high-level; further work is needed to investigate their structural implications. This could eventually lead to a reasoned classification of the graphs that satisfy them, and therefore of (at least a large class of) exotic processes.

\begin{acknowledgments}
It is a pleasure to thank Alastair Abbott, Giulio Chiribella, Ognyan Oreshkov, Nicola Pinzani, Eleftherios Tselentis, Julian Wechs, Matt Wilson, and Wataru Yokojima for helpful discussions and comments. AV also thanks Alexandra Elbakyan for her help to access the scientific literature. AV is supported by the EPSRC Centre for Doctoral Training in Controlled Quantum Dynamics. HK acknowledges funding from the UK Engineering and Physical Sciences Research Council (EPSRC) through grant EP/R513295/1. This publication was made possible through the
support of the grant 61466 `The Quantum Information Structure of Spacetime (QISS)' (qiss.fr) from the John Templeton
Foundation. The opinions expressed in this publication
are those of the authors and do not necessarily reflect the views of the John Templeton Foundation.
\end{acknowledgments}

\bibliographystyle{utphys}
\bibliography{refs}

\appendix
\section{The relationship between the supermap and process matrix representations} \label{app: process matrices and supermaps}

Two equivalent but distinct mathematical frameworks are in use in the indefinite causal order literature to represent higher-order processes, stemming from two independent lines of work: one is that of supermaps \cite{chiribella2008transforming,chiribella2013quantum} (also called superchannels), and the other is that of process matrices \cite{oreshkov2012quantum, araujo2017purification} (also called W-matrices). This can lead to some confusion. In this Appendix, we spell out the equivalence between the two pictures, in order to help readers more accustomed to the process matrix picture to translate our results and concepts from the supermap picture, that we use in this paper.

In broad terms, supermaps and process matrices are equivalent mathematical representations of a same higher-order process, connected by the Choi-Jamiołkowski (CJ) isomorphism. What can add to the confusion is also that they stem from different conceptual points of view on the situations being modelled, and that the equivalence between these points of view might not be obvious at first sight. We will thus start with a conceptual discussion, before spelling out the mathematical equivalence. We will then further comment on how super\textit{unitaries}, which are the focus of this paper, can be translated to super\textit{channels}: the jump is simply the standard one between the linear representation of pure quantum theory and the completely positive representation of mixed quantum theory.

\subsection{At the conceptual level}
The point of superchannels is to model higher-order transformations, mapping channels to channels, in the same way that channels map states to states. More precisely, in analogy with the fact that channels can be characterised as the only linear mappings $\cc \in \Lin \left[ \Lin(\ch_A^\inn) \to \Lin(\ch_A^\out) \right]$ that preserve all quantum states -- including quantum states on an extended system $\rho \in \Lin[\ch_A^\inn \otimes \ch_X]$ --, superchannels are characterised \cite{chiribella2008transforming} as the linear mappings 

\be \cs \in \Lin \left[ \Lin \left[ \Lin(\ch_A^\inn) \to \Lin(\ch_A^\out) \right] \to \Lin \left[ \Lin(\ch_P) \to \Lin(\ch_F) \right] \right] \ee
that preserve all quantum channels -- including quantum channels on an extended system $\cc \in \Lin \left[ \Lin(\ch_A^\inn \otimes \ch_X) \to \Lin(\ch_A^\out \otimes \ch_Y) \right]$. Moreover, \textit{multipartite} superchannels \cite{chiribella2009beyond, chiribella2013quantum} can act on pairs, or generally tuples, of channels, mapping them to one `global' channel\footnote{More precisely, bipartite superchannels were originally defined as acting on the larger space of all non-signalling channels on the tensor product of their two slots. However, it was proven at the same time that the well-defined superchannels on pairs of channels are exactly the same ones as superchannels on non-signalling channels, so we can overlook this difference.}. The conceptual idea is thus to combine `Alice's channel' and `Bob's channel' into a larger channel; it stems from an emphasis on a \textit{computational} picture, focused on the study of architectures for quantum computation.

Another line of research, developed independently, adopts an \textit{operational} picture, insisting instead on the idea of local agents performing quantum measurements -- and crucially, getting  classical outcomes \cite{oreshkov2012quantum}. Therefore, rather than on a notion of combining operations, it focuses on the task of computing joint probability distributions for these local outcomes. This is where process matrices come in naturally: taking $\cm_i \in \Lin[\Lin[\ch_A^\inn] \to \Lin[\ch_A^\out]]$ as the CP map corresponding to Alice obtaining outcome $i$, and $M_i \in \Lin[\ch_A^\inn \otimes \ch_A^\out]$ as its CJ representation (see below for its mathematical definition) -- and similarly taking $N_j$ for Bob obtaining the outcome $j$ --, one can write the joint probability compactly as

\be \label{eq: proba process matrix} \cp(i,j) = \Tr \left[ (M_i^T \otimes N_j^T) \circ W  \right] \, , \ee
where $W \in \Lin[\ch_A^\inn \otimes \ch_A^\out \otimes \ch_B^\inn \otimes \ch_B^\out] $ is the \textit{process matrix}, which one asks to yield well-defined probabilities, through (\ref{eq: proba process matrix}), for any choice of measurements on Alice and Bob's parts.

In order to allow for a notion of purification, the process matrix formalism was then extended \cite{araujo2017purification} to model general higher-order operations, with $W$ now also acting on a global past $P$ and a global future $F$, and the RHS of (\ref{eq: proba process matrix}) being taken to be a partial trace on all other systems, so that the LHS yields (the CJ representation of) a quantum operation $P \to F$. One can now see how this gets us closer, at least conceptually, to the notion of a superchannel. The original process matrices as defined in \cite{oreshkov2012quantum} can then be understood as akin to superchannels with a trivial output.

One might, however, worry that this overlooks the key conceptual difference: that process matrices relate not only channels, but also the probabilities of measurement outcomes. Shouldn't that be more general than superchannels? The key idea to understand why this worry is, in fact, unwarranted, is the fact that the obtaining of any measurement outcome can equivalently be modelled as the implementation of a deterministic channel acting on the system at hand and on an ancilla, together with a measurement outcome having been obtained by subsequently measuring the ancilla. In this way, one can recover the probabilities for measurements in the superchannel picture as well.

Before we turn to the mathematical equivalence, let us briefly comment on the respective strengths and weaknesses of the two representations. The main advantage of process matrices is the ability, in an operational picture, to compute joint probabilities in a straightforward and compact way, via the Hilbert-Schmidt inner product of (\ref{eq: proba process matrix}). However, this strength becomes a weakness once one is interested in higher-order processes with a non-trivial output: because the process matrix is then a `flattened' CJ representation of the process -- i.e., an operator acting indistinctly on all the input and output Hilbert spaces involved, e.g.\ on $A^\inn \otimes A^\out \otimes B^\inn \otimes B^\out \otimes P \otimes F$ --, it smears out the distinction i) between the inputs and outputs of local operations (e.g.\ between $A^\inn$ and $A^\out$), and ii) between the inputs and the output of the higher-order transformation (e.g.\ between $A^\inn$ and $A^\out$, which correspond to one of the input channels, and $P$ and $F$, which correspond to the output channel).

This flattening is the reason why the connectivity of a higher-order process is particularly difficult to parse in a process matrix: identities between systems, for instance, have to be represented not by identity operators but by maximally entangled states. Similarly, the tracing-out move of (\ref{eq: proba process matrix}) lends itself badly to intuition and generally messes up the input/output structure. By contrast, the superchannel's type, as a map ${\rm Chan}(A^\inn \to A^\out) \times {\rm Chan}(B^\inn \to B^\out) \to {\rm Chan}(P \to F)$, neatly encodes the different roles of the different spaces, and facilitates intuitions about the connectivity. In particular, as we are especially interested in a faithful and direct representation of the connectivity of processes, we found the superchannel picture more practical for the needs of the present paper.

\subsection{At the mathematical level}

The process matrix picture relies on the Choi-Jamiolkowski (CJ) representation of CP maps \cite{jamiolkowski1972linear, choi1975completely}, which can be defined as follows. Consider a CP map $\cm_A: A^{\rm{in}} \rightarrow A^{\rm{out}}$. We make a copy of the input system $A^{\rm{in}}$, and consider the (unnormalised) maximally entangled state $\ket{\Phi^+}=\sum_i{\ket{ii}}$ on $A^{\rm{in}} \otimes A^{\rm{in}}$. The CJ representation $M_A$ of $\cm_A$ is then the positive operator on $A^{\rm{in}} \otimes A^{\rm{out}}$ obtained by feeding one half of this entangled state into $\cm$:
\begin{equation}
    M_A := (\ci \otimes \cm_A)\ket{\Phi^+}\bra{\Phi^+}
\end{equation}

Process matrices were originally defined as operators mapping CJ representations of CP maps to probabilities via (\ref{eq: proba process matrix}).
In the bipartite case, they were therefore required to satisfy:
\begin{equation} \label{originalpms}
\begin{split}
     W \in \Lin[\ch_A^\inn \otimes \ch_A^\out \otimes \ch_B^\inn \otimes \ch_B^\out] \\
     W \geq 0 \\
     {\rm Tr} \left[(M^T  \otimes N^T) \circ W \right] =1  \ \ \ \ \ \forall M\forall N
\end{split} 
\end{equation}
where $M$ and $N$ can be CJ matrices for any pair of \textit{channels} for Alice and Bob. The positivity requirement buys us positive probabilities; the last requirement ensures that our probability distributions are normalised.

In \cite{araujo2017purification}, the definition of process matrices was extended so that they output a CJ matrix for a CP map from a `past' system $P$ to a `future' system $F$. In the bipartite case, the extended process matrix $W \in \Lin[\ch_P \otimes \ch_A^\inn \otimes \ch_A^\out \otimes \ch_B^\inn \otimes \ch_B^\out \otimes \ch_F]$ maps CJ matrices $M$ and $N$ to a CJ matrix $G := {\rm Tr}_{A^\inn A^\out B^\inn B^\out}\left[(M^T  \otimes N^T) \circ W \right] \in \Lin[\ch_P \otimes \ch_F]$. Now, rather than requiring that we map CP maps to positive and normalised probabilities, we need to require that we map CP maps to CP maps, and channels to channels. This is guaranteed by the following conditions:
\begin{equation} \label{extendedpms}
\begin{split}
     W \in \Lin[\ch_P \otimes \ch_A^\inn \otimes \ch_A^\out \otimes \ch_B^\inn \otimes \ch_B^\out \otimes \ch_F] \\
     W \geq 0 \\
     {\rm Tr}_{A^\inn A^\out B^\inn B^\out} \left[(M^T  \otimes N^T) \circ W \right] =G  \ \ \ \ \ \forall M\forall N
\end{split} 
\end{equation}
where $M$ and $N$ represent any channels for Alice and Bob, and we require that $G$ represents a channel from $P$ to $F$. The definitions for the original and the extended process matrices generalise in an obvious way to the multipartite case.

The original process matrices are special cases of these more general process matrices, in which the global `past' and `future' systems $P$ and $F$ are one-dimensional, since in this case probabilities are CP maps and the number 1 is a channel. On the other hand, when one considers any particular state-preparation at $P$, and traces out $F$, any of these new, extended process matrices gives rise to a process matrix as originally defined in (\ref{originalpms})\cite{wechs2021quantum}.

We now demonstrate the equivalence of the extended process matrices and superchannels. More precisely, we show that every extended process matrix uniquely defines a valid superchannel, and vice versa. A bipartite process matrix $W$ defines a superchannel $\cs$ in the following way:
\begin{equation}
\cs(\cm, \cn) := {\rm Choi}^{-1} \left( {\rm Tr}_{A^\inn A^\out B^\inn B^\out}\left[(M^T \otimes N^T) \circ W  \right]  \right)
\end{equation}
where $\cm: \tilde{A}^\inn \otimes A^\inn \rightarrow \tilde{A}^\out \otimes A^\out$ is Alice's channel, which also acts on ancillas, and $M$ is its CJ representation, and similarly for $\cn$ and $N$. 

One might initially worry that $\cs$ need not always be a superchannel, since the extended process matrices were only defined with respect to input CP maps without ancillas, but a superchannel must also preserve channels with ancillas. However, the positivity of $W$ ensures that $G \geq 0$ where
\begin{equation}
\begin{split}
    G:&={\rm Tr}_{A^\inn A^\out B^\inn B^\out}\left[(M^T \otimes N^T) \circ W \right] \\
    & \in \Lin[\ch_P \otimes \ch_{\tilde{A}}^\inn \otimes \ch_{\tilde{B}}^\inn  \otimes \ch_F \otimes \ch_{\tilde{A}}^\out \otimes \ch_{\tilde{B}}^\out]
\end{split}
\end{equation}
A positive CJ matrix always represents a CP map, meaning that ${\rm Choi}^{-1}(G)$ is indeed a CP map from the past and ancillary inputs, to the future and ancillary outputs. Then the last condition in (\ref{extendedpms}) ensures that $\cs$ maps channels with ancillas to other channels, meaning that $\cs$ is indeed a superchannel.

To see how a bipartite superchannel $\cs$ defines a process matrix $W$, we suppose that Alice and Bob both insert swap channels into the superchannel. The process matrix is then the CJ representation of the resulting channel:
\begin{equation}
    W := {\rm Choi}\left(\cs(\texttt{SWAP}, \texttt{SWAP})\right)
\end{equation}
The positivity of $W$ follows from the complete positivity of the channel $\cs(\texttt{SWAP}, \texttt{SWAP})$. One can show that the mapping on (CJ representations of) channels provided by $W$ is the same as the mapping provided by $\cs$
\begin{equation}
    {\rm Tr}_{A^\inn A^\out B^\inn B^\out}\left[(M^T \otimes N^T) \circ W \right]= G = {\rm Choi} (\cs(\cm, \cn))
\end{equation}
meaning the last requirement of (\ref{extendedpms}) will also be satisfied.

Rather than superchannels, this work concentrates on what we call \textit{superunitaries} -- that is, linear mappings from a set of unitary operators to another unitary operator. To connect these to process matrices, we note that every superunitary uniquely defines a `unitary superchannel' -- that is, a superchannel that always returns a unitary channel when you feed it a set of unitary channels. Given a superunitary $S$, we define a unitary superchannel in an obvious way:
\begin{equation}
    \cs(\cu, \cv) := \left[S(U, V) \right](\cdot) \left[S(U, V) \right]^\dag
\end{equation}
That is, the action of $\cs$ on the \changes{unitary} channels  $\cu:=U(\cdot)U^\dag$ and $\cv:=V(\cdot)V^\dag$ \changes{(including unitary channels that additionally act on local ancillary systems)} is just the channel corresponding to the action of $S$ on $U$ and $V$.
The action of this superchannel on more general channels can then be calculated using the Stinespring dilation. Conversely, any unitary superchannel defines a superunitary, up to an irrelevant global phase.

We have shown that i) superchannels are equivalent to process matrices and ii) superunitaries are equivalent to unitary superchannels (up to phase). Now it follows trivially that superunitaries are equivalent to unitary process matrices -- that is, the process matrices that map unitary channels (possibly with ancillas) to unitary channels. As proven in \cite{araujo2017purification}, these are precisely the process matrices that can be written in the form $W=\ket{U_W}\bra{U_W}$ where 
\begin{equation}
    \ket{U_W}:= (I \otimes U_W)\ket{\Phi^+}
\end{equation}
is a CJ \textit{vector} representing a unitary operator $U_W$.

\onecolumn

\section{Technical definitions and proofs} \label{app: Theorem}

\subsection{Notations}

In this appendix, we spell out our framework in fully technical terms, and prove its central theorem, ensuring that the routed graphs that satisfy the two principles define routed superunitaries.

To do this, and in particular to prove the theorem, we will have to deal with complicated operators, often acting on arbitrary numbers of factored spaces. We therefore introduced some notational techniques to avoid unnecessary clutter, which we will present and motivate here.

The first of these techniques is \textit{padding}. The idea is to have operators act on spaces larger than the ones they were originally defined to act on, simply by tensoring them with identity operators. For instance, we can have $f \in \cl(\ch_A)$ act on $\ch_A \otimes \ch_B$, by considering $f \otimes I_B$. However, when -- as will often be the case in what follows -- there is an arbitrarily large number of factors, and $f$ formally only acts on an arbitrary subset of them, it becomes very heavy notationally, and of limited mathematical interest, to keep track explicitly of which identity operators we should tensor $f$ with. We will thus allow ourselves to make this procedure implicit.

The idea will then be the following: for an operator $f$, its padded version $f_\pad$ will be defined as `$f$ tensored with the identity operators required to make its action meaningful, in the context of the expression at hand'. For instance, taking $f$ to act on $\ch_A \otimes \ch_B$, and $g$ to act on $\ch_B \otimes \ch_C $, $g_\pad \circ f_\pad$ will be taken to mean $(I_A \otimes g) \circ (f \otimes I_C)$, an operator acting on $\ch_A \otimes \ch_B \otimes \ch_C$. For another example, taking $h$ acting on $\ch_C \otimes \ch_D$, the equation $g_\pad \circ f_\pad = h_\pad$ will be taken to mean $(I_A \otimes g \otimes I_D) \circ (f \otimes I_C \otimes I_D) = I_A \otimes I_B \otimes I_C \otimes h$. This notation will extend to supermaps as well: for instance, if $\cs$ is a supermap of type $(\ch_A^\inn \to \ch_A^\out) \to (\ch_P \to \ch_F)$, we will define its action on a map $f: \ch_A^\inn \otimes \ch_B^\inn \to \ch_A^\out \otimes \ch_B^\out$ as $\cs_\pad[f]$, which should be understood as $(\cs \otimes \ci)[f]$, where $\ci$ is the identity supermap on $\cl(\ch_B^\inn \to \ch_B^\out)$.

Another related technique we will use in order to avoid clutter is to disregard the ordering of factors. Indeed, factors in a given tensor product are usually regarded as being labelled by ordered lists, rather than by sets, of indices. For instance, $\ch_A \otimes \ch_B$ and $\ch_B \otimes \ch_A$ are usually regarded as different (albeit isomorphic) spaces. Accordingly, suppose we take a map $h$ over $\ch_A \otimes \ch_B \otimes \ch_C$ which decomposes as a tensor product of a map $f$ on $\ch_A \otimes \ch_C$ and a map $g$ on $\ch_B$. This fact, in the usual picture, could not be expressed as $h = f \otimes g$, as the RHS there acts on $\ch_A \otimes \ch_C \otimes \ch_B$. Rather, one should write $h = (I \otimes \texttt{swap}_{C,B}) \circ (f \otimes g) \circ (I \otimes \texttt{swap}_{B,C})$. Another feature of the standard view is that it is not possible to write $\bigotimes_{X \in \{A, B\}} \ch_X$, as this expression would leave the order of the factors ambiguous.

For the expressions we will consider, keeping with this use would force us to 1) explicitly introduce arbitrary orderings of all the sets of indices we use to label factors in tensor products, and 2) overload our expressions with swaps, in order to always place next to each other the spaces on which a given map acts. This would once again create a lot of clutter with little relevance. We will therefore abstain ourselves from that constraint, and take the view that tensor products are labelled with sets, rather than ordered lists, of indices. This will allow us to write $h = f \otimes g$ in the case described above, or to write Hilbert spaces of the form $\bigotimes_{X \in \{A, B\}} \ch_X$. The expressions we will write in this way could always be recast in the standard view, using arbitrary orderings of the sets at hand and large amounts of swaps.

One might wonder whether either padding or disregarding the ordering of factors might lead to ambiguities. In fact, such ambiguities only arise  if some Hilbert spaces in a tensor product are labelled with the same index, for instance if one is dealing with a Hilbert space like $\ch_A \otimes \ch_B \otimes \ch_A$. This is why we will carefully avoid such situations, by only ever tensoring different -- although possibly isomorphic -- Hilbert spaces. For instance, if we need a tensor product of $\ch_A$ with itself, we will write it as $\ch_A \otimes \ch_{A'}$, with $\ch_A \cong \ch_{A'}$.

Note that the same techniques and notations will be used for relations.

\subsection{Technical definitions on supermaps}

\begin{definition}[Superrelation]
A \emph{superrelation} of type $\bigtimes_{N} (K_N \to L_N) \to (P \to F)$, where $P$, $F$, the $K_N$'s and the $L_N$'s are all sets, is a relation $\cs^\Rel:  \Rel(\bigtimes_N K_N, \bigtimes_{N'} L_{N'}) \to \Rel(P,F)$. With a slight abuse of notation, we denote $\cs^\Rel[(\laN)_N]:= \cs^\Rel[\bigotimes_N \laN]$.
\end{definition}

\begin{definition}[Supermap]
A \emph{supermap} of type $\bigtimes_{N} (\ch^\inn_N \to \ch^\out_N) \to (\ch_P \to \ch_F)$, where $\ch_P$, $\ch_F$, the $\ch^\inn_N$'s and the $\ch^\out_N$'s are all finite-dimensional Hilbert spaces, is a linear map $\cs:  \cl(\bigotimes_N \ch^\inn_N, \bigotimes_{N'} \ch^\out_{N'}) \to \cl(\ch_P,\ch_F)$. With a slight abuse of notation, we denote $\cs[(f_N)_N]:= \cs[\bigotimes_N f_N]$.
\end{definition}

\begin{definition}[Superunitary]
A \emph{superunitary} of type $\bigtimes_{N} (\ch^\inn_N \to \ch^\out_N) \to (\ch_P \to \ch_F)$ is a supermap of the same type such that, for any choice of ancillary input and output spaces $\ch^{\inn,\anc}_N$ and $\ch^{\out,\anc}_N$ at every $N$, and any choice of unitary maps $U_N : \ch^\inn_N \otimes \ch^{\inn,\anc}_N \to \ch^\out_N \otimes \ch^{\out,\anc}_N$ at every $N$, one has:

\be \label{eq:superunitary} \cs_\pad[(U_N)_N] \textrm{ is a unitary from } \ch_P \otimes \left(\bigotimes_N \ch^{\inn,\anc}_N \right) \textrm{ to } \ch_F \otimes \left(\bigotimes_N \ch^{\out,\anc}_N \right) \, . \ee
\end{definition}

\begin{definition}[Routed Supermap]
A \emph{routed supermap} of type $\bigtimes_{N} (\ch^\inn_N \overset{\laN}{\to} \ch^\out_N) \to (\ch_P \overset{\mu}{\to} \ch_F)$, where the $\ch^\inn_N$'s and the $\ch^\out_N$'s are sectorised finite-dimensional Hilbert spaces and the $\laN$'s are relations $\Indin_N \to \Indout_N$, is a supermap: i) which is restricted to only act on the maps of $\cl(\bigotimes_N \ch^\inn_N, \bigotimes_{N'} \ch^\out_{N'})$ that follow the route $\bigtimes_N \lambda_N$; and ii) whose output always follows the route $\mu$.

We say it is superunitary if it satisfies (\ref{eq:superunitary}) when acting on routed unitaries $U_N : \ch^\inn_N \otimes \ch^{\inn,\anc}_N \overset{\laN}{\to} \ch^\out_N \otimes \ch^{\out,\anc}_N$ that follow the routes.
\end{definition}

\subsection{Technical presentation of the framework}

\begin{definition}[Indexed graph]
An \emph{indexed graph} $\Ga$ consists of

\begin{itemize}
    \item a finite set of nodes (or vertices) $\Nodes$;
    \item a finite set of arrows (or edges) $\Arr = \Arr^\inn \sqcup \Arr^\int \sqcup \Arr^\out $;
    \item functions $\head : \Arr^\inn \sqcup \Arr^\int \to \Nodes$ and $\tail : \Arr^\int \sqcup \Arr^\out \to \Nodes$;
    \item for each arrow $A \in \Arr$, a finite set of indices $\Ind_A$, satisfying: $A \not\in \Arr^\int \implies \Ind_A$ is trivial (i.e. is a singleton);
    \item a function $\dim: \bigsqcup_{A \in \Arr} \Ind_A \to \mathbb{N}^*$.
\end{itemize}
\end{definition}

We further define $\Indin_\Ga := \bigtimes_{A \in \Arr^\inn} \Ind_A$, $\Indout_\Ga := \bigtimes_{A \in \Arr^\out} \Ind_A$, and for any $N \in \Nodes$: $\inn(N) := \head\inv (N)$, $\out(N) := \tail\inv (N)$, $\Indin_N := \bigtimes_{A \in \inn(N)} \Ind_A$ and $\Indout_N := \bigtimes_{A \in \out(N)} \Ind_A$.

To prepare for the interpretation of this graph in terms of complex linear maps, we also define the following sectorised Hilbert spaces: for all $A \in \Arr$, $\ch_A := \bigoplus_{k \in \Ind_A} \ch_A^k$, where $\ch_A^k \cong \mathbb{C}^{\dim(k)}$; $\ch_P := \bigotimes_{A \in \Arr^\inn} \ch_A = \bigoplus_{\Vec{k} \in \Indin_\Ga} \bigotimes \ch_A^{k_A}$ and $\ch_F := \bigotimes_{A \in \Arr^\out} \ch_A = \bigoplus_{\Vec{k} \in \Indout_\Ga} \bigotimes \ch_A^{k_A}$; and for all $N \in \Nodes$, $\chinN := \bigotimes_{A \in \inn(N)} \ch_A = \bigoplus_{\Vec{k} \in \Indin_N} \bigotimes \ch_A^{k_A}$ and $\choutN := \bigotimes_{A \in \out(N)} \ch_A = \bigoplus_{\Vec{k} \in \Indout_N} \bigotimes \ch_A^{k_A}$.

\begin{definition}[Branched relation]
A relation $\la: K \to L$ is said to be \emph{branched} if, when seen as a function $K \to \cp(L)$, it satisfies

\be \forall k,k' \in K, \la(k) \cap \la(k') = \la(k) \textrm{ or } \emptyset \, , \ee

i.e. $\la(k)$ and $\la(k')$ are disjoint or the same.
\end{definition}

Note that $\la$ is branched if and only if $\la^\top$ is branched. Branched relations define compatible, non-complete partitions of their domain and codomain, corresponding to \textit{branches}. Formally, a branch $\al$ of the branched relation $\la: K \to L$ is a pair of non-empty sets $K^\al \subseteq K$ and $L^\al \subseteq L$ such that $\la(K^\al) = L^\al$ and $\la^\top(L^\al) = K^\al$. We denote the set of branches of $\la$ as $\Bran(\la)$. Note that the partitions are not complete, i.e. $\bigsqcup_{\al \in\Bran(\la)} K^\al$ might not be equal to $K$ (and the same goes for the outputs); the discrepancy corresponds to the indices that are sent by $\la$ to the empty set, as we consider these indices to be part of no branch at all.

\begin{definition}[Routed graph]
A \emph{routed graph} $(\Ga, (\laN)_{N \in \Nodes})$ consists of an indexed graph $\Ga$ and, for every node $N$, of a branched relation $\la : \Indin_N \to \Indout_N$, called the route for node $N$.
\end{definition}

We will write routed graphs as $(\Ga, (\laN)_{N} )$ for brevity. We denote elements of $\Bran(\laN)$ as $\Nal$, and denote the set of input (resp. output) indices of $\Nal$ as $\IndinNal \subseteq \Indin_N$ (resp. $\IndoutNal \subseteq \Indout_N$). We also define $\Bran_{(\Ga,(\laN)_N)} := \bigsqcup_{N \in \Nodes} \Bran(\laN)$, the set of all branches in the whole routed graph.

We will now define the notion of a branch being a \textit{strong parent} of another: this will correspond to solid arrows in the branch graph. First, we introduce the set of possible tuple of values of indices, in order to exclude inconsistent assignments of values.

\begin{definition}
We define $\PossVal$ as the subset of $\bigtimes_{A \in \Arr} \Ind_A$ defined by

\be \forall (k_A)_{A \in \Arr} \in \PossVal, \, \forall N \in \Nodes, \,\,  (k_A)_{A \in \inn(N)} \overset{\laN}{\sim} (k_A)_{A \in \out(N)} \,.\ee
\end{definition}

A tuple of values is possible if and only if for every node, it yields a input and an output values that are in the same branch.

\begin{lemma} \label{lem: mu}
Let $\veck = (k_A)_{A \in \Arr} \in \bigtimes_{A \in \Arr} \Ind_A$. $\veck \in \PossVal$ if and only if it meets the following two conditions:

\begin{itemize}
    \item $\forall N \in \Nodes$, $(k_A)_{A \in \inn(N)}$ is  in $\laN$'s practical inputs and $(k_A)_{A \in \out(N)}$ is in $\laN$'s practical outputs;
    \item denoting, for every $N$, $\mu_N^\inn(\veck)$ as the element of $\Bran_N$ such that $(k_A)_{A \in \inn(N)} \in \Indin_{N^{\mu_N^\inn(\veck)}}$, and $\mu_N^\out(\veck)$ similarly, we have: $\forall N, \mu_N^\inn(\veck) = \mu_N^\out(\veck)$.
\end{itemize}
\end{lemma}

\begin{proof}
This derives directly from the structure of the routes: a branched route relates an input value to an output value if and only if they are both in the same branch (and in particular, are not outside of its practical in/outputs).
\end{proof}

For $\veck \in \PossVal$, we can therefore denote for every $N$ the branch $\muNveck$ equal to both $\mu^\inn_N(\veck)$ and $\mu^\out_N(\veck)$.

\begin{definition}[Strong parents]
Let $(\Ga, (\laN)_{N \in \Nodes})$ be a routed graph, and $\Nal$ and $\Mbe$ two branches in it. We define the set of arrows from $N$ to $M$ as $\Link(N,M):= \out(N) \cap \inn(M)$. We define $\LinkVal(\Nal,\Mbe)$, the set of values linking $\Nal$ to $\Mbe$, as

\be \LinkVal(\Nal,\Mbe) := \left\{ (k_A)_{A \in \Link(N,M)} | \exists (k_A)_{A \in \Arr \setminus \Link(N,M)} \textrm{ such that } \begin{cases} \mu_N\left((k_A)_{A \in \Arr} \right) = \al \\
 \mu_M\left((k_A)_{A \in \Arr} \right) = \bet \end{cases} \right\}\ee

We say that the branch $\Nal$ is \emph{not a strong parent} of the branch $\Mbe$ if at least one of the following holds:

\begin{itemize}
\item  $\Link(N,M) = \emptyset$;
\item $\LinkVal(\Nal,\Mbe) = \emptyset$;
\item $\LinkVal(\Nal,\Mbe)$ is a singleton and its unique element $(k_A)_{A \in \Link(N,M)}$ satisfies $\forall A \in \Link(N,M), \, \dim(k_A) =1 $.
\end{itemize}
\end{definition}

\begin{definition}[Adjoint of a graph] \label{def:adjoint of a graph}
If $\Ga$ is an indexed graph, its \emph{adjoint} $\Gatop$ is the indexed graph given by swapping the roles of $\Arr^\inn$ and $\Arr^\out$ and those of $\head$ and $\tail$, and leaving the rest invariant.

If $\GalaN$ is a routed graph, its adjoint is the routed graph $(\Gatop, (\laN^\top)_N)$.
\end{definition}

\begin{definition}[Skeletal superrelation of an indexed graph]
Given an indexed graph $\Ga$, its associated \emph{skeletal superrelation} is the superrelation $\SRelGa: \bigtimes_{N} (\Indin_N \to \Indout_N) \to (\Indin_\Ga \to \Indout_\Ga)$ defined by

\be \label{eq:skeletal superrelation}\SRelGa[(\laN)_N] := \Tr_{\Ind_A, A \in \Arr^\int} \left[ \bigotimes_N \laN \right] \, .\ee

Note that this is well-typed because $\bigtimes_N \Indin_N = \bigtimes_{A \in \Arr^\inn \sqcup \Arr^\int} \Ind_A$ and $\bigtimes_N \Indout_N = \bigtimes_{A \in \Arr^\out \sqcup \Arr^\int} \Ind_A$.
\end{definition}

\begin{definition}[Skeletal supermap of a routed graph]
Given a routed graph $\GalaN$, its associated (routed) \emph{skeletal supermap} is the supermap $\SGalaN$ of type $\bigtimes_{N} (\ch^\inn_N \overset{\laN}{\to} \ch^\out_N) \to (\ch_P \overset{\SRelGa[(\laN)_N]}{\to} \ch_F)$ defined by

\be \label{eq:skeletal supermap}\SGalaN[(f_N)_N] := \Tr_{\ch_A, A \in \Arr^\int} \left[ \bigotimes_N f_N \right] \, .\ee

Note that the fact that $\SGalaN[(f_N)_N]$ follows the route $\SRelGa[(\laN)_N]$ when the $f_N$'s follow the $\laN$'s is ensured by the fact that routed maps form a compact closed category \cite{vanrietvelde2021routed,wilson2021composable}.

\end{definition}

\begin{definition}[Augmented relation]
Given a relation $\laN : \Indin_N \to \Indout_N$ serving as a route for node $N$, its \emph{augmented} version is the partial function (encoded as a relation) $\laNaug : \Indin_N \times \left( \bigtimes_{\al \in \Bran(\laN)} \IndoutNal \right) \to \Indout_N \times \left( \bigtimes_{\al \in \Bran(\laN)} \HappNal \right)$ -- where $\forall \al, \HappNal \cong \{0,1\}$ -- given by

\be \label{eq:augmented relation} \laNaug(k, (l^\al)_{\al \in \Bran(\laN)}) = \begin{cases} \{ (l^\al, (\delta^\al_{\al'})_{\al \in \Bran(\laN)} ) \} &\textrm{ if } k \in \IndinNal \\ \emptyset &\textrm{ if } \forall \al, k \not\in \IndinNal \,. \end{cases}\ee

\end{definition}

\begin{definition}[Univocality] \label{def:univocality}
A routed graph $\GalaN$ is \emph{univocal} if

\be \label{eq:univocality} \SRelGapad \left[{(\laNaug)}_N \right] \textrm{ is a function.} \ee

We then note this function as $\Lambda_\GalaN$.

$\GalaN$ is \emph{bi-univocal} if both it and its adjoint $\left( \Gatop, (\laN^\top)_N \right)$ are univocal.
\end{definition}

\begin{definition}[Branch graph]
If $\GalaN$ is a bi-univocal routed graph, its \emph{branch graph} $\GaBran$ is the graph in which

\begin{itemize}
    \item the nodes are the branches of $\GalaN$, i.e. the elements of $\Bran_\GalaN$;
    \item there is a green dashed arrow from $\Nal$ to $\Mbe$ if $\Lambda_\GalaN$ features influence from $\IndoutNal$ to $\HappMbe$;
    \item there is a red dashed arrow from $\Nal$ to $\Mbe$ if $\Lambda_{\left( \Gatop, (\laN^\top)_N \right)}$ features influence from $\IndinMbe$ to $\HappNal$;
    \item there is a solid arrow from $\Nal$ to $\Mbe$ if $\Nal$ is a strong parent of $\Mbe$.
\end{itemize}
\end{definition}

\begin{definition}[Weak loops]
Let $\GalaN$ be a bi-univocal routed graph. We say that a loop in $\GaBran$ is \emph{weak} if it only contains green dashed arrows, or if it only contains red dashed arrows.
\end{definition}

\begin{theorem}[Main theorem]\label{thm:Main Theorem}
Let $\GalaN$ be a routed graph which is bi-univocal, and whose branch graph $\GaBran$ only features weak loops. Then its associated skeletal supermap $\SGalaN$ is a superunitary.
\end{theorem}

The rest of this Appendix is dedicated to the proof of this theorem.

\subsection{Proof}

\subsubsection{Preliminary lemmas and definitions} \label{sec: preliminary}

\begin{lemma}\label{lem: simplification}
To prove Theorem \ref{thm:Main Theorem}, it is sufficient to prove that, for any valid routed graph $\GalaN$, $\SGalaN$ preserves unitarity when acting on input operations without ancillas.
\end{lemma}

\begin{proof}
Suppose it was proven that for any valid $\GalaN$, and for any set of routed unitaries $U_N : \ch^\inn_N \overset{\laN}{\to} \ch^\out_N$, $\cs_{\GalaN}[(U_N)_N]$ is a unitary.

Taking now a valid $\GalaN$ and, for every $N$, a choice of ancillary input and output spaces $\ch_N^{\inn, \anc}$ and $\ch_N^{\out, \anc}$, and a routed unitary $U_N : \ch^\inn_N \otimes \ch^{\inn,\anc}_N \overset{\laN}{\to} \ch^\out_N \otimes \ch^{\out,\anc}_N$. One can then define a new indexed graph $\tilde{\Gamma}$ by adding, for each $N$, a new arrow in $\Arr^\inn$, with Hilbert space $\ch^{\inn,\anc}_N$, and a new arrow in $\Arr^\out$, with Hilbert space $\ch^{\out,\anc}_N$. The routed graph $(\tilde{\Gamma}, (\laN)_N)$ then has the same choice relation and the same branch graph as $\GalaN$; it is therefore valid as well. We can thus apply our assumption to it, which entails that $\cs_{(\tilde{\Gamma}, (\laN)_N)}[(U_N)_N] = \cs^\pad_\GalaN[(U_N)_N]$ is unitary. This thus proves the theorem in the general case.
\end{proof}

From now on, we will therefore work with a fixed routed graph $\GalaN$ (which we will often denote as $\Ga$ for simplicity) satisfying bi-univocality and weak loops, and a fixed collection of routed unitary maps $U_N : \chinN \overset{\laN}{\to} \choutN$ following the $\laN$'s. Writing $\cs:= \SGa$ for simplicity, our goal is to prove that $\cs[(U_N)_N]: \ch_P \to \ch_F$ is a unitary.

For each branch $\Nal$, we define $\chinNal := \bigoplus_{(k_A)_{A \in \inn(N)} \in \IndinNal} \bigotimes_{A \in \inn(N)} \ch_A^{k_A} \subseteq \chinN$, $\choutNal := \bigoplus_{(k_A)_{A \in \out(N)} \in \IndoutNal} \bigotimes_{A \in \out(N)} \ch_A^{k_A} \subseteq \choutN$. We also define the projection $p_N^\al: \chinN \to \chinNal$ and the injection $i_N^\al : \choutNal \to \choutN$.

We define the \textit{exchange} gate for $N$, $\ex_N : \chinN \otimes \left(\bigotimes_{\al \in \Bran(\laN)} \choutNal \right) \to \chout \otimes \left( \bigotimes_{\al \in \Bran(\laN)} \chinNal \right)$, by

\be \label{eq: def exchange} \ex_N := \sum_{\al \in \Bran(\laN)} \iNalpad \circ \left( \swap_{\Nalin, \Nalout} \otimes (\bigotimes_{\bet \neq \al} \Th_\Nbe) \right) \circ \pNalpad \, , \ee

where $\forall \Nbe$, $\Th_\Nbe$ is an arbitrarily chosen unitary from $\chout_\Nbe$ to $\chin_\Nbe$. Note that $\ex_N$ follows $\laN$ by construction.

We note that the fact that the $U_N$'s follow the $\laN$'s entails that one can find a block decomposition for them, i.e., one can define unitaries $\UNal: \chinNal \to \choutNal$ such that

\be \forall N, U_N = \sum_{\al \in \Bran(\laN)} \iNalpad \circ\UNal \circ \pNalpad \, . \ee

As a first preliminary to the proof, we will study in detail how bifurcation choices are in correspondence with assignments of values to the arrow's indices.

The following definition and lemma prove two things. First, univocality implies that any tuple of bifurcation choices fixes not only the branch at every node, but also the specific index values picked in that branch. Second, for a fixed tuple of bifurcation choices, the bifurcation choices at the branches not happening have no effect -- i.e.\ modifying them to any other value wouldn't affect the any of the index values in the graph; while, on the contrary, modifying the bifurcation choice at any of the branches happening always changes at least one of the index values in the graph. In that sense, any tuple of values of the graph's indices corresponds either to no tuple of bifurcation choices at all, or to exactly one bifurcation choice at the branches that happen for this tuple of values, with no dependence on the bifurcation choices at branches that don't happen.

\begin{definition}
For every $N$ in $\Nodes$, we take $\Indout_N{}' \cong \Indout_N$ and define the partial function (encoded as a relation) $\laNsec : \Indin_N \times \left( \bigtimes_{\al \in \Bran_N} \IndoutNal \right) \to \Indout_N \times \Indout_N{}'$ given by

\be \label{eq:sec relation} \laNsec(k, (l^\al)_{\al \in \Bran_N}) = \begin{cases} \{ (l^\al, l^\al) \} &\textrm{ if } k \in \IndinNal \\ \emptyset &\textrm{ if } \forall \al, \, k \not\in \IndinNal \,. \end{cases}\ee
\end{definition}

\begin{lemma} \label{lem: lambdasec}
If $\GalaN$ is univocal, then $\cs^\Rel \left[{(\laNsec)}_N \right]$ is an injective function $\bigtimes_{\Nal \in \Bran_\Ga}\IndoutNal \to \bigtimes_{N} \Indout_N{}'$, which we denote $\Lambda^\sec$. Furthermore, its preimage sets are given by



\be \label{eq: reverse lambdasec} \forall (k_N)_{N}, \, \left( \Lambda^\sec \right) \inv \left((k_N)_{N}\right) = \textrm{ either } \emptyset \textrm{ or } \left( \bigtimes_N \{k_N \} \right) \times \left( \bigtimes_{N^\al | k_N \not\in \IndoutNal} \IndoutNal \right)\, . \ee
\end{lemma}

\begin{proof}


We will use bra-ket notations for relational states and effects. For every branch $\Nal$, we define $\copyy_\Nalout: \IndoutNal \to \IndoutNal \times \IndoutNal{}'$, with $\IndoutNal{}' \cong \IndoutNal$, by $\copyy_\Nalout \ket{l} = \ket{l} \otimes \ket{l}$. For every node $N$, we define the partial function (encoded as a relation) $\sigma^N: \bigtimes_{\al \in \Bran(\laN)} \HappNal \to \Bran_N$  by 

\be \label{eq: sigma def} \sigma^N((\varepsilon^{\al'})_{\al' \in \Bran_N}) = \begin{cases} \{ \al \} &\textrm{ if } \forall \al', \varepsilon^{\al'} = \delta^{\al'}_{\al} \\ \emptyset &\textrm{ otherwise.} \end{cases}\ee

For every node $N$, we define the function $\nu^N: \Bran_N \times \left( \bigtimes_{\al \in \Bran_N} \IndoutNal{}' \right) \to \Indout_N{}'$, with $\Indout_N{}' \cong \Indout_N$, by 

\be \nu^N(\Vec{l}, \al) = l^\al \, . \ee

One can then compute that $\laNsec = \nu^N_\pad \circ \sigma^N_\pad \circ \la^\aug_{N, \pad} \circ \copyy_{\Nalout, \pad}$; we can thus reexpress $\Lambda^\sec$ in terms of the choice function $\Lambda$,

\be \Lambda^\sec = \cs^\Rel\left[{(\laNsec)}_N \right] = \left( \prod_N \nu^N_\pad \right) \circ \left( \prod_N \sigma^N_\pad \right) \circ \Lambda \circ \left( \prod_{\Nal} \copyy_{\Nalout, \pad} \right) \, . \ee

Given that the outputs of a $\la^\aug_N$ are within the domain of definition of the corresponding $\sigma^N$, the fact that $\Lambda$ is a function thus implies that $\left( \prod_N \sigma^N_\pad \right) \circ \Lambda$ is a function as well.
Given that the $\copyy_{\Nalout}$'s and $\nu^N$'s are functions, $\La^\sec$ is a function as well
.

Furthermore, let us fix an $N$ and $k_N \in \Indout_N$. If $k_N$ is outside of $\laN$'s practical outputs, it immediately has no preimage through $\nu^N$. Taking the other case, we denote $\al$ as the branch of $N$ such that $k_N \in \IndoutNal$. Then,

\be \begin{split}
    &\bra{k_N}_{\IndoutNal{}', \pad} \circ \laNsec \\
    &= \bra{k_N}_{\IndoutNal{}', \pad} \circ \nu^N_\pad \circ \sigma^N_\pad \circ \la^\aug_{N, \pad} \circ \left( \bigotimes_{\al' \in \Bran_N} \copyy_{N^{\al'}_\out} \right)_\pad\\
    &= \bra{k_N}_{\IndoutNal{}', \pad} \circ \left( \bigotimes_{\al' \in \Bran_N \setminus \{\al\}} \bra{\Indout_{N^{\al'}}} \right)_\pad \circ \sigma^N_\pad \circ \la^\aug_{N, \pad} \circ \left( \bigotimes_{\al' \in \Bran_N} \copyy_{N^{\al'}_\out} \right)_\pad\\
    &= \bra{k_N}_{\IndoutNal{}', \pad} \circ \sigma^N_\pad \circ \la^\aug_{N, \pad} \circ \left( \, \ketbra{k_N}{k_N} \, \right)_{\IndoutNal, \pad} \\
    &= \left( \bigotimes_{\al' \in \Bran_N} \bra{\delta^{\al'}_{\al}}_{\Happ_{N^{\al'}}} \right)_\pad \circ \la^\aug_{N, \pad} \circ \left( \, \ketbra{k_N}{k_N} \, \right)_{\IndoutNal, \pad} \\
    &= \ket{k_N}_{\Indout_N} \, \left(  \bra{\IndinNal}_{\Indin_N} \otimes \bra{k_N}_{\IndoutNal} \otimes \left( \bigotimes_{\al' \in \Bran_N \setminus \{\al\}} \bra{\Indout_{N^{\al'}}} \right) \right) \, .
\end{split}
\ee

Therefore, we find that  $\left( \Lambda^\sec \right) \inv ((k_N)_N)$ is empty if at least one of the $k_N$'s is outside of the practical outputs of the corresponding $\laN$'s, and that otherwise -- denoting, for every $N$, $\al(k_N)$ as the branch such that $k_N \in \Indout_{N^{\al(k_N)}}$--,

\be \begin{split}
    &\left( \bigotimes_{N \in \Nodes} \bra{k_N}_{\Indout_N{}'} \right) \circ \Lambda^\sec \\
    &= \left( \bigotimes_{N \in \Nodes} \bra{k_N}_{\Indout_N{}'} \right) \circ \cs^\Rel \left[{(\laNsec)}_N \right] \\
    &= \cs^\Rel \left[ \left(\ket{k_N}_{\Indout_N} \bra{\Indin_{N^{\al(k_N)}}}_{\Indin_N} \right)_N \right] \otimes \left( \bigotimes_{N \in \Nodes} \bra{k_N}_{\Indout_{N^{\al(k_N)}}} \otimes \left( \bigotimes_{\al' \in \Bran_N \setminus \{\al(k_N)\}} \bra{\Indout_{N^{\al'}}} \right)\right)  \\
\end{split}
\ee

$\cs^\Rel \left[ \left(\ket{k_N}_{\Indout_N} \bra{\Indin_{N^{\al(k_N)}}}_{\Indin_N} \right)_N \right]$ is just a scalar in the theory of relations, i.e.\ $0$ or $1$; $\left( \Lambda^\sec \right) \inv ((k_N)_N)$ is thus non-empty if and only if this scalar is equal to $1$, and the rest of the expression yields (\ref{eq: reverse lambdasec}). This also shows that $\La^\sec_\Ga$ is injective.

\end{proof}

Note that we defined $\Lambda^\sec_\Ga$ as having codomain $\bigtimes_{N \in \Nodes} \Indout_N$; but, given that for each $N$ we have $\Indout_N = \bigtimes_{A \in \out(N)} \Ind_A$, we can also see it as a function to $\bigtimes_{A \in \Arr} \Ind_A$ (we neglect the discrepancy due to global input arrows of the graph, as their sets of index values are trivial). $\Lambda^\sec_\Ga$ can thus be interpreted as telling us how bifurcation choices fix all indices in the graph. $\Lambda^\sec_{\Gatop}$, obtained from considering the adjoint graph, tells us the same about reverse bifurcation choices.

From that perspective, in the above Lemma, the case of an empty set of preimages corresponds exactly to impossible assignments of values to the arrows, i.e.\ to ones that are outside of $\PossVal$.

\begin{lemma} \label{lem: PossVal}
Given $\veck = (k_A)_{A \in \Arr}$, $\left( \Lambda^\sec \right) \inv (\veck)$ is empty if and only if $\veck \not\in \PossVal$. 
\end{lemma}

\begin{proof}
First, if there exists an $N$ such that $k_N = (k_A)_{A \in \out(N)}$ is outside $\laN$'s practical outputs, then $\left( \Lambda^\sec \right) \inv (\veck)$ is empty (as pointed out in the previous proof), and $\veck \not\in \PossVal$ (as pointed out in Lemma \ref{lem: mu}).

Otherwise, we know from the previous proof that $\left( \Lambda^\sec \right) \inv (\veck)$ is not empty if and only if $\cs^\Rel \left[ \left(\ket{k_N}_{\Indout_N} \bra{\Indin_{N^{\al(k_N)}}}_{\Indin_N} \right)_N \right] = 1$. But given how $\cs^\Rel$ was defined in (\ref{eq:skeletal superrelation}), and the form of the $\laN$'s, this is the case if and only if for all $N$, $(k_A)_{A \in \inn(N)}$ is in the branch $\al(k_N)$. As the function $\al_N$ is precisely the function $\mu^\out_N$ defined in Lemma \ref{lem: mu}, we thus find the condition $\mu^\out_N(\veck) = \mu^\inn_N(\veck)$ showed in this Lemma to be necessary and sufficient for $\veck \in \PossVal$. 
\end{proof}

Finally, we draw the consequences of the fact that branches satisfy the weak loops condition. Given a branch $\Nal$, we define the following subsets of $\BranGa$. By a `path' in $\GaBran$, we mean any sequence of arrows, without a distinction between the solid, green dashed or red dashed types.
\begin{itemize}
    \item $\cp(\Nal) := \{\Oga \neq \Nal \, | \, \exists \textrm{ a path } \Oga \to \Nal \textrm{ in } \GaBran \}$, $\Nal$'s past;
    \item $\cf(\Nal) := \{\Oga \neq \Nal \, | \, \exists \textrm{ a path } \Nal \to \Oga \textrm{ in } \GaBran \}$, $\Nal$'s future;
    \item $\cl(\Nal) := \cp(\Nal) \cap \cf(\Nal)$, $\Nal$'s layer (i.e. the branches that form a loop with $\Nal$);
    \item $\cpst(\Nal) := \cp(\Nal) \setminus \cl(\Nal)$, $\Nal$'s  strict past;
    \item $\cfst(\Nal) := \cf(\Nal) \setminus \cl(\Nal)$, $\Nal$'s strict future;
\end{itemize}

It is clear that the relation $\sim$, defined by: $\Nal \sim \Oga$ if $\Nal = \Oga$ or $\Oga \in \cl(\Nal)$, is an equivalence relation on $\BranGa$, partitioning it into a collection of layers. The fact that all loops in $\BranGa$ are weak then allows us to say that a given layer either only contains green dashed arrows between its elements (in which case we will call it a green layer), or only contains red dashed arrows (in which case we will call it a red layer)\footnote{Note that single-branch layers can be considered to be either green or red: the choice will not affect the proof.}.

Furthermore, merging the nodes of each layer transforms $\GaBran$ into an acyclic graph. One can thus define a partial order between layers. Arbitrarily turning it into a total order, and picking arbitrary orderings within each layer, leads to a total ordering $<$ of $\BranGa$ in which branches of a same layer are all next to each other, and in which $\Nal < \Oga \implies \Nal \not\in \cfst(\Oga)$. We can use this total ordering to label the branches with natural numbers, as $\BranGa = \{ B(i) \, | \, 1 \leq i \leq n \}$. For a given $i$ and a given branch $\Nal > B(i)$, we define $\cpi(\Nal) := \cp(\Nal) \cap \{\Oga > B(i) \}$, $\cfi(\Nal) := \cf(\Nal) \cap \{\Oga > B(i) \}$, etc.

\subsubsection{The induction hypothesis}
This ordering of $\Ga$'s branches will allows us to define an induction. The point is to start from $\Spad[(\ex_N)_N]$, and then to `refill' the branches one by one, making sure that the unitary obtained at each step is sufficiently well-behaved to be able to move to the next step. To define it, we will first need to define these `partially filled exchanges' that are being used at every step $i$ in the induction, which we shall call $V_{N,i}$'s. We do so by defining how they act on each branch: i.e., $\forall i, \forall \Nal$, we define $V_{N,i}^\al : \chinN \otimes \left( \bigotimes_{\bet | \Nbe > B(i)} \chout_\Nbe \right) \to \choutN \otimes \left( \bigotimes_{\bet | \Nbe > B(i)} \chin_\Nbe \right)$ by

\be  V_{N,i}^\al = \begin{cases}
\iNalpad \circ \left( \swap_{\Nalin, \Nalout} \otimes (\bigotimes_{\bet> B(i), \bet \neq \al} \Th_\Nbe ) \right) \circ \pNalpad &\textrm{ if } \Nal > B(i) \, ,\\
 \left( i_N^\al \circ \UNal \circ p_N^\al \right) \otimes\left( \bigotimes_{\bet> B(i)} \Th_\Nbe \right)  &\textrm{ if } \Nal \leq B(i) \,,
\end{cases}\ee

and we use them to define

\be  V_{N,i} := \sum_{\al \in \Bran(\laN)}  V_{N,i}^\al \, .\ee

We will write the input (resp. output) space of $\Spad[(V_{N,i})_N]$ as $\ch_i^\out := \ch_P \otimes \left( \bigotimes_{\Nbe > B(i)} \ch^\out_{\Nbe} \right)$ (resp. $\ch_i^\inn:= \ch_F \otimes \left( \bigotimes_{\Nbe > B(i)} \ch^\inn_{\Nbe} \right)$). We also write $\Bar{V}_{N,i}^\al:= V_{N,i} - V_{N,i}^\al$. Note that the $V_{N,i}$'s follow the $\laN$'s by construction, and that one has $V_{N,0} = \ex_N$ and $V_{N,n} = U_N$.

The core of the induction will be the hypothesis that, at step $i$, $\Spad[(V_{N,i})_N]$ is unitary. However, this will not be sufficient: we will also need other conditions ensuring that $\Spad[(V_{N,i})_N]$ features structural properties which allow us to move to step $i+1$. More precisely, these conditions will encode the fact that at every step $i$, and for every branch $\Nal$ that hasn't been filled yet (i.e. such that $\Nal > B(i)$), one can find projectors on $\Spad[(V_{N,i})_N]$'s inputs and outputs that control whether $\Nal$ happens or not, and that all these projectors will play well with one another.



One subtlety is that, if $B(i)$ is in a \textit{red} layer and if there are still unfilled branches in that layer, then the projectors controlling the status of branches above that layer cannot be defined. This is ultimately not problematic, as one can wait for the whole layer to have been filled to redefine them; but this will force us to amend parts of the induction hypothesis when it is the case.

Finally, another part of the induction hypothesis will rely on the causal properties of $\Spad \left[(V_{N,i})_N \right]$. We will describe these by using the behaviour of $\Spad[(V_{N,i})_N]$ seen as an isomorphism of operator algebras, defining $\cv_i : \Lin \left[ \ch_P \otimes \left( \bigotimes_{\Oga > B(i)} \ch_{\Ogaout} \right) \right] \to \Lin \left[ \ch_F \otimes \left( \bigotimes_{\Oga > B(i)} \ch_{\Ogain} \right) \right]$ by

\be \forall f, \cv_i[f] := \Spad[(V_{N,i})_N] \circ f \circ \Spad[(V_{N,i})_N]^\dagger \, . \ee

When $\Spad[(V_{N,i})_N]$ is unitary, this defines an isomorphism of operator algebras, preserving sums, compositions, and the dagger. This implies that, more generally, $\cv_i$ will preserve commutation relations, self-adjointness, idempotency, etc.

We now turn to our induction hypotheses at step $i$.

\begin{hyp}[H1]\label{H1}
$\Spad[(V_{N,i})_N]$ is unitary.
\end{hyp}

As we mentioned, H1 is the core of the induction, and will allow us to conclude in the end that $\cs[(V_{N,n})_N] = \cs[(U_{N})_N]$ is indeed unitary.

\begin{hyp}[H2]
One has defined, for all $\Nal > B(i)$, orthogonal projectors:
\begin{itemize}
    \item $\zeiout(\Nal)$, acting on $\begin{cases}
    \ch_P \otimes \left( \bigotimes_{\Oga \in \cpi(\Nal)} \ch_\Ogaout  \right) \, &\textrm{ if } \Nal \textrm{ is in a green layer;}\\
    \ch_P \otimes \left( \bigotimes_{\Oga \in \cpist(\Nal)} \ch_\Ogaout  \right) \, &\textrm{ if } \Nal \textrm{ is in a red layer;} \end{cases}$
    \item $\zeiin(\Nal)$, acting on $\begin{cases}
    \ch_F \otimes \left( \bigotimes_{\Oga \in \cfist(\Nal)} \ch_\Ogain  \right) \, &\textrm{ if } \Nal \textrm{ is in a green layer;}\\
    \ch_F \otimes \left( \bigotimes_{\Oga \in \cfi(\Nal)} \ch_\Ogain  \right) \, &\textrm{ if } \Nal \textrm{ is in a red layer;} \end{cases}$
\end{itemize}

such that (once correctly padded) the $\zeipadout(\Nal)$'s for different $\Nal$'s all commute pairwise, and the $\zeipadin(\Nal)$'s commute as well, and such that

\be \label{eq:zetain to zetaout}
    \forall \Nal, \, \zeipadin(\Nal) = \cv_i[\zeipadout(\Nal)] \, .
\ee

If $B(i)$ and $B(i+1)$ are in the same red layer, then all of the former definitions have only been made for the $\Nal$'s in that same layer, i.e. in $\cl_i(B(i))$. When this happens, we say that \emph{$i$ is a special step}.
\end{hyp}

H2 introduces the projectors that will be used to control the status of the still-unfilled branches. The fact that the $\ze$'s commute pairwise ensures that these controls can always be meaningfully combined. Note that the out-projector for $\Nal$ only acts on $\Nal$'s past, while its in-projector only acts on $\Nal$'s future; and furthermore, that for $\Nal$ in a green layer its in-projector only acts on $\Nal$'s \textit{strict} future, while for $\Nal$ in a red layer its out-projector only acts on its \textit{strict} past. In particular, the $\ze(\Nal)$'s never act on $\Nal$ itself: this ensures that at any step, a branch never holds some part of its own controls.

We will also write $\barzeiin(\Nal):= \id - \zeiin(\Nal) $ and $\barzeiout(\Nal):= \id - \zeiout(\Nal) $.

\begin{hyp}[H3]
The $\zeiout$'s satisfy

\be \forall \Nal, \forall \Oga, \zeipadout(\Nal) \circ \barzeipadout(\Oga) \textrm{ acts trivially on } \ch_\Ogaout \, ,  \ee

and the $\zeiin$'s satisfy

\be \forall \Nal, \forall \Oga, \zeipadin(\Nal) \circ \barzeipadin(\Oga) \textrm{ acts trivially on } \ch_\Ogain \, .  \ee
\end{hyp}

This hypothesis encodes the fact that, when a branch $\Oga$ doesn't happen, it doesn't hold any control on other branches $\Nal$. Note that the $\forall \Nal, \Oga$ only runs over the branches for which the $\ze$'s have been defined in H2, i.e. it only runs over $\cl_i(B(i))$ if $i$ is a special step. The same will apply in the other hypotheses.

\begin{hyp}[H4]
The $\zeiout$'s satisfy:

\be \forall \Nal, \Nbe, \textrm{ branches of the same node}, \, \zeipadout(\Nal) \circ \zeipadout(\Nbe) = 0 \,.  \ee
\end{hyp}

The meaning is that two branches of the same node are incompatible. Note that one can infer, using (\ref{eq:zetain to zetaout}), that the $\zeiin$'s then satisfy the same property.

\begin{hyp}[H5]
Let $Q \subseteq \{B(i') \, | \, i' \geq i\}$ a set of branches on different nodes; i.e., one can define $\Tilde{Q} \subseteq \Nodes$ and a function $\al$ such that $Q = \{ N^{\al(N)} \, | \, N \in \Tilde{Q}\}$. Then,

\be \Spad[(V_{N,i})_N] \circ \prod_{N \in \Tilde{Q}} \zeipadout \left( N^{\al(N)} \right) = \Spad \left[(V_{N,i})_{N \in \Nodes \setminus \Tilde{Q}} \times \left(V_{N,i}^{\al(N)} \right)_{N \in \Tilde{Q}} \right] \, . \ee
\end{hyp}

H5 formalises the fact that the $\zeiout$'s control whether branches happen or not. Note that, using (\ref{eq:zetain to zetaout}), one could have written the same equation using $\zeiin$'s.

\begin{hyp}[H6]
For a branch $\Nal$ in a green layer, we have

\be \label{eq:(H6)} \begin{split}
    &\forall f \in \Lin[\ch_\Nalin], \, \exists f' \in \Lin \left[\ch_P \otimes \left( \bigotimes_{\Oga \in \cpist (\Nal)} \ch_\Ogaout \right) \right] \, \textrm{ such that} \\
    &\cv_i^\dagger[f_\pad] \circ \zeipadout(\Nal) = {f'}_\pad \circ \zeipadout(\Nal) \,.
\end{split} \ee
\end{hyp}

H6 means that, for a branch $\Nal$ in a green layer, provided that one is in the subspace in which branch $\Nal$ happens, $\Spad[(V_{N,i})_N]$'s causal structure only has the \textit{strict} past of $\Nal$ signalling to $\Nalin$. This will be important to ensure that, when $\Nal$ is `refilled', the action of any $\zeiout$'s on it becomes an action on its strict past.

\begin{hyp}[H7]
For a branch $\Nal$ in a red layer, we have

\be \label{eq:(H7)} \begin{split}
    &\forall f \in \Lin[\ch_\Nalout], \, \exists f' \in \Lin \left[\ch_F \otimes \left( \bigotimes_{\Oga \in \cfist (\Nal)} \ch_\Ogain \right) \right] \, \textrm{ such that} \\
    &\cv_i[f_\pad] \circ \zeipadin(\Nal) = {f'}_\pad \circ \zeipadin(\Nal) \,.
\end{split} \ee
\end{hyp}

H7 plays the same role as H6 in the reverse time direction.

\subsubsection{Proof of the base case}



\paragraph{H1} The proof that $\cs_\pad[(\ex_N)_N]$ is unitary will rely on the lemmas of Section \ref{sec: preliminary}. To use them, we will first introduce a way to show how bifurcation choices are enforced through the use of the $\ex$'s. For every $A$ in $\Arr$, we define $\witness_A: \ch_A \to \ch_A \otimes \mathbb{C}^\abs{\Ind_A}$ by

\be \witness_A := \sum_{k_A \in \Ind_A} \pi^{k_A}_A \otimes \ket{k_A} \, , \ee

where the $\pi^{k_A}_A$'s are the projectors on the different sectors of $A$, and we've introduced an arbitrary basis of $\mathbb{C}^\abs{\Ind_A}$ labelled by $A$'s index values. The point of $\witness_A$ is simply to channel out the information about each arrow's index value.

For a given $N$, with respect to the sectorisations of the $\ch_\Nalout$, of the $\ch_A$'s for the $A$'s in $\inn(N)$ and $\out(N)$, and to the sectorisation of the $\mathbb{C}^\abs{\Ind_A}$'s given by the previous basis, $\laNsec$ is a route for $\left( \bigotimes_{A \in \out(N)} \witness_{A} \right)_\pad \circ \ex_N$. Thus (because the compatibility with routes is preserved by the dagger compact structure \cite{vanrietvelde2021routed}), $\cs^\Rel_\pad[(\laNsec)_N] = \La^\sec_\Ga$ is a route for $\cs_\pad \left[\left (\left( \bigotimes_{A \in \out(N)} \witness_{A} \right)_\pad \circ \ex_N \right)_N \right]$.  Therefore,

\be \label{eq: base case H1 1}
\begin{split}
    &\cs_\pad \left[\left (\left( \bigotimes_{A \in \out(N)} \witness_{A} \right)_\pad \circ \ex_N \right)_N \right] \\
    &= \sum_{\veck \in \bigtimes_A \Ind_A} \left( \bigotimes_A \ketbra{k_A}{k_A} \right)_\pad \circ \cs_\pad \left[\left (\left( \bigotimes_{A \in \out(N)} \witness_{A} \right)_\pad \circ \ex_N \right)_N \right] \\
    &\circ \left( \sum_{\Vec{q} \in \left( \Lambda^\sec \right) \inv (\veck)} \bigotimes_{\Nal \in \BranGa} \pi^{q_\Nal}_\Nalout \right)_\pad \\
    &\overset{\textrm{Lemma \ref{lem: PossVal}}}{=} \sum_{\veck \in \PossVal} \left( \bigotimes_A \ketbra{k_A}{k_A} \right)_\pad \circ \cs_\pad \left[\left (\left( \bigotimes_{A \in \out(N)} \witness_{A} \right)_\pad \circ \ex_N \right)_N \right] \\
    &\circ \left( \sum_{\Vec{q} \in \left( \Lambda^\sec \right) \inv (\veck)} \bigotimes_{\Nal \in \BranGa} \pi^{q_\Nal}_\Nalout \right)_\pad \\
    &\overset{\textrm{Lemma \ref{lem: lambdasec}}}{=} \sum_{\veck \in \PossVal} \left( \bigotimes_A \ketbra{k_A}{k_A} \right)_\pad \circ \cs_\pad \left[\left (\left( \bigotimes_{A \in \out(N)} \witness_{A} \right)_\pad \circ \ex_N \right)_N \right] \\
    &\circ \left(\bigotimes_{N \in \Nodes} \pi^{(k_A)_{A \in \out(N)}}_{N^\muNveck_\out} \right)_\pad \, ,
\end{split}
\ee

where $\pi^{(k_A)_{A \in \out(N)}}_{N^\muNveck_\out}$ is the projector on $\ch_{N^\muNveck_\out}$'s sector labelled by $(k_A)_{A \in \out(N)}$ (remember that for a given $\Nal$, we have $\ch_\Nalout = \bigoplus_{(k_A)_{A \in \out(N)} \in \IndoutNal} \bigotimes_{A \in \out(N)} \ch_A^{k_A}$).

Moreover, we have $\left( \sum_{k \in \Ind_A} \bra{k} \right)_\pad \circ \witness_A = \id_A$, so

\be \label{eq: base case H1 2}
\begin{split}
    &\cs_\pad \left[(\ex_N)_N \right] \\
    &= \cs_\pad \left[\left( \sum_{(k_A)_A \in \Indout_N} \left(\bigotimes_{A \in \out(N)} \bra{k_A} \circ \witness_A \right)_\pad \ex_N \right)_N \right] \\
    &= \sum_{\veck \in \bigtimes_A \Ind_A} \left( \bigotimes_{A} \bra{k_A} \right)_\pad \cs_\pad \left[\left(  \left( \witness_A \right)_\pad \circ \ex_N \right)_N \right] \\
    &\overset{\textrm{(\ref{eq: base case H1 1})}}{=} \sum_{\veck \in \PossVal} \left( \bigotimes_A \bra{k_A} \right)_\pad \circ \cs_\pad \left[\left (\left( \bigotimes_{A \in \out(N)} \witness_{A} \right)_\pad \circ \ex_N \right)_N \right] \\
    &\circ \left(\bigotimes_{N \in \Nodes} \pi^{(k_A)_{A \in \out(N)}}_{N^\muNveck_\out} \right)_\pad \\
    &= \sum_{\veck \in \PossVal}  \cs_\pad \left[\left (\left( \bigotimes_{A \in \out(N)} \pi^{k_A}_{A} \right)_\pad \circ \ex_N \right)_N \right]    \circ \left(\bigotimes_{N \in \Nodes} \pi^{(k_A)_{A \in \out(N)}}_{N^\muNveck_\out} \right)_\pad \, .
\end{split}
\ee

A symmetric argument relying on $\Gatop$ leads to

\be \label{eq: base case H1 2 sym}
\begin{split}
    &\cs_\pad \left[(\ex_N)_N \right] \\
    &= \sum_{\veck \in \PossVal}  \left(\bigotimes_{N \in \Nodes} \pi^{(k_A)_{A \in \inn(N)}}_{N^\muNveck_\inn} \right)_\pad   \circ \cs_\pad \left[\left( \ex_N \circ \left( \otimes_{A \in \inn(N)} \pi^{k_A}_{A} \right)_\pad \right)_N \right]  \, .
\end{split}
\ee

Furthermore, the projectors $\left(\bigotimes_{N \in \Nodes} \pi^{(k_A)_{A \in \out(N)}}_{N^\muNveck_\out} \right)_\pad$, for $\veck \in \PossVal$, form a sectorisation of the input space of $\cs_\pad \left[(\ex_N)_N \right]$. Indeed, by Lemma \ref{lem: lambdasec}, the 

\be \left( \bigtimes_N \{(k_A)_{A \in \out(N)} \}_{\Indout_{N^\muNveck}} \right) \times \left( \bigtimes_{N^\al | \al \neq \muNveck} \IndoutNal \right) \ee

are the preimage sets of the injective function $\La^\sec_\Ga$, and therefore form a partition of its domain $\bigtimes_\Nal \IndoutNal$. The sectorisation is thus obtained as a coarse-graining of that given by the $\bigotimes_\Nal \pi^{(k_A)_{A \in \out(N)}}_\Nalout$. Symmetrically, the  $\left(\bigotimes_{N \in \Nodes} \pi^{(k_A)_{A \in \inn(N)}}_{N^\muNveck_\inn} \right)_\pad$ form a sectorisation of $\cs_\pad \left[(\ex_N)_N \right]$'s codomain. Crucially, $\cs_\pad \left[(\ex_N)_N \right]$ is block diagonal with respect to these two sectorisations: indeed, for a given $\veck$,

\be \label{eq: base block diagonal}\begin{split}
     &\cs_\pad \left[(\ex_N)_N \right] \circ \left(\bigotimes_{N \in \Nodes} \pi^{(k_A)_{A \in \out(N)}}_{N^\muNveck_\out} \right)_\pad \\
     &\overset{\textrm{(\ref{eq: base case H1 2})}}{=} \cs_\pad \left[\left (\left( \bigotimes_{A \in \out(N)} \pi^{k_A}_{A} \right)_\pad \circ \ex_N \right)_N \right] \\
     &\overset{\textrm{(\ref{eq:skeletal supermap})}}{=} \cs_\pad \left[\left( \ex_N \circ \left( \bigotimes_{A \in \inn(N)} \pi^{k_A}_{A} \right)_\pad \right)_N \right] \\
     &\overset{\textrm{(\ref{eq: base case H1 2 sym})}}{=}  \left(\bigotimes_{N \in \Nodes} \pi^{(k_A)_{A \in \inn(N)}}_{N^\muNveck_\inn} \right)_\pad \circ \cs_\pad \left[(\ex_N)_N \right] \\
\end{split} \ee

All that is left for us to prove is that all of these blocks, which we will denote as $T^{\veck}$'s, are unitary (with respect to the suitable restrictions of their domain and codomain). We start by computing

\be \begin{split}
    T^\veck &= \cs_\pad \left[\left (\left( \bigotimes_{A \in \out(N)} \pi^{k_A}_{A} \right)_\pad \circ \ex_N \right)_N \right] \\
    &\overset{\textrm{(\ref{eq: def exchange})}}{=} \cs_\pad \left[\left (\left( \bigotimes_{A \in \out(N)} \pi^{k_A}_{A} \right)_\pad \circ i_{N,\pad}^\muNveck \circ \left( \swap_{N^\muNveck_\inn, N^\muNveck_\out} \otimes \left( \bigotimes_{\bet \neq \muNveck} \Th_\Nbe \right) \right) \circ p_{N,\pad}^\muNveck \right)_N \right] \\
    &\overset{\textrm{(\ref{eq:skeletal supermap})}}{=} \Tr_{A \in \Arr^\int} \left[ \bigotimes_N \left( \bigotimes_{A \in \out(N)} \pi^{k_A}_{A} \right)_\pad \circ i_{N,\pad}^\muNveck \circ \swap_{N^\muNveck_\inn, N^\muNveck_\out} \circ p_{N,\pad}^\muNveck \right] \otimes \left( \bigotimes_{M, \bet \neq \mu_M(\veck)} \Th_\Mbe \right) \\
    &= \left( \left( \bigotimes_N p_N^\muNveck \right) \circ \left( \bigotimes_{A \in \Arr} \pi^{k_A}_{A} \right) \circ \left( \bigotimes_N i_{N}^\muNveck \right)  \right) \otimes \left( \bigotimes_{M, \bet \neq \mu_M(\veck)} \Th_\Mbe \right) \,.
\end{split} \ee

Remember that $i_N^\muNveck$ is the injection $\ch^\out_{N^\muNveck} \to \ch^\out_N = \bigotimes_{A \in \out(N)} \ch_A$, and $p_N^\muNveck$ is the projection $\ch^\inn_N = \bigotimes_{A \in \inn(N)} \ch_A \to \ch^\inn_{N^\muNveck} $. We will also define the injection $i_N^{(k_A)_{A \in \out(N)}}: \bigotimes_{A \in \out(N)} \ch_A^{k_A} \to \ch^\out_{N^\muNveck}$, and the projection $p_N^{(k_A)_{A \in \inn(N)}}: \ch^\inn_{N^\muNveck} \to \bigotimes_{A \in \inn(N)} \ch_A^{k_A}$: these map $T^\veck$'s to the suitable domains and codomains. Note that we then have

\begin{subequations}
\be \bigotimes_N i_N^\muNveck \circ i_N^{(k_A)_{A \in \out(N)}} = \bigotimes_A i_A^{k_A} \, , \ee
\be \bigotimes_N p_N^{(k_A)_{A \in \out(N)}} \circ p_N^\muNveck = \bigotimes_A p_A^{k_A} \, , \ee
\end{subequations}

where $i_A^{k_A}$ is the injection $\ch_A^{k_A} \to \ch_A$ and $p_A^{k_A}$ is the projection $\ch_A \to \ch_A^{k_A}$. Thus,

\be \label{eq: Tk H1} \begin{split}
    &\left( \bigotimes_N p_N^{(k_A)_{A \in \inn(N)}} \right)_\pad \circ  T^\veck \circ \left( \bigotimes_N i_N^{(k_A)_{A \in \out(N)}} \right)_\pad \\
    &= \left( \bigotimes_A p_A^\kA \circ \pi_A^\kA \circ i_A^\kA \right) \otimes \left( \bigotimes_{M, \bet \neq \mu_M(\veck)} \Th_\Mbe \right)\\
    &= \left( \bigotimes_A \id_{A^\kA} \right) \otimes \left( \bigotimes_{M, \bet \neq \mu_M(\veck)} \Th_\Mbe \right) \,.
\end{split} \ee

Each of the blocks composing $\cs_\pad \left[(\ex_N)_N \right]$ is thus unitary once restricted to the suitable subspaces, so $\cs_\pad \left[(\ex_N)_N \right]$ is unitary as well.

\paragraph{H2} We define, for all branches $\Nal$,

\begin{subequations}
\be Z^\out(\Nal) \label{eq: Z base} := \La_\Ga\inv \left( \{1\}_\HappNal \times \bigtimes_{\Mbe \neq \Nal} \HappMbe \right) \, , \ee
\be \zeout(\Nal) \label{eq: zeta base} := \sum_{(l_\Mbe)_{\Mbe \in \BranGa} \in Z^\out(\Nal)} \left( \bigotimes_\Mbe \pi^{l_\Mbe}_{\Mbeout} \right) \, , \ee
\end{subequations}

and similarly for the $Z^\inn$'s and $\zein$'s. Note that 

\be \label{eq: Z from Lasec} Z^\out(\Nal) = \bigsqcup_{\substack{\veck \in \PossVal \\ \muNveck = \al}} \left( \La^\sec_\Ga \right) \inv \left( \veck \right) \, .\ee

Given their definition, the $\zeout$'s are commuting orthogonal projectors. Furthermore, as green dashed arrows in $\GaBran$ represent $\La_\Ga$'s causal structure, we have, for any branch $\Nal$,

\be Z^\out(\Nal) := \tilde{Z}^\out(\Nal) \times \left( \bigtimes_{\exists \textrm{ no green dashed arrow } \Mbe \to \Nal} \Indout_\Mbe \right)  \ee

Through (\ref{eq: zeta base}), this implies that $\zeout(\Nal)$ acts trivially on the $\Mbe$'s that are not linked to $\Nal$ by a green dashed arrow. We can thus in particular see it as the padding of an operator acting only on $\cp(\Nal)$, or on $\cp^\str(\Nal)$ if $\Nal$ is in a red layer. The same applies symmetrically for the $\zein$'s. Finally, (\ref{eq: base block diagonal}) implies

\be \label{eq: base zein/zeout} \begin{split}
     \cs_\pad \left[(\ex_N)_N \right] \circ \zeout(\Nal) &= \sum_{\substack{\veck \in \PossVal \\ \muNveck = \al}} \cs_\pad \left[(\ex_N)_N \right] \circ \left(\bigotimes_{N \in \Nodes} \pi^{(k_A)_{A \in \out(N)}}_{N^\muNveck_\out} \right)_\pad \\
     &= \sum_{\substack{\veck \in \PossVal \\ \muNveck = \al}} \left(\bigotimes_{N \in \Nodes} \pi^{(k_A)_{A \in \out(N)}}_{N^\muNveck_\out} \right)_\pad \circ \cs_\pad \left[(\ex_N)_N \right] \\
     &= \zein(\Nal) \circ \cs_\pad \left[(\ex_N)_N \right] \, ; \\
\end{split}\ee

thus, $\cv_0[\zeout(\Nal)] = \zein(\Nal)$.

\paragraph{H3}

We will prove that $Z^\out(\Nal) \cap \bar{Z}^\out(\Mbe)$ is of the form $\tilde{Z} \times \Indout_\Mbe$, from which (H3) derives. This set can be computed, using bra-ket notations in $\Rel$, as $\left( \bra{1}_{\HappNal} \otimes \bra{0}_\HappMbe \right)_\pad \circ \La_\Ga$. Yet, one can see from the definition of the $\la^\aug$'s that

\be \bra{0}_{\HappMbe, \pad} \circ \la_M^\aug = \bra{0}_{\HappMbe, \pad} \circ \la_M^\aug \circ \ketbra{\Indout_\Mbe}{\Indout_\Mbe}_{\Indout_\Mbe, \pad} \, ; \ee

thus,

\be \begin{split}
    &\left( \bra{1}_{\HappNal} \otimes \bra{0}_\HappMbe \right)_\pad \circ \La_\Ga \\
    &= \left( \bra{1}_{\HappNal} \otimes \bra{0}_\HappMbe \right)_\pad \circ \SRelGa\left[(\laNaug)_N \right] \\
    &= \left( \bra{1}_{\HappNal} \otimes \bra{0}_\HappMbe \right)_\pad \circ \SRelGa\left[(\laNaug)_{N\neq M} \times \left(\la_M^\aug \circ \ketbra{\Indout_\Mbe}{\Indout_\Mbe}_{\Indout_\Mbe, \pad} \right) \right] \\
    &= \left( \left( \bra{1}_{\HappNal} \otimes \bra{0}_\HappMbe \right)_\pad \circ \SRelGa\left[(\laNaug)_N \right] \circ \ket{\Indout_\Mbe}_{\Indout_\Mbe} \right) \otimes \bra{\Indout_\Mbe}_{\Indout_\Mbe} \, ,
\end{split}  \ee

which shows that indeed $Z^\out(\Nal) \cap \bar{Z}^\out(\Mbe) = \tilde{Z} \times \Indout_\Mbe$. The proof for the $Z^\inn$'s is symmetric.

\paragraph{H4} (H4) comes from the fact that, for $\al \neq \bet$, one has $Z^\out(\Nal) \cap Z^\out(\Nbe) = \emptyset$, which can be derived directly from (\ref{eq: Z from Lasec}).

\paragraph{H5} We take $Q = \{ N^{\al(N)} \, | \, N \in \Tilde{Q}\} \subseteq \BranGa$. Then,

\be \begin{split}
    &\cs_\pad \left[(\ex_N)_N \right] \circ \prod_{N \in \tilde{Q}} \zeout(N^{\al(N)}) \\
    &\overset{\textrm (\ref{eq: Z from Lasec})}{=} \sum_{\substack{\veck \in \PossVal \\ \forall N \in \tilde{Q}, \muNveck = \al(N)}} \cs_\pad \left[(\ex_N)_N \right] \circ \left(\bigotimes_{N \in \Nodes} \pi^{(k_A)_{A \in \out(N)}}_{N^\muNveck_\out} \right)_\pad \\
    &\overset{\textrm (\ref{eq: base case H1 2})}{=} \sum_{\substack{\veck \in \PossVal \\ \forall N \in \tilde{Q}, \muNveck = \al(N)}} \cs_\pad \left[\left (\left( \bigotimes_{A \in \out(N)} \pi^{k_A}_{A} \right)_\pad \circ \ex_N \right)_N \right] \\
    &= \cs_\pad \left[ \left( \ex_N \right)_{N \not\in \tilde{Q}} \times \left ( \sum_{(k_A)_{A \in \out(N)} \in \IndoutNal} \left( \bigotimes_{A \in \out(N)} \pi^{k_A}_{A} \right)_\pad \circ \ex_N \right)_{N \in \tilde{Q}} \right] \\
    &= \cs_\pad \left[ \left( \ex_N \right)_{N \not\in \tilde{Q}} \times \left ( \pi^\al_{N^\out} \circ \ex_N \right)_{N \in \tilde{Q}} \right] \\
    &\overset{\textrm (\ref{eq: def exchange})}{=} \cs_\pad \left[ \left( \ex_N \right)_{N \not\in \tilde{Q}} \times \left ( V_{0,N}^\al \right)_{N \in \tilde{Q}} \right] \,.
\end{split} \ee

\paragraph{H6} We take $\Nal$ in a green layer, and $f \in \Lin \left[ \ch_{\Nalin} \right]$. We then have (note that $\zein(\Nal)$ doesn't act on $\Nalin$)

\be
\begin{split}
    &\cv_{0}^\dagger \left[ f_\pad \right] \circ \zepadout(\Nal) = \cs_\pad \left[(\ex_N)_N \right]^\dagger \circ f_\pad \circ \cs_\pad \left[(\ex_N)_N \right] \circ \zepadout(\Nal) \\
    &\overset{\textrm (\ref{eq: base zein/zeout})}{=} \cs_\pad \left[(\ex_N)_N \right]^\dagger \circ f_\pad \circ  \zepadin(\Nal) \circ \cs_\pad \left[(\ex_N)_N \right] \circ \zepadout(\Nal) \\
    &= \cs_\pad \left[(\ex_N)_N \right]^\dagger \circ  \zepadin(\Nal) \circ f_\pad  \circ \cs_\pad \left[(\ex_N)_N \right] \circ \zepadout(\Nal) \\
    &\overset{\textrm (\ref{eq: base zein/zeout})}{=} \zepadout(\Nal) \circ \cs_\pad \left[(\ex_N)_N \right]^\dagger \circ f_\pad \circ \cs_\pad \left[(\ex_N)_N \right] \circ \zepadout(\Nal) \\
    &\overset{\textrm (\ref{eq: zeta base}), (\ref{eq: Z from Lasec})}{=} \sum_{\substack{\veck, \vecl \in \PossVal \\ \muNveck = \muNvecl = \al}}  \left( T^\vecl \, \right)^\dagger \circ f_\pad \circ T^\veck\, .
\end{split}
\ee

Furthermore, taking $\Mbe \not\in \cpst(\Nal)$, because $\Nal$ is in a green layer we know that there is no red dashed arrow from $\Mbe$ to $\Nal$, and thus $\zein(\Mbe)$ doesn't act on $\Nalin$. We can thus apply the same computation to it as well, which leads to

\be
\begin{split}
    &\cv_{0}^\dagger \left[ f_\pad \right] \circ \zepadout(\Nal) = \cv_{0} \left[ f_\pad \right] \circ \zepadout(\Nal) \circ \left( \zepadout(\Mbe) + \barzepadout (\Mbe) \right)  \\
    &= \zepadout(\Nal) \circ \zepadout(\Mbe) \circ \cv_{0} \left[ f_\pad \right] \circ \zepadout(\Nal) \circ \zepadout(\Mbe)\\
    &+ \zepadout(\Nal) \circ \barzepadout(\Mbe) \circ \cv_{0} \left[ f_\pad \right] \circ \zepadout(\Nal) \circ \barzepadout(\Mbe)\\
    &=  \sum_{\substack{\veck, \vecl \in \PossVal \\ \muNveck = \muNvecl = \al \\
\mu_M(\veck) = \bet \iff \mu_M(\vecl) = \bet}}  \left( T^\vecl \, \right)^\dagger \circ f_\pad \circ T^\veck\, ;
\end{split}
\ee

in other words, in the sum above, the values of $\veck$ and $\vecl$ that lead to attributing different statuses to $\Mbe$ correspond to null terms, so that one can skip them in the summation. More generally, one can apply this reasoning to all branches $\Mbe \not\in \cpst(\Nal)$, leading to

\be \label{eq: base H6 1}
\cv_{0} \left[ f_\pad \right] \circ \zepadout(\Nal) =  \sum_{\substack{\veck, \vecl \in \PossVal \\ \muNveck = \muNvecl = \al \\
\forall \Mbe \not\in \cpst(\Nal), \mu_M(\veck) = \bet \iff \mu_M(\vecl) = \bet}}  \left( T^\vecl \, \right)^\dagger \circ f_\pad \circ T^\veck\, .
\ee

Using (\ref{eq: Tk H1}), we rewrite $T^\veck$, for an arbitrary $\veck$, as 

\be \begin{split}
    T^\veck &= \left( \bigotimes_M {p_M^{(k_A)_{A \in \inn(M)}}} \right)^\dagger_\pad \circ \left( \left( \bigotimes_A \id_{A^{k_A}} \right) \otimes \left( \bigotimes_{M, \bet \neq \mu_M(\veck) }\Theta_\Mbe \right)\right) \\
    &\circ \left( \bigotimes_M {i_M^{(k_A)_{A \in \out(M)}}} \right)^\dagger_\pad \, . 
\end{split}
\ee



Now, we take $\veck, \vecl \in \PossVal$ satisfying the requirements we pinned down earlier; we can then compute

\be \begin{split}
&{T^\vecl}^\dagger \circ f_\pad \circ T^\veck \\
&= \left( \bigotimes_{M}  i_M^{(l_A)_{A \in \out(M)}}  \right)_\pad \circ \left( \left( \bigotimes_A \id_{A^{l_A}} \right) \otimes \left( \bigotimes_{M, \bet \neq \mu_M(\vecl) }\Theta_\Mbe^\dagger \right)\right) \\
&\circ \left( \bigotimes_M {p_M^{(l_A)_{A \in \inn(M)}}}  \right)_\pad \circ f_\pad \circ \left( \bigotimes_M {p_M^{(k_A)_{A \in \inn(M)}}}  \right)^\dagger_\pad \\
&\circ \left( \left( \bigotimes_A \id_{A^{\kA}} \right) \otimes \left( \bigotimes_{M, \bet \neq \mu_M(\veck) }\Theta_\Mbe \right)\right) \circ \left( \bigotimes_M {i_M^{(k_A)_{A \in \out(M)}}}  \right)^\dagger_\pad \\
&= \left( \bigotimes_{M}  i_M^{(l_A)_{A \in \out(M)}}  \right)_\pad \circ \left[ \left( p_N^{(l_A)_{A \in \inn(N)}} \circ f \circ  \left( p_N^{(k_A)_{A \in \inn(N)}} \right)^\dagger \right) \right. \\
&\otimes \left( \bigotimes_{\substack{M \neq N \\ \mu_M(\vecl) \neq \mu_M(\veck)}} \left( \Theta_{M^{\mu_M(\veck)}}^\dagger \circ \left( p_M^{(k_A)_{A \in \inn(M)}} \right)^\dagger \right) \otimes \left( p_M^{(l_A)_{A \in \inn(M)}}  \circ \Theta_{M^{\mu_M(\vecl)}} \right)  \right) \\
&\otimes \left. \left( \bigotimes_{\substack{M \neq N \\ \mu_M(\vecl) = \mu_M(\veck) = 1}} {p_M^{(l_A)_{A \in \inn(M)}}} \circ \left( p_M^{(k_A)_{A \in \inn(M)}} \right)^\dagger \right)  \right]_\pad \circ \left( \bigotimes_M {i_M^{(k_A)_{A \in \out(M)}}}  \right)^\dagger_\pad \, . 
\end{split}
\ee

Note that each of the ${p_M^{(l_A)_{A \in \inn(M)}}} \circ \left(p_M^{(k_A)_{A \in \inn(M)}} \right)^\dagger$ terms, for $M$ such that $M^{\mu_M(\veck)} \not\in \cpst(\Nal) \cup \{\Nal\}$, can be rewritten as $\bigotimes_{A \in \inn(M)} p_A^{l_A} \circ i_A^{k_A}$ (which is the identity if $k_A = l_A \forall A \in \inn(M)$, and zero otherwise).

Now, for any $M$ such that $M^{\mu_M(\veck)} \not\in \cpst(\Nal)$, and for any $O$ such that $O^{\mu_O(\veck)} \in \cpst(\Nal) \cup \{ \Nal \}$, there is no arrow $A \in \Link(M,O)$ such that $\abs{A^\kA} \neq 1$, as that would imply the existence of a solid arrow from $M^{\mu_M(\veck)}$ to $O^{\mu_O(\veck)}$, which would contradict $M^{\mu_M(\veck)} \not\in \cpst(\Nal)$. Thus, all of the non-trivial arrows in $\out(M)$ go to $O$'s such that $O \neq N$ and $\mu_O(\vecl) = \mu_O(\veck)$. Thus this implies that, if one doesn't have $k_A = l_A \forall A \in \out(M)$ then the whole expression is null; while otherwise, the term in square brackets acts trivially on each of the $A \in \out(M)$ -- in other words, the arrows coming out of $M$ are never acted on and simply link $i_M^{(k_A)_{A \in \out(M)}}{}^\dagger$ directly to $i_M^{(l_A)_{A \in \out(M)}}$.  One can thus reorganise this expression (neglecting the existence of all the trivial spaces) as

\be \label{eq: comp H6 base}\begin{split}
&{T^\vecl}^\dagger \circ f_\pad \circ T^\veck =  \left[ \left( \bigotimes_{\substack{M \\ M^{\mu_M(\vecl)} \in \cpst(\Nal)}}  i_M^{(l_A)_{A \in \out(M)}}  \right)_\pad \circ \left[ \vphantom{\bigotimes_{\substack{a \\ b \\ c}}} \left( p_N^{(l_A)_{A \in \inn(N)}} \circ f \circ  \left( p_N^{(k_A)_{A \in \inn(N)}} \right)^\dagger \right) \right. \right. \\
&\otimes \left( \bigotimes_{\substack{M \\ M^{\mu_M(\vecl)}, M^{\mu_M(\veck)} \in \cpst(\Nal) \\ \mu_M(\vecl) \neq \mu_M(\veck)}} \left( \Theta_{M^{\mu_M(\veck)}}^\dagger \circ \left( p_M^{(k_A)_{A \in \inn(M)}} \right)^\dagger \right) \otimes \left( p_M^{(l_A)_{A \in \inn(M)}}  \circ \Theta_{M^{\mu_M(\vecl)}} \right)  \right) \\
&\otimes \left. \left( \bigotimes_{\substack{M \\ M^{\mu_M(\vecl)}, M^{\mu_M(\veck)} \in \cpst(\Nal) \\ \mu_M(\vecl) = \mu_M(\veck)}} \bigotimes_{A \in \inn(M)} p_A^{l_A} \circ i_A^{k_A} \right) \otimes \left( \bigotimes_{\substack{O \\ O^{\mu_O(\veck)} \not\in \cpst(\Nal)}} \bigotimes_{\substack{M \\ M^{\mu_M(\veck)} \in \cpst(\Nal)}} \bigotimes_{A \in \Link(M,O)} p_A^{l_A} \circ i_A^{k_A} \right)  \right]  \\
&\circ \left. \left( \bigotimes_{\substack{M \\ M^{\mu_M(\veck)} \in \cpst(\Nal)}} {i_M^{(k_A)_{A \in \out(M)}}}  \right)^\dagger_\pad \right]_\pad \otimes \left( \left( \bigotimes_{\substack{M \\ M^{\mu_M(\veck)} \not\in \cpst(\Nal)}} \pi^{(\kA)_{A \in \out(M)}}_{M^{\mu_M(\veck)}_\out} \right) \otimes \left( \bigotimes_{\substack{M, \bet \neq \mu_M(\veck) \\ \Mbe \not\in \cpst(\Nal)}} \id_{\Mbeout} \right)\right) \, ;
\end{split}
\ee

note how the action on the $\Mbeout$'s for $\Mbe \not\in \cpst(\Nal)$ now only consists of a projector independent of $f$. Note that in the second bracket of the third line, the condition $M^{\mu_M(\veck)} \in \cpst(\Nal)$ could equivalently have been replaced with $M^{\mu_M(\vecl)} \in \cpst(\Nal)$, because we know that $M^{\mu_M(\veck)} \not\in \cpst(\Nal) \implies M^{\mu_M(\vecl)} = M^{\mu_M(\veck)} \not\in \cpst(\Nal)$, and conversely; so that we have the equivalence $M^{\mu_M(\veck)} \not\in \cpst(\Nal) \iff M^{\mu_M(\vecl)} \not\in \cpst(\Nal)$.

Both bracketed terms in the third line can be rewritten simply as Kronecker deltas, of the form:

\be \left( \prod_{\substack{M \\ M^{\mu_M(\vecl)}, M^{\mu_M(\veck)} \not\in \cpst(\Nal) \\ \mu_M(\vecl) = \mu_M(\veck)}} \prod_{A \in \inn(M)} \delta_{k_A, l_A} \right)  \cdot  \left( \prod_{\substack{O \\ O^{\mu_O(\veck)} \not\in \cpst(\Nal)}} \prod_{\substack{M \\ M^{\mu_M(\veck)} \in \cpst(\Nal)}} \prod_{A \in \Link(M,O)} \delta_{k_A, l_A} \right) \, .\ee

We now take $\vec{q} = (q^\Oga)_{\Oga \in \BranGa}, \vec{r} = (r^\Oga)_{\Oga \in \BranGa} \in \bigtimes_{\Oga \in \BranGa} \Indout_\Oga$. We then have

\be \label{eq: comp H6 base 2}\begin{split}
    &\left( \bigotimes_{\Oga \in \BranGa} \pi^{r^\Oga}_\Ogaout \right) \circ \cv_0^\dagger[f_\pad] \circ \zepadout(\Nal) \circ \left( \bigotimes_{\Oga \in \BranGa} \pi^{q^\Oga}_\Ogaout \right) \\
    &= \left( \bigotimes_{\Oga \in \BranGa} \pi^{r^\Oga}_\Ogaout \right) \circ \cs_\pad \left[(\ex_N)_N \right]^\dagger \circ f_\pad \circ \cs_\pad \left[(\ex_N)_N \right] \circ  \left( \bigotimes_{\Oga \in \BranGa} \pi^{q^\Oga}_\Ogaout \right) \circ \zepadout(\Nal) \\
    &= \left( \bigotimes_{\Oga \in \BranGa} \pi^{r^\Oga}_\Ogaout \right) \circ \left( T^{\La^\sec_\Ga \left(\vec{r} \right)} \right)^\dagger \circ f_\pad \circ T^{\La_\Ga^\sec \left( \vec{q} \right)} \circ  \left( \bigotimes_{\Oga \in \BranGa} \pi^{q^\Oga}_\Ogaout \right) \circ \zepadout(\Nal) \, .
\end{split}
\ee

Note that $\mu_M \circ \La_\Ga^\sec \left( \vec{q} \right)$ denotes the only branch $\bet$ of $M$ such that the $\Mbe$ term of $\La_\Ga \left( \vec{q} \right)$, which we denote $\La^\Mbe_\Ga \left( \vec{q} \right)$, is $1$. By the previous considerations, the term above can thus be non null only if $\La^\Nal_\Ga \left( \vec{q} \right) = \La^\Nal_\Ga \left( \vec{r} \right) = 1$ and if $\La^\Mbe_\Ga \left( \vec{q} \right) = \La^\Mbe_\Ga \left( \vec{r} \right) \forall \Mbe \not\in \cpst(\Nal)$. Furthermore, when this is the case, then we can use (\ref{eq: comp H6 base}) and get 

\be \label{eq: comp H6 base 3}\begin{split}
    &\left( \bigotimes_{\Oga \in \BranGa} \pi^{r^\Oga}_\Ogaout \right) \circ \cv_0^\dagger[f_\pad] \circ \zepadout(\Nal) \circ \left( \bigotimes_{\Oga \in \BranGa} \pi^{q^\Oga}_\Ogaout \right) \\
    &= F^{(r^\Oga)_{\Oga \in \cpst(\Nal)}, (q^\Oga)_{\Oga \in \cpst(\Nal)}} \otimes \left(\bigotimes_{\Mbe \not\in \cpst(\Nal)} \pi^{q^\Mbe}_\Mbeout \right) \\
    &\cdot \prod_{\Oga \not\in \cpst(\Nal)} \delta_{q^\Oga, r^\Oga}\, ,
\end{split}
\ee

where

\be \label{eq: F^rq} \begin{split}
&F^{(r^\Oga)_{\Oga \in \cpst(\Nal)}, (q^\Oga)_{\Oga \in \cpst(\Nal)}} \\
&= \left( \bigotimes_{\Oga \in \cpst(\Nal)} \pi^{r^\Oga}_\Ogaout \right) \circ  \left( \bigotimes_{\substack{\Mbe \in \cpst(\Nal) \\ \La^\Mbe_\Ga \left( \vec{r} \right) =1}}  i_M^{(l_A)_{A \in \out(M)}}  \right)_\pad \circ \left[ \vphantom{\bigotimes_{\substack{a \\ b \\ c}}} \left( p_N^{(l_A)_{A \in \inn(N)}} \circ f \circ  \left( p_N^{(k_A)_{A \in \inn(N)}} \right)^\dagger \right)  \right. \\
&\otimes \left. \left( \bigotimes_{\substack{M \textrm{ s.t.\ } \exists \bet, \bet': \\ \Mbe, M^{\bet'} \in \cpst(\Nal) \\ \La^\Mbe_\Ga \left( \vec{q} \right) = \La^{M^{\bet'}}_\Ga \left( \vec{r} \right) = 1}} \left( \Theta_{\Mbe}^\dagger \circ \left( p_M^{(k_A)_{A \in \inn(M)}} \right)^\dagger \right) \otimes \left( p_M^{(l_A)_{A \in \inn(M)}}  \circ \Theta_{M^{\bet'}} \right)  \right) \right] \\
&\circ \left( \bigotimes_{\substack{\Mbe \in \cpst(\Nal) \\ \La^\Mbe_\Ga \left( \vec{q} \right) =1}} {i_M^{(k_A)_{A \in \out(M)}}}  \right)^\dagger_\pad \circ \left( \bigotimes_{\Oga \in \cpst(\Nal)} \pi^{q^\Oga}_\Ogaout \right) \cdot \left( \prod_{\substack{\Mbe \in \cpst(\Nal) \\ \La^\Mbe_\Ga \left( \vec{r} \right) = \La^\Mbe_\Ga \left( \vec{q} \right) = 1}} \prod_{A \in \inn(M)} \delta_{k_A, l_A} \right)  \\
&\cdot  \left( \prod_{\substack{O \\ \sum_{\Oga \in \cpst(\Nal)} \La^\Oga_\Ga \left( \vec{q} \right) = 0}} \quad \prod_{\substack{M \\ \sum_{\Mbe \in \cpst(\Nal)} \La^\Mbe_\Ga \left( \vec{q} \right) = 1}} \quad \prod_{A \in \Link(M,O)} \delta_{k_A, l_A} \right)\, ,
\end{split}\ee

with $\veck := \La_\Ga^\sec(\vec{q})$, $\vecl := \La_\Ga^\sec(\vec{r})$. Note how $F^{(r^\Oga)_{\Oga \in \cpst(\Nal)}, (q^\Oga)_{\Oga \in \cpst(\Nal)}}$, which comes from the term in square brackets in (\ref{eq: comp H6 base}) and acts only on $P$ and on $\cpst(\Nal)$, doesn't depend on the value of $(q^\Oga)_{\Oga \not\in \cpst(\Nal)}$: indeed, the values of $\La^\Mbe_\Ga \left( \vec{q} \right)$ and of $\La^\Mbe_\Ga \left( \vec{r} \right)$ don't depend on the values of the $q^\Oga$ for $\Oga \not\in \cpst(\Nal)$, as there are no green dashed arrows from these $\Oga$'s to any $\Mbe \in \cpst(\Nal)$.  Similarly, the relevant values of the $k_A$'s are those such that $A \in \out(M)$ for some $\Mbe \in \cpst(\Nal)$ such that $\La^\Mbe_\Ga \left( \vec{q} \right) = 1$; thus they are just equal to the $A$-value of $q^\Mbe$. The same goes for the $l_A$'s.

So far, we proved that (\ref{eq: comp H6 base 3}) holds for $\vec{q}$ and $\vec{r}$ satisfying: $\La^\Nal_\Ga \left( \vec{q} \right) = \La^\Nal_\Ga \left( \vec{r} \right) = 1$ and $\La^\Mbe_\Ga \left( \vec{q} \right) = \La^\Mbe_\Ga \left( \vec{r} \right) \forall \Mbe \not\in \cpst(\Nal)$. We now want to prove that the same holds when the latter condition is not satisfied -- or in other words, that in this case the RHS is also null. We will thus prove that if $\La^\Mbe_\Ga \left( \vec{q} \right) = \La^\Mbe_\Ga \left( \vec{r} \right) \forall \Mbe \not\in \cpst(\Nal)$ does not hold, then (\ref{eq: F^rq}) is null, and thus the RHS in (\ref{eq: comp H6 base 3}) is null as well.

We suppose (\ref{eq: F^rq}) is not null, and take $\Mbe \not\in \cpst(\Nal)$ such that $\La^\Mbe_\Ga \left( \vec{q} \right) = 1$. Taking $A \in \out(M)$, and denoting $O := \head(A)$ and $\ga$ such that $\La^\Oga_\Ga \left( \vec{q} \right) = 1$, we have: either $\Oga \in \cpst(\Nal)$, in which case $k_A = l_A$ by the penultimate term in (\ref{eq: F^rq}); or $\Oga \not\in \cpst(\Nal)$, in which case $k_A = l_A$ by the last term in (\ref{eq: F^rq}). Thus we have $k_A = l_A \forall A \in \out(M)$, and thus $\La^\Mbe_\Ga \left( \vec{r} \right) = 1$. Symmetrically, $\La^\Mbe_\Ga \left( \vec{r} \right) = 1$ implies $\La^\Mbe_\Ga \left( \vec{q} \right) = 1$, so we indeed get $\La^\Mbe_\Ga \left( \vec{r} \right) = \La^\Mbe_\Ga \left( \vec{q} \right)$.







Therefore, (\ref{eq: comp H6 base 3}) holds for $\vec{q}$ satisfying $\La^\Nal_\Ga \left( \vec{q} \right) = 1$.  We can thus finally compute

\be \begin{split}
    &\cv_0^\dagger[f_\pad] \circ \zepadout(\Nal) \\
    &= \sum_{\substack{\vec{q}, \vec{r} \\ \La^\Nal_\Ga \left( \vec{q} \right) = 1}} \left( \bigotimes_{\Oga \in \BranGa} \pi^{r^\Oga}_\Ogaout \right) \circ \cv_0^\dagger[f_\pad] \circ \zepadout(\Nal) \circ \left( \bigotimes_{\Oga \in \BranGa} \pi^{q^\Oga}_\Ogaout \right) \\
    &\overset{\textrm (\ref{eq: comp H6 base 3})}{=} \sum_{\substack{\vec{q}, \vec{r} \\ \La^\Nal_\Ga \left( \vec{q} \right) = 1}} F^{(r^\Oga)_{\Oga \in \cpst(\Nal)}, (q^\Oga)_{\Oga \in \cpst(\Nal)}}\\
    &\otimes \left(\bigotimes_{\Mbe \not\in \cpst(\Nal)} \pi^{q^\Mbe}_\Mbeout \right) \cdot \prod_{\Oga \not\in \cpst(\Nal)} \delta_{q^\Oga, r^\Oga} \\
    &=  \sum_{\substack{\left( q^\Oga \right)_{\Oga \in \cpst(\Nal)} \\ \left( r^\Oga \right)_{\Oga \in \cpst(\Nal)}}} \left( F^{(r^\Oga)_{\Oga \in \cpst(\Nal)}, (q^\Oga)_{\Oga \in \cpst(\Nal)}} \otimes \sum_{\substack{\left( q^\Oga \right)_{\Oga \not\in \cpst(\Nal)} \\ \La^\Nal_\Ga \left( \vec{q} \right) = 1}} \left(\bigotimes_{\Mbe \not\in \cpst(\Nal)} \pi^{q^\Mbe}_\Mbeout \right) \right) \\
    &= \left( \sum_{\substack{\left( q^\Oga \right)_{\Oga \in \cpst(\Nal)} \\  \left( r^\Oga \right)_{\Oga \in \cpst(\Nal)}}}  F^{(r^\Oga)_{\Oga \in \cpst(\Nal)}, (q^\Oga)_{\Oga \in \cpst(\Nal)}} \right) \circ \left( \sum_{\substack{\left( s^\Oga \right)_{\Oga \in \BranGa} \\ \La^\Nal_\Ga \left( \vec{s} \right) = 1}} \bigotimes_{\Mbe \not\in \cpst(\Nal)} \pi^{s^\Mbe}_\Mbeout \right)_\pad  \\
    &= \left( \sum_{\substack{\left( q^\Oga \right)_{\Oga \in \cpst(\Nal)} \\  \left( r^\Oga \right)_{\Oga \in \cpst(\Nal)}}}  F^{(r^\Oga)_{\Oga \in \cpst(\Nal)}, (q^\Oga)_{\Oga \in \cpst(\Nal)}} \right) \circ \zepadout(\Nal) \, ;
\end{split}
\ee

denoting the left-hand factor as $f'$ yields (\ref{eq:(H6)}).

\paragraph{H7} The proof of (H7) is the symmetric of that of (H6).

\subsubsection{Proof of the induction step}

We suppose the induction hypotheses are all satisfied up to step $i$. We write $\Mbe := B(i+1)$, the branch we have to refill in this induction step.

We first consider the case: neither $i$ nor $i+1$ are special steps. Note that, because the branches have been ordered so that all branches of a same layer are next to each other, the fact that $i+1$ is not a special step entails: $\Mbe$ is in a red layer $\implies \cpi(\Mbe) = \emptyset$.

\paragraph*{H1}

Let us first prove H1 at step $i+1$. From (\ref{eq:zetain to zetaout}) and (H5) applied to $Q = \{\Mbe\}$, we have

\be \begin{split}
    \Spad[(V_{N,i})_N] &= \, \zeipadin(\Mbe) \circ \Spad[(V_{N,i})_{N \neq M} \times (V_{M,i}^\bet)] \circ \zeipadout(\Mbe) \\
    &+ \, \barzeipadin(\Mbe) \circ \Spad[(V_{N,i})_{N \neq M} \times (\Bar{V}_{M,i}^\bet)] \circ \barzeipadout(\Mbe) \,.    \end{split}  \ee

Furthermore,

\begin{subequations}
\be \Spad[(V_{N,i})_{N \neq M} \times (V_{M,i+1}^\bet)] = \Tr_{\Mbeout} [U_\pad^\bet \circ \Spad[(V_{N,i})_{N \neq M} \times (V_{M,i}^\bet)] ] \,, \ee
\be \Spad[(V_{N,i})_{N \neq M} \times (\bar{V}_{m,i+1}^\bet)] = \frac{1}{\dim (\Mbe)} \Tr_{\Mbeout} [\Theta_{\Mbe,\pad}^\dagger \circ \Spad[(V_{N,i})_{N \neq M} \times (\bar{V}_{M,i}^\bet)]] \,,  \ee
\end{subequations}

so, because $\zeiin(\Mbe)$ doesn't act on $\Mbein$, and $\zeiout(\Mbe)$ doesn't act on $\Mbeout$,

\be \label{eq: step for H1} \begin{split}
    \Spad[(V_{N,i+1})_N] &= \, \Spad[(V_{N,i})_{N \neq M} \times (V_{M,i+1}^\bet)] + \Spad[(V_{N,i})_{N \neq M} \times (\bar{V}_{M,i+1}^\bet)] \\
    &= \, \zeipadin(\Mbe) \circ \Tr_{\Mbeout} [U_\pad^\bet \circ \Spad[(V_{N,i})_{N \neq M} \times (V_{M,i}^\bet)] ] \circ \zeipadout(\Mbe) \\
    &+ \, \barzeipadin(\Mbe) \circ \frac{1}{\dim (\Mbe)} \Tr_{\Mbeout} [\Theta_{\Mbe,\pad}^\dagger \circ \Spad[(V_{N,i})_{N \neq M} \times (\bar{V}_{M,i}^\bet)]] \circ \barzeipadout(\Mbe) \,.    \end{split}  \ee

Therefore, $\Spad[(V_{N,i+1})_N]$ can be decomposed into two terms: one that can be considered as a linear map from the subspace of $\ch_i^\out$ defined by the projector $\zeipadout(\Mbe)$, to the subspace of $\ch_i^\inn$ defined by the projector $\zeipadin(\Mbe)$; and one that can be considered a linear map from, and to, the subspaces orthogonal to these. We now have to prove that each of these two terms is unitary.

We start with the first term. (H6) implies that $U_\pad^\bet \circ \Spad[(V_{N,i})_{N \neq M} \times (V_{M,i}^\bet)]$ features no causal influence from $\Mbeout$ to $\Mbein$, via the characterisation of causal influence in terms of algebras \cite{Ormrod2022} (note that it makes sense to talk about the factors $\Mbeout$ and $\Mbein$ of its input and output spaces because the $\ze(\Mbe)$'s do not act on these). Therefore, one can find a unitary causal decomposition of it as $W^2 \circ (\swap_{\Mbeout, \Mbein} \otimes \id) \circ W^1$, where $W^1$ doesn't act on $\Mbeout$ and $W^2$ doesn't act on $\Mbein$. The first term in (\ref{eq: step for H1}) -- with its input and output spaces suitably restricted -- is thus of the form $W^2 \circ (U^\bet \otimes \id) \circ W^1$, which is unitary.

As for the second term, one can see from the definition of the $V_{M,i}^\al$'s that $\Spad[(V_{N,i})_{N \neq M} \times (\bar{V}_{M,i}^\bet)]$ is of the form $\Theta_{\Mbe} \otimes W$, with $W$ a unitary (once restricted to the suitable subspaces). Therefore, the term can simply be rewritten as $W$.

We have therefore proven (H1) at rank $i+1$.

\paragraph*{A Lemma.}

Before turning to the other induction hypotheses, we prove a Lemma that we will need to use a few times to compute how $\cv_{i+1}^\dagger$ acts on sufficiently well-behaved linear operators.

\begin{lemma} \label{lem: MovingThrough}
Let $g \in \Lin[\ch_i^\inn]$, not acting (i.e.\ acting trivially) on $\Mbein$, commuting with $\zeipadin(\Mbe)$, and satisfying: $\cv_i^\dagger[g] \circ \barzeipadout(\Mbe)$ doesn't act on $\Mbeout$. We fix an orthonormal (with respect to the Hilbert-Schmidt inner product) basis $(E_m)_{1 \leq m \leq \dim(\Mbeout)^2}$ of $\Lin[\ch_\Mbeout]$, with $E_0 = \id$, and decompose $\cv_i^\dagger[g]$ as

\be \label{eq: g decomposition} \cv_i^\dagger[g] = \sum_m \chi_m \otimes E_m \, ,\ee

with the $\chi_m$'s acting on $\ch_P \otimes \left( \bigotimes_{\Oga > B(i+1)} \ch_\Ogaout  \right)$. With padding, we can also write $\cv_i^\dagger[g] = \sum_m \chi_{m, \pad} \circ E_{m, \pad}$, with the terms commuting. We then have

\be \label{eq: Lemma's result} \cv_{i+1}^\dagger[g] = \chi_{0, \pad} + \sum_{m \neq 0} E'_{m, \pad} \circ \chi_{m, \pad} \, , \ee

where the $E'_m$'s are defined, through the use of (H6) at step $i$, by $\cv_{i}^\dagger[({U^{\bet}}^\dagger E_m U ^\bet)_\pad] \circ \zeipadout(\Mbe) = E'_{m,\pad} \circ \zeipadout(\Mbe)$, with the $E'_m$'s only acting on $\ch_P$ (because $\cpist(\Mbe) = \emptyset$).
\end{lemma}

\begin{proof}
We will compute $\cv_{i+1}^\dagger[g]$ by looking at how $g$ `moves through' $\Spad[(V_{N,i+1})_N]$. First, we rewrite (\ref{eq: step for H1}) more compactly as

\be \label{eq: compact rewriting} \begin{split}
    &\Spad[(V_{N,i+1})_N] \\
    &= \, \Tr_{\Mbeout} \left[ \left(\zeipadin(\Mbe) \otimes U_\pad^\bet + \barzeipadin(\Mbe) \otimes \frac{1}{\dim (\Mbe)} \Theta_{\Mbe,\pad}^\dagger \right)     \circ \Spad[(V_{N,i})_{N}] \right] \, .
    \end{split}  \ee

Thus (because $g$ doesn't act on $\Mbein$, and commutes with $\zeipadin(\Mbe)$),

\be \label{eq: computation for Lemma} \begin{split}
    & g \circ \Spad[(V_{N,i+1})_N] \\
    &= \, g \circ \Tr_{\Mbeout} \left[ \left( \zeipadin(\Mbe) \circ  U_\pad^\bet + \barzeipadin(\Mbe) \circ  \frac{1}{\dim (\Mbe)} \Theta_{\Mbe,\pad}^\dagger \right)     \circ \Spad[(V_{N,i})_{N}] \right] \\
    &= \, \Tr_{\Mbeout} \left[ \left(\zeipadin(\Mbe) \circ  U_\pad^\bet + \barzeipadin(\Mbe) \circ  \frac{1}{\dim (\Mbe)} \Theta_{\Mbe,\pad}^\dagger \right)     \circ g \circ \Spad[(V_{N,i})_{N}] \right] \\
    &= \, \Tr_{\Mbeout} \left[ \left(\zeipadin(\Mbe) \circ  U_\pad^\bet + \barzeipadin(\Mbe) \circ  \frac{1}{\dim (\Mbe)} \Theta_{\Mbe,\pad}^\dagger \right)      \circ \Spad[(V_{N,i})_{N}] \circ \cv_i^\dagger[g] \right] \, .
    \end{split}  \ee

We now consider the decomposition (\ref{eq: g decomposition}) of $\cv_i^\dagger[g]$, and we look at

\be \begin{split}
    \sum_{m \neq 0} (\chi_{m,\pad} \circ \barzeipadout(\Mbe)) \otimes E_m  &= \sum_{m \neq 0} (\chi_{m} \otimes E_m)_\pad \circ \barzeipadout(\Mbe) \\ &= \cv_i^\dagger[g] \circ \barzeipadout(\Mbe) - (\chi_0 \otimes \id_\Mbeout)_\pad \circ \barzeipadout(\Mbe) \, .
\end{split}\ee

Both terms of the second line's RHS act trivially on $\Mbeout$: the first term by assumption, and the second because it is a composition of operators acting trivially on $\Mbeout$. From the form of the LHS, we can thus deduce: $\forall m \neq 0, \chi_{m,\pad} \circ \barzeipadout(\Mbe) = 0$, which can be rewritten as

\be \label{eq: absorbing zeta} \forall m \neq 0, \chi_{m,\pad} = \chi_{m,\pad} \circ \zeipadout(\Mbe) \, . \ee

In the same way we can prove that $\chi_{m,\pad} = \zeipadout(\Mbe) \circ \chi_{m,\pad}$. We are now in a position to continue the computation started in (\ref{eq: computation for Lemma}); we write $\kappa_i^\inn(\Mbe):= \zeipadin(\Mbe) \otimes U_\pad^\bet + \barzeipadin(\Mbe) \otimes \frac{1}{\dim (\Mbe)} \Theta_{\Mbe,\pad}^\dagger$.

\be \begin{split}
    & g \circ \Spad[(V_{N,i+1})_N] \\
    &\overset{(\ref{eq: absorbing zeta})}{=} \, \Tr_{\Mbeout} \left[ \left(\zeipadin(\Mbe) \circ  U_\pad^\bet + \barzeipadin(\Mbe) \circ  \frac{1}{\dim (\Mbe)} \Theta_{\Mbe,\pad}^\dagger \right)      \circ \Spad[(V_{N,i})_{N}] \circ \chi_{0, \pad} \right] \\
    &+ \, \sum_{m \neq 0} \Tr_{\Mbeout} \left[ \zeipadin(\Mbe) \circ  U_\pad^\bet      \circ \Spad[(V_{N,i})_{N}] \circ \chi_{m, \pad} \circ E_{m, \pad} \right] \\
    &= \, \Spad[(V_{N,i+1})_N] \circ \chi_{0, \pad} \\
    &+ \, \sum_{m \neq 0} \Tr_{\Mbeout} \left[ E_{m, \pad} \circ \zeipadin(\Mbe) \circ  U_\pad^\bet      \circ \Spad[(V_{N,i})_{N}] \circ \chi_{m, \pad} \right] \\
    &= [\ldots] +  \sum_{m \neq 0} \Tr_{\Mbeout} \left[ \zeipadin(\Mbe) \circ U_\pad^\bet      \circ \Spad[(V_{N,i})_{N}] \circ \cv_{i}^\dagger[({U^{\bet}}^\dagger E_m U ^\bet)_\pad] \right] \circ \chi_{m, \pad}\\
    &= [\ldots] +  \sum_{m \neq 0} \Tr_{\Mbeout} \left[ \kappa_{i,\pad}^\inn(\Mbe) \circ \zeipadin(\Mbe)  \circ \Spad[(V_{N,i})_{N}] \circ \cv_{i}^\dagger[({U^{\bet}}^\dagger E_m U ^\bet)_\pad] \right] \circ \chi_{m, \pad}\\
    &= [\ldots] +  \sum_{m \neq 0} \Tr_{\Mbeout} \left[ \kappa_{i,\pad}^\inn(\Mbe)  \circ \Spad[(V_{N,i})_{N}] \circ \cv_{i}^\dagger[({U^{\bet}}^\dagger E_m U ^\bet)_\pad] \circ \zeipadout(\Mbe) \right] \circ \chi_{m, \pad}\\
    &\overset{\textrm{(H6)}}{=} [\ldots] +  \sum_{m \neq 0} \Tr_{\Mbeout} \left[ \kappa_{i,\pad}^\inn(\Mbe)  \circ \Spad[(V_{N,i})_{N}] \circ E'_{m, \pad} \circ \zeipadout(\Mbe) \right] \circ \chi_{m, \pad}\\
    &= \, [\ldots] + \sum_{m \neq 0} \Tr_{\Mbeout} \left[ \kappa_{i,\pad}^\inn(\Mbe)  \circ \Spad[(V_{N,i})_{N}] \right] \circ E'_{m, \pad} \circ \zeipadout(\Mbe) \circ \chi_{m, \pad}\\
    &\overset{(\ref{eq: compact rewriting})}{=} \, [\ldots] +  \sum_{m \neq 0} \Spad[(V_{N,i+1})_N]  \circ E'_{m, \pad}  \circ \zeipadout(\Mbe) \circ  \chi_{m, \pad}\\
    &\overset{(\ref{eq: absorbing zeta})}{=} \, [\ldots] +  \sum_{m \neq 0} \Spad[(V_{N,i+1})_N]  \circ E'_{m, \pad}  \circ \chi_{m, \pad}\\
    &= \, \Spad[(V_{N,i+1})_N] \circ  \left( \chi_{0, \pad} + \sum_{m \neq 0} E'_{m, \pad} \circ \chi_{m, \pad} \right) \,.
    \end{split}\ee
    
In the previous computation, we used (H6) to replace  $\cv_{i}^\dagger[({U^{\bet}}^\dagger E_m U ^\bet)_\pad] \circ \zeipadout(\Mbe)$ with $E'_{m,\pad} \circ \zeipadout(\Mbe)$, with the $E'_m$'s only acting on $\ch_P$ (because $\cpist(\Mbe) = \emptyset$). This then allowed to get the term out of the trace. The computation allows us to conclude that (\ref{eq: Lemma's result}) holds.

\end{proof}

\paragraph*{H2}

We now turn to (H2). We will take: $\forall \Nal > B(i+1), \zeiplin(\Nal):= \zeiin(\Nal)$. There are two things to check in order to ensure that this makes sense. The first is that the $\zeiin(\Nal)$ are indeed all defined, which holds here because $i$ is not a special step. The second thing to check is that for an arbitrary $\Nal$, $\zeiin(\Nal)$ wasn't acting on $\Mbein$. This follows from the fact that $i+1$ is not a special step. Indeed, the way we defined the ordering of the branches ensures that $\Mbe \not\in \cfist(\Nal)$. This ensures that $\zeiin(\Nal)$ doesn't act on $\Mbein$ if $\Nal$ is in a green layer; while if $\Nal$ is in a red layer, then the fact that $i+1$ is not a special step implies that $\Mbe$ is not in this red layer, i.e. that $\Mbe \not\in \cfi(\Nal)$ and thus that $\zeiin(\Nal)$ doesn't act on $\Mbein$.

We then want to define, from there, $\zeiplpadout(\Nal) := \cv_{i+1}^\dagger[\zeiplpadin(\Nal)], \, \forall \Nal$. The fact (which derives from (H1)) that $\cv_{i+1}$ is an isomorphism of operator algebras will then ensure that the $\zeiplpadout$'s are pairwise commuting orthogonal projectors, as the $\zeiplpadin$'s are. What is left to prove is that, fixing an $\Nal$ whose layer is green, $\cv_{i+1}^\dagger[\zeiplpadin(\Nal)]$ (which formally acts on the whole $\ch_{i+1}^\out$) can  indeed be seen as the padding of an operator acting on $\ch_P \otimes \left( \bigotimes_{\Oga \in \cp_{i+1}(\Nal)} \ch_\Ogaout  \right)$ -- i.e., that it acts trivially on other factors; and similarly, that for an $\Nal$ whose layer is red, $\cv_{i+1}^\dagger[\zeiplpadin(\Nal)]$ can be seen as only acting on $\ch_P \otimes \left( \bigotimes_{\Oga \in \cp_{i+1}^\str(\Nal)} \ch_\Ogaout  \right)$.

For this, fixing an $\Nal > B(i+1)$, we will make use of Lemma \ref{lem: MovingThrough} to compute $\cv_{i+1}^\dagger[\zeipadin(\Nal)]$. The latter satisfies the lemma's assumptions: $\zeipadin(\Nal)$ doesn't act on $\Mbein$ and commutes with $\zeipadin(\Mbe)$ by (H1) at step $i$, and $\cv_i^\dagger[\zeipadin(\Nal)] = \zeipadout(\Nal)$, by (H3) at step $i$, satisfies: $\zeipadout(\Nal) \circ \barzeipadout(\Mbe)$ acts trivially on $\Mbeout$.

By Lemma \ref{lem: MovingThrough}, writing $\zeiout(\Nal) = \sum_m \chi_m \otimes E_m $, we can thus conclude

\be \label{eq: result for (H2)} \cv_{i+1}^\dagger[\zeiplpadin(\Nal)] = \chi_{0, \pad} + \sum_{m \neq 0} E'_{m, \pad} \circ \chi_{m, \pad} \, . \ee


If $\Nal$ is in a green layer, then in this expression, the $\chi_{m}$'s act on $\ch_P \otimes \left( \bigotimes_{\Oga \in \cp_{i+1}(\Nal)} \ch_\Ogaout  \right)$ and the $E'_m$'s act on $\ch_P$; thus, $\cv_{i+1}^\dagger [\zeiplpadin(\Nal)]$ only acts non-trivially on $\ch_P \otimes \left( \bigotimes_{\Oga \in \cp_{i+1}(\Nal)} \ch_\Ogaout  \right)$. If $\Nal$ is in a red layer, the same can be said replacing $\cp$'s with $\cp^\str$'s. This concludes the proof of (H2).

\paragraph*{H3}

The proof of (H3) is direct for the $\zeiplin$'s, as they are equal to the $\zeiin$'s. For the $\zeiplout$'s, fixing $\Nal$ and $\Oga$, one can compute $\zeiplpadout(\Nal) \circ \barzeiplpadout(\Oga) = \cv_{i+1}^\dagger[\zeipadin(\Nal) \circ \barzeipadin(\Oga)]$ by once again invoking Lemma \ref{lem: MovingThrough}, writing

\be \zeiout(\Nal) \circ \barzeiout(\Mbe) = \sum_m \xi_m \otimes E_m \, ,\ee

where the $\xi_m$'s act trivially on $\Ogaout$ because $\zeiout(\Nal) \circ \barzeiout(\Mbe)$ does, by (H3) at step $i$. $\zeipadin(\Nal) \circ \barzeipadin(\Oga)$ commutes with $\zeipadin(\Mbe)$ and doesn't act on $\Mbein$; to apply the Lemma, we thus have to prove that $\zeipadout(\Nal) \circ \barzeipadout(\Oga) \circ \barzeipadout(\Mbe)$ acts trivially on $\Mbeout$. This follows from the rewriting

\be \begin{split}
    &\zeipadout(\Nal) \circ \barzeipadout(\Oga) \circ \barzeipadout(\Mbe) \\ 
    &= \left( \zeipadin(\Nal) \circ \barzeipadin(\Mbe) \right)
    \circ \left( \barzeipadout(\Mbe) -  \zeipadout(\Oga) \circ \barzeipadout(\Mbe) \right) \, ,
\end{split}  \ee 

all of the terms in which, one can conclude by (H3) at step $i$, act trivially on $\Mbeout$. Lemma \ref{lem: MovingThrough} thus leads to

\be  \zeiplpadout(\Nal) \circ \barzeiplpadout(\Oga) = \cv_{i+1}^\dagger[\zeipadin(\Nal) \circ \barzeipadin(\Oga)] = \xi_{0, \pad} + \sum_{m \neq 0} E'_{m, \pad} \circ \xi_{m, \pad} \, , \ee

with neither the $\xi_m$'s nor the $E'_m$'s acting on $\Ogaout$, which concludes the proof of (H3).

\paragraph*{H4}

The proof of (H4) at step $i+1$ is immediate as it derives from (H4) at step $i$ for the $\zeiplin$'s, which are equal to the $\zeiin$'s.

\paragraph*{H5}

For the proof of (H5), we take $Q \subseteq \{B(i') \, | \, i' \geq i +1 \}$ a set of branches on different nodes, and $\Tilde{Q} \subseteq \Nodes$ and a function $\al$ such that $Q = \{ N^{\al(N)} \, | \, N \in \Tilde{Q}\}$. We will prove the version of (H5) written with $\zeiplin$'s. We first consider the case $M \in \Tilde{Q}$. Then by (H4) we have $\zeipadin(M^{\al(M)}) = \zeipadin(M^{\al(M)}) \circ \barzeipadin(\Mbe)$, and we can therefore write

\be \begin{split}
    &\prod_{N \in \Tilde{Q}} \zeiplpadin(N^{\al(N)}) \circ \Spad[(V_{N,i+1})_N] \\
    &= \, \prod_{N \in \Tilde{Q}} \zeipadin(N^{\al(N)}) \circ \zeipadin(\Mbe) \circ \Spad[(V_{N,i+1})_N] \\
    &= \, \prod_{N \in \Tilde{Q}} \zeipadin(N^{\al(N)}) \circ \zeipadin(\Mbe) \circ \Spad[(V_{N,i+1})_N] \\
    &= \, \prod_{N \in \Tilde{Q}} \zeipadin(N^{\al(N)}) \circ \zeipadin(\Mbe) \circ \frac{1}{\dim(\Mbe))} \Tr_\Mbeout \left[ \Theta_{\Mbe,\pad}^\dagger \circ \Spad[(V_{N,i})_N] \right] \\
    &= \, \prod_{N \in \Tilde{Q}} \zeipadin(N^{\al(N)}) \circ \frac{1}{\dim(\Mbe))} \Tr_\Mbeout \left[ \Theta_{\Mbe,\pad}^\dagger \circ \Spad[(V_{N,i})_N] \right] \\
    &\overset{\textrm{(H5)}}{=} \, \frac{1}{\dim(\Mbe))} \Tr_\Mbeout \left[ \Theta_{\Mbe,\pad}^\dagger \circ \Spad[(V_{N,i})_{N \not\in \Tilde{Q}} \times (V^{\al(N)}_{N,i})_{N \in \Tilde{Q}}] \right] \\
    &= \, \frac{1}{\dim(\Mbe))} \Tr_\Mbeout \left[ \Theta_{\Mbe,\pad}^\dagger \circ \Theta_{\Mbe,\pad} \circ \Spad[(V_{N,i+1})_{N \not\in \Tilde{Q}} \times (V^{\al(N)}_{N,i+1})_{N \in \Tilde{Q}}] \right] \\
    &= \,  \Spad[(V_{N,i+1})_{N \not\in \Tilde{Q}} \times (V^{\al(N)}_{N,i+1})_{N \in \Tilde{Q}}] \, .
\end{split} \ee

In the case $M \not\in \Tilde{Q}$, then defining $\al(M) := \bet$, we have

\be \begin{split}
    &\prod_{N \in \Tilde{Q}} \zeiplpadin(N^{\al(N)}) \circ \Spad[(V_{N,i+1})_N] \\
    &= \, \Tr_{\Mbeout} \left[ \prod_{N \in \Tilde{Q} \cup \{M\}} \zeipadin(N^{\al(N)}) \circ U_\pad^\bet \circ \Spad[(V_{N,i})_{N}] \right] \\ 
    &+ \,\frac{1}{\dim (\Mbe)} \Tr_{\Mbeout} \left[ \prod_{N \in \Tilde{Q}} \zeipadin(N^{\al(N)}) \circ  \barzeipadin(\Mbe) \circ  \Theta_{\Mbe,\pad}^\dagger \circ \Spad[(V_{N,i})_{N}] \right] \\
    &\overset{\textrm{(H5)}}{=} \, \Tr_{\Mbeout} \left[ U_\pad^\bet \circ \Spad[(V_{N,i})_{N \not\in \Tilde{Q} \cup \{M\}} \times (V_{N,i}^{\al(N)})_{N \in \Tilde{Q} \cup \{M\}}] \right] \\ 
    &+ \, \frac{1}{\dim (\Mbe)} \Tr_{\Mbeout} \left[  \Theta_{\Mbe,\pad}^\dagger \circ \Spad[(V_{N,i})_{N \not\in \Tilde{Q} \cup \{M\}} \times (V_{N,i}^{\al(N)})_{N \in \Tilde{Q}} \times (\bar{V}_{M,i}^\bet)]  \right] \\
    &= \, \Spad[(V_{N,i+1})_{N \not\in \Tilde{Q} \cup \{M\}} \times (V_{N,i+1}^{\al(N)})_{N \in \Tilde{Q}} \times (V_{M,i}^\bet)]  \\ 
    &+ \, \Spad[(V_{N,i+1})_{N \not\in \Tilde{Q} \cup \{M\}} \times (V_{N,i+1}^{\al(N)})_{N \in \Tilde{Q}} \times (\bar{V}_{M,i}^\bet)]\\
    &= \, \Spad[(V_{N,i+1})_{N \not\in \Tilde{Q}} \times (V_{N,i+1}^{\al(N)})_{N \in \Tilde{Q}}] \,.
\end{split} \ee

\paragraph*{H6}


To prove (H6), we fix a branch $\Nal > B(i+1)$ in a green layer, and $f \in \Lin[\ch_\Nalin]$. We first consider the case $\Nal \not\in \cfist(\Mbe)$. $\zeiin(\Mbe)$ then doesn't act on $\Nalin$: indeed, either $\Mbe$ is in a green layer and $\zeiin(\Mbe)$ doesn't act outside of $\cfist(\Mbe)$, or $\Mbe$ is in a red layer and then, because $i+1$ is not a special step, $\Nal$ is not in this layer and thus $\Nal \not\in \cfi(\Mbe)$. Furthermore, as we saw that $\Spad[(V_{N,i})_N)] \circ \barzeipadout(\Mbe)$ was of the form $W \otimes \Theta_\Mbe$, and $f$ doesn't act on $\Mbein$, it follows that $\cv_i^\dagger[f]\circ \barzeipadout(\Mbe)$ doesn't act on $\Mbeout$. We can therefore apply Lemma \ref{lem: MovingThrough} and get

\be \cv_{i+1}^\dagger[f_\pad] = \phi_{0, \pad} + \sum_{m \neq 0} E'_{m, \pad} \circ \phi_{m, \pad} \, \overset{\textrm{(H6)}}{=} \sum_m \cv_{i}^\dagger[({U^{\bet}}^\dagger E_m U ^\bet)_\pad] \circ \phi_{m, \pad} \, , \ee

where $\cv_i^\dagger[f_\pad] = \sum_m \phi_m \otimes E_m$. Furthermore, as we've seen, we have $\zeiplpadout(\Nal) = \sum_n \cv_{i}^\dagger[({U^{\bet}}^\dagger E_n U ^\bet)_\pad] \circ \chi_n $. We are therefore led to

\be \label{eq: comp H6}
\begin{split}
    &\cv_{i+1}^\dagger[f_\pad] \circ \zeiplpadout(\Nal) = \sum_{mn} \phi_m \circ \chi_n \circ  \cv_{i}^\dagger[({U^{\bet}}^\dagger \circ E_m \circ E_n \circ U^\bet)_\pad]\\
    &= \, \left( \sum_{mn} \sigma_{lmn}  \phi_m \circ \chi_n \right) \circ \sum_l \cv_{i}^\dagger[({U^{\bet}}^\dagger \circ E_l \circ U^\bet)_\pad] \, ,
\end{split}
\ee 

where the $\sigma_{lmn}$'s are the structure constants on $\Lin[\Mbeout]$, i.e. $E_m \circ E_n = \sum_l \sigma_{lmn} E_l$. Yet (H6) at step $i$ gives us that there exists $f'$ acting on $P$ and $\cpist(\Nal)$ (and therefore not on $\Mbeout$) such that $\cv_i^\dagger[f_\pad] \circ \zeipadout(\Nal) = {f'}_\pad \circ \zeipadout(\Nal)$, which can be rewritten as

\be \sum_l \left( \sum_{mn} \sigma_{lmn}  \phi_m \circ \chi_n \right) \otimes E_l = \sum_l \left( f'_\pad \circ \chi_l \right) \otimes E_l \, , \ee

leading to 

\be \forall l, \sum_{mn} \sigma_{lmn}  \phi_m \circ \chi_n = f'_\pad \circ \chi_l \, . \ee

Reinserting this into (\ref{eq: comp H6}), we find

\be \begin{split}
    \cv_{i+1}^\dagger[f_\pad] \circ \zeiplpadout(\Nal) &= \, f'_\pad \circ \sum_l \chi_l \circ \cv_{i}^\dagger[({U^{\bet}}^\dagger \circ E_l \circ U^\bet)_\pad] \\
    &= f'_\pad \circ \zeiplpadout(\Nal) \, , 
\end{split} \ee 

where $f'$ acts on $P$ and on $\cpist(\Nal)$, and the latter is equal to $\cp_{i+1}^\str(\Nal)$ as $\Mbe \not\in \cpist(\Nal)$.

We now consider the case $\Nal \in \cfist(\Mbe)$. We will use the fact that (\ref{eq:(H6)}) can be equivalently written as $\zeipadout(\Nal) \circ \cv_i^\dagger[f_\pad] = f'_\pad \circ \zeipadout(\Nal)$; indeed, $\zeipadin(\Nal)$ doesn't act on $\Nalin$, so $f_\pad$ and $\zeipadin(\Nal)$ commute, so $\cv_i^\dagger[f]$ and $\cv_i^\dagger[\zeipadin(\Nal)] = \zeipadout(\Nal)$ commute as well. We write $\kappa_i^\out(\Mbe):= \zeipadout(\Mbe) \otimes U_\pad^\bet + \barzeipadout(\Mbe) \otimes \frac{1}{\dim (\Mbe)} \Theta_{\Mbe,\pad}^\dagger$.

\be \label{eq: H6 computation 1} \begin{split}
    &\zeiplpadout(\Nal) \circ \Spad[(V_{N,i+1})_N]^\dagger \circ f_\pad \\
    &= \, \Spad[(V_{N,i+1})_N]^\dagger \circ \zeipadin(\Nal) \circ f_\pad \\
    &= \, \Tr_\Mbein \left[ \kappa^\out_{i,\pad}(\Mbe)^\dagger \circ \Spad[(V_{N,i})_N]^\dagger \right] \circ \zeipadin(\Nal) \circ f_\pad \\
    &= \, \Tr_\Mbein \left[ \kappa_{i,\pad}^\out(\Mbe)^\dagger \circ \zeipadout(\Nal) \circ \cv_i^\dagger[f_\pad] \circ \Spad[(V_{N,i})_N]^\dagger \right] \\
    &\overset{\textrm{(H6)}}{=} \, \Tr_\Mbein \left[ \kappa_{i,\pad}^\out(\Mbe)^\dagger \circ f'_\pad \circ \zeipadout(\Nal) \circ \Spad[(V_{N,i})_N]^\dagger \right] \\
    &= \, \Tr_\Mbein \left[ \kappa_{i,\pad}^\out(\Mbe)^\dagger \circ f'_\pad \circ \Spad[(V_{N,i})_N]^\dagger \right] \circ \zeiplpadin(\Nal) \, , \\
\end{split} \ee

from which we get

\be \label{eq: H6 computation 2} \begin{split}
    &\zeiplpadout(\Nal) \circ \cv_{i+1}^\dagger[f_\pad] \\
    &= \, \zeiplpadout(\Nal) \circ \Spad[(V_{N,i+1})_N]^\dagger \circ f_\pad \circ \Spad[(V_{N,i+1})_N] \\
    &= \, \Tr_\Mbein \left[ \kappa_{i,\pad}^\out(\Mbe)^\dagger \circ f'_\pad \circ \Spad[(V_{N,i})_N]^\dagger \right] \\
    &\circ \Tr_\Mbein \left[ \Spad[(V_{N,i})_N] \circ {\kappa_{i,\pad}^\out}(\Mbe) \right] \circ \zeiplpadout(\Nal) \, .
\end{split} \ee

We will rewrite the traces in another way, defining $\ket{\Phi^+(\Mbe)} := \sum_k \ket{k}_\Mbein \otimes \ket{k}_{\Mbein'}$, where $\Mbein' \cong \Mbein$, and $(\ket{k})_k$ is an arbitrary orthonormal basis. The above can then be expressed as

\be \label{eq: H6 computation 3}\begin{split}
    &\zeiplpadout(\Nal) \circ \cv_{i+1}^\dagger[f_\pad] \\
    &= \, \bra{\Phi^+(\Mbe)}_\pad \circ \kappa_{i,\pad}^\out(\Mbe)^\dagger \circ f'_\pad \circ \Spad[(V_{N,i})_N]^\dagger \circ \ket{\Phi^+(\Mbe)}_\pad \\
    &\circ \bra{\Phi^+(\Mbe)}_\pad \circ \Spad[(V_{N,i})_N] \circ {\kappa_{i,\pad}^\out}(\Mbe) \circ \ket{\Phi^+(\Mbe)}_\pad  \circ \zeiplpadout(\Nal)\\
    &= \, \bra{\Phi^+(\Mbe)}_\pad \circ \kappa_{i,\pad}^\out(\Mbe)^\dagger \circ f'_\pad \circ \cv_{i, \pad}^\dagger \left[\ket{\Phi^+(\Mbe)}_\pad \bra{\Phi^+(\Mbe)}_\pad \right]   \\
    & \circ {\kappa_{i,\pad}^\out}(\Mbe) \circ \ket{\Phi^+(\Mbe)}_\pad \circ \zeiplpadout(\Nal) \, . \\
\end{split} \ee

In this expression, $\bra{\Phi^+(\Mbe)}$ and $\ket{\Phi^+(\Mbe)}$ act on $\Mbeout$ and $\Mbeout'$; $\kappa^\out(\Mbe)$ acts on $P$, $\Mbeout$ / $\Mbein$ (on its domain/codomain), and $\cpi(\Mbe) \subseteq \cpist(\Nal)$; $f'$ acts on $P$ and $\cpist(\Nal)$; and $\cv_{i, \pad}^\dagger \left[\ket{\Phi^+(\Mbe)}_\pad \bra{\Phi^+(\Mbe)}_\pad \right]$ acts on $P$, $\Mbeout$, $\Mbein'$ and $\cpi(\Mbe) \subseteq \cpist(\Nal)$. Their composition -- which doesn't act on $\Mbeout$ and $\Mbeout'$ as these are explicitly terminated by $\bra{\Phi^+(\Mbe)}$ and $\ket{\Phi^+(\Mbe)}$ -- thus acts trivially outside of $P$ and $\cpist(\Nal) \setminus \{ \Mbe \} = \cp_{i+1}^\str(\Nal)$. Therefore, we can write

\be \zeiplpadout(\Nal) \circ \cv_{i+1}^\dagger[f_\pad] = f''_\pad \circ \zeiplpadout(\Nal) \, , \ee

with $f'' \in \Lin \left[\ch_P \otimes (\bigotimes_{\Oga \in \cp_{i+1}^\str(\Nal)} \ch_\Ogaout) \right]$.

\paragraph*{H7}

We take $\Nal > \Mbe$ in a red layer, and $f \in \Lin[\ch_{\Nalout}]$. Because $\Nal$ is in a red layer and $i+1$ is not a special step, we have $\Nal \not\in \cpi(\Mbe)$. Thus $f$ commutes with $\kappa_i^\out(\Mbe)$, as the latter only acts non trivially on $P$, $\cpi(\Mbe)$ and $\Mbe$. Thus,

\be \label{eq: computation H7 1} \begin{split}
    &f_\pad \circ \Spad[(V_{N,i+1})_N]^\dagger \circ \zeiplpadin(\Nal) \\
    &= \, f_\pad \circ \Spad[(V_{N,i+1})_N]^\dagger \circ \zeipadin(\Nal) \\
    &= \, \Tr_\Mbein \left[ \kappa_{i,\pad}^\out(\Mbe)^\dagger \circ \Spad[(V_{N,i})_N]^\dagger \circ \cv_i^\dagger[f_\pad] \circ \zeipadin(\Nal) \right] \\
    &= \, \Tr_\Mbein \left[ \kappa_{i,\pad}^\out(\Mbe)^\dagger \circ \Spad[(V_{N,i})_N]^\dagger \circ \zeipadin(\Nal) \circ \cv_i^\dagger[f_\pad] \right] \\
    &\overset{\textrm{(H7)}}{=} \, \Tr_\Mbein \left[ \kappa_{i,\pad}^\out(\Mbe)^\dagger \circ \Spad[(V_{N,i})_N]^\dagger \circ \zeipadin(\Nal) \circ f'_\pad  \right] \\
    &= \, \Tr_\Mbein \left[ \kappa_{i,\pad}^\out(\Mbe)^\dagger \circ \Spad[(V_{N,i})_N]^\dagger \right] \circ \zeipadin(\Nal) \circ f'_\pad \\
    &= \, \Spad[(V_{N,i+1})_N]^\dagger \circ \zeipadin(\Nal) \circ f'_\pad \, ,
\end{split} \ee

where $f'$ acts on $\cfist(\Nal) = \cf_{i+1}^\str(\Nal)$. We can use this to find (noting that $\cv_{i+1}[f]$ and $\zeiplpadin(\Nal)$ commute because $f$ and $\zeiplpadout(\Nal)$ do)

\be \label{eq: computation H7 2} \zeiplpadin(\Nal) \circ \cv_{i+1}[f_\pad] = \cv_{i+1}[f_\pad] \circ \zeiplpadin(\Nal) = \zeiplpadin(\Nal) \circ f'_\pad \, .   \ee

\paragraph*{The special steps}

The previous proofs were relying on the assumption that neither $i$ nor $i+1$ were special steps. We now consider the other cases. Note first that the proof of (H1) presented earlier is valid in these cases as well.

We start with the case: $i+1$ is a special step (regardless of the status of $i$). One then only has to define, and check the properties of, the $\ze_{i+1}(\Nal)$'s for $\Nal \in \cl_{i+1}(\Mbe)$. Note that for such $\Nal$'s, the $\ze_{i}(\Nal)$'s were defined at step $i$; indeed, either $i$ is not a special step and all the $\ze_i$'s were defined, or $i$ is a special step, which entails that $\cl_{i+1}(\Mbe) \subset \cl_{i}(B(i))$, and the $\ze_i$'s were defined for elements of this latter set.

One can then follow a proof strategy that is a time-reversed version of the one presented earlier, except that only the $\Nal$'s in $\cl_{i+1}(\Mbe)$ are considered. Namely, we will define, for $\Nal \in \cl_{i+1}(\Mbe)$, $\zeiplout(\Nal) := \zeiout(\Nal)$ and $\zeiplpadin(\Nal) := \cv_i[\zeiplpadout(\Nal)]$, and the rest of the proof can be obtained by following the earlier proof, simply replacing in's with out's, looking at daggered versions of maps, etc. Indeed, the induction hypotheses are fully invariant under time-symmetry, except for one crucial thing: the fact that, when looking in the forward direction, all branches of the layers in the strict past of the branch under consideration have been refilled already, and thus have no $\ze$'s. Here, however, we are only redefining, and proving properties of, the $\ze_{i+1}(\Nal)$'s of the layer under consideration; everything thus goes as if the branches in its strict future didn't exist. Moreover, the fact that $i+1$ is a special step implies that $B(i+1)$ is in a red layer, which means that, when considering things from a time-reversed perspective, $B(i+1)$ is in a green layer and thus neither $i$ not $i+1$ are special steps. Thus, for the purposes of defining these $\ze_{i+1}(\Nal)$'s, the situation is exactly symmetric to the one considered previously.

A final case to consider is: $i$ is a special step but $i+1$ is not. In this case, $B(i)$ and $\Mbe = B(i+1)$ are in the same red layer, but $B(i+2)$ and the rest of the $\Nal > B(i+1)$ aren't. The interpretation is that we just finished filling up the branches of a red layer, a procedure during which we didn't define $\zeta$'s for the branches above it; so that we now have to redefine them.  The strategy for this case is to define the $\zeiplin(\Nal)$'s to be equal, not to the $\zeiin(\Nal)$'s -- which were not defined --, but to the $\zejin(\Nal)$, where $j$ is the latest step that was not special, i.e.\ the latest step at which these were defined. $B(j)$ is then the first branch of the red layer we finished refilling.

One can then follow a strategy similar to the previous proof, now deriving that (H2)-(H7) hold at step $i+1$ from the fact that they hold at step $j$. Let us highlight the main steps. First, from (\ref{eq: step for H1}) holding at all steps between $j$ and $i+1$, we can deduce

\be \label{eq: expression for special step} \Spad[(V_{N,i+1})_N] = \Tr_{B(t)^\inn, \, j < t \leq i+1} \left[ \Spad[(V_{N,j})_N] \circ \left( \prod_{t=j+1}^{i+1} \kappa_{j,\pad}^\out(B(t)) \right) \right] \, , \ee

where we also used the fact that, due to how we defined the $\ze^\out$'s to remain the same when filling up a red layer, we have $\forall t \in \llbracket j+1, i+1 \rrbracket, \kappa_t^\out(B(t)) = \kappa_j^\out(B(t))$.

For (H2), we fix $\Nal > B(i+1)$; note that we then have $\Nal \in \cfist(B(i+1))$, because $B(i+1)$ is the last branch in its (red) layer. As before, we have to prove that, defining $\zeiplpadin(\Nal) := \ze_{j,\pad}^\inn$ and $\zeiplpadout(\Nal) := \cv^{i+1}[\zeiplpadin(\Nal)]$, the latter doesn't act outside of $P$ and $\cp_{i+1}(\Nal)$ (or $\cp_{i+1}^\str(\Nal)$ if $\Nal$ is in a red layer). Using (\ref{eq: expression for special step}) and techniques similar to before, we are led to

\be \label{eq: computation H2 special} \begin{split}
    \cv_{i+1}^\dagger[\ze_{j,\pad}^\inn(\Nal)] &= \left( \prod_{t=j+1}^{i+1} \bra{\phi^+(B(t))}_\pad \right) \circ \left( \prod_{t=j+1}^{i+1} \kappa_{j,\pad}^\out(B(t))^\dagger \right) \\
    &\circ \cv_{j, \pad}^\dagger \left[  \prod_{t=j+1}^{i+1} \ket{\Phi^+(B(t))}_\pad \bra{\Phi^+(B(t))}_\pad  \right] \circ \ze_{j,\pad}^\out(\Nal) \\
    &\circ \left( \prod_{t=j+1}^{i+1} \kappa_{j,\pad}^\out(B(t)) \right) \circ \left( \prod_{t=j+1}^{i+1} \ket{\phi^+(B(t))}_\pad \right) \, ,
\end{split} \ee

in which $\ze_{j,\pad}^\out(\Nal)$ acts only on $P$ and $\cp_j(\Nal)$ ($\cp_j^\str(\Nal)$ if $\Nal$ is in a red layer), and the other terms act only on $\cl(B(i+1))$ and on $P$. Given that all the wires in $\cl(B(i+1))$ are explicitly terminated by the $\phi^+$'s, it follows that $\cv_{i+1}^\dagger[\ze_{j,\pad}^\inn(\Nal)]$ only acts on $P$ and on $\cp_j(\Nal) \setminus \cl(B(i+1)) = \cp_{i+1}(\Nal)$ (or $\cp_j^\str(\Nal) \setminus \cl(B(i+1)) = \cp_{i+1}^\str(\Nal)$ if $\Nal$ is in a red layer).

The proof of (H3) is fully analogous: computing $\cv_{i+1}^\dagger[\ze_{j,\pad}^\inn(\Nal) \circ \bar{\zeta}_{j,\pad}^\inn(\Mbe)]$ leads to (\ref{eq: computation H2 special}) with $\ze_{j,\pad}^\out(\Nal)$ replaced with $\ze_{j,\pad}^\out(\Nal) \circ \bar{\zeta}_{j,\pad}^\out(\Mbe)$, so that invoking (H3) at step $j$ leads to (H3) at step $i+1$. (H4), as before, is direct, and the proof of (H5) is analogous to the one for the non-special cases.

For the proof of (H6), we take $\Nal > B(i+1)$ in a green layer and $f \in \Lin[\ch_{\Nalin}]$. Then, the computation is similar to (\ref{eq: H6 computation 1}), (\ref{eq: H6 computation 2}) and (\ref{eq: H6 computation 3}), yielding

\be \label{eq: computation H6 special} \begin{split}
    &\zeiplpadout(\Nal) \circ \cv_{i+1}^\dagger[f_\pad] \\
    &= \left( \prod_{t=j+1}^{i+1} \bra{\phi^+(B(t))}_\pad \right) \circ \left( \prod_{t=j+1}^{i+1} \kappa_{j,\pad}^\out(B(t))^\dagger \right) \circ \cv_{j, \pad}^\dagger \left[  \prod_{t=j+1}^{i+1} \ket{\Phi^+(B(t))}_\pad \bra{\Phi^+(B(t))}_\pad  \right] \\ &\circ f'_\pad \circ \left( \prod_{t=j+1}^{i+1} \kappa_{j,\pad}^\out(B(t)) \right) \circ \left( \prod_{t=j+1}^{i+1} \ket{\phi^+(B(t))}_\pad \right) \circ \zeiplpadout(\Nal) \, ,
\end{split} \ee

where $f' \in \Lin[\ch_P \otimes (\bigotimes_{\Oga \in \cp_j^\str(\Nal)})]$ was defined through $\ze_{j, \pad}^\out(\Nal) \circ \cv_{j}^\dagger[f_\pad] = f'_\pad \circ \ze_{j, \pad}^\out(\Nal)$ by using (H6) at step $j$. Once again, by looking at where the operators are acting, we can conclude that this defines a $f''$ acting only on $P$ and on $\cp_j^\str(\Nal) \setminus \cl(B(i+1)) = \cp_{i+1}^\str(\Nal)$.

Finally, for (H7), one can follow computations (\ref{eq: computation H7 1}) and (\ref{eq: computation H7 2}), to get

\be \zeiplpadin(\Nal) \circ \cv_{i+1}[f_\pad] =  \ze_{j, \pad}^\inn(\Nal) \circ f'_\pad \, ,   \ee

where $f'$, obtained by the use of (H7) at step $j$, only acts on $\cf_j^\str(\Nal) = \cf_{i+1}^\str(\Nal)$. This concludes the proof of the induction step.

\paragraph*{Conclusion}

As the base case and induction step are true, the induction hypotheses are true at every step up to $n$. In particular, (H1) at step $n$ then reads:

\be \cs[(U_N)_N] \textrm{ is unitary.} \ee
As this was done for $\cs = \cs_\GalaN$ for an arbitrary valid routed graph $\GalaN$, and for an arbitrary collection of routed unitaries $U_N : \chinN \overset{\laN}{\to} \choutN$ following the $\laN$'s, we can invoke Lemma \ref{lem: simplification} and conclude that Theorem \ref{thm:Main Theorem} holds.

\qed


\end{document}